\documentclass[oneside,reqno]{amsart}%
\usepackage[T1]{fontenc}
\usepackage[utf8]{inputenc}
\usepackage{amsthm}
\usepackage{amstext}
\usepackage{amssymb}
\usepackage{esint}
\usepackage[numbers]{natbib}
\usepackage[unicode=true,pdfusetitle,
bookmarks=true,bookmarksnumbered=false,bookmarksopen=false,
breaklinks=false,pdfborder={0 0 1},backref=section,colorlinks=false]%
{hyperref}
\usepackage{mathrsfs}
\usepackage{amsthm}
\usepackage{amstext}
\usepackage{datetime}
\usepackage{mathrsfs}
\usepackage{amsthm}
\usepackage{amstext}
\usepackage{amsfonts}
\usepackage{comment}
\usepackage{amsmath}
\usepackage{graphicx}%
\providecommand{\U}[1]{\protect\rule{.1in}{.1in}}
\providecommand{\U}[1]{\protect\rule{.1in}{.1in}}
\makeatletter
\numberwithin{equation}{section}
\numberwithin{figure}{section}
\numberwithin{equation}{section}
\numberwithin{figure}{section}
\numberwithin{equation}{section}
\numberwithin{figure}{section}
\theoremstyle{plain}
\newtheorem{theorem}{Theorem}[section]
\newtheorem{corollary}[theorem]{Corollary}
\newtheorem{lemma}[theorem]{Lemma}

\newtheorem{proposition}[theorem]{Proposition}

\renewenvironment{proof}[1][Proof]{\textbf{#1.} }{\ \rule{0.5em}{0.5em}}
\theoremstyle{definition}
\newtheorem{definition}[theorem]{Definition}

\newtheorem{notation}[theorem]{Notation}

\newtheorem{example}[theorem]{Example}
\theoremstyle{remark}
\newtheorem{remark}[theorem]{Remark}

\theoremstyle{plain}

\theoremstyle{definition}
\newtheorem{assumption}[]{Assumption}

\numberwithin{equation}{section}

\excludecomment{comments}
\makeatother

\begin{document}
\author[Driver]{Bruce K. Driver}
\author[Tong]{Pun Wai Tong}
\address[Driver]{Department of Mathematics, University of California, San Diego, La
Jolla, California 92093 USA.}
\email{bdriver@ucsd.edu}
\urladdr{www.math.ucsd.edu/$\sim$bdriver/}

\address[Tong]{Department of Mathematics, University of California, San Diego, La
Jolla, California 92093 USA. }
\email{p1tong@ucsd.edu}
\urladdr{www.math.ucsd.edu/$\sim$p1tong/}

\thanks{The research of Bruce K. Driver was supported in part by NSF Grant
DMS-0739164. The author also greatfully acknowledges the generosity and
hospitality of the Imperial College Mathematics department where the author
was a visiting Nelder fellow in the Fall of 2014.}
\thanks{The research of Pun Wai Tong was supported in part by NSF Grant DMS-0739164}

\subjclass{Primary  81Q20 , 81S05 ; Secondary 47D08, 47A63}
\keywords{Quantum Mechanics, Classical limit, Hepp's method}

\title{On the classical limit of quantum mechanics I.}
\date{November 7, 2015}
\subjclass{Primary  81Q20 , 81S05 ; Secondary 47D08, 47A63}
\keywords{Quantum Mechanics, Classical limit, Hepp's method}

\begin{abstract}
This paper is devoted to the study of the classical limit of quantum
mechanics. In more detail we will elaborate on a method introduced by Hepp in
1974 for studying the asymptotic behavior of quantum expectations in the limit
as Plank's constant ($\hbar)$ tends to zero. Our goal is to allow for
unbounded observables which are (non-commutative) polynomial functions of the
position and momentum operators. This is in contrast to Hepp's original paper
where the observables were, roughly speaking, required to be bounded functions
of the position and momentum operators. As expected the leading order
contributions of the quantum expectations come from evaluating the observables
along the classical trajectories while the next order contributions are
computed by evolving the $\hbar=1$ observables by a linear canonical
transformations which is determined by the second order pieces of the
quantum mechanical Hamiltonian.

\end{abstract}
\maketitle
\tableofcontents

\section{Introduction\label{sec.1}}

In the limit where Plank's constant $\left(  \hbar\right)  $ tends to zero,
quantum mechanics is supposed to reduce to the laws of classical mechanics and
their connection was first shown by P. Ehrenfest in \citep{Ehrenfest1927}.
There is in fact a very large literature devoted in one way or another to this
theme. Although it is not our intent nor within our ability to review this
large literature here, nevertheless the interested reader can find more
information by searching for terms like, correspondence principle, WKB
approximation, pseudo-differential operators, micro-local analysis, Moyal
brackets, star products, deformation quantization, Gaussian wave packet and
stationary phase approximation in the context of Feynmann path integrals to
name a few. Also,
\citep{Brian2013,Folland1989,Knowles2009,Littlejohn1986,Zworski2012,Heller1975}
may introduce readers a broad background on the subject of semi-classical
limit in one aspect or another. In this paper we wish to concentrate on a
formulation and a method to understand the classical limit of quantum
mechanics which was introduced by Hepp \citep{Hepp1974} in 1974.

This paper is an elaboration on Hepp's method to allow for unbounded
observables which was motivated by Rodnianski and Schlein's
\citep{Rodnianski2009} treatment of the mean field dynamics associated to Bose
Einstein condensation. In fact, some of the ideas in
\citep{Grillakis2011,Hagedorn1980,Hagedorn1981,Hagedorn1981III,Hagedorn1985,Knowles2010,Ammari2008, Rodnianski2009}
and \citep{Chiara2015} already appeared in Hepp's \citep{Hepp1974} paper. In
order to emphasize the main ideas and to not be needlessly encumbered by more
complicated notation we will restrict our attention to systems with only one
degree of freedom. Before summarizing the main results of this paper, we first
need to introduce some notation. [See section \ref{sec.2} below for more
details on the basic setup-used in this paper.]

\subsection{Basic Setup\label{sec.1.1}}

Let $\alpha_{0}=\left(  \xi+i\pi\right)  /\sqrt{2}\in\mathbb{C~}%
$($\mathbb{C\cong~}T^{\ast}\mathbb{R}$ is to be thought of as phase space)$,$
$H\left(  \theta,\theta^{\ast}\right)  $ be a symmetric [see Notation
\ref{not.2.12}] non-commutative polynomial in two indeterminates, $\left\{
\theta,\theta^{\ast}\right\}  ,$ $H^{\mathrm{cl}}\left(  z\right)  :=H\left(
z,\bar{z}\right)  $ for all $z\in\mathbb{C}$ be the \textbf{symbol} of $H.$
[By Remark \ref{rem.2.19} below, we know $H^{\mathrm{cl}}$ is real valued.] A
differentiable function, $\alpha\left(  t\right)  \in\mathbb{C},$ is said to
satisfy Hamilton's equations of motion with an initial condition $\alpha
_{0}\in\mathbb{C}$ if
\begin{equation}
i\dot{\alpha}\left(  t\right)  =\left(  \frac{\partial}{\partial\bar{\alpha}%
}H^{\mathrm{cl}}\right)  \left(  \alpha\left(  t\right)  \right)  \text{ and
}\alpha\left(  0\right)  =\alpha_{0}. \label{equ.1.1}%
\end{equation}
[See Section \ref{sub.2.1} where we recall that Eq. (\ref{equ.1.1}) is
equivalent to the standard real form of Hamilton's equations of motion.]
Further, let $\Phi\left(  t,\alpha_{0}\right)  =\alpha\left(  t\right)  $
(where $\alpha\left(  t\right)  $ is the solution to Eq. (\ref{equ.1.1}) ) be
the flow associated to Eq. (\ref{equ.1.1}) and $\Phi^{\prime}\left(
t,\alpha_{0}\right)  :\mathbb{C\rightarrow C}$ be the real-linear differential
of this flow relative to its starting point, i.e. for all $z\in\mathbb{C}$
let
\begin{equation}
\Phi^{\prime}\left(  t,\alpha_{0}\right)  z:=\frac{d}{ds}|_{s=0}\Phi\left(
t,\alpha_{0}+sz\right)  . \label{equ.1.2}%
\end{equation}
As $z\rightarrow\Phi^{\prime}\left(  t,\alpha_{0}\right)  z$ is a real-linear
function of $z,$ for each $\alpha_{0}\in\mathbb{C}$ there exists unique
complex valued functions $\gamma\left(  t\right)  $ and $\delta\left(
t\right)  $ such that
\begin{equation}
\Phi^{\prime}\left(  t,\alpha_{0}\right)  z=\gamma\left(  t\right)
z+\delta\left(  t\right)  \bar{z}. \label{equ.1.3}%
\end{equation}
where $\gamma\left(  0\right)  =1$ and $\delta\left(  0\right)  =0.$

We now turn to the quantum mechanical setup. Let $L^{2}\left(  m\right)  $
$:=L^{2}\left(  \mathbb{R},m\right)  $ be the Hilbert space of square
integrable complex valued functions on $\mathbb{R}$ relative to Lebesgue
measure, $m.$ The inner product on $L^{2}\left(  m\right)  $ is taken to be
\begin{equation}
\left\langle f,g\right\rangle :=\int_{\mathbb{R}}f\left(  x\right)  \bar
{g}\left(  x\right)  dm\left(  x\right)  ~\forall~f,g\in L^{2}\left(
m\right)  \label{equ.1.4}%
\end{equation}
and the corresponding norm is $\left\Vert f\right\Vert =\left\Vert
f\right\Vert _{2}=\sqrt{\left\langle f,f\right\rangle }.$ [Note that we are
using the mathematics convention that $\left\langle f,g\right\rangle $ is
linear in the first variable and conjugate linear in the second.] We say $A$
is an operator on $L^{2}\left(  m\right)  $ if $A$ is a linear (possibly
unbounded) operator from a dense subspace, $D\left(  A\right)  ,$ to
$L^{2}\left(  m\right)  .$ As usual if $A$ is closable, then its adjoint,
$A^{\ast},$ also has a dense domain and $A^{\ast\ast}=$ $\bar{A}$ where
$\bar{A}$ is the closure of $A.$

\begin{notation}
\label{not.1.1}As is customary, let $\mathcal{S}:=\mathcal{S}\left(
\mathbb{R}\right)  \subset L^{2}\left(  m\right)  $ denote Schwartz space of
smooth rapidly decreasing complex valued functions on $\mathbb{R}.$
\end{notation}

\begin{definition}
[Formal Adjoint]\label{def.1.2}If $A$ is a closable operator on $L^{2}\left(
m\right)  $ such that $D\left(  A\right)  =\mathcal{S}$ and $\mathcal{S}%
\subset D\left(  A^{\ast}\right)  ,$ then we define the \textbf{formal adjoint
}of $A$ to be the operator, $A^{\dag}:=A^{\ast}|_{\mathcal{S}}.$ Thus
$A^{\dag}$ is the unique operator with $D\left(  A^{\dag}\right)
=\mathcal{S}$ such that $\left\langle Af,g\right\rangle =\left\langle
f,A^{\dag}g\right\rangle $ for all $f,g\in\mathcal{S}.$
\end{definition}

\begin{definition}
[Annihilation and Creation operators]\label{def.1.3}For $\hbar>0,$ let
$a_{\hbar}$ be the \textbf{annihilation operator} acting on $L^{2}\left(
m\right)  $ defined so that $D\left(  a_{\hbar}\right)  =\mathcal{S}$ and
\begin{equation}
\left(  a_{\hbar}f\right)  \left(  x\right)  :=\sqrt{\frac{\hbar}{2}}\left(
xf\left(  x\right)  +\partial_{x}f\left(  x\right)  \right)  \text{ for }%
f\in\mathcal{S}. \label{equ.1.5}%
\end{equation}
The corresponding \textbf{creation operator }is $a_{\hbar}^{\dag}$ -- the
formal adjoint of $a_{\hbar},$ i.e.
\begin{equation}
\left(  a_{\hbar}^{\dag}f\right)  \left(  x\right)  :=\sqrt{\frac{\hbar}{2}%
}\left(  xf\left(  x\right)  -\partial_{x}f\left(  x\right)  \right)  \text{
for }f\in\mathcal{S}. \label{equ.1.6}%
\end{equation}
We write $a$ and $a^{\dag}$ for $a_{\hbar}$ and $a_{\hbar}^{\dag}$
respectively when $\hbar=1.$
\end{definition}

Notice that both the creation $\left(  a_{\hbar}^{\dag}\right)  $ and
annihilation $\left(  a_{\hbar}\right)  $ operators preserve $\mathcal{S}$ and
satisfy the canonical commutation relations (CCRs),
\begin{equation}
\left[  a_{\hbar},a_{\hbar}^{\dag}\right]  =\hbar I|_{\mathcal{S}}.
\label{equ.1.7}%
\end{equation}

For each $t\in\mathbb{R}$ and $\alpha_{0}\in\mathbb{C}$ we also define two
operators, $a\left(  t,\alpha_{0}\right)  $ and $a^{\dag}\left(  t,\alpha
_{0}\right)  $ acting on $\mathcal{S}$ by,
\begin{align}
a\left(  t,\alpha_{0}\right)   &  =\gamma\left(  t\right)  a+\delta\left(
t\right)  a^{\dag}\text{ and}\label{equ.1.8}\\
a^{\dag}\left(  t,\alpha_{0}\right)   &  =\bar{\gamma}\left(  t\right)
a^{\dag}+\bar{\delta}\left(  t\right)  a, \label{equ.1.9}%
\end{align}
where $\gamma\left(  t\right)  $ and $\delta\left(  t\right)  $ are determined
as in Eq. (\ref{equ.1.3}). Because we are going to fix $\alpha_{0}%
\in\mathbb{C}$ once and for all in this paper we will simply write $a\left(
t\right)  $ and $a^{\dag}\left(  t\right)  $ for $a\left(  t,\alpha
_{0}\right)  $ and $a^{\dag}\left(  t,\alpha_{0}\right)  $ respectively. These
operators still satisfy the CCRs, indeed making use of Eq. (\ref{equ.2.12})
below we find,
\begin{align}
\left[  a\left(  t\right)  ,a^{\dag}\left(  t\right)  \right]   &  =\left[
\bar{\gamma}\left(  t\right)  a^{\dag}+\bar{\delta}\left(  t\right)
a,\gamma\left(  t\right)  a+\delta\left(  t\right)  a^{\dag}\right]
\nonumber\\
&  =\left(  \left\vert \gamma\left(  t\right)  \right\vert ^{2}-\left\vert
\delta\left(  t\right)  \right\vert ^{2}\right)  I=I. \label{equ.1.10}%
\end{align}
This result also may be deduced from Theorem \ref{the.5.18} below.

\begin{definition}
[Harmonic Oscillator Hamiltonian]\label{def.1.4}The \textbf{Harmonic
Oscillator Hamiltonian }is the self-adjoint operator on $L^{2}\left(
m\right)  $ defined by
\begin{equation}
\mathcal{N}_{\hbar}:=a_{\hbar}^{\ast}\bar{a}_{\hbar}=\hbar a^{\ast}\bar{a}.
\label{equ.1.11}%
\end{equation}
As above we write $\mathcal{N}$ for $\mathcal{N}_{1}$ and refer to
$\mathcal{N}$ as the \textbf{Number operator}.
\end{definition}

\begin{remark}
\label{rem.1.5}The operator, $\mathcal{N}_{\hbar},$ is self-adjoint by a well
know theorem of Von Neumann (see for example \citep[Theorem 3.24, p.
275 in ][]{Kato1966}). It is also standard and well known (or see Corollary
\ref{cor.3.36} below) that
\[
D\left(  a_{\hbar}^{\ast}\right)  =D\left(  \bar{a}_{\hbar}\right)  =D\left(
\mathcal{N}_{\hbar}^{1/2}\right)  =D\left(  \partial_{x}\right)  \cap D\left(
M_{x}\right)  .
\]

\end{remark}

\begin{definition}
[Weyl Operators]\label{def.1.6}For $\alpha:=\left(  \xi+i\pi\right)  /\sqrt
{2}\in\mathbb{C}$ as in Eq. (\ref{equ.2.1}), define the unitary \textbf{Weyl
Operator} $U\left(  \alpha\right)  $ on $L^{2}\left(  m\right)  $ by
\begin{equation}
U\left(  \alpha\right)  =e^{\left(  \overline{\alpha\cdot a^{\dag}-\bar
{\alpha}\cdot a}\right)  }=e^{i\left(  \overline{\pi M_{x}-\frac{\xi}%
{i}\partial_{x}}\right)  .} \label{equ.1.12}%
\end{equation}
More generally, if $\hbar>0,$ let
\begin{equation}
U_{\hbar}\left(  \alpha\right)  =U\left(  \frac{\alpha}{\sqrt{\hbar}}\right)
=\exp\left(  \frac{1}{\hbar}\left(  \overline{\alpha\cdot a_{\hbar}^{\dag
}-\bar{\alpha}\cdot a_{\hbar}}\right)  \right)  . \label{equ.1.13}%
\end{equation}

\end{definition}

The symmetric operator, $i\left(  \alpha\cdot a_{\hbar}^{\dag}-\bar{\alpha
}\cdot a_{\hbar}\right)  ,$ can be shown to be essentially self adjoint on
$\mathcal{S}$ by the same methods used to show $\frac{1}{i}\partial_{x}$ is
essentially self adjoint on $C_{c}^{\infty}\left(  \mathbb{R}\right)  $ in
\citep[Proposition 9.29]{Brian2013}. Hence the Weyl operators, $U_{\hbar
}\left(  \alpha\right)  ,$ are well defined unitary operators by Stone's
theorem. Alternatively, see Proposition \ref{pro.2.6} below for an explicit
description of $U_{\hbar}\left(  \alpha\right)  .$

\begin{definition}
\label{def.1.8}Given an operator $A$ on $L^{2}\left(  m\right)  $ let
\[
\left\langle A\right\rangle _{\psi}:=\,\left\langle A\psi,\psi\right\rangle
\]
denote the \textbf{expectation }of $A$ relative to a normalized state $\psi\in
D\left(  A\right)  .$ The \textbf{variance} of $A$ relative to a normalized
state $\psi\in D\left(  A^{2}\right)  $ is then defined as
\[
\operatorname*{Var}\nolimits_{\psi}\left(  A\right)  :=\left\langle
A^{2}\right\rangle _{\psi}-\left\langle A\right\rangle _{\psi}^{2}.
\]

\end{definition}

From Corollary \ref{cor.3.10} below; if $\psi\in\mathcal{S}$ is a normalized
state and $P\left(  \theta,\theta^{\ast}\right)  $ is a non-commutative
polynomial in two variables $\left\{  \theta,\theta^{\ast}\right\}  ,$ then
\begin{align*}
\left\langle P\left(  a_{\hbar},a_{\hbar}^{\dag}\right)  \right\rangle
_{U_{\hbar}\left(  \alpha\right)  \psi}  &  =P\left(  \alpha,\bar{\alpha
}\right)  +O\left(  \sqrt{\hbar}\right) \\
\operatorname*{Var}\nolimits_{U_{\hbar}\left(  \alpha\right)  \psi}\left(
P\left(  a_{\hbar},a_{\hbar}^{\ast}\right)  \right)   &  =O\left(  \sqrt
{\hbar}\right)  .
\end{align*}
Consequently, $U_{\hbar}\left(  \alpha\right)  \psi$ is a state which is
concentrated in phase space near the $\alpha$ and are therefore reasonable
quantum mechanical approximations of the classical state $\alpha.$

\begin{definition}
[Non-Commutative Laws]\label{def.1.9}If $A_{1},\dots,A_{k}$ are operators on
$L^{2}\left(  m\right)  $ having a common dense domain $D$ such that
$A_{j}D\subset D,$ $D\subset D\left(  A_{j}^{\ast}\right)  ,$ and $A_{j}%
^{\ast}D\subset D$ for $1\leq j\leq k,$ then for a unit vector, $\psi\in D,$
and a non-commutative polynomial,
\[
\mathbf{P}:=P\left(  \theta_{1},\dots,\theta_{k},\theta_{1}^{\ast}%
,\dots,\theta_{k}^{\ast}\right)
\]
in $2k$ indeterminants, we let
\[
\mu\left(  \mathbf{P}\right)  :=\left\langle P\left(  A_{1},\dots,A_{k}%
,A_{1}^{\ast},\dots,A_{k}^{\ast}\right)  \right\rangle _{\psi}=\left\langle
P\left(  A_{1},\dots,A_{k},A_{1}^{\ast},\dots,A_{k}^{\ast}\right)  \psi
,\psi\right\rangle .
\]
The linear functional, $\mu,$ on the linear space of non-commutative
polynomials in $2k$ -- variables is referred to as the \textbf{law} of
$\left(  A_{1},\dots,A_{k}\right)  $ relative to $\psi$ and we will in the
sequel denote $\mu$ by $\operatorname*{Law}\nolimits_{\psi}\left(  A_{1}%
,\dots,A_{k}\right)  .$
\end{definition}

\subsection{Main results\label{sec.1.2}}

Theorem \ref{the.1.20} and Corollaries \ref{cor.1.22} and \ref{cor.1.24} below
on the convergence of correlation functions are the main results of this
paper. [The proofs of these results will be given Section \ref{sec.9}.] The
results of this paper will be proved under the Assumption \ref{ass.1}
described below. First we need a little more notation.

\begin{definition}
\label{def.1.10}Let $S$ be a dense subspace of a Hilbert space $\mathcal{K}$
and $A$ be an operator on $\mathcal{K}.$ We say $A$ is \textbf{symmetric on
}$S$ provided, $S\subseteq D\left(  A\right)  $ and $A|_{\mathcal{S}}\subseteq
A|_{\mathcal{S}}^{\ast},$ i.e. $\left\langle Af,g\right\rangle =\left\langle
f,Ag\right\rangle $ for all $f,g\in S.$
\end{definition}

We now introduce three different partial ordering on symmetric operators on a
Hilbert space.

\begin{notation}
\label{not.1.11}Let $S$ be a dense subspace of a Hilbert space, $\mathcal{K},
$ and $A$ and $B$ be two densely defined operators on $\mathcal{K}.$

\begin{enumerate}
\item We write $A\preceq_{S}B$ if both $A$ and $B$ are symmetric on $S$ and
\[
\left\langle A\psi,\psi\right\rangle _{\mathcal{K}}\leq\left\langle B\psi
,\psi\right\rangle _{\mathcal{K}}\text{ for all }\psi\in S.
\]

\item We write $A\preceq B$ if $A\preceq_{D\left(  B\right)  }B,$ i.e.
$D\left(  B\right)  \subset D\left(  A\right)  ,$ $A$ and $B$ are both
symmetric on $D\left(  B\right)  ,$ and
\[
\left\langle A\psi,\psi\right\rangle _{\mathcal{K}}\leq\left\langle B\psi
,\psi\right\rangle _{\mathcal{K}}\text{ for all }\psi\in D\left(  B\right)  .
\]

\item If $A$ and $B$ are non-negative (i.e. $0\preceq A$ and $0\preceq B$)
self adjoint operators on a Hilbert space $\mathcal{K},$ then we say $A\leq B
$ if and only if $D\left(  \sqrt{B}\right)  \subseteq D\left(  \sqrt
{A}\right)  $ and
\[
\left\Vert \sqrt{A}\psi\right\Vert \leq\left\Vert \sqrt{B}\psi\right\Vert
\text{ for all }\psi\in D\left(  \sqrt{B}\right)  .
\]

\end{enumerate}
\end{notation}

Interested readers may read Section 10.3 of \citep{Schmudgen2012} to learn
more properties and relations among these different partial orderings. Let us
now record the main assumptions which will be needed for the main theorems in
this paper. In this assumption, $\mathbb{R}\left\langle \theta,\theta^{\ast
}\right\rangle $ denotes the subspace of non-commutative polynomials with real
coefficients, see Subsection \ref{sub.2.3}.

\begin{assumption}
\label{ass.1}We say $H\left(  \theta,\theta^{\ast}\right)  \in\mathbb{R}%
\left\langle \theta,\theta^{\ast}\right\rangle $ satisfies Assumption 1. if,
$H$ is symmetric (see Definition \ref{def.2.14}), $d=\deg_{\theta}H\geq2$ (see
Notation \ref{not.2.12}) is even and $H_{\hbar}:=\overline{H\left(  a_{\hbar
},a_{\hbar}^{\dag}\right)  }$ satisfies; there exists constants $C>0, $
$C_{\beta}>0$ for $\beta\geq0,$ and $1\geq\eta>0$ such that for all $\hbar
\in\left(  0,\eta\right)  ,$

\begin{enumerate}
\item $H_{\hbar}$ is self-adjoint and $H_{\hbar}+C\geq I,$ and

\item for all $\beta\geq0,$
\begin{equation}
\mathcal{N}_{\hbar}^{\beta}\preceq C_{\beta}(H_{\hbar}+C)^{\beta}.
\label{equ.1.14}%
\end{equation}

\end{enumerate}
\end{assumption}

The next Proposition provides a simple class of example $H\in\mathbb{R}%
\left\langle \theta,\theta^{\ast}\right\rangle $ satisfying Assumption
\ref{ass.1} whose infinite dimensional analogues feature in some of the papers
involving Bose-Einstein condensation, see for example, \citep{Ammari2008,
Rodnianski2009}.

\begin{proposition}
[$p\left(  \theta^{\ast}\theta\right)  $ -- examples]\label{pro.1.14}Let
$p\left(  x\right)  \in\mathbb{R}\left[  x\right]  $ (the polynomials in $x$
with real coefficients) and suppose $\operatorname{deg}\left(  p\right)
\geq1$ and the leading order coefficient is positive. Then $H\left(
\theta,\theta^{\ast}\right)  =p\left(  \theta^{\ast}\theta\right)
\in\mathbb{R}\left\langle \theta,\theta^{\ast}\right\rangle $ will satisfy the
hypothesis of Assumption \ref{ass.1}.
\end{proposition}

\begin{proof}
First we will show
\[
H_{\hbar}=\overline{p\left(  a_{\hbar}^{\dagger}a_{\hbar}\right)  }=p\left(
\mathcal{N}_{\hbar}\right)  .
\]

We know that $p\left(  \mathcal{N}_{\hbar}\right)  $ is self-adjoint and by
Corollaries \ref{cor.3.26} and \ref{cor.3.40} we have
\[
p\left(  \mathcal{N}_{\hbar}\right)  =p\left(  a_{\hbar}^{\ast}\bar{a}_{\hbar
}\right)  =p\left(  \overline{a_{\hbar}^{\dagger}}\bar{a}_{\hbar}\right)
\subset\overline{p\left(  a_{\hbar}^{\dagger}a_{\hbar}\right)  }.
\]
Taking adjoint of this inclusion implies
\[
p\left(  a_{\hbar}^{\dagger}a_{\hbar}\right)  ^{\ast}=\overline{p\left(
a_{\hbar}^{\dagger}a_{\hbar}\right)  }^{\ast}\subset p\left(  \mathcal{N}%
_{\hbar}\right)  ^{\ast}=p\left(  \mathcal{N}_{\hbar}\right)  .
\]
However, since $p\left(  a_{\hbar}^{\dagger}a_{\hbar}\right)  $ is symmetric
we also have
\[
p\left(  a_{\hbar}^{\dagger}a_{\hbar}\right)  \subset p\left(  a_{\hbar
}^{\dagger}a_{\hbar}\right)  ^{\ast}=\overline{p\left(  a_{\hbar}^{\dagger
}a_{\hbar}\right)  }^{\ast}\subset p\left(  \mathcal{N}_{\hbar}\right)
\]
which implies
\[
\overline{p\left(  a_{\hbar}^{\dagger}a_{\hbar}\right)  }\subset p\left(
\mathcal{N}_{\hbar}\right)  .
\]
Since there exists $C>0$ and $C_{\beta}$ for any $\beta\geq0$ such that $x\leq
C_{\beta}\left(  p\left(  x\right)  +C\right)  $ for $x\geq0,$ it follows by
the spectral theorem that $H_{\hbar}$ satisfies Eq. (\ref{equ.1.14}).
\end{proof}

The next example provides a much broader class of $H\in\mathbb{R}\left\langle
\theta,\theta^{\ast}\right\rangle $ satisfying Assumption \ref{ass.1} while
the corresponding operators, $H_{\hbar},$ no longer typically commute with the
number operator.

\begin{example}
[Example Hamiltonians]\label{exa.1.15}Let $m\geq1,$ $b_{k}\in\mathbb{R}\left[
x\right]  $ for $0\leq k\leq m,$ and
\begin{equation}
H\left(  \theta,\theta^{\ast}\right)  :=\sum_{k=0}^{m}\frac{\left(  -1\right)
^{k}}{2^{k}}\left(  \theta-\theta^{\ast}\right)  ^{k}b_{k}\left(  \frac
{1}{\sqrt{2}}\left(  \theta+\theta^{\ast}\right)  \right)  \left(
\theta-\theta^{\ast}\right)  ^{k}. \label{equ.1.15}%
\end{equation}
With the use of Eqs. (\ref{equ.1.5}) and (\ref{equ.1.6}), it follows
\begin{equation}
H_{\hbar}=\sum_{k=0}^{m}\hbar^{k}\partial_{x}^{k}M_{b_{k}\left(  \sqrt{\hbar
}x\right)  }\partial_{x}^{k}\text{ on }\mathcal{S} \label{equ.1.16}%
\end{equation}
If

\begin{enumerate}
\item each $b_{k}\left(  x\right)  $ is an even polynomial in $x$ with
positive leading order coefficient, and $b_{m}>0,$ and

\item $\deg_{x}(b_{0})\geq2$ and $\deg_{x}(b_{k})\leq\deg_{x}(b_{k-1})$ for
$1\leq k\leq m,$
\end{enumerate}

then by Corollary 1.10 in \citep{BruceDriver20151st} $H\left(  \theta
,\theta^{\ast}\right)  $ satisfies Assumption \ref{ass.1}. In particular, if
$m>0$ and $V\in\mathbb{R}\left[  x\right]  $ such that $\deg_{x}
V\in2\mathbb{N} $ such that $\lim_{x\rightarrow\infty}V\left(  x\right)
=\infty,$ then
\begin{align}
H\left(  \theta,\theta^{\ast}\right)   &  =-\frac{m}{2}\left(  \frac
{\theta-\theta^{\ast}}{\sqrt{2}}\right)  ^{2}+V\left(  \frac{1}{\sqrt{2}%
}\left(  \theta+\theta^{\ast}\right)  \right)  \text{ and}\label{equ.1.17}\\
H\left(  a_{\hbar},a_{\hbar}^{\dag}\right)   &  =-\frac{1}{2}\hbar
m\partial_{x}^{2}+V\left(  \sqrt{\hbar}x\right)  \label{equ.1.18}%
\end{align}
satisfies Assumption \ref{ass.1}.
\end{example}

\begin{remark}
\label{rem.1.16}The essential self-adjointness of $H\left(  a_{\hbar}%
,a_{\hbar}^{\dag}\right)  $ in Eq. (\ref{equ.1.18}) and all of its
non-negative integer powers on $\mathcal{S}$ may be deduced using results in
Kato \citep{Kato1973} and Chernoff \citep{Chernoff1973}. This fact along with
the Eq. (\ref{equ.1.14}) restricted to hold on $\mathcal{S}$ and for $\beta
\in\mathbb{N}$ could be combined together to prove Eq. (\ref{equ.1.14}) for
all $\beta\geq0$ as is explained in Lemma 6.13 in \citep{BruceDriver20151st}.

Using Theorem A.1 of \citep{BruceDriver20151st}, for any symmetric
noncommutative polynomial, $H\left(  \theta,\theta^{\ast}\right)
\in\mathbb{R}\left\langle \theta,\theta^{\ast}\right\rangle ,$ there exists
polynomials, $b_{l}\left(  \sqrt{\hbar},x\right)  \in\mathbb{R}\left[
\sqrt{\hbar},x\right]  , $ (polynomials in $\sqrt{\hbar}$ and $x$ with real
coefficients), such that
\[
H\left(  a_{\hbar},a_{\hbar}^{\dag}\right)  =\sum_{k=0}^{m}\hbar^{k}%
\partial_{x}^{k}M_{b_{k}\left(  \sqrt{\hbar},\sqrt{\hbar}x\right)  }%
\partial_{x}^{k}\text{ on }\mathcal{S}.
\]
If it so happens that these $b_{k}\left(  \sqrt{\hbar},\sqrt{\hbar}x\right)  $
satisfy the assumptions of Corollary 1.10 of \citep{BruceDriver20151st}, then
Assumption \ref{ass.1} will hold for this $H.$
\end{remark}

\begin{example}
\label{ex.1.17}Let
\begin{equation}
H\left(  \theta,\theta^{\ast}\right)  =\theta^{4}+\theta^{\ast4}-\frac{7}%
{8}\left(  \theta-\theta^{\ast}\right)  \left(  \theta+\theta^{\ast}\right)
^{2}\left(  \theta-\theta^{\ast}\right)  \in\mathbb{R}\left\langle
\theta,\theta^{\ast}\right\rangle . \label{equ.1.19}%
\end{equation}
By using product rule repeatedly with Eqs. (\ref{equ.1.5}) and (\ref{equ.1.6}%
), it follows that
\[
H\left(  a_{h},a_{h}^{\dag}\right)  =\hbar^{2}\partial_{x}^{2}b_{2}\left(
\sqrt{\hbar},\sqrt{\hbar}x\right)  \partial_{x}^{2}-\hbar\partial_{x}%
b_{1}\left(  \sqrt{\hbar},\sqrt{\hbar}x\right)  \partial_{x}+b_{0}\left(
\sqrt{\hbar},\sqrt{\hbar}x\right)
\]
where
\[
b_{0}\left(  \sqrt{\hbar},x\right)  =\frac{1}{2}x^{4}+\frac{3h^{2}}{2},\text{
}b_{1}\left(  \sqrt{\hbar},x\right)  =\frac{1}{2}x^{2},\text{ and }%
b_{2}\left(  \sqrt{\hbar},x\right)  =\frac{1}{2}.
\]
These polynomials satisfy the assumptions of Corollary 1.10 of
\citep{BruceDriver20151st} and therefore $H\left(  \theta,\theta^{\ast
}\right)  $ in Eq. (\ref{equ.1.19}) satisfies Assumption \ref{ass.1}.
\end{example}

\begin{notation}
\label{not.1.19} Given a non-commutative polynomial
\begin{equation}
P\left(  \left\{  \theta_{i},\theta_{i}^{\ast}\right\}  _{i=1}^{n}\right)
:=P\left(  \theta_{1},\dots,\theta_{n},\theta_{1}^{\ast},\dots,\theta
_{n}^{\ast}\right)  \in\mathbb{C}\left\langle \theta_{1},\dots,\theta
_{n},\theta_{1}^{\ast},\dots,\theta_{n}^{\ast}\right\rangle , \label{equ.1.20}%
\end{equation}
in $2n$ -- indeterminants,
\begin{equation}
\Lambda_{n}:=\left\{  \theta_{1},\dots,\theta_{n},\theta_{1}^{\ast}%
,\dots,\theta_{n}^{\ast}\right\}  , \label{equ.1.21}%
\end{equation}
let $p_{\min}$ denote the minimum degree among all non-constant monomials
terms appearing in $P\left(  \left\{  \theta_{i},\theta_{i}^{\ast}\right\}
_{i=1}^{n}\right)  .$ In more detail there is a constant, $P_{0}\in
\mathbb{C},$ such that $P\left(  \theta_{1},\dots,\theta_{n},\theta_{1}^{\ast
},\dots,\theta_{n}^{\ast}\right)  -P_{0}$ may be written as a linear
combination in words in the alphabet, $\Lambda_{n},$ which have length no
smaller than $p_{\min}.$
\end{notation}

\begin{theorem}
\label{the.1.20}Suppose $H\left(  \theta,\theta^{\ast}\right)  \in
\mathbb{R}\left\langle \theta,\theta^{\ast}\right\rangle ,$ $d=\deg_{\theta
}H>0$ and $1\geq\eta>0$ satisfy Assumptions \ref{ass.1}, $\alpha_{0}%
\in\mathbb{C},$ $\psi\in\mathcal{S}$ is an $L^{2}\left(  m\right)  $ --
normalized state and then let;

\begin{enumerate}
\item $\alpha\left(  t\right)  \in\mathbb{C}$ be the solution (which exists
for all time by Proposition \ref{pro.3.13}) to Hamilton's (classical)
equations of motion (\ref{equ.1.1}),

\item $a\left(  t\right)  =a\left(  t,\alpha_{0}\right)  $ be the annihilation
operator on $L^{2}\left(  m\right)  $ as in Eq. (\ref{equ.1.8}), and

\item $A_{\hbar}\left(  t\right)  $ denote $a_{\hbar}$ in the Heisenberg
picture, i.e.
\begin{equation}
A_{\hbar}\left(  t\right)  :=e^{iH_{\hbar}t/\hbar}a_{\hbar}e^{-iH_{\hbar
}t/\hbar}. \label{equ.1.22}%
\end{equation}

\end{enumerate}

If $\left\{  t_{i}\right\}  _{i=1}^{n}\subset\mathbb{R}$ and $P\left(
\left\{  \theta_{i},\theta_{i}^{\ast}\right\}  _{i=1}^{n}\right)
\in\mathbb{C}\left\langle \theta_{1},\dots,\theta_{n},\theta_{1}^{\ast}%
,\dots,\theta_{n}^{\ast}\right\rangle $ is a non-commutative polynomial in
$2n$ -- indeterminants, then for $0<\hbar<\eta,$ we have
\begin{align}
&  \left\langle P\left(  \left\{  A_{\hbar}\left(  t_{i}\right)
-\alpha\left(  t_{i}\right)  ,A_{\hbar}^{\dag}\left(  t_{i}\right)
-\overline{\alpha}\left(  t_{i}\right)  \right\}  _{i=1}^{n}\right)
\right\rangle _{U_{\hbar}\left(  \alpha_{0}\right)  \psi}\nonumber\\
&  \quad\quad=\left\langle P\left(  \left\{  \sqrt{\hbar}a\left(
t_{i}\right)  ,\sqrt{\hbar}a^{\dag}\left(  t_{i}\right)  \right\}  _{i=1}%
^{n}\right)  \right\rangle _{\psi}+O\left(  \hbar^{\frac{p_{\min}+1}{2}%
}\right)  . \label{equ.1.23}%
\end{align}

\end{theorem}

\begin{remark}
\label{rem.1.21}The left member of Eq. (\ref{equ.1.23}) is well defined
because; 1) $U_{\hbar}\left(  \alpha_{0}\right)  \mathcal{S}=\mathcal{S}$ (see
Proposition \ref{pro.2.6}) and 2) $e^{itH_{\hbar}/\hbar}\mathcal{S}%
=\mathcal{S}$ (see Proposition \ref{pro.6.5}) from which it follows that
$A_{\hbar}\left(  t\right)  $ and $A_{\hbar}\left(  t\right)  ^{\dag
}=e^{iH_{\hbar}t/\hbar}a_{\hbar}^{\dag}e^{-iH_{\hbar}t/\hbar}$ both preserve
$\mathcal{S}$ for all $t\in\mathbb{R}.$
\end{remark}

This theorem is a variant of the results in Hepp \citep{Hepp1974} which now
allows for unbounded observables. It should be emphasized that the operators,
$a\left(  t\right)  ,$ are constructed using only knowledge of solutions to
the classical ordinary differential equations of motions while the
construction of $A_{\hbar}\left(  t\right)  $ requires knowledge of the
quantum mechanical evolution. As an easy consequence of Theorem \ref{the.1.20}
we may conclude that
\begin{equation}
\operatorname*{Law}\nolimits_{U_{\hbar}\left(  \alpha_{0}\right)  \psi}\left(
\left\{  A_{\hbar}\left(  t_{i}\right)  \right\}  _{i=1}^{n}\right)
\cong\operatorname*{Law}\nolimits_{\psi}\left(  \left\{  \alpha\left(
t_{i}\right)  +\sqrt{\hbar}a\left(  t_{i}\right)  \right\}  _{i=1}^{n}\right)
\text{ for }0<\hbar\ll1. \label{equ.1.24}%
\end{equation}
The precise meaning of Eq. (\ref{equ.1.24}) is given in the following corollary.

\begin{corollary}
\label{cor.1.22}If we assume the same conditions and notations as in Theorem
\ref{the.1.20}, then (for $0<\hbar<\eta)$%
\begin{align}
&  \left\langle P\left(  \left\{  A_{\hbar}\left(  t_{i}\right)  ,A_{\hbar
}^{\dag}\left(  t_{i}\right)  \right\}  _{i=1}^{n}\right)  \right\rangle
_{U_{\hbar}\left(  \alpha_{0}\right)  \psi}\nonumber\\
&  \quad\quad=\left\langle P\left(  \left\{  \alpha\left(  t_{i}\right)
+\sqrt{\hbar}a\left(  t_{i}\right)  ,\overline{\alpha}\left(  t_{i}\right)
+\sqrt{\hbar}a^{\dag}\left(  t_{i}\right)  \right\}  _{i=1}^{n}\right)
\right\rangle _{\psi}+O\left(  \hbar\right)  . \label{equ.1.25}%
\end{align}

\end{corollary}

By expanding out the right side of Eq.(\ref{equ.1.25}), it follows that
\begin{align}
&  \left\langle P\left(  \left\{  A_{\hbar}\left(  t_{i}\right)  ,A_{\hbar
}^{\dag}\left(  t_{i}\right)  \right\}  _{i=1}^{n}\right)  \right\rangle
_{U_{\hbar}\left(  \alpha_{0}\right)  \psi}\nonumber\\
&  \quad=P\left(  \left\{  \alpha\left(  t_{i}\right)  ,\bar{\alpha}\left(
t_{i}\right)  \right\}  _{i=1}^{n}\right)  +\sqrt{\hbar}\left\langle
P_{1}\left(  \left\{  \alpha\left(  t_{i}\right)  :a\left(  t_{i}\right)
,a^{\dag}\left(  t_{i}\right)  \right\}  _{i=1}^{n}\right)  \right\rangle
_{\psi}+O\left(  \hbar\right)  \label{equ.1.26}%
\end{align}
where $P_{1}\left(  \left\{  \alpha\left(  t_{i}\right)  :\theta_{i}%
,\theta_{i}^{\ast}\right\}  _{i=1}^{n}\right)  $ is a degree one homogeneous
polynomial of $\left\{  \theta_{i},\theta_{i}^{\ast}\right\}  _{i=1}^{n}$ with
coefficients depending smoothly on $\left\{  \alpha\left(  t_{i}\right)
\right\}  _{i=1}^{n}.$ Equation (\ref{equ.1.26}) states that the quantum
expectation values,
\begin{equation}
\left\langle P\left(  \left\{  A_{\hbar}\left(  t_{i}\right)  ,A_{\hbar}%
^{\dag}\left(  t_{i}\right)  \right\}  _{i=1}^{n}\right)  \right\rangle
_{U_{\hbar}\left(  \alpha_{0}\right)  \psi}, \label{equ.1.27}%
\end{equation}
closely track the corresponding classical values $P\left(  \left\{
\alpha\left(  t_{i}\right)  ,\bar{\alpha}\left(  t_{i}\right)  \right\}
_{i=1}^{n}\right)  .$ The $\sqrt{\hbar}$ term in Eq. (\ref{equ.1.26})
represent the first quantum corrections (or fluctuations ) beyond the leading
order classical behavior.

\begin{remark}
\label{rem.1.23}If both $H\left(  \theta,\theta^{\ast}\right)  ,$
$\widetilde{H}\left(  \theta,\theta^{\ast}\right)  \in\mathbb{R}\left\langle
\theta,\theta^{\ast}\right\rangle $ both satisfy Assumption \ref{ass.1} and
are such that $H^{\mathrm{cl}}\left(  \alpha\right)  :=H\left(  \alpha
,\overline{\alpha}\right)  $ and $\widetilde{H}^{\mathrm{cl}}\left(
\alpha\right)  :=\widetilde{H}\left(  \alpha,\overline{\alpha}\right)  $ are
equal modulo a constant, then Eq. (\ref{equ.1.26}) also holds with the
$A_{\hbar}\left(  t_{i}\right)  $ and $A_{\hbar}^{\dag}\left(  t_{i}\right)  $
appearing on the left side of this equation being replaced by
\[
e^{i\tilde{H}_{\hbar}t_{i}/\hbar}a_{\hbar}e^{-i\tilde{H}_{\hbar}t_{i}/\hbar
}\text{ and }e^{i\tilde{H}_{\hbar}t_{i}/\hbar}a_{\hbar}^{\dag}e^{-i\tilde
{H}_{\hbar}t_{i}/\hbar}%
\]
where $\tilde{H}_{\hbar}:=\overline{\tilde{H}\left(  a_{\hbar},a_{\hbar}%
^{\dag}\right)  }.$ In other words, if we view $H$ and $\tilde{H}$ as two
\textquotedblleft quantizations\textquotedblright of $H^{\mathrm{cl}},$ then
the quantum expectations relative to $H$ and $\tilde{H}$ agree up to order
$\sqrt{\hbar}.$
\end{remark}

\begin{corollary}
\label{cor.1.24} Under the same conditions in Theorem \ref{the.1.20}, we let
$\psi_{\hbar}=U_{\hbar}\left(  \alpha_{0}\right)  \psi.$ As $\hbar
\rightarrow0^{+}$, we have
\begin{equation}
\left\langle P\left(  \left\{  A_{\hbar}\left(  t_{i}\right)  ,A_{\hbar}%
^{\dag}\left(  t_{i}\right)  \right\}  _{i=1}^{n}\right)  \right\rangle
_{\psi_{\hbar}}\rightarrow P\left(  \left\{  \alpha_{i}\left(  t\right)
,\bar{\alpha}_{i}\left(  t\right)  \right\}  _{i=1}^{n}\right)  ,
\label{equ.1.28}%
\end{equation}
and
\begin{equation}
\left\langle P\left(  \left\{  \frac{A_{\hbar}\left(  t_{i}\right)
-\alpha\left(  t_{i}\right)  }{\sqrt{\hbar}},\frac{A_{\hbar}^{\dag}\left(
t_{i}\right)  -\bar{\alpha}\left(  t_{i}\right)  }{\sqrt{\hbar}}\right\}
_{i=1}^{n}\right)  \right\rangle _{\psi_{\hbar}}\rightarrow\left\langle
P\left(  \left\{  a\left(  t_{i}\right)  ,a^{\dag}\left(  t_{i}\right)
\right\}  _{i=1}^{n}\right)  \right\rangle _{\psi}. \label{equ.1.29}%
\end{equation}
We abbreviate this convergence by saying
\[
\operatorname*{Law}\nolimits_{\psi_{\hbar}}\left(  \left\{  \frac{A_{\hbar
}\left(  t_{i}\right)  -\alpha\left(  t_{i}\right)  }{\sqrt{\hbar}}%
,\frac{A_{\hbar}^{\dag}\left(  t_{i}\right)  -\bar{\alpha}\left(
t_{i}\right)  }{\sqrt{\hbar}}\right\}  _{i=1}^{n}\right)  \rightarrow
\operatorname*{Law}\nolimits_{\psi}\left(  \left\{  a\left(  t_{i}\right)
,a^{\dag}\left(  t_{i}\right)  \right\}  _{i=1}^{n}\right)  .
\]

\end{corollary}

\subsection{Comparison with Hepp\label{sec.1.3}}

The primary difference between our results and Hepp's results in
\citep{Hepp1974} is that we allow for non-bounded (polynomial in $a_{\hbar}$
and $a_{\hbar}^{\dag})$ observables where as Hepp's \textquotedblleft
observables\textquotedblright\ are unitary operators of the form
\[
U_{\hbar}\left(  z\right)  =\exp\left(  \overline{za_{\hbar}-\bar{z}a_{\hbar
}^{\dag}}\right)  \text{ for }z\in\mathbb{C}.
\]
As these observables are bounded operators, Hepp is able to prove his results
under less restrictive assumptions than those in Assumption \ref{ass.1} of
this paper. For the most part Hepp primarily works with Hamiltonian operators
in the Schr\"{o}dinger form of Eq. (\ref{equ.1.17}) where the potential
function, $V,$ is not necessarily restricted to be a polynomial function.
[Hepp does however allude to being able to allow for more general Hamiltonian
operators which are not necessarily of the Schr\"{o}dinger form in Eq.
(\ref{equ.1.17}).] The analogue of Corollary \ref{cor.1.24} (for $n=1)$ in
Hepp \citep{Hepp1974}, is his Theorem 2.1 which states; if $z\in\mathbb{C}$
and $\psi\in L^{2}\left(  \mathbb{R}\right)  $, then%
\[
\lim_{\hbar\downarrow0}\left\langle \exp\left(  \overline{z\frac{a_{\hbar
}-\alpha\left(  t\right)  }{\sqrt{\hbar}}-\bar{z}\frac{a_{\hbar}^{\dag}%
-\bar{\alpha}\left(  t\right)  }{\sqrt{\hbar}}}\right)  \right\rangle
_{\psi_{\hbar\left(  t\right)  }}=\left\langle \exp\left(  \overline{za\left(
t\right)  -\bar{z}a^{\dag}\left(  t\right)  }\right)  \right\rangle _{\psi},
\]
where $\psi_{\hbar}\left(  t\right)  :=e^{-iH_{\hbar}t/\hbar}U_{\hbar}\left(
\alpha_{0}\right)  \psi.$

\textit{Acknowledgment.} Both authors would like to thank Ioan Bejenaru, Brian
C. Hall, Rupert L. Frank, and Jacob Sterbenz for helpful discussions at
various stages of this work.

\section{Background and Setup\label{sec.2}}

In this section we will expand on the basic setup described above and recall
some basic facts that will be needed throughout the paper.

\subsection{Classical Setup\label{sub.2.1}}

In this paper, we take configuration space to be $\mathbb{R}$ so that our
classical state space is $T^{\ast}\mathbb{R\cong R}^{2}.$ [Extensions to
higher and to infinite dimensions will be considered elsewhere.] Following
Hepp \citep{Hepp1974}, we identify $T^{\ast}\mathbb{R}$ with $\mathbb{C}$ via
\begin{equation}
T^{\ast}\mathbb{R\ni}\left(  \xi,\pi\right)  \rightarrow\alpha:=\frac{1}%
{\sqrt{2}}\left(  \xi+i\pi\right)  . \label{equ.2.1}%
\end{equation}
Taking in account the \textquotedblleft$\sqrt{2}$\textquotedblright above, we
set
\[
\frac{\partial}{\partial\alpha}:=\frac{1}{\sqrt{2}}\left(  \partial_{\xi
}-i\partial_{\pi}\right)  \text{ and }\frac{\partial}{\partial\bar{\alpha}%
}:=\frac{1}{\sqrt{2}}\left(  \partial_{\xi}+i\partial_{\pi}\right)
\]
so that $\frac{\partial}{\partial\alpha}\alpha=1=\frac{\partial}{\partial
\bar{\alpha}}\bar{\alpha}$ and $\frac{\partial}{\partial\alpha}\bar{\alpha
}=0=\frac{\partial}{\partial\bar{\alpha}}\alpha.$ As usual given a smooth real
valued function,\footnote{Later $H^{cl}$ will be the symbol of a symmetric
element of $H\in\mathbb{C}\left\langle \theta,\theta^{\ast}\right\rangle $ as
described in subsection \ref{sub.2.3}.} $H^{\mathrm{cl}}\left(  \xi
,\pi\right)  ,$ on $T^{\ast}\mathbb{R}$ we say $\left(  \xi\left(  t\right)
,\pi\left(  t\right)  \right)  $ solves Hamilton's equations of motion
provided,
\begin{equation}
\dot{\xi}\left(  t\right)  =H_{\pi}^{\mathrm{cl}}\left(  \xi\left(  t\right)
,\pi\left(  t\right)  \right)  \text{ and }\dot{\pi}\left(  t\right)
=-H_{\xi}^{\mathrm{cl}}\left(  \xi\left(  t\right)  ,\pi\left(  t\right)
\right)  \label{equ.2.2}%
\end{equation}
where $H_{\pi}^{\mathrm{cl}}:=\partial H^{\mathrm{cl}}/\partial\pi$ and
$H_{\xi}^{\mathrm{cl}}:=\partial H^{\mathrm{cl}}/\partial\xi$ . A simple
verifications shows; if
\[
\alpha\left(  t\right)  :=\frac{1}{\sqrt{2}}\left(  \xi\left(  t\right)
+i\pi\left(  t\right)  \right)  ,
\]
then $\left(  \xi\left(  t\right)  ,\pi\left(  t\right)  \right)  $ solves
Hamilton's Eqs. (\ref{equ.2.2}) iff $\alpha\left(  t\right)  $ satisfies
\begin{equation}
i\dot{\alpha}\left(  t\right)  =\left(  \frac{\partial}{\partial\bar{\alpha}%
}\tilde{H}^{\mathrm{cl}}\right)  \left(  \alpha\left(  t\right)  \right)
\label{equ.2.3}%
\end{equation}
where
\[
\tilde{H}^{\mathrm{cl}}\left(  \alpha\right)  :=H^{\mathrm{cl}}\left(  \xi
,\pi\right)  \text{ where }\alpha=\frac{1}{\sqrt{2}}\left(  \xi+i\pi\right)
\in\mathbb{C}.
\]
In the future we will identify $\tilde{H}^{\mathrm{cl}}$ with $H^{\mathrm{cl}%
}$ and drop the tilde from our notation.

\begin{example}
\label{exa.2.2}If $H\left(  \alpha\right)  =\left\vert \alpha\right\vert
^{2}+\frac{1}{2}\left\vert \alpha\right\vert ^{4},$ then the associated
Hamiltonian equations of motion are given by
\[
i\dot{\alpha}=\frac{\partial}{\partial\bar{\alpha}}\left(  \alpha\bar{\alpha
}+\frac{1}{2}\alpha^{2}\bar{\alpha}^{2}\right)  =\alpha+\alpha^{2}\bar{\alpha
}=\alpha+\left\vert \alpha\right\vert ^{2}\alpha.
\]

\end{example}

\begin{proposition}
\label{pro.2.3} Let $z\left(  t\right)  :=\Phi^{\prime}\left(  t,\alpha
_{0}\right)  z$ be the \textbf{real} differential of the flow associated to
Eq. (\ref{equ.1.1}) as in Eq. (\ref{equ.1.2}). Then $z\left(  t\right)  $
satisfies $z\left(  0\right)  =z$ and
\begin{equation}
i\dot{z}\left(  t\right)  =u\left(  t\right)  \bar{z}\left(  t\right)
+v\left(  t\right)  z\left(  t\right)  , \label{equ.2.4}%
\end{equation}
where
\begin{equation}
u\left(  t\right)  :=\left(  \frac{\partial^{2}}{\partial\bar{\alpha}^{2}%
}H^{\mathrm{cl}}\right)  \left(  \alpha\left(  t\right)  \right)
\in\mathbb{C}\text{ and }v\left(  t\right)  =\left(  \frac{\partial^{2}%
}{\partial\alpha\partial\bar{\alpha}}H^{\mathrm{cl}}\right)  \left(
\alpha\left(  t\right)  \right)  \in\mathbb{R}. \label{equ.2.5}%
\end{equation}
Moreover, if we express $z\left(  t\right)  =\gamma\left(  t\right)
z+\delta\left(  t\right)  \bar{z}$ as in Eq. (\ref{equ.1.3}) and let
\begin{equation}
\Lambda\left(  t\right)  :=\left[
\begin{array}
[c]{cc}%
\gamma\left(  t\right)  & \delta\left(  t\right) \\
\bar{\delta}\left(  t\right)  & \bar{\gamma}\left(  t\right)
\end{array}
\right]  , \label{equ.2.6}%
\end{equation}
then
\[
\det\Lambda\left(  t\right)  =\left\vert \gamma\left(  t\right)  \right\vert
^{2}-\left\vert \delta\left(  t\right)  \right\vert ^{2}=1~\forall
~t\in\mathbb{R}%
\]
and
\begin{equation}
i\dot{\Lambda}\left(  t\right)  =\left[
\begin{array}
[c]{cc}%
v\left(  t\right)  & u\left(  t\right) \\
-\bar{u}\left(  t\right)  & -\bar{v}\left(  t\right)
\end{array}
\right]  \Lambda\left(  t\right)  \text{ and }\Lambda\left(  0\right)  =I.
\label{equ.2.7}%
\end{equation}

\end{proposition}

\begin{proof}
First recall if $f:\mathbb{C\rightarrow C}$ is a smooth function (not analytic
in general), then the \textbf{real }differential, $z\rightarrow f^{\prime
}\left(  \alpha\right)  z:=\frac{d}{ds}|_{0}f\left(  \alpha+sz\right)  ,$ of
$f$ at $\alpha$ satisfies
\begin{equation}
f^{\prime}\left(  \alpha\right)  z=\left(  \frac{\partial}{\partial\alpha
}f\right)  \left(  \alpha\right)  z+\left(  \frac{\partial}{\partial
\bar{\alpha}}f\right)  \left(  \alpha\right)  \bar{z}.\label{equ.2.8}%
\end{equation}
By definition $\Phi\left(  t,\alpha_{0}\right)  $ satisfies the differential
equation,
\[
i\dot{\Phi}\left(  t,\alpha_{0}\right)  =\left(  \frac{\partial}{\partial
\bar{\alpha}}H^{\mathrm{cl}}\right)  \left(  \Phi\left(  t,\alpha_{0}\right)
\right)  \text{ and }\Phi\left(  0,\alpha_{0}\right)  =\alpha_{0}.
\]
Differentiating this equation relative to $\alpha_{0}$ using the chain rule
along with Eq. (\ref{equ.2.8}) shows $z\left(  t\right)  :=\Phi^{\prime
}\left(  t,\alpha_{0}\right)  z$ satisfies Eq. (\ref{equ.2.4}). The fact that
$v\left(  t\right)  $ is real valued follows from its definition in Eq.
(\ref{equ.2.5}) and the fact that $H^{\mathrm{cl}}$ is a real valued function.

Inserting the expression, $z\left(  t\right)  =\gamma\left(  t\right)
z+\delta\left(  t\right)  \bar{z},$ into Eq. (\ref{equ.2.4}) one shows after a
little algebra that,
\[
i\dot{\gamma}\left(  t\right)  z+i\dot{\delta}\left(  t\right)  \bar
{z}=\left(  u\left(  t\right)  \bar{\delta}\left(  t\right)  +v\left(
t\right)  \gamma\left(  t\right)  \right)  z+\left(  u\left(  t\right)
\bar{\gamma}\left(  t\right)  +v\left(  t\right)  \delta\left(  t\right)
\right)  \bar{z}%
\]
from which we conclude that $\left(  \gamma\left(  t\right)  ,\delta\left(
t\right)  \right)  \in\mathbb{C}^{2}$ satisfy the equations
\begin{align}
i\dot{\gamma}\left(  t\right)   &  =u\left(  t\right)  \bar{\delta}\left(
t\right)  +v\left(  t\right)  \gamma\left(  t\right)  \text{ and}%
\label{equ.2.9}\\
i\dot{\delta}\left(  t\right)   &  =u\left(  t\right)  \bar{\gamma}\left(
t\right)  +v\left(  t\right)  \delta\left(  t\right)  . \label{equ.2.10}%
\end{align}
Using these equations we then find;
\begin{align}
\frac{d}{dt}\left(  \left\vert \gamma\right\vert ^{2}-\left\vert
\delta\right\vert ^{2}\right)   &  =2\operatorname{Re}\left(  \dot{\gamma}%
\bar{\gamma}-\dot{\delta}\bar{\delta}\right) \nonumber\\
&  =2\operatorname{Re}\left(  -i\left(  u\bar{\delta}+v\gamma\right)
\bar{\gamma}+i\left(  u\bar{\gamma}+v\delta\right)  \bar{\delta}\right)
\nonumber\\
&  =2\operatorname{Re}\left(  -iv\left\vert \gamma\right\vert ^{2}%
+iv\left\vert \delta\right\vert ^{2}\right)  =0. \label{equ.2.11}%
\end{align}
Since $z\left(  0\right)  =z,$ $\gamma\left(  0\right)  =1$ and $\delta\left(
0\right)  =1$ and so from Eq. (\ref{equ.2.11}) we learn
\begin{equation}
\left(  \left\vert \gamma\right\vert ^{2}-\left\vert \delta\right\vert
^{2}\right)  \left(  t\right)  =\left(  \left\vert \gamma\right\vert
^{2}-\left\vert \delta\right\vert ^{2}\right)  \left(  0\right)  =1^{2}%
-0^{2}=1. \label{equ.2.12}%
\end{equation}
Finally, Eq. (\ref{equ.2.7}) is simply the vector form of Eqs. (\ref{equ.2.9})
and (\ref{equ.2.10}).
\end{proof}

\begin{remark}
\label{rem.2.5}Equation (\ref{equ.2.4}) may be thought of as the time
dependent Hamiltonian flow,
\[
i\dot{z}\left(  t\right)  =\frac{\partial q\left(  t,\cdot\right)  }%
{\partial\bar{z}}\left(  z\left(  t\right)  \right)
\]
where $q\left(  t,z\right)  \in\mathbb{R}$ is the quadratic time dependent
Hamiltonian defined by
\begin{align*}
q\left(  t:z\right)  =  &  \frac{1}{2}u\left(  t\right)  z^{2}+\frac{1}{2}%
\bar{u}\left(  t\right)  \bar{z}^{2}+v\left(  t\right)  \bar{z}z\\
=  &  \frac{1}{2}\left(  \frac{\partial^{2}}{\partial\alpha^{2}}%
H^{\mathrm{cl}}\right)  \left(  \alpha\left(  t\right)  \right)  z^{2}%
+\frac{1}{2}\left(  \frac{\partial^{2}}{\partial\bar{\alpha}^{2}%
}H^{\mathrm{cl}}\right)  \left(  \alpha\left(  t\right)  \right)  \bar{z}%
^{2}\\
&  +\left(  \frac{\partial}{\partial\alpha}\frac{\partial}{\partial\bar
{\alpha}}H^{\mathrm{cl}}\right)  \left(  \alpha\left(  t\right)  \right)
\left\vert z\right\vert ^{2}.
\end{align*}

\end{remark}

\subsection{Quantum Mechanical Setup\label{sub.2.2}}

Recall that our quantum mechanical Hilbert space is taken to be the space of
Lebesgue square integrable complex valued functions on $\mathbb{R}$
($L^{2}\left(  m\right)  $ $:=L^{2}\left(  \mathbb{R},m\right)  )$ equipped
with the usual $L^{2}\left(  m\right)  $-inner product as in Eq.
(\ref{equ.1.4}). To each $\hbar>0$ ($\hbar$ is to be thought of as Planck's
constant), let
\begin{equation}
q_{\hbar}:=\sqrt{\hbar}M_{x}\text{ and }p_{\hbar}:=\sqrt{\hbar}\frac{1}%
{i}\frac{d}{dx} \label{e.2.13}%
\end{equation}
interpreted as self-adjoint operators on $L^{2}\left(  m\right)
:=L^{2}\left(  \mathbb{R},m\right)  $ with domains
\begin{align*}
D\left(  q_{\hbar}\right)   &  =\left\{  f\in L^{2}\left(  m\right)
:x\rightarrow xf\left(  x\right)  \in L^{2}\left(  m\right)  \right\}  \text{
and}\\
D\left(  p_{\hbar}\right)   &  =D\left(  \frac{d}{dx}\right)  =\left\{  f\in
L^{2}\left(  m\right)  :x\rightarrow f\left(  x\right)  \text{ is A.C. and
}f^{\prime}\in L^{2}\left(  m\right)  \right\}
\end{align*}
where A.C. is an abbreviation of absolutely continuous. Using Corollary
\ref{cor.3.36} below, the annihilation and creation operators in Definition
\ref{def.1.3} may be expressed as
\begin{align}
\bar{a}_{\hbar}  &  :=\frac{q_{\hbar}+ip_{\hbar}}{\sqrt{2}}=\sqrt{\frac{\hbar
}{2}}\left(  M_{x}+\frac{d}{dx}\right)  \text{ and }\label{e.2.14}\\
a_{\hbar}^{\ast}  &  :=\frac{q_{\hbar}-ip_{\hbar}}{\sqrt{2}}=\sqrt{\frac
{\hbar}{2}}\left(  M_{x}-\frac{d}{dx}\right)  . \label{e.2.15}%
\end{align}

\subsubsection{Weyl Operator\label{sub.2.2.1}}

\begin{proposition}
\label{pro.2.6}Let $\alpha:=\left(  \xi+i\pi\right)  /\sqrt{2}\in\mathbb{C},$
$\hbar>0,$ and $U\left(  \alpha\right)  $ and $U_{\hbar}\left(  \alpha\right)
$ be as in Definition \ref{def.1.6}. Then
\begin{equation}
\left(  U\left(  \alpha\right)  f\right)  \left(  x\right)  =\exp\left(
i\pi\left(  x-\frac{1}{2}\xi\right)  \right)  f\left(  x-\xi\right)
~\forall~f\in L^{2}\left(  m\right)  , \label{equ.2.16}%
\end{equation}
$U\left(  \alpha\right)  \mathcal{S}=\mathcal{S},$
\begin{align}
U_{\hbar}\left(  \alpha\right)  ^{\ast}a_{\hbar}U_{\hbar}\left(
\alpha\right)   &  =a_{\hbar}+\alpha,\text{ and}\label{equ.2.17}\\
U_{\hbar}\left(  \alpha\right)  ^{\ast}a_{\hbar}^{\dag}U_{\hbar}\left(
\alpha\right)   &  =a_{\hbar}^{\dag}+\bar{\alpha}, \label{equ.2.18}%
\end{align}
as identities on $\mathcal{S}.$
\end{proposition}

\begin{proof}
Given $f\in\mathcal{S}$ let $F\left(  t,x\right)  :=\left(  U\left(
t\alpha\right)  f\right)  \left(  x\right)  $ so that
\begin{equation}
\frac{\partial}{\partial t}F\left(  t,x\right)  =\left(  i\pi x-\xi
\frac{\partial}{\partial x}\right)  F\left(  t,x\right)  \text{ with }F\left(
0,x\right)  =f\left(  x\right)  . \label{equ.2.19}%
\end{equation}
Solving this equation by the method of characteristics then gives Eq.
(\ref{equ.2.16}). [Alternatively one easily verifies directly that
\[
F\left(  t,x\right)  :=\exp(it\pi(x-\frac{1}{2}t\xi))f(x-t\xi)
\]
solves Eq. (\ref{equ.2.19}).] It is clear from Eq. (\ref{equ.2.16}) that
$U\left(  \alpha\right)  \mathcal{S\subset S}$ and $U\left(  -\alpha\right)
U\left(  \alpha\right)  =I$ for all $\alpha\in\mathbb{C}.$ Therefore
$\mathcal{S}\subset U\left(  -\alpha\right)  \mathcal{S}.$ Replacing $\alpha$
by $-\alpha$ in this last inclusion allows us to conclude that $U\left(
\alpha\right)  \mathcal{S=S}.$ The formula in Eq. (\ref{equ.2.19}) also
directly extends to $L^{2}\left(  m\right)  $ where it defines a unitary
operator. The identities in Eqs. (\ref{equ.2.17}) and (\ref{equ.2.18}) for
$\hbar=1$ follows by simple direct calculations using Eq. (\ref{equ.2.16}).
The case of general $\hbar>0$ then follows by simple scaling arguments.
\end{proof}

\begin{remark}
\label{rem.2.8}Another way to prove Eq. (\ref{equ.2.17}) is to integrate the
identity,
\[
\frac{d}{dt}U_{\hbar}\left(  t\alpha\right)  ^{\ast}a_{\hbar}U_{\hbar}\left(
t\alpha\right)  =-U_{\hbar}\left(  t\alpha\right)  ^{\ast}\left[  \frac
{1}{\hbar}\left(  \alpha\cdot a_{\hbar}^{\dag}-\bar{\alpha}\cdot a_{\hbar
}\right)  ,a_{\hbar}\right]  U_{\hbar}\left(  t\alpha\right)  =\alpha,
\]
with respect to $t$ on $\mathcal{S}$ and the initial condition $U\left(
0\right)  =I.$
\end{remark}

\begin{definition}
\label{def.2.9} Suppose that $\left\{  W\left(  t\right)  \right\}
_{t\in\mathbb{R}}$ is a one parameter family of (possibly) unbounded operators
on a Hilbert space $\left\langle \mathcal{K},\left\langle \cdot,\cdot
\right\rangle _{\mathcal{K}}\right\rangle .$ Given a dense subspace,
$D\subset\mathcal{K},$ we say $W\left(  t\right)  $ is \textbf{strongly
}$\left\Vert \cdot\right\Vert _{\mathcal{K}}$-norm\textbf{\ differentiable} on
$D$ if 1) $D\subset D\left(  W\left(  t\right)  \right)  $ for all
$t\in\mathbb{R}$ and 2) for all $\psi\in D,$ $t\rightarrow W\left(  t\right)
\psi$ is $\left\Vert \cdot\right\Vert _{\mathcal{K}}$-norm differentiable. For
notational simplicity we will write $\dot{W}\left(  t\right)  \psi$ for
$\frac{d}{dt}\left[  W\left(  t\right)  \psi\right]  .$
\end{definition}

\begin{proposition}
\label{pro.2.10}If $\mathbb{R}\ni t\rightarrow\alpha\left(  t\right)
\in\mathbb{C}$ is a $C^{1}$ function and $\mathcal{N}:=\mathcal{N}_{\hbar
}|_{\hbar=1}$ the number operator defined in Eq. (\ref{equ.1.11}), then
$\left\{  U\left(  \alpha\left(  t\right)  \right)  \right\}  _{t\in
\mathbb{R}}$ is strongly $L^{2}\left(  m\right)  $-norm differentiable on
$D\left(  \sqrt{\mathcal{N}}\right)  $ as in the Definition \ref{def.2.9} and
for all $f\in D\left(  \sqrt{\mathcal{N}}\right)  $ we have
\begin{align*}
\frac{d}{dt}(U\left(  \alpha\left(  t\right)  \right)  f)  &  =\left(
\dot{\alpha}\left(  t\right)  a^{\ast}-\overline{\dot{\alpha}\left(  t\right)
}\bar{a}+i\operatorname{Im}\left(  \alpha\left(  t\right)  \overline
{\dot{\alpha}\left(  t\right)  }\right)  \right)  U\left(  \alpha\left(
t\right)  \right)  f\\
&  =U\left(  \alpha\left(  t\right)  \right)  \left(  \dot{\alpha}\left(
t\right)  a^{\ast}-\overline{\dot{\alpha}\left(  t\right)  }\bar
{a}-i\operatorname{Im}\left(  \alpha\left(  t\right)  \overline{\dot{\alpha
}\left(  t\right)  }\right)  \right)  f.
\end{align*}
Moreover, $U\left(  \alpha\left(  t\right)  \right)  $ preserves $D\left(
\sqrt{\mathcal{N}}\right)  $, $C_{c}(\mathbb{R}),$ and $\mathcal{S}.$
\end{proposition}

\begin{proof}
From Corollary \ref{cor.3.36} below we know $D(\partial_{x})\cap
D(M_{x})=D\left(  \sqrt{\mathcal{N}}\right)  .$ Using this fact, the
proposition is a straightforward verification based on Eq. (\ref{equ.2.16}).
The reader not wishing to carry out these computations may find it instructive
to give a formal proof based on the algebraic fact that $e^{A+B}=e^{A}%
e^{B}e^{-\frac{1}{2}\left[  A,B\right]  }$ where $A$ and $B$ are operators
such that the commutator, $\left[  A,B\right]  :=AB-BA,$ commutes with both
$A$ and $B.$
\end{proof}

As we do not wish to make any particular choice of quantization scheme, in
this paper we will describe all operators as a non-commutative polynomial
functions of $a_{\hbar}$ and $a_{\hbar}^{\dag}.$ This is the topic of the next subsection.

\subsection{Non-commutative Polynomial Expansions\label{sub.2.3}}

\begin{notation}
\label{not.2.12}Let $\mathbb{C}\left\langle \theta,\theta^{\ast}\right\rangle
$ be the space of non-commutative polynomials in the non-commutative
indeterminates. That is to say $\mathbb{C}\left\langle \theta,\theta^{\ast
}\right\rangle $ is the vector space over $\mathbb{C}$ whose basis consists of
words in the two letter alphabet, $\Lambda_{1}=\left\{  \theta,\theta^{\ast
}\right\}  ,$ cf. Eq. (\ref{equ.1.21}). The general element, $P\left(
\theta,\theta^{\ast}\right)  ,$ of $\mathbb{C}\left\langle \theta,\theta
^{\ast}\right\rangle $ may be written as
\begin{equation}
P\left(  \theta,\theta^{\ast}\right)  =\sum_{k=0}^{d}\sum_{\mathbf{b}=\left(
b_{1},\dots,b_{k}\right)  \in\Lambda_{1}^{k}}c_{k}\left(  \mathbf{b}\right)
b_{1}\dots b_{k}, \label{equ.2.20}%
\end{equation}
where $d\in\mathbb{N}_{0}$ and
\[
\left\{  c_{k}\left(  \mathbf{b}\right)  :0\leq k\leq d\text{ and }%
\mathbf{b}=\left(  b_{1},\dots,b_{k}\right)  \mathbf{\in}\Lambda_{1}%
^{k}\right\}  \subset\mathbb{C}.
\]
If $c_{d}:\Lambda_{1}^{d}\rightarrow\mathbb{C}$ is not the zero function, we
say $d=:\deg_{\theta}P$ is the degree of $P.$
\end{notation}

It is sometimes convenient to decompose $P\left(  \theta,\theta^{\ast}\right)
$ in Eq. (\ref{equ.2.20}) as
\begin{equation}
P\left(  \theta,\theta^{\ast}\right)  =\sum_{k=0}^{d}P_{k}\left(
\theta,\theta^{\ast}\right)  \label{equ.2.21}%
\end{equation}
where
\begin{equation}
P_{k}\left(  \theta,\theta^{\ast}\right)  =\sum_{b_{1},\dots,b_{k}\in
\Lambda_{1}}c_{k}\left(  b_{1},\dots,b_{k}\right)  b_{1}\dots b_{k}.
\label{equ.2.22}%
\end{equation}
Polynomials of the form in Eq. (\ref{equ.2.22}) are said to be
\textbf{homogeneous }of degree $k.$ By convention, $P_{0}:=P_{0}\left(
\theta,\theta^{\ast}\right)  $ is just an element of $\mathbb{C}.$ We endow
$\mathbb{C}\left\langle \theta,\theta^{\ast}\right\rangle $ with its $\ell
^{1}$ -- norm, $\left\vert \cdot\right\vert ,$ defined for $P$ as in Eq.
(\ref{equ.2.20}) by
\begin{equation}
\left\vert P\right\vert :=\sum_{k=0}^{d}\left\vert P_{k}\right\vert \text{
where }\left\vert P_{k}\right\vert =\sum_{\mathbf{b}=\left(  b_{1},\dots
,b_{k}\right)  \in\Lambda_{1}^{k}}\left\vert c_{k}\left(  \mathbf{b}\right)
\right\vert . \label{equ.2.23}%
\end{equation}

\begin{definition}
[Monomials]\label{def.2.13}For $\mathbf{b}=\left(  b_{1},\dots,b_{k}\right)
\in\left\{  \theta,\theta^{\ast}\right\}  ^{k}$ let $u_{\mathbf{b}}%
\in\mathbb{C}\left\langle \theta,\theta^{\ast}\right\rangle $ be the
monomial,
\begin{equation}
u_{\mathbf{b}}\left(  \theta,\theta^{\ast}\right)  =b_{1}\dots b_{k}
\label{equ.2.24}%
\end{equation}
with the convention that for $k=0$ we associate the unit element $u_{0}=1.$
\end{definition}

As usual, we make $\mathbb{C}\left\langle \theta,\theta^{\ast}\right\rangle $
into a non-commutative algebra with its natural multiplication determined on
the word basis elements $\cup_{k=0}^{\infty}\left\{  u_{\mathbf{b}%
}:\mathbf{b\in}\left\{  \theta,\theta^{\ast}\right\}  ^{k}\right\}  $ by
concatenation of words, i.e. $u_{\mathbf{b}}u_{\mathbf{d}}=u_{\left(
\mathbf{b,d}\right)  }$ where if $\mathbf{d}=\left(  d_{1},\dots,d_{l}\right)
\in\left\{  \theta,\theta^{\ast}\right\}  ^{l}$
\[
\left(  \mathbf{b,d}\right)  :=\left(  b_{1},\dots,b_{k},d_{1},\dots
,d_{l}\right)  \in\left\{  \theta,\theta^{\ast}\right\}  ^{k+l}.
\]
For example, $\theta\theta\theta^{\ast}\cdot\theta^{\ast}\theta=\theta
\theta\theta^{\ast}\theta^{\ast}\theta.$ We also define a \textbf{natural
involution} on $\mathbb{C}\left\langle \theta,\theta^{\ast}\right\rangle $
determined by $\left(  \theta\right)  ^{\ast}=\theta^{\ast},$ $\left(
\theta^{\ast}\right)  ^{\ast}=\theta,$ $z^{\ast}=\bar{z}$ for $z\in
\mathbb{C},$ and $\left(  \alpha\cdot\beta\right)  ^{\ast}=\beta^{\ast}%
\alpha^{\ast}$ for $\alpha,\beta\in$ $\mathbb{C}\left\langle \theta
,\theta^{\ast}\right\rangle .$ Formally, if $\mathbf{b}=\left(  b_{1}%
,\dots,b_{k}\right)  \in\left\{  \theta,\theta^{\ast}\right\}  ^{k},$ then
\begin{equation}
u_{\mathbf{b}}^{\ast}=b_{k}^{\ast}b_{k-1}^{\ast}\dots b_{1}^{\ast
}=u_{\mathbf{b}^{\ast}}\text{ where }\mathbf{b}^{\ast}:=\left(  b_{k}^{\ast
},b_{k-1}^{\ast},\dots,b_{1}^{\ast}\right)  . \label{equ.2.25}%
\end{equation}
In what follows we will often denote an $P\in\mathbb{C}\left\langle
\theta,\theta^{\ast}\right\rangle $ by $P\left(  \theta,\theta^{\ast}\right)
. $

\begin{definition}
[Symmetric Polynomials]\label{def.2.14}We say $P\in\mathbb{C}\left\langle
\theta,\theta^{\ast}\right\rangle $ is \textbf{symmetric} provided $P=P^{\ast
}.$
\end{definition}

If $\mathcal{A}$ is any unital algebra equipped with an involution,
$\xi\rightarrow\xi^{\dag},$ and $\xi$ is any fixed element of $\mathcal{A},$
then there exists a unique algebra homomorphism
\[
P\left(  \theta,\theta^{\ast}\right)  \in\mathbb{C}\left\langle \theta
,\theta^{\ast}\right\rangle \rightarrow P\left(  \xi,\xi^{\dag}\right)
\in\mathcal{A}%
\]
determined by substituting $\xi$ for $\theta$ and $\xi^{\dag}$ for
$\theta^{\ast}.$ Moreover, the homomorphism preserves involutions, i.e.
$\left[  P\left(  \xi,\xi^{\dag}\right)  \right]  ^{\dag}=P^{\ast}\left(
\xi,\xi^{\dag}\right)  .$ The two special cases of this construction that we
need here are contained in the following two definitions.

\begin{definition}
[Classical Symbols]\label{def.2.15}The symbol (or\textbf{\ classical }residue)
of $P\in\mathbb{C}\left\langle \theta,\theta^{\ast}\right\rangle $ is the
function $P^{\mathrm{cl}}\in\mathbb{C}\left[  z,\bar{z}\right]  $ ($=$ the
commutative polynomials in $z$ and $\bar{z}$ with complex coefficients)
defined by $P^{\mathrm{cl}}\left(  \alpha\right)  :=P\left(  \alpha
,\bar{\alpha}\right)  $ where we view $\mathbb{C}$ as a commutative algebra
with an involution given by complex conjugation.
\end{definition}

\begin{definition}
[Polynomial Operators]\label{def.2.16}If $P\left(  \theta,\theta^{\ast
}\right)  \in\mathbb{C}\left\langle \theta,\theta^{\ast}\right\rangle $ is a
non-commutative polynomial and $\hbar>0,$ then $P\left(  a_{\hbar},a_{\hbar
}^{\dag}\right)  $ is a differential operator on $L^{2}\left(  m\right)  $
whose domain is $\mathcal{S}.$ [Notice that $P\left(  a_{\hbar},a_{\hbar
}^{\dag}\right)  $ preserves $\mathcal{S},$ i.e. $P\left(  a_{\hbar},a_{\hbar
}^{\dag}\right)  \mathcal{S}\subset\mathcal{S}.]$ We further let $P_{\hbar
}:=\overline{P\left(  a_{\hbar},a_{\hbar}^{\dag}\right)  }$ be the closure of
$P\left(  a_{\hbar},a_{\hbar}^{\dag}\right)  .$ Any linear differential
operator of the form $P\left(  a_{\hbar},a_{\hbar}^{\dag}\right)  $ for some
$P\left(  \theta,\theta^{\ast}\right)  \in\mathbb{C}\left\langle \theta
,\theta^{\ast}\right\rangle $ will be called a \textbf{polynomial operator}.
\end{definition}

We introduce the following notation in order to write out $P\left(  a_{\hbar
},a_{\hbar}^{\dag}\right)  $ more explicitly.

\begin{notation}
\label{not.2.17}For any $\hbar>0$ let $\Xi_{\hbar}:\left\{  \theta
,\theta^{\ast}\right\}  \rightarrow\left\{  a_{\hbar},a_{\hbar}^{\dag
}\right\}  $ be define by
\begin{equation}
\Xi_{\hbar}\left(  b\right)  =\left\{
\begin{array}
[c]{ccc}%
a_{\hbar} & \text{if} & b=\theta\\
a_{\hbar}^{\dag} & \text{if} & b=\theta^{\ast}%
\end{array}
\right.  . \label{equ.2.26}%
\end{equation}
In the special case where $\hbar=1$ we will simply denote $\Xi_{1}$ by $\Xi.$
\end{notation}

With this notation if $P\in\mathbb{C}\left\langle \theta,\theta^{\ast
}\right\rangle $ is as in Eq. (\ref{equ.2.20}), then $P\left(  a_{\hbar
},a_{\hbar}^{\dag}\right)  $ may be written as,
\begin{equation}
P\left(  a_{\hbar},a_{\hbar}^{\dag}\right)  =\sum_{k=0}^{d}\sum_{\mathbf{b}%
=\left(  b_{1},\dots,b_{k}\right)  \in\Lambda_{1}^{k}}c_{k}\left(
\mathbf{b}\right)  \Xi_{\hbar}\left(  b_{1}\right)  \dots\Xi_{\hbar}\left(
b_{k}\right)  \label{equ.2.27}%
\end{equation}
or as
\begin{equation}
P\left(  a_{\hbar},a_{\hbar}^{\dag}\right)  =\sum_{k=0}^{d}\sum_{\mathbf{b}%
=\left(  b_{1},\dots,b_{k}\right)  \in\Lambda_{1}^{k}}\hbar^{k/2}c_{k}\left(
\mathbf{b}\right)  u_{\mathbf{b}}\left(  a,a^{\dag}\right)  \label{equ.2.28}%
\end{equation}

\begin{definition}
[Monomial Operators]\label{def.2.18}Any linear differential operator of the
form $u_{\mathbf{b}}\left(  a,a^{\dag}\right)  =\Xi_{1}\left(  b_{1}\right)
\dots\Xi_{1}\left(  b_{k}\right)  $ for some $\mathbf{b}=\left(  b_{1}%
,\dots,b_{k}\right)  \in\left\{  \theta,\theta^{\ast}\right\}  ^{k}$ and
$k\in\mathbb{N}_{0}$ will be called a \textbf{monomial operator.}
\end{definition}

\begin{remark}
\label{rem.2.19}If $H\left(  \theta,\theta^{\ast}\right)  \in\mathbb{C}%
\left\langle \theta,\theta^{\ast}\right\rangle $ is symmetric (i.e.
$H=H^{\ast})$, then;

\begin{enumerate}
\item $H\left(  a_{\hbar},a_{\hbar}^{\dag}\right)  $ is a symmetric operator
on $\mathcal{S}$ (i.e. $\left[  H\left(  a_{\hbar},a_{\hbar}^{\dag}\right)
\right]  ^{\dag}=H\left(  a_{\hbar},a_{\hbar}^{\dag}\right)  )$ for any
$\hbar>0 $ and

\item $H^{\mathrm{cl}}\left(  z\right)  :=H\left(  z,\overline{z}\right)  $ is
a real valued function on $\mathbb{C}.$
\end{enumerate}

Indeed,
\[
\left[  H\left(  a_{\hbar},a_{\hbar}^{\dag}\right)  \right]  ^{\dag}=H^{\ast
}\left(  a_{\hbar},a_{\hbar}^{\dag}\right)  =H\left(  a_{\hbar},a_{\hbar
}^{\dag}\right)
\]
and
\[
\overline{H^{\mathrm{cl}}\left(  \alpha\right)  }:=\overline{H\left(
\alpha,\bar{\alpha}\right)  }=H^{\ast}\left(  \alpha,\bar{\alpha}\right)
=H\left(  \alpha,\bar{\alpha}\right)  =H^{\mathrm{cl}}\left(  \alpha\right)
.
\]
The main point of this paper is to show under Assumption \ref{ass.1} on $H$
that classical Hamiltonian dynamics associated to $H^{\mathrm{cl}}$ determine
the limiting quantum mechanical dynamics determined by $H_{\hbar}%
:=\overline{H\left(  a_{\hbar},a_{\hbar}^{\dag}\right)  }.$
\end{remark}

We have analogous definitions and statements for the non-commutative algebra,
$\mathbb{C}\left\langle \theta_{1},\dots,\theta_{n},\theta_{1}^{\ast}%
,\dots,\theta_{n}^{\ast}\right\rangle ,$ of non-commuting polynomials in $2n$
-- indeterminants, $\Lambda_{n}=\left\{  \theta_{1},\dots,\theta_{n}%
,\theta_{1}^{\ast},\dots,\theta_{n}^{\ast}\right\}  ,$ as in Eq.
(\ref{equ.1.21}).

\begin{notation}
\label{not.2.20}Let $\mathbb{C}\left[  x\right]  \left\langle \theta
,\theta^{\ast}\right\rangle $ and $\mathbb{C}\left[  \alpha,\bar{\alpha
}\right]  \left\langle \theta,\theta^{\ast}\right\rangle $ denote the
non-commutative polynomials in $\left\{  \theta,\theta^{\ast}\right\}  $ with
coefficients in the commutative polynomial rings, $\mathbb{C}\left[  x\right]
$ and $\mathbb{C}\left[  \alpha,\bar{\alpha}\right]  $ respectively. For
$P\in\mathbb{C}\left[  x\right]  \left\langle \theta,\theta^{\ast
}\right\rangle $ or $P\in\mathbb{C}\left[  \alpha,\bar{\alpha}\right]
\left\langle \theta,\theta^{\ast}\right\rangle $ we will write $\deg_{\theta
}P$ to indicate that we are computing the degree relative to $\left\{
\theta,\theta^{\ast}\right\}  $ and not relative to $x$ or $\left\{
\alpha,\bar{\alpha}\right\}  .$
\end{notation}

For any $\alpha\in\mathbb{C}$ and $P\left(  \theta,\theta^{\ast}\right)
\in\mathbb{C}\left\langle \theta,\theta^{\ast}\right\rangle \ $with
$d=\deg_{\theta}P,$ let $\left\{  P_{k}\left(  \alpha:\theta,\theta^{\ast
}\right)  \right\}  _{k=0}^{d}\subset\mathbb{C}\left[  \alpha,\bar{\alpha
}\right]  \left\langle \theta,\theta^{\ast}\right\rangle $ denote the unique
homogeneous polynomials in $\mathbb{C}\left\langle \theta,\theta^{\ast
}\right\rangle $ with coefficients which are polynomials in $\alpha$ and
$\bar{\alpha}$ such that $\deg_{\theta}P_{k}\left(  \alpha:\theta,\theta
^{\ast}\right)  =k$ and
\begin{equation}
P\left(  \theta+\alpha,\theta^{\ast}+\bar{\alpha}\right)  =\sum_{k=0}^{d}%
P_{k}\left(  \alpha:\theta,\theta^{\ast}\right)  . \label{equ.2.29}%
\end{equation}

\begin{example}
\label{exa.2.21}If
\[
P\left(  \theta,\theta^{\ast}\right)  =\theta\theta^{\ast}\theta+\theta^{\ast
}\theta\theta^{\ast}%
\]
then
\begin{align*}
P\left(  \theta+\alpha,\theta^{\ast}+\bar{\alpha}\right)   &  =\left(
\theta+\alpha\right)  \left(  \theta^{\ast}+\bar{\alpha}\right)  \left(
\theta+\alpha\right)  +\left(  \theta^{\ast}+\bar{\alpha}\right)  \left(
\theta+\alpha\right)  \left(  \theta^{\ast}+\bar{\alpha}\right) \\
&  =P_{0}+P_{1}+P_{2}+P_{\geq3}%
\end{align*}
where
\begin{align*}
P_{0}\left(  \alpha,\theta,\theta^{\ast}\right)   &  =\alpha^{2}\bar{\alpha
}+\bar{\alpha}^{2}\alpha=P^{\mathrm{cl}}\left(  \alpha\right) \\
P_{1}\left(  \alpha,\theta,\theta^{\ast}\right)   &  =\left(  2\left\vert
\alpha\right\vert ^{2}+\overline{\alpha}^{2}\right)  \theta+\left(
2\left\vert \alpha\right\vert ^{2}+\alpha^{2}\right)  \theta^{\ast}\\
&  =\frac{\partial P^{\mathrm{cl}}}{\partial\alpha}\left(  \alpha\right)
\theta+\frac{\partial P^{\mathrm{cl}}}{\partial\bar{\alpha}}\left(
\alpha\right)  \theta^{\ast}\\
P_{2}\left(  \alpha,\theta,\theta^{\ast}\right)   &  =\overline{\alpha}%
\theta^{2}+\alpha\theta^{\ast2}+\left(  \alpha+\overline{\alpha}\right)
\theta^{\ast}\theta+\left(  \alpha+\overline{\alpha}\right)  \theta
\theta^{\ast}\\
&  =\frac{1}{2}\left(  \frac{\partial^{2}P^{\mathrm{cl}}}{\partial\alpha^{2}%
}\left(  \alpha\right)  \theta^{2}+\frac{\partial^{2}P^{\mathrm{cl}}}%
{\partial\bar{\alpha}^{2}}\left(  \alpha\right)  \theta^{\ast2}\right)
+\frac{d}{dt}|_{t=0}\frac{d}{ds}|_{s=0}P\left(  s\theta+\alpha,t\theta^{\ast
}+\bar{\alpha}\right) \\
P_{\geq3}\left(  \alpha,\theta,\theta^{\ast}\right)   &  =\theta\theta^{\ast
}\theta+\theta^{\ast}\theta\theta^{\ast}.
\end{align*}
This example is generalized in the following theorem.
\end{example}

\begin{theorem}
\label{the.2.22}Let $P\left(  \theta,\theta^{\ast}\right)  \in\mathbb{C}%
\left\langle \theta,\theta^{\ast}\right\rangle $ and $\alpha\in\mathbb{C},$
then
\begin{align}
P_{0}\left(  \alpha:\theta,\theta^{\ast}\right)   &  =P^{\mathrm{cl}}\left(
\alpha\right) \nonumber\\
P_{1}\left(  \alpha:\theta,\theta^{\ast}\right)   &  =\left[  \frac{\partial
P^{\mathrm{cl}}}{\partial\alpha}\left(  \alpha\right)  \theta+\frac{\partial
P^{\mathrm{cl}}}{\partial\bar{\alpha}}\left(  \alpha\right)  \theta^{\ast
}\right]  \text{ and}\nonumber\\
P_{2}\left(  \alpha:\theta,\theta^{\ast}\right)   &  =\frac{1}{2}\left(
\frac{\partial^{2}P^{\mathrm{cl}}}{\partial\alpha^{2}}\left(  \alpha\right)
\theta^{2}+\frac{\partial^{2}P^{\mathrm{cl}}}{\partial\bar{\alpha}^{2}}\left(
\alpha\right)  \theta^{\ast2}\right) \nonumber\\
&  +\frac{d}{dt}|_{t=0}\frac{d}{ds}|_{s=0}P\left(  s\theta+\alpha
,t\theta^{\ast}+\bar{\alpha}\right)  . \label{equ.2.30}%
\end{align}
where
\[
\frac{d}{dt}|_{t=0}\frac{d}{ds}|_{s=0}P\left(  s\theta+\alpha,t\theta^{\ast
}+\bar{\alpha}\right)  =\frac{\partial^{2}P^{\mathrm{cl}}}{\partial
\alpha\partial\bar{\alpha}}\left(  \alpha\right)  \theta^{\ast}\theta\text{
}\operatorname{mod}~\theta^{\ast}\theta=\theta\theta^{\ast}%
\]
for all $\alpha\in\mathbb{C}.$ So we have
\begin{align}
P  &  \left(  \theta+\alpha,\theta^{\ast}+\bar{\alpha}\right) \nonumber\\
&  =P^{\mathrm{cl}}\left(  \alpha\right)  +\left[  \frac{\partial
P^{\mathrm{cl}}}{\partial\alpha}\left(  \alpha\right)  \theta+\left(
\frac{\partial}{\partial\bar{\alpha}}P^{\mathrm{cl}}\right)  \left(
\alpha\right)  \theta^{\ast}\right]  +P_{2}\left(  \alpha:\theta,\theta^{\ast
}\right)  +P_{\geq3}\left(  \alpha:\theta,\theta^{\ast}\right)
\label{equ.2.31}%
\end{align}
where the remainder term, $P_{\geq3}$ is a sum of homogeneous terms of degree
$3$ or more. Moreover if $P=P^{\ast},$ then $P_{2}^{\ast}=P_{2}$ and
$P_{\geq3}^{\ast}=P_{\geq3}.$
\end{theorem}

\begin{proof}
If $p=\deg_{\theta}P,$ then
\[
P\left(  t\theta+\alpha,t\theta^{\ast}+\bar{\alpha}\right)  =\sum_{k=0}%
^{p}t^{k}P_{k}\left(  \alpha:\theta,\theta^{\ast}\right)  \text{ }\forall
~t\in\mathbb{R},
\]
and it follows (by Taylor's theorem) that
\begin{equation}
P_{k}\left(  \alpha:\theta,\theta^{\ast}\right)  =\frac{1}{k!}\left(  \frac
{d}{dt}\right)  ^{k}|_{t=0}P\left(  t\theta+\alpha,t\theta^{\ast}+\bar{\alpha
}\right)  . \label{equ.2.32}%
\end{equation}
From Eq. (\ref{equ.2.32}),
\begin{align*}
P_{0}\left(  \alpha:\theta,\theta^{\ast}\right)   &  =P\left(  \alpha
,\bar{\alpha}\right)  =P^{\mathrm{cl}}\left(  \alpha\right)  \text{ and}\\
P_{1}\left(  \alpha:\theta,\theta^{\ast}\right)   &  =\frac{d}{dt}%
|_{t=0}P\left(  t\theta+\alpha,t\theta^{\ast}+\bar{\alpha}\right) \\
&  =\frac{d}{dt}|_{t=0}P\left(  t\theta+\alpha,\bar{\alpha}\right)  +\frac
{d}{dt}|_{t=0}P\left(  \alpha,t\theta^{\ast}+\bar{\alpha}\right) \\
&  =\frac{\partial P^{\mathrm{cl}}}{\partial\alpha}\left(  \alpha\right)
\theta+\frac{\partial P^{\mathrm{cl}}}{\partial\bar{\alpha}}\left(
\alpha\right)  \theta^{\ast}.
\end{align*}
Similarly from Eq. (\ref{equ.2.32}),
\begin{align*}
P_{2}\left(  \alpha:\theta,\theta^{\ast}\right)  =  &  \frac{1}{2}\left(
\frac{d}{dt}\right)  ^{2}|_{t=0}P\left(  t\theta+\alpha,t\theta^{\ast}%
+\bar{\alpha}\right) \\
=  &  \frac{1}{2}\left(  \frac{d}{dt}\right)  ^{2}|_{t=0}\left[  P\left(
t\theta+\alpha,\bar{\alpha}\right)  +P\left(  \alpha,t\theta^{\ast}%
+\bar{\alpha}\right)  \right] \\
&  +\frac{d}{dt}|_{t=0}\frac{d}{ds}|_{s=0}P\left(  s\theta+\alpha
,t\theta^{\ast}+\bar{\alpha}\right) \\
=  &  \frac{1}{2}\left(  \frac{\partial^{2}P^{\mathrm{cl}}}{\partial\alpha
^{2}}\left(  \alpha\right)  \theta^{2}+\frac{\partial^{2}P^{\mathrm{cl}}%
}{\partial\bar{\alpha}^{2}}\left(  \alpha\right)  \theta^{\ast2}\right) \\
&  +\frac{d}{dt}|_{t=0}\frac{d}{ds}|_{s=0}P\left(  s\theta+\alpha
,t\theta^{\ast}+\bar{\alpha}\right)  .
\end{align*}

If $P\left(  \theta,\theta^{\ast}\right)  \in\mathbb{C}\left\langle
\theta,\theta^{\ast}\right\rangle $ is symmetric, then $P\left(
t\theta+\alpha,t\theta^{\ast}+\bar{\alpha}\right)  \in\mathbb{C}\left\langle
\theta,\theta^{\ast}\right\rangle $ is symmetric and hence from Eq.
(\ref{equ.2.32}) it follows that $P_{k}\left(  \alpha:\theta,\theta^{\ast
}\right)  \in\mathbb{C}\left\langle \theta,\theta^{\ast}\right\rangle $ is
still symmetric and therefore so is the remainder term,
\[
P_{\geq3}\left(  \alpha:\theta,\theta^{\ast}\right)  =\sum_{k=3}^{p}%
P_{k}\left(  \alpha:\theta,\theta^{\ast}\right)  .
\]

\end{proof}

\section{Polynomial Operators\label{sec.3}}

\subsection{Algebra of Polynomial Operators\label{sub.3.1}}

\begin{notation}
\label{not.3.1}For $\mathbf{b}=\left(  b_{1},\dots,b_{k}\right)  \in\left\{
\theta,\theta^{\ast}\right\}  ^{k},$ $p\left(  \mathbf{b}\right)  ,$ $q\left(
\mathbf{b}\right)  ,$ and $\ell\left(  \mathbf{b}\right)  $ be the
$\mathbb{Z}$ -- valued functions defined by
\begin{align}
p\left(  \mathbf{b}\right)   &  :=\#\left\{  i:b_{i}=\theta\right\}  ,\text{
}q\left(  \mathbf{b}\right)  :=\#\left\{  i:b_{i}=\theta^{\ast}\right\}
,\text{ and}\label{equ.3.1}\\
\ell\left(  \mathbf{b}\right)   &  :=\sum_{i=1}^{k}\left(  1_{b_{i}%
=\theta^{\ast}}-1_{b_{i}=\theta}\right)  =q\left(  \mathbf{b}\right)
-p\left(  \mathbf{b}\right)  . \label{equ.3.2}%
\end{align}
Thus $p\left(  \mathbf{b}\right)  $ $\left(  q\left(  \mathbf{b}\right)
\right)  $ is the number of $\theta$'s $\left(  \theta^{\ast}\text{'s}\right)
$ in $\mathbf{b}$ and $\ell\left(  \mathbf{b}\right)  $ counts the excess
number of $\theta^{\ast}$'s over $\theta$'s in $\mathbf{b.}$
\end{notation}

\begin{lemma}
[Normal Ordering]\label{lem.3.2}If $P\left(  \theta,\theta^{\ast}\right)  \in$
$\mathbb{C}\left\langle \theta,\theta^{\ast}\right\rangle $ with
$d=\deg_{\theta}P,$ then there exists $R\left(  \hbar:\theta,\theta^{\ast
}\right)  \in\mathbb{C}\left[  \hbar\right]  \left\langle \theta,\theta^{\ast
}\right\rangle $ (a non-commutative polynomial in $\left\{  \theta
,\theta^{\ast}\right\}  $ with polynomial coefficients in $\hbar)$ such that
$\deg_{\theta}R\left(  \hbar:\theta,\theta^{\ast}\right)  \leq d-2$ and
\[
P\left(  a_{\hbar},a_{\hbar}^{\dag}\right)  =\sum_{0\leq k,l;~k+l\leq d}%
\frac{1}{k!\cdot l!}\left(  \frac{\partial^{k+l}P^{\mathrm{cl}}}%
{\partial\overline{\alpha}^{k}\partial\alpha^{l}}\right)  \left(  0\right)
a_{\hbar}^{\dag k}a_{\hbar}^{l}+\hbar R\left(  \hbar:a_{\hbar},a_{\hbar}%
^{\dag}\right)  ~\forall~\hbar>0.
\]

\end{lemma}

\begin{proof}
By linearity it suffices to consider the case here $P\left(  \theta
,\theta^{\ast}\right)  $ is a homogeneous polynomial of degree $d$ which may
be written as
\begin{equation}
P\left(  \theta,\theta^{\ast}\right)  =\sum_{\mathbf{b}\in\left\{
\theta,\theta^{\ast}\right\}  ^{d}}c\left(  \mathbf{b}\right)  u_{\mathbf{b}%
}\left(  \theta,\theta^{\ast}\right)  =\sum_{p=0}^{d}\sum_{\mathbf{b}%
\in\left\{  \theta,\theta^{\ast}\right\}  ^{d}}1_{p\left(  \mathbf{b}\right)
=p}c\left(  \mathbf{b}\right)  u_{\mathbf{b}}\left(  \theta,\theta^{\ast
}\right)  . \label{equ.3.3}%
\end{equation}
Since
\[
P\left(  \alpha,\bar{\alpha}\right)  =\sum_{p=0}^{d}\left[  \sum
_{\mathbf{b}\in\left\{  \theta,\theta^{\ast}\right\}  ^{d}}1_{p\left(
\mathbf{b}\right)  =p}c\left(  \mathbf{b}\right)  \right]  \alpha^{p}%
\bar{\alpha}^{d-p}%
\]
it follows that
\[
\frac{1}{\left(  d-p\right)  !\cdot p!}\left(  \frac{\partial^{d}P^{cp}%
}{\partial\overline{\alpha}^{d-p}\partial\alpha^{p}}\right)  \left(  0\right)
=\sum_{\mathbf{b}\in\left\{  \theta,\theta^{\ast}\right\}  ^{d}}1_{p\left(
\mathbf{b}\right)  =p}c\left(  \mathbf{b}\right)  .
\]
On the other hand, if $\mathbf{b}\in\left\{  \theta,\theta^{\ast}\right\}
^{d} $ and $p:=p\left(  \mathbf{b}\right)  ,$ then making use of the CCRs of
Eq. (\ref{equ.1.7}) it is easy to show there exists $R_{\mathbf{b}}\left(
\hbar,\theta,\theta^{\ast}\right)  \in\mathbb{C}\left[  \hbar\right]
\left\langle \theta,\theta^{\ast}\right\rangle $ such that $\deg_{\theta
}R_{\mathbf{b}}\left(  \hbar,\theta,\theta^{\ast}\right)  \leq d-2$ such that
\begin{equation}
u_{\mathbf{b}}\left(  a_{\hbar},a_{\hbar}^{\dag}\right)  =a_{\hbar}%
^{\dag\left(  d-p\right)  }a_{\hbar}^{p}+\hbar R_{\mathbf{b}}\left(
\hbar,a_{\hbar},a_{\hbar}^{\dag}\right)  . \label{equ.3.4}%
\end{equation}
Replacing $\theta$ by $a_{\hbar}$ and $\theta^{\ast}$ by $a_{\hbar}^{\dag}$ in
Eq. (\ref{equ.3.3}) and using Eq. (\ref{equ.3.4}) we find,
\begin{align*}
P\left(  a_{\hbar},a_{\hbar}^{\dag}\right)   &  =\sum_{p=0}^{d}\sum
_{\mathbf{b}\in\left\{  \theta,\theta^{\ast}\right\}  ^{d}}1_{p\left(
\mathbf{b}\right)  =p}c\left(  \mathbf{b}\right)  u_{\mathbf{b}}\left(
a_{\hbar},a_{\hbar}^{\dag}\right) \\
&  =\sum_{p=0}^{d}\left[  \sum_{\mathbf{b}\in\left\{  \theta,\theta^{\ast
}\right\}  ^{d}}1_{p\left(  \mathbf{b}\right)  =p}c\left(  \mathbf{b}\right)
\right]  a_{\hbar}^{\dag\left(  d-p\right)  }a_{\hbar}^{p}+\hbar
\sum_{\mathbf{b}\in\left\{  \theta,\theta^{\ast}\right\}  ^{d}}c\left(
\mathbf{b}\right)  R_{\mathbf{b}}\left(  \hbar,a_{\hbar},a_{\hbar}^{\dag
}\right) \\
&  =\sum_{p=0}^{d}\frac{1}{\left(  d-p\right)  !\cdot p!}\left(
\frac{\partial^{d}P^{\mathrm{cl}}}{\partial\overline{\alpha}^{d-p}%
\partial\alpha^{p}}\right)  \left(  0\right)  a_{\hbar}^{\dag\left(
d-p\right)  }a_{\hbar}^{p}+\hbar R\left(  \hbar,a_{\hbar},a_{\hbar}^{\dag
}\right)
\end{align*}
where
\[
R\left(  \hbar,\theta,\theta^{\ast}\right)  =\sum_{\mathbf{b}\in\left\{
\theta,\theta^{\ast}\right\}  ^{d}}c\left(  \mathbf{b}\right)  R_{\mathbf{b}%
}\left(  \hbar,\theta,\theta^{\ast}\right)  .
\]

\end{proof}

\begin{corollary}
\label{cor.3.3} If $P\left(  \theta,\theta^{\ast}\right)  $ and $Q\left(
\theta,\theta^{\ast}\right)  $ are non-commutative polynomials such that
$P^{\mathrm{cl}}=Q^{\mathrm{cl}},$ then there exists $R\left(  \hbar
:\theta,\theta^{\ast}\right)  \in\mathbb{C}\left[  \hbar\right]  \left\langle
\theta,\theta^{\ast}\right\rangle $ with $\deg_{\theta}R\left(  \hbar
:\theta,\theta^{\ast}\right)  \leq\deg_{\theta}\left(  P-Q\right)  \left(
\theta,\theta^{\ast}\right)  -2$ such that
\[
P\left(  a_{\hbar},a_{\hbar}^{\dag}\right)  =Q\left(  a_{\hbar},a_{\hbar
}^{\dag}\right)  +\hbar R\left(  \hbar,a_{\hbar},a_{\hbar}^{\dag}\right)  .
\]

\end{corollary}

\begin{proof}
Apply Lemma \ref{lem.3.2} to the non-commutative polynomial, $P\left(
\theta,\theta^{\ast}\right)  -Q\left(  \theta,\theta^{\ast}\right)  .$
\end{proof}

\begin{proposition}
\label{pro.3.5}For all $H\in\mathbb{C}\left\langle \theta,\theta^{\ast
}\right\rangle ,$ there exists a polynomial, $p_{H}\in\mathbb{C}\left[
z,\bar{z}\right]  $ such that
\begin{align*}
H_{2}  &  \left(  \alpha:a,a^{\dag}\right) \\
&  =\frac{1}{2}\left(  \frac{\partial^{2}H^{\mathrm{cl}}}{\partial\alpha^{2}%
}\right)  \left(  \alpha\right)  a^{2}+\frac{1}{2}\left(  \frac{\partial
^{2}H^{\mathrm{cl}}}{\partial\overline{\alpha}^{2}}\right)  \left(
\alpha\right)  a^{\dag2}+\left(  \frac{\partial^{2}H^{\mathrm{cl}}}%
{\partial\alpha\partial\overline{\alpha}}\right)  \left(  \alpha\right)
a^{\dag}a+p_{H}\left(  \alpha,\bar{\alpha}\right)  I
\end{align*}
for all $\alpha\in\mathbb{C}$ where $H_{2}\left(  \alpha:\theta,\theta^{\ast
}\right)  $ is defined in Eq. (\ref{equ.2.29}).
\end{proposition}

\begin{proof}
As we have seen the structure of $H_{2}\left(  \alpha:\theta,\theta^{\ast
}\right)  $ implies there exists $\rho,\gamma,\delta\in\mathbb{C}\left[
\alpha,\bar{\alpha}\right]  $ such that
\[
2H_{2}\left(  \alpha:\theta,\theta^{\ast}\right)  =\rho\left(  \alpha
,\bar{\alpha}\right)  \theta^{2}+\overline{\rho\left(  \alpha,\bar{\alpha
}\right)  }\theta^{\ast2}+\gamma\left(  \alpha,\bar{\alpha}\right)
\theta^{\ast}\theta+\delta\left(  \alpha,\bar{\alpha}\right)  \theta
\theta^{\ast}.
\]
From this equation we find,
\[
2H_{2}\left(  \alpha:z,\bar{z}\right)  =\rho\left(  \alpha,\bar{\alpha
}\right)  z^{2}+\overline{\rho\left(  \alpha,\bar{\alpha}\right)  }\bar{z}%
^{2}+\left[  \gamma\left(  \alpha,\bar{\alpha}\right)  +\delta\left(
\alpha,\bar{\alpha}\right)  \right]  z\bar{z}%
\]
while form Eq. (\ref{equ.2.30}) we may conclude that
\begin{equation}
2H_{2}\left(  \alpha:z,\bar{z}\right)  =\left(  \frac{\partial^{2}%
H^{\mathrm{cl}}}{\partial\alpha^{2}}\right)  \left(  \alpha\right)
z^{2}+\left(  \frac{\partial^{2}H^{\mathrm{cl}}}{\partial\overline{\alpha}%
^{2}}\right)  \left(  \alpha\right)  \bar{z}^{2}+2\left(  \frac{\partial
^{2}H^{\mathrm{cl}}}{\partial\alpha\partial\overline{\alpha}}\right)  \left(
\alpha\right)  \bar{z}z. \label{equ.3.5}%
\end{equation}
Comparing these last two equations shows,
\begin{align*}
\left(  \frac{\partial^{2}H^{\mathrm{cl}}}{\partial\alpha^{2}}\right)  \left(
\alpha\right)   &  =\rho\left(  \alpha,\bar{\alpha}\right)  ,\text{ }\left(
\frac{\partial^{2}H^{\mathrm{cl}}}{\partial\overline{\alpha}^{2}}\right)
\left(  \alpha\right)  =\overline{\rho\left(  \alpha,\bar{\alpha}\right)
},\text{ and}\\
\left(  \frac{\partial^{2}H^{\mathrm{cl}}}{\partial\alpha\partial
\overline{\alpha}}\right)  \left(  \alpha\right)   &  =\frac{1}{2}\left[
\gamma\left(  \alpha,\bar{\alpha}\right)  +\delta\left(  \alpha,\bar{\alpha
}\right)  \right]  .
\end{align*}
Using these last identities and the canonical commutations relations we find,
\begin{align*}
2  &  H_{2}\left(  \alpha:a,a^{\dag}\right) \\
&  =\rho\left(  \alpha,\bar{\alpha}\right)  a^{2}+\overline{\rho\left(
\alpha,\bar{\alpha}\right)  }a^{\dag2}+\gamma\left(  \alpha,\bar{\alpha
}\right)  a^{\dag}a+\delta\left(  \alpha,\bar{\alpha}\right)  aa^{\dag}\\
&  =\rho\left(  \alpha,\bar{\alpha}\right)  a^{2}+\overline{\rho\left(
\alpha,\bar{\alpha}\right)  }a^{\dag2}+\left[  \gamma\left(  \alpha
,\bar{\alpha}\right)  +\delta\left(  \alpha,\bar{\alpha}\right)  \right]
a^{\dag}a+\delta\left(  \alpha,\bar{\alpha}\right)  I\\
&  =\left(  \frac{\partial^{2}H^{\mathrm{cl}}}{\partial\alpha^{2}}\right)
\left(  \alpha\right)  a^{2}+\left(  \frac{\partial^{2}H^{\mathrm{cl}}%
}{\partial\overline{\alpha}^{2}}\right)  \left(  \alpha\right)  a^{\dag
2}+2\left(  \frac{\partial^{2}H^{\mathrm{cl}}}{\partial\alpha\partial
\overline{\alpha}}\right)  \left(  \alpha\right)  a^{\dag}a+p_{H}\left(
\alpha,\bar{\alpha}\right)  I
\end{align*}
with $p_{H}\left(  \alpha,\bar{\alpha}\right)  =\delta\left(  \alpha
,\bar{\alpha}\right)  .$
\end{proof}

Proposition \ref{pro.3.5} and the following simple commutator formulas,
\begin{align*}
\left[  a^{\dag}a,a\right]   &  =-a,\quad\left[  a^{\dag2},a\right]
=-2a^{\dag},\quad\\
\left[  a^{\dag}a,a^{\dag}\right]   &  =a^{\dag},\text{ and }\left[
a^{2},a^{\dag}\right]  =2a,
\end{align*}
immediately give the following corollary.

\begin{corollary}
\label{cor.3.8}If $H\in\mathbb{C}\left\langle \theta,\theta^{\ast
}\right\rangle $ and $\alpha\in\mathbb{C},$ then
\begin{align*}
\left[  H_{2}\left(  \alpha:a,a^{\dag}\right)  ,a\right]   &  =-\left(
\frac{\partial^{2}H^{\mathrm{cl}}}{\partial\alpha\partial\overline{\alpha}%
}\right)  \left(  \alpha\right)  a-\left(  \frac{\partial^{2}H^{\mathrm{cl}}%
}{\partial\overline{\alpha}^{2}}\right)  \left(  \alpha\right)  a^{\dag}\\
\left[  H_{2}\left(  \alpha:a,a^{\dag}\right)  ,a^{\dag}\right]   &  =\left(
\frac{\partial^{2}H^{\mathrm{cl}}}{\partial\alpha^{2}}\right)  \left(
\alpha\right)  a+\left(  \frac{\partial^{2}H^{\mathrm{cl}}}{\partial
\alpha\partial\overline{\alpha}}\right)  \left(  \alpha\right)  a^{\dag}.
\end{align*}

\end{corollary}

\subsection{Expectations and variances for translated states\label{sub.3.2}}

The next result is a fairly easy consequence of Proposition \ref{pro.2.6} and
the expansion of non-commutative polynomials into their homogeneous components.

\begin{corollary}
[Concentrated states]\label{cor.3.10}Let $P\left(  \theta,\theta^{\ast
}\right)  \in\mathbb{C}\left\langle \theta,\theta^{\ast}\right\rangle ,$
$\psi\in\mathcal{S},$ $\hbar>0,$ and $\alpha\in\mathbb{C},$ then
\begin{align}
\left\langle P\left(  a_{\hbar},a_{\hbar}^{\dag}\right)  \right\rangle
_{U_{\hbar}\left(  \alpha\right)  \psi}  &  =P\left(  \alpha,\bar{\alpha
}\right)  +O\left(  \sqrt{\hbar}\right) \label{equ.3.6}\\
\operatorname*{Var}\nolimits_{U_{\hbar}\left(  \alpha\right)  \psi}\left(
P\left(  a_{\hbar},a_{\hbar}^{\dag}\right)  \right)   &  =O\left(  \sqrt
{\hbar}\right)  , \label{equ.3.7}%
\end{align}
and
\begin{equation}
\lim_{\hbar\downarrow0}\left\langle P\left(  \frac{a_{\hbar}-\alpha}%
{\sqrt{\hbar}},\frac{a_{\hbar}^{\dag}-\bar{\alpha}}{\sqrt{\hbar}}\right)
\right\rangle _{U_{\hbar}\left(  \alpha\right)  \psi}=\left\langle P\left(
a,a^{\dag}\right)  \right\rangle _{\psi} \label{equ.3.8}%
\end{equation}
where $\left\langle \cdot\right\rangle _{U_{\hbar}\left(  \alpha\right)  \psi}
$ is defined in Definition \ref{def.1.8}. [In fact, the equality in the last
equation holds before taking the limit as $\hbar\rightarrow0.]$
\end{corollary}

\begin{proof}
From Proposition \ref{pro.2.6} and Eq. (\ref{equ.2.29}),
\begin{equation}
U_{\hbar}\left(  \alpha\right)  ^{\ast}P\left(  a_{\hbar},a_{\hbar}^{\dag
}\right)  U_{\hbar}\left(  \alpha\right)  =P\left(  a_{\hbar}+\alpha,a_{\hbar
}^{\dag}+\bar{\alpha}\right)  =\sum_{k=0}^{d}P_{k}\left(  \alpha:a_{\hbar
},a_{\hbar}^{\dag}\right)  \label{equ.3.9}%
\end{equation}
and hence
\begin{align*}
\left\langle P\left(  a_{\hbar},a_{\hbar}^{\dag}\right)  \right\rangle
_{U_{\hbar}\left(  \alpha\right)  \psi}  &  =\left\langle U_{\hbar}\left(
\alpha\right)  ^{\ast}P\left(  a_{\hbar},a_{\hbar}^{\dag}\right)  U_{\hbar
}\left(  \alpha\right)  \right\rangle _{\psi}=\left\langle P\left(  a_{\hbar
}+\alpha,a_{\hbar}^{\dag}+\bar{\alpha}\right)  \right\rangle _{\psi}\\
&  =\left\langle \sum_{k=0}^{d}P_{k}\left(  \alpha:a_{\hbar},a_{\hbar}^{\dag
}\right)  \right\rangle _{\psi}=P_{0}\left(  \alpha\right)  +\sum_{k=1}%
^{d}\hbar^{k/2}\left\langle P_{k}\left(  \alpha:a,a^{\dag}\right)
\right\rangle _{\psi}%
\end{align*}
from which Eq. (\ref{equ.3.6}) follows where $P_{0}\left(  \alpha\right)  $ is
defined in Notation \ref{not.2.12}. Similarly, making use of the fact that
$\left(  P^{2}\right)  _{0}\left(  \alpha\right)  =\left(  P_{0}^{2}\right)
\left(  \alpha\right)  $
\begin{equation}
\left\langle P^{2}\left(  a_{\hbar},a_{\hbar}^{\dag}\right)  \right\rangle
_{U_{\hbar}\left(  \alpha\right)  \psi}=\left(  P_{0}^{2}\right)  \left(
\alpha\right)  +\sum_{k=1}^{2d}\hbar^{k/2}\left\langle \left(  P^{2}\right)
_{k}\left(  \alpha:a,a^{\dag}\right)  \right\rangle _{\psi} \label{equ.3.10}%
\end{equation}
and hence
\begin{align*}
\operatorname*{Var}\nolimits_{U_{\hbar}\left(  \alpha\right)  \psi}\left(
P\left(  a_{\hbar},a_{\hbar}^{\dag}\right)  \right)   &  =\left(  P_{0}%
^{2}\right)  \left(  \alpha\right)  +\sum_{k=1}^{2d}\hbar^{k/2}\left\langle
\left(  P^{2}\right)  _{k}\left(  \alpha:a,a^{\dag}\right)  \right\rangle
_{\psi}\\
&  -\left[  \left(  P_{0}\left(  \alpha\right)  +\sum_{k=0}^{d}\hbar
^{k/2}\left\langle P_{k}\left(  \alpha:a,a^{\dag}\right)  \right\rangle
_{\psi}\right)  \right]  ^{2}\\
&  =O\left(  \sqrt{\hbar}\right)  .
\end{align*}
Lastly, using Eq. (\ref{equ.3.9}) one shows,
\[
\left\langle P\left(  \frac{a_{\hbar}-\alpha}{\sqrt{\hbar}},\frac{a_{\hbar
}^{\dag}-\bar{\alpha}}{\sqrt{\hbar}}\right)  \right\rangle _{U_{\hbar}\left(
\alpha\right)  \psi}=\left\langle P\left(  \frac{a_{\hbar}+\alpha-\alpha
}{\sqrt{\hbar}},\frac{a_{\hbar}^{\dag}+\bar{\alpha}-\bar{\alpha}}{\sqrt{\hbar
}}\right)  \right\rangle _{\psi}=\left\langle P\left(  a,a^{\dag}\right)
\right\rangle _{\psi}%
\]
which certainly implies Eq. (\ref{equ.3.8}).
\end{proof}

\begin{remark}
\label{rem.3.11}If $\psi\in\mathcal{S}$ and $\alpha\in\mathbb{C},$ Eqs.
(\ref{equ.3.6}) and (\ref{equ.3.7}) should be interpreted to say that for
small $\hbar>0,$ $U_{\hbar}\left(  \alpha\right)  \psi$ is a state which is
concentrated in phase space near $\alpha.$ Consequently, these are good
initial states for discussing the classical ($\hbar\rightarrow0)$ limit of
quantum mechanics.
\end{remark}

The next result shows that, under Assumption \ref{ass.1}, the classical
equations of motions in Eq. (\ref{equ.1.1}) have global solutions which remain
bounded in time.

\begin{proposition}
\label{pro.3.13}If $C$ and $C_{1}$ are the constants appearing in Eq.
(\ref{equ.1.14}) of Assumption \ref{ass.1}, $\alpha_{0}\in\mathbb{C},$ and
$\alpha\left(  t\right)  \in\mathbb{C}$ is the maximal solution of Hamilton's
ordinary differential equations (\ref{equ.1.1}), then $\alpha\left(  t\right)
$ is defined for all time $t$ and moreover,
\begin{equation}
\left\vert \alpha\left(  t\right)  \right\vert ^{2}\leq C_{1}\left(
H^{\text{cl}}\left(  \alpha\left(  0\right)  \right)  +C\right)  ,
\label{equ.3.11}%
\end{equation}
where $H^{\text{cl}}\left(  \alpha\right)  :=H\left(  \alpha,\bar{\alpha
}\right)  .$
\end{proposition}

\begin{proof}
Equation (\ref{equ.1.14}) with $\beta=1$ implies
\begin{equation}
\left\langle \mathcal{N}_{\hbar}\right\rangle _{\psi}\leq C_{1}\left\langle
H_{\hbar}+C\right\rangle _{\psi}\text{ for all }\psi\in\mathcal{S}.
\label{equ.3.12}%
\end{equation}
Replacing $\psi$ by $U_{\hbar}\left(  \alpha\right)  \psi$ in Eq.
(\ref{equ.3.12}) and then letting $\hbar\downarrow0$ gives (with the aid of
Corollary \ref{cor.3.10}) the estimate,
\begin{equation}
\left\vert \alpha\right\vert ^{2}\leq C_{1}\left(  H^{\text{cl}}\left(
\alpha\right)  +C\right)  \text{ for all }\alpha\in\mathbb{C}.
\label{equ.3.13}%
\end{equation}
If $\alpha\left(  t\right)  $ solves Hamilton's Eq. (\ref{equ.1.1}) then
$H^{\text{cl}}\left(  \alpha\left(  t\right)  \right)  =H^{\text{cl}}\left(
\alpha\left(  0\right)  \right)  $ for all $t.$ As the level sets of
$H^{\text{cl}}$ are compact because of the estimate in Eq. (\ref{equ.3.13})
there is no possibility for $\alpha\left(  t\right)  $ to explode and hence
solutions will exist for all times $t$ and moreover must satisfy the estimate
in Eq. (\ref{equ.3.11}).
\end{proof}

\subsection{Analysis of Monomial Operators of $a$ and $a^{\dag}$
\label{sub.3.3}}

In this subsection, recall that $a=a_{1}$ and $a^{\dag}=a_{1}^{\dag}$ as in
Definition \ref{def.1.3}. Let
\begin{equation}
\Omega_{0}\left(  x\right)  :=\frac{1}{\sqrt[4]{4\pi}}\exp\left(  -\frac{1}%
{2}x^{2}\right)  \text{ and }\left\{  \Omega_{n}:=\frac{1}{\sqrt{n!}}a^{\dag
n}\Omega_{0}\right\}  _{n=0}^{\infty}. \label{equ.3.14}%
\end{equation}
\textbf{Convention: }$\Omega_{n}\equiv0$ for all $n\in\mathbb{Z}$ with $n<0.$

The following theorem summarizes the basic well known and easily verified
properties of these functions which essentially are all easy consequences of
the canonical commutation relations, $\left[  a,a^{\dag}\right]  =I$ on
$\mathcal{S}.$ We will provide a short proof of these well known results for
the readers convenience.

\begin{theorem}
\label{the.3.14}The functions $\left\{  \Omega_{n}\right\}  _{n=0}^{\infty
}\subset\mathcal{S}$ form an orthonormal basis for $L^{2}\left(  m\right)  $
which satisfy for all $n\in\mathbb{N}_{0},$
\begin{align}
a\Omega_{n}  &  =\sqrt{n}\Omega_{n-1},\text{ }\label{equ.3.15}\\
a^{\dag}\Omega_{n}  &  =\sqrt{n+1}\Omega_{n+1}\text{ and}\label{equ.3.16}\\
a^{\dag}a\Omega_{n}  &  =n\Omega_{n}. \label{equ.3.17}%
\end{align}

\end{theorem}

\begin{proof}
First observe that $\Omega_{n}\left(  x\right)  $ is a polynomial $\left(
p_{n}\left(  x\right)  \right)  $ of degree $n$ times $\Omega_{0}\left(
x\right)  .$ Therefore the span of $\left\{  \Omega_{n}\right\}
_{n=0}^{\infty} $ are all functions of the form $p\left(  x\right)  \Omega
_{0}\left(  x\right)  $ where $p\in\mathbb{C}\left[  x\right]  .$ As
$\mathbb{C}\left[  x\right]  $ is dense in $L^{2}\left(  \Omega_{0}^{2}\left(
x\right)  dx\right)  $ it follows that $\left\{  \Omega_{n}\right\}
_{n=0}^{\infty}$ is total in $L^{2}\left(  m\right)  .$

For the remaining assertions let us recall, if $A$ and $B$ are operators on
some vector space (like $\mathcal{S})$ and $ad_{A}B:=\left[  A,B\right]  ,$
then $ad_{A}$ acts as a derivation, i.e.
\begin{equation}
ad_{A}\left(  BC\right)  =\left(  ad_{A}B\right)  C+B\left(  ad_{A}C\right)  .
\label{equ.3.18}%
\end{equation}
Combining this observation with $ad_{a}a^{\dag}=I$ then shows $ad_{a}a^{\dag
n}=na^{\dag n-1}$ so that
\[
a\Omega_{n}=a\frac{1}{\sqrt{n!}}a^{\dag n}\Omega_{0}=\frac{1}{\sqrt{n!}%
}\left(  ad_{a}a^{\dag n}\right)  \Omega_{0}=\frac{n}{\sqrt{n!}}a^{\dag\left(
n-1\right)  }\Omega_{0}=\sqrt{n}\Omega_{n-1}%
\]
which proves Eq. (\ref{equ.3.15}). Equation (\ref{equ.3.16}) is obvious from
the definition of $\left\{  \Omega_{n}\right\}  _{n=0}^{\infty}$ and Eq.
(\ref{equ.3.17}) follows from Eqs. (\ref{equ.3.15}) and (\ref{equ.3.16}). As
$\left\{  \Omega_{n}\right\}  _{n=0}^{\infty}$ are eigenvectors of the
symmetric operator $a^{\dag}a$ with distinct eigenvalues it follows that
$\left\langle \Omega_{n},\Omega_{m}\right\rangle =0$ if $m\neq n.$ So it only
remains to show $\left\Vert \Omega_{n}\right\Vert ^{2}=1$ for all $n.$
However, taking the $L^{2}\left(  m\right)  $ -norm of Eq. (\ref{equ.3.16})
gives
\begin{align*}
\left(  n+1\right)  \left\Vert \Omega_{n+1}\right\Vert ^{2}  &  =\left\Vert
a^{\dag}\Omega_{n}\right\Vert ^{2}=\left\langle \Omega_{n},aa^{\dag}\Omega
_{n}\right\rangle =\left\langle \Omega_{n},\left(  a^{\dag}a+I\right)
\Omega_{n}\right\rangle \\
&  =\left(  n+1\right)  \left\Vert \Omega_{n}\right\Vert ^{2},
\end{align*}
i.e. $n\rightarrow\left\Vert \Omega_{n}\right\Vert ^{2}$ is constant in $n.$
As we normalized $\Omega_{0}$ to be a unit vector, the proof is complete.
\end{proof}

\begin{notation}
\label{not.3.16}For $N\in\mathbb{N}_{0},$ let $\mathcal{P}_{N}$ denote
orthogonal projection of $L^{2}\left(  m\right)  $ onto $\operatorname*{span}%
\left\{  \Omega_{n}:0\leq n\leq N\right\}  ,$ i.e.
\begin{equation}
\mathcal{P}_{N}f:=\sum_{n=0}^{N}\left\langle f,\Omega_{n}\right\rangle
\Omega_{n}\text{ for all }f\in L^{2}\left(  m\right)  . \label{equ.3.19}%
\end{equation}

\end{notation}

\begin{notation}
[Standing Notation]\label{not.3.17}For the remainder of this section let
$k,j\in\mathbb{N},$ $\mathbf{b}=\left(  b_{1},\dots,b_{k}\right)  \in\left\{
\theta,\theta^{\ast}\right\}  ^{k},$ $q:=q\left(  \mathbf{b}\right)  ,$
$l:=\mathcal{\ell}\left(  \mathbf{b}\right)  ,$ $\mathbf{d}=\left(
d_{1},\dots,d_{j}\right)  \in\left\{  \theta,\theta^{\ast}\right\}  ^{j},$ and
$\mathcal{\ell}\left(  \mathbf{d}\right)  $ be as in Notation \ref{not.3.1}.
We further let $\mathcal{A}$ and $\mathcal{D}$ be the two monomial operators,
\begin{align*}
\mathcal{A}  &  :=u_{\mathbf{b}}\left(  a,a^{\dag}\right)  =\Xi\left(
b_{1}\right)  \dots\Xi\left(  b_{k}\right)  \text{ and }\\
\mathcal{D}  &  :=u_{\mathbf{d}}\left(  a,a^{\dag}\right)  =\Xi\left(
d_{1}\right)  \dots\Xi\left(  d_{j}\right)  .
\end{align*}

\end{notation}

\begin{lemma}
\label{lem.3.18}To each monomial operator $\mathcal{A}=u_{\mathbf{b}}\left(
a,a^{\dag}\right)  $ as in Notation \ref{not.3.17}, there exists
$c_{\mathcal{A}}:\mathbb{N}_{0}\rightarrow\lbrack0,\infty)$ such that
\begin{equation}
\mathcal{A}\Omega_{n}=c_{\mathcal{A}}\left(  n\right)  \cdot\Omega_{n+l}\text{
for all }n\in\mathbb{N}_{0} \label{equ.3.20}%
\end{equation}
where (as above) $\Omega_{m}:=0$ if $m<0.$ Moreover, $c_{\mathcal{A}}$
satisfies $c_{\mathcal{A}^{\dag}}\left(  n\right)  =c_{\mathcal{A}}\left(
n-l\right)  $ (where by convention $c_{\mathcal{A}}\left(  n\right)  \equiv0$
if $n<0),$
\begin{equation}
0\leq c_{\mathcal{A}}\left(  n\right)  \leq\left(  n+q\right)  ^{\frac{k}{2}%
}\text{ and }c_{\mathcal{A}}\left(  n\right)  \asymp n^{k/2}\text{ (i.e. }%
\lim_{n\rightarrow\infty}\frac{c_{\mathcal{A}}\left(  n\right)  }{n^{k/2}}=1).
\label{equ.3.21}%
\end{equation}

\end{lemma}

\begin{proof}
Since $a$ and $a^{\dag}$ shift $\Omega_{n}$ to its adjacent $\Omega_{n-1}$ and
$\Omega_{n+1}$ respectively from Theorem \ref{the.3.14}, it is easy to see
that Eq. (\ref{equ.3.20}) holds for some constants $c_{\mathcal{A}}\left(
n\right)  \in\mathbb{R}.$ Moreover a simple induction argument on $k$ shows
there exists $\delta_{i}\in\mathbb{Z}$ with $\delta_{i}\leq q$ such that
\begin{equation}
c_{\mathcal{A}}\left(  n\right)  =\left(  \sqrt{\prod_{i=1}^{k}\left(
n+\delta_{i}\right)  }\right)  \geq0. \label{equ.3.22}%
\end{equation}
The estimate and the limit statement in Eq. (\ref{equ.3.21}) now follows
directly from the Eq. (\ref{equ.3.22}).

Since $\mathcal{A}^{\dag}\Omega_{n}=c_{\mathcal{A}^{\dag}}\left(  n\right)
\Omega_{n-l},$ we find
\[
c_{\mathcal{A}^{\dag}}\left(  n\right)  =\left\langle \mathcal{A}^{\dag}%
\Omega_{n},\Omega_{n-l}\right\rangle =\left\langle \Omega_{n},\mathcal{A}%
\Omega_{n-l}\right\rangle =\left\langle \Omega_{n},c_{\mathcal{A}}\left(
n-l\right)  \Omega_{n}\right\rangle =c_{\mathcal{A}}\left(  n-l\right)  .
\]

\end{proof}

\begin{example}
\label{exa.3.20}Suppose that $p,q\in\mathbb{N}_{0},$ $k=p+q,$ $\ell=q-p,$ and
$\mathcal{A}=a^{p}a^{\dag q}.$ Then
\begin{align*}
\mathcal{A}\Omega_{n}  &  =a^{p}a^{\dag q}\Omega_{n}=a^{p}\sqrt{\prod
_{i=1}^{q}\left(  n+i\right)  }\cdot\Omega_{n+q}\\
&  =\sqrt{\prod_{i=1}^{q}\left(  n+i\right)  }\cdot a^{p}\Omega_{n+q}%
=\sqrt{\prod_{i=1}^{q}\left(  n+i\right)  }\sqrt{\prod_{j=0}^{p-1}\left(
n+q-j\right)  }\Omega_{n+\ell}%
\end{align*}
where
\begin{equation}
0\leq c_{\mathcal{A}}\left(  n\right)  =\sqrt{\prod_{i=1}^{q}\left(
n+i\right)  }\cdot\sqrt{\prod_{j=0}^{p-1}\left(  n+q-j\right)  }\leq\left(
n+q\right)  ^{\frac{k}{2}}. \label{equ.3.23}%
\end{equation}

\end{example}

\begin{definition}
\label{def.3.21}For $\beta\geq0,$ let
\[
D_{\beta}:=\left\{  f\in L^{2}\left(  \mathbb{R}\right)  :\sum_{n=0}^{\infty
}\left\vert \left\langle f,\Omega_{n}\right\rangle \right\vert ^{2}n^{2\beta
}<\infty\right\}  .
\]
[We will see shortly that $D_{\beta}=D\left(  \mathcal{N}^{\beta}\right)  ,$
see Example \ref{exa.3.28}.]
\end{definition}

\begin{theorem}
\label{the.3.22}Let $k=\deg_{\theta}u_{\mathbf{b}}\left(  \theta,\theta^{\ast
}\right)  ,$ $\mathcal{A}=u_{\mathbf{b}}\left(  a,a^{\dag}\right)  ,$
$l=\ell\left(  \mathbf{b}\right)  \in\mathbb{Z}$ be as in Notations
\ref{not.3.17} and \ref{not.3.1} and $c_{\mathcal{A}}\left(  n\right)  $ be
coefficients in Lemma \ref{lem.3.18}. Then $\mathcal{A}$ and $\mathcal{A}%
^{\dag}$ are closable operators satisfying;

\begin{enumerate}
\item $\mathcal{\bar{A}=A}^{\dag\ast}$ and $\overline{\mathcal{A}^{\dag}%
}=\mathcal{A}^{\ast}$ where we write $\mathcal{A}\mathbf{^{\dag}}^{\ast}$ for
$\left(  \mathcal{A}^{\dag}\right)  ^{\ast}.$

\item $D\left(  \mathcal{\bar{A}}\right)  =D_{k/2}=D\left(  \overline
{\mathcal{A}^{\dag}}\right)  $ and if $g\in D_{k/2},$ then
\begin{align}
\mathcal{A}^{\ast}g  &  =\sum_{n=0}^{\infty}\left\langle g,\Omega
_{n}\right\rangle \mathcal{A}^{\dag}\Omega_{n}=\sum_{n=0}^{\infty}\left\langle
g,\Omega_{n}\right\rangle c_{\mathcal{A}}\left(  n-l\right)  \Omega
_{n-l}\text{ and}\label{equ.3.24}\\
\mathcal{A}\mathbf{^{\dag}}^{\ast}g  &  =\mathcal{\bar{A}}g=\sum_{n=0}%
^{\infty}\left\langle g,\Omega_{n}\right\rangle \mathcal{A}\Omega_{n}%
=\sum_{n=0}^{\infty}\left\langle g,\Omega_{n}\right\rangle c_{\mathcal{A}%
}\left(  n\right)  \Omega_{n+l} \label{equ.3.25}%
\end{align}
with the conventions that $c_{\mathcal{A}}\left(  n\right)  $ and $\Omega
_{n}=0$ if $n<0.$

\item The subspace,
\begin{equation}
\mathcal{S}_{0}:=\operatorname{span}\left\{  \Omega_{n}\right\}
_{n=0}^{\infty}\subset\mathcal{S}\subset L^{2}\left(  m\right)
\label{equ.3.26}%
\end{equation}
is a core of both $\mathcal{\bar{A}}$ and $\overline{\mathcal{A}^{\dag}}.$
More explicitly if $g\in D_{k/2},$ then
\[
\mathcal{\bar{A}}g=\lim_{N\rightarrow\infty}\mathcal{AP}_{N}g\text{ and
}\overline{\mathcal{A}^{\dag}}g=\lim_{N\rightarrow\infty}\mathcal{A}^{\dag
}\mathcal{P}_{N}g\text{ }%
\]
where $\mathcal{P}_{N}$ is the orthogonal projection operator onto
$\operatorname*{span}\left\{  \Omega_{k}\right\}  _{k=0}^{n}$ as in Notation
\ref{not.3.16}.
\end{enumerate}
\end{theorem}

\begin{proof}
Since $\left\langle \mathcal{A}f,g\right\rangle =\left\langle f,\mathcal{A}%
^{\dag}g\right\rangle $ for all $f,g\in\mathcal{S}=D\left(  \mathcal{A}%
\right)  =D\left(  \mathcal{A}^{\dag}\right)  ,$ it follows that
$\mathcal{A}\subset\mathcal{A}\mathbf{^{\dag}}^{\ast}$ and $\mathcal{A}^{\dag
}\subset\mathcal{A}^{\ast}$ and therefore both $\mathcal{A}$ and
$\mathcal{A}^{\dag}$ are closable (see
\citep[Theorem VIII.1 on p.252]{Reed1980}) and
\begin{equation}
\overline{\mathcal{A}^{\dag}}\subset\mathcal{A}^{\ast}\text{ and
}\mathcal{\bar{A}}\subset\mathcal{A}^{\dag\ast}. \label{equ.3.27}%
\end{equation}
If $g\in D\left(  \mathcal{A}\mathbf{^{\ast}}\right)  \subset L^{2}\left(
m\right)  ,$then from Theorem \ref{the.3.14} and Lemma \ref{lem.3.18}, we
have
\begin{align}
\mathcal{A}^{\ast}g  &  =\sum_{n=0}^{\infty}\left\langle \mathcal{A}^{\ast
}g,\Omega_{n}\right\rangle \Omega_{n}=\sum_{n=0}^{\infty}\left\langle
g,\mathcal{A}\Omega_{n}\right\rangle \Omega_{n}\nonumber\\
&  =\sum_{n=0}^{\infty}\left\langle g,c_{\mathcal{A}}\left(  n\right)
\Omega_{n+l}\right\rangle \Omega_{n}=\sum_{n=0}^{\infty}\left\langle
g,\Omega_{n+l}\right\rangle c_{\mathcal{A}}\left(  n\right)  \Omega
_{n}\nonumber\\
&  =\sum_{n=0}^{\infty}\left\langle g,\Omega_{n}\right\rangle c_{\mathcal{A}%
}\left(  n-l\right)  \Omega_{n-l}=\sum_{n=0}^{\infty}\left\langle g,\Omega
_{n}\right\rangle \mathcal{A}\mathbf{^{\dag}}\Omega_{n}, \label{equ.3.28}%
\end{align}
wherein we have used the conventions stated after Eq. (\ref{equ.3.25})
repeatedly. Since, by Lemma \ref{lem.3.18}, $\left\{  \mathcal{A}%
\mathbf{^{\dag}}\Omega_{n}=c_{\mathcal{A}}\left(  n-l\right)  \Omega
_{n-l}\right\}  _{n=0}^{\infty}$ is an orthogonal set such that
\[
\left\Vert \mathcal{A}\mathbf{^{\dag}}\Omega_{n}\right\Vert _{2}%
^{2}=\left\vert c_{\mathcal{A}}\left(  n-l\right)  \right\vert ^{2}\asymp
n^{k},
\]
it follows that the last sum in Eq. (\ref{equ.3.28}) is convergent iff
\[
\sum_{n=0}^{\infty}\left\vert \left\langle g,\Omega_{n}\right\rangle
\right\vert ^{2}n^{k}<\infty\iff g\in D_{k/2}.
\]
Conversely if $g\in D_{k/2}$ and $f\in\mathcal{S}=D\left(  \mathcal{A}\right)
$ we have,
\begin{align*}
\left\langle \sum_{n=0}^{\infty}\left\langle g,\Omega_{n}\right\rangle
\mathcal{A}\mathbf{^{\dag}}\Omega_{n},f\right\rangle  &  =\sum_{n=0}^{\infty
}\left\langle g,\Omega_{n}\right\rangle \left\langle \mathcal{A}%
\mathbf{^{\dag}}\Omega_{n},f\right\rangle \\
&  =\sum_{n=0}^{\infty}\left\langle g,\Omega_{n}\right\rangle \left\langle
\Omega_{n},\mathcal{A}f\right\rangle =\left\langle g,\mathcal{A}%
f\right\rangle
\end{align*}
from which it follows that $g\in D\left(  \mathcal{A}\mathbf{^{\ast}}\right)
$ and $\mathcal{A}^{\ast}g$ is given as in Eq. (\ref{equ.3.28}).

In summary, we have shown $D\left(  \mathcal{A}\mathbf{^{\ast}}\right)
=D_{k/2}$ and $\mathcal{A}^{\ast}g$ is given by Eq. (\ref{equ.3.28}).
Moreover, from Eq. (\ref{equ.3.28}), if $g\in D_{k/2}$ then
\[
\mathcal{A}^{\ast}g=\lim_{N\rightarrow\infty}\sum_{n=0}^{N}\left\langle
g,\Omega_{n}\right\rangle \mathcal{A}\mathbf{^{\dag}}\Omega_{n}=\lim
_{N\rightarrow\infty}\mathcal{A}^{\dag}\mathcal{P}_{N}g
\]
which implies $g\in D\left(  \overline{\mathcal{A}^{\dag}}\right)  $ and
$\mathcal{A}^{\ast}g=\overline{\mathcal{A}^{\dag}}g,$ i.e. $\mathcal{A}^{\ast
}\subset\overline{\mathcal{A}^{\dag}}.$ Combining this last assertion with the
first inclusion in Eq. (\ref{equ.3.27}) implies and $\mathcal{A}^{\ast
}=\overline{\mathcal{A}^{\dag}}.$ This proves all of the assertions involving
$\mathcal{A}^{\ast}$ and $\overline{\mathcal{A}^{\dag}}.$ We may now complete
the proof by applying these assertions with $\mathcal{A}=u_{\mathbf{b}}\left(
a,a^{\dag}\right)  $ replaced by $\mathcal{A}^{\dag}=u_{\mathbf{b}^{\ast}%
}\left(  a,a^{\dag}\right)  $ and using the facts that $\mathcal{A}^{\dag\dag
}=\mathcal{A},$ $\ell\left(  \mathbf{b}^{\ast}\right)  =-\ell\left(
\mathbf{b}\right)  =-l,$ and $c_{\mathcal{A}^{\dag}}\left(  n\right)
=c_{\mathcal{A}}\left(  n-l\right)  .$
\end{proof}

\begin{theorem}
\label{the.3.23}Let $k=\deg_{\theta}u_{\mathbf{b}}\left(  \theta,\theta^{\ast
}\right)  ,$ $j=\deg_{\theta}u_{\mathbf{d}}\left(  \theta,\theta^{\ast
}\right)  ,$ $\mathcal{A}=u_{\mathbf{b}}\left(  a,a^{\dag}\right)
,\mathcal{\ D}=u_{\mathbf{d}}\left(  a,a^{\dag}\right)  ,$ $\ell\left(
\mathbf{b}\right)  ,$ and $\ell\left(  \mathbf{d}\right)  $ be as in Notations
\ref{not.3.17} and \ref{not.3.1}. Then;

\begin{enumerate}
\item $\overline{\mathcal{AD}}=\mathcal{\bar{A}\bar{D}},$

\item $\left(  \mathcal{AD}\right)  ^{\ast}=\mathcal{D}^{\ast}\mathcal{A}%
^{\ast},$ and

\item $\mathcal{\bar{A}}:=\overline{u_{\mathbf{b}}\left(  a,a^{\dag}\right)
}=u_{\mathbf{b}}\left(  \bar{a},a^{\ast}\right)  ,$ i.e. if $\mathcal{A}$ is a
monomial operator in $a$ and $a^{\dag},$ then $\mathcal{\bar{A}}$ is the
operator resulting from replacing $a$ by $\bar{a}$ and $a^{\dag}$ by $a^{\ast
}$ everywhere in $\mathcal{A}.$
\end{enumerate}
\end{theorem}

\begin{proof}
Because of the conventions described after Eq. (\ref{equ.3.25}), in the
argument below it will be easier to view all sums over $n\in\mathbb{Z}$
instead of $n\in\mathbb{N}_{0}.$ We will denote all of these infinite sums
simply as $\sum_{n}.$ We now prove each item in turn.

\begin{enumerate}
\item Since $\mathcal{AD}$ is a monomial operator of degree $k+j$ it follows
from Theorem \ref{the.3.22} that $D\left(  \overline{\mathcal{AD}}\right)
=D_{\left(  k+j\right)  /2}.$ On the other hand, $f\in D\left(  \mathcal{\bar
{A}\bar{D}}\right)  $ iff $f\in D\left(  \mathcal{\bar{D}}\right)  =D_{j/2}$
and $\overline{\mathcal{D}}f\in D\left(  \mathcal{\bar{A}}\right)  =D_{k/2}.$
Moreover, $\overline{\mathcal{D}}f=\mathcal{D}^{\dag\ast}f\in D\left(
\mathcal{\bar{A}}\right)  =D_{k/2}$ iff
\begin{align}
\infty &  >\sum_{n}\left\vert \left\langle \overline{\mathcal{D}}f,\Omega
_{n}\right\rangle \right\vert ^{2}n^{k}=\sum_{n}\left\vert \left\langle
f,\mathcal{D^{\dag}}\Omega_{n}\right\rangle \right\vert ^{2}n^{k}\nonumber\\
&  =\sum_{n}\left\vert \left\langle f,\Omega_{n-\ell\left(  \mathbf{d}\right)
}\right\rangle \right\vert ^{2}\left\vert c_{\mathcal{D}^{\dag}}\left(
n\right)  \right\vert ^{2}n^{k}. \label{equ.3.29}%
\end{align}
However, by Lemma \ref{lem.3.18} we know $\left\vert c_{\mathcal{D}^{\dag}%
}\left(  n\right)  \right\vert ^{2}\asymp n^{j}$ and so the condition in Eq.
(\ref{equ.3.29}) is the same as saying $f\in D_{\left(  k+j\right)  /2}. $
Thus we have shown $D\left(  \mathcal{\bar{A}\bar{D}}\right)  =D\left(
\overline{\mathcal{AD}}\right)  .$ Moreover, if $f\in D_{\left(  k+j\right)
/2},$ then by Theorem \ref{the.3.22} and Lemma \ref{lem.3.18} we find,
\begin{align}
\mathcal{\bar{A}\bar{D}}f  &  =\sum_{n}\left\langle \overline{\mathcal{D}%
}f,\Omega_{n}\right\rangle \mathcal{A}\Omega_{n}=\sum_{n}\left\langle
f,\mathcal{D}^{\dag}\Omega_{n}\right\rangle \mathcal{A}\Omega_{n}\nonumber\\
&  =\sum_{n}\left\langle f,c_{\mathcal{D}}\left(  n-\ell\left(  \mathbf{d}%
\right)  \right)  \Omega_{n-\ell\left(  \mathbf{d}\right)  }\right\rangle
\mathcal{A}\Omega_{n}\nonumber\\
&  =\sum_{n}\left\langle f,\Omega_{n}\right\rangle \mathcal{A}c_{\mathcal{D}%
}\left(  n\right)  \Omega_{n+\ell\left(  \mathbf{d}\right)  }\nonumber\\
&  =\sum_{n}\left\langle f,\Omega_{n}\right\rangle \mathcal{AD}\Omega
_{n}=\overline{\mathcal{AD}}f. \label{equ.3.30}%
\end{align}

\item By item 1. of Theorem \ref{the.3.22} and item 1. of this theorem,
\[
\left(  \mathcal{AD}\right)  ^{\ast}=\overline{\left(  \mathcal{AD}\right)
^{\dag}}=\overline{\mathcal{D}^{\dag}\mathcal{A^{\dag}}}=\overline
{\mathcal{D}^{\dag}}\overline{\mathcal{A^{\dag}}}=\mathcal{D}^{\ast
}\mathcal{A}^{\ast}.
\]

\item This follows by induction on $k=\deg_{\theta}u_{\mathbf{b}}$ making use
of item 1. of Theorem \ref{the.3.22} and item 1.
\end{enumerate}
\end{proof}

\begin{corollary}
[Diagonal form of the Number Operator]\label{cor.3.26}If $\mathcal{N=}%
u_{\left(  \theta^{\ast},\theta\right)  }\left(  \bar{a},a^{\ast}\right)
=a^{\ast}\bar{a}$ as in Definition \ref{def.1.4}, then by $\mathcal{N}%
=\overline{a^{\dag}a},$
\[
D\left(  \mathcal{N}\right)  =D_{1}=\left\{  f\in L^{2}\left(  m\right)
:\sum_{n=0}^{\infty}n^{2}\left\vert \left\langle f,\Omega_{n}\right\rangle
\right\vert ^{2}<\infty\right\}  ,
\]
and for $f\in D\left(  \mathcal{N}\right)  ,$
\[
\mathcal{N}f=\sum_{n=0}^{\infty}n\left\langle f,\Omega_{n}\right\rangle
\Omega_{n}.
\]

\end{corollary}

\begin{proof}
Since $\mathcal{N=}u_{\left(  \theta^{\ast},\theta\right)  }\left(  \bar
{a},a^{\ast}\right)  ,$ it follows by Theorem \ref{the.3.23} that
\[
\mathcal{N=}\overline{u_{\left(  \theta^{\ast},\theta\right)  }\left(
a,a^{\dag}\right)  }=\overline{a^{\dag}a}%
\]
and then by Theorem \ref{the.3.22} that $D\left(  \mathcal{N}\right)  =D_{1}.$
Moreover, by items 1 and 2 in the Theorem \ref{the.3.22}, if $f\in D\left(
\mathcal{N}\right)  ,$ then
\[
\mathcal{N}f=\sum_{n=0}^{\infty}\left\langle f,\Omega_{n}\right\rangle
a^{\dag}a\Omega_{n}=\sum_{n=0}^{\infty}n\left\langle f,\Omega_{n}\right\rangle
\Omega_{n}.
\]

\end{proof}

\begin{definition}
[Functional Calculus for $\mathcal{N}$]\label{def.3.27}Given a function
$G:\mathbb{N}_{0}\rightarrow\mathbb{C}$ let $G\left(  \mathcal{N}\right)  $ be
the unique closed operator on $L^{2}\left(  m\right)  $ such that $G\left(
\mathcal{N}\right)  \Omega_{n}:=G\left(  n\right)  \Omega_{n}$ for all
$n\in\mathbb{N}_{0}.$ In more detail,
\begin{equation}
D\left(  G\left(  \mathcal{N}\right)  \right)  :=\left\{  u\in L^{2}\left(
m\right)  :\sum_{n=0}^{\infty}\left\vert G\left(  n\right)  \right\vert
^{2}\left\vert \left\langle u,\Omega_{n}\right\rangle \right\vert ^{2}%
<\infty\right\}  \label{equ.3.31}%
\end{equation}
and for $u\in D\left(  G\left(  \mathcal{N}\right)  \right)  ,$
\[
G\left(  \mathcal{N}\right)  u:=\sum_{n=0}^{\infty}G\left(  n\right)
\left\langle u,\Omega_{n}\right\rangle \Omega_{n}.
\]

\end{definition}

\begin{example}
\label{exa.3.28}If $\beta\geq0,$ then $D\left(  \mathcal{N}^{\beta}\right)
=D_{\beta}$ where $D_{\beta}$ was defined in Definition \ref{def.3.21}.
\end{example}

\begin{notation}
\label{not.3.29}If $J\subset\mathbb{N}_{0}$ and
\[
\mathbf{1}_{J}\left(  n\right)  :=\left\{
\begin{array}
[c]{cc}%
1 & \text{if }n\in J\\
0 & \text{otherwise}%
\end{array}
,\right.
\]
then
\begin{equation}
\mathbf{1}_{J}\left(  \mathcal{N}\right)  f=\sum_{n\in J}\left\langle
f,\Omega_{n}\right\rangle \Omega_{n} \label{equ.3.32}%
\end{equation}
is orthogonal projection onto $\overline{\operatorname*{span}\left\{
\Omega_{n}:n\in J\right\}  }.$ When $J=\left\{  0,1,\dots,N\right\}  ,$ then
$\mathbf{1}_{J}\left(  \mathcal{N}\right)  $ (or also write $1_{\mathcal{N}%
\leq N})$ is precisely the orthogonal projection operator already defined in
Eq. (\ref{equ.3.19}) above.
\end{notation}

At this point it is convenient to introduce a scale of Sobolev type norms on
$L^{2}\left(  m\right)  .$

\begin{notation}
[$\beta$ -- Norms]\label{not.3.30}For $\beta\geq0$ and $f\in L^{2}\left(  m\right),$ let
\begin{equation}
\left\Vert f\right\Vert _{\beta}^{2}:=\sum_{n=0}^{\infty}\left\vert
\left\langle f,\Omega_{n}\right\rangle \right\vert ^{2}\left(  n+1\right)
^{2\beta}. \label{equ.3.33}%
\end{equation}
\end{notation}

\begin{remark}
\label{rem.3.31}From Definition \ref{def.3.27} and Notation \ref{not.3.30}, it
is readily seen that
\begin{align*}
D_{\beta}  &  =D\left(  \mathcal{N}^{\beta}\right)  =\left\{  f\in
L^{2}\left(  m\right)  :\left\Vert f\right\Vert _{\beta}^{2}<\infty\right\}
,\\
\left\Vert f\right\Vert _{\beta}^{2}  &  =\left\Vert \left(  \mathcal{N}%
+I\right)  ^{\beta}f\right\Vert _{L^{2}\left(  m\right)  }^{2}\text{ }%
\forall~f\in D\left(  \mathcal{N}^{\beta}\right)  ,\\
D\left(  \mathcal{N}^{\beta}\right)   &  =D\left(  \left(  \mathcal{N}%
+1\right)  ^{\beta}\right)  \text{ for all }\beta\geq0,\text{ and}\\
\left\Vert \cdot\right\Vert _{\beta_{1}}  &  \leq\left\Vert \cdot\right\Vert
_{\beta_{2}}\text{ and }D\left(  \mathcal{N}^{\beta_{2}}\right)  \subseteq
D\left(  \mathcal{N}^{\beta_{1}}\right)  \text{ for all }0\leq\beta_{1}%
\leq\beta_{2}.
\end{align*}
The normed space, $\left(  D\left(  \mathcal{N}^{\beta}\right)  ,\left\Vert
\cdot\right\Vert _{\beta}\right)  ,$ is a Hilbertian space which is isomorphic
to $\ell^{2}\left(  \mathbb{N}_{0},\mu_{\beta}\right)  $ where $\mu_{\beta
}\left(  n\right)  :=\left(  1+n\right)  ^{2\beta}.$ The isomorphism is given
by the unitary map,
\[
f\in D\left(  \mathcal{N}^{\beta}\right)  \rightarrow\left\{  \left\langle
f,\Omega_{n}\right\rangle \right\}  _{n=0}^{\infty}\in\ell^{2}\left(
\mathbb{N}_{0},\mu_{\beta}\right)  .
\]

\end{remark}

It is well known (see for example, \citep[Theorem 1]{Simon1971}) that
\begin{equation}
\mathcal{S}=\bigcap_{n=0}^{\infty}D\left(  \mathcal{N}^{n}\right)
=\bigcap_{\beta\geq0}D\left(  \mathcal{N}^{\beta}\right)  . \label{equ.3.34}%
\end{equation}
The inclusion $\mathcal{S}\subset\bigcap_{n=0}^{\infty}D\left(  \mathcal{N}%
^{n}\right)  $ is easy to understand since if $n\in\mathbb{N}_{0},$ $\left(
a^{\dag}a+I\right)  ^{n}$ is symmetric on $\mathcal{S}$ and therefore if
$f\in\mathcal{S}$ we have,
\begin{align*}
\left\Vert f\right\Vert _{n}^{2}  &  =\sum_{n=0}^{\infty}\left\vert
\left\langle f,\Omega_{n}\right\rangle \right\vert ^{2}\left(  n+1\right)
^{2n}=\sum_{n=0}^{\infty}\left\vert \left\langle f,\left(  a^{\dag}a+I\right)
^{n}\Omega_{n}\right\rangle \right\vert ^{2}\\
&  =\sum_{n=0}^{\infty}\left\vert \left\langle \left(  a^{\dag}a+I\right)
^{n}f,\Omega_{n}\right\rangle \right\vert ^{2}=\left\Vert \left(  a^{\dag
}a+I\right)  ^{n}f\right\Vert _{L^{2}\left(  m\right)  }^{2}<\infty.
\end{align*}
The following related result will be useful in the sequel.

\begin{proposition}
\label{pro.3.32}The subspace $\mathcal{S}_{0}$ in Eq. (\ref{equ.3.26}) is
dense (and so is $\mathcal{S)}$ in $\left(  D\left(  \mathcal{N}^{\beta
}\right)  ,\left\Vert \cdot\right\Vert _{\beta}\right)  $ for all $\beta\geq0.
$ Moreover, if $f\in D\left(  \mathcal{N}^{\beta}\right)  ,$ then
$\mathcal{P}_{N}f\in\mathcal{S}_{0}$ and $\left\Vert f-\mathcal{P}%
_{N}f\right\Vert _{\beta}\rightarrow0$ as $N\rightarrow\infty.$
\end{proposition}

\begin{proof}
If $f\in D\left(  \mathcal{N}^{\beta}\right)  ,$ then
\[
\sum_{n=0}^{\infty}\left\vert \left\langle f,\Omega_{n}\right\rangle
\right\vert ^{2}\left(  1+n\right)  ^{2\beta}=\left\Vert f\right\Vert _{\beta
}^{2}<\infty
\]
and hence
\[
\left\Vert f-\mathcal{P}_{N}f\right\Vert _{\beta}^{2}=\sum_{n=N+1}^{\infty
}\left\vert \left\langle f,\Omega_{n}\right\rangle \right\vert ^{2}\left(
1+n\right)  ^{2\beta}\rightarrow0\text{ as }N\rightarrow\infty.
\]

\end{proof}

\begin{remark}
\label{rem.3.33}The zero norm, $\left\Vert \cdot\right\Vert _{0},$ is just a
standard $L^{2}\left(  m\right)  $-norm and we will typically drop the
subscript $0$ and simply write $\left\Vert \cdot\right\Vert $ for $\left\Vert
\cdot\right\Vert _{0}=\left\Vert \cdot\right\Vert _{L^{2}\left(  m\right)  }.$
\end{remark}

\begin{remark}
\label{rem.3.34}If $\mathcal{A}=u_{\mathbf{b}}\left(  a,a^{\dag}\right)  $ and
$k=\deg_{\theta}u_{\mathbf{b}}\left(  \theta,\theta^{\ast}\right)  ,$ then by
Eq. (\ref{equ.3.33}) and the Theorem \ref{the.3.22}, we have
\begin{equation}
D\left(  \mathcal{\bar{A}}\right)  =D_{k/2}=D\left(  \mathcal{N}^{\frac{k}{2}%
}\right)  . \label{equ.3.35}%
\end{equation}

\end{remark}

\begin{corollary}
\label{cor.3.36}The following domain statement holds;
\begin{equation}
D\left(  \bar{a}\right)  =D\left(  a^{\ast}\right)  =D\left(  \mathcal{N}%
^{1/2}\right)  =D\left(  M_{x}\right)  \cap D\left(  \partial_{x}\right)  .
\label{equ.3.36}%
\end{equation}
Moreover for $f\in D\left(  \mathcal{N}^{1/2}\right)  ,$
\begin{align}
\bar{a}f  &  =\sum_{n=1}^{\infty}\sqrt{n}\left\langle f,\Omega_{n}%
\right\rangle \Omega_{n-1}\text{ and }\label{equ.3.37}\\
a^{\ast}f  &  =\sum_{n=0}^{\infty}\sqrt{n+1}\left\langle f,\Omega
_{n}\right\rangle \Omega_{n+1}. \label{equ.3.38}%
\end{align}

\end{corollary}

\begin{proof}
$D\left(  \bar{a}\right)  =D\left(  a^{\ast}\right)  =D\left(  \mathcal{N}%
^{1/2}\right)  $ is followed by the Eq. (\ref{equ.3.35}) in the Remark
\ref{rem.3.34}. Eqs (\ref{equ.3.37}) and (\ref{equ.3.38}) are consequence from
Theorem \ref{the.3.22}. The only new statement to prove here is that $D\left(
\mathcal{N}^{1/2}\right)  =D\left(  M_{x}\right)  \cap D\left(  \partial
_{x}\right)  .$ If $f\in D\left(  M_{x}\right)  \cap D\left(  \partial
_{x}\right)  $ we have
\begin{align*}
\sqrt{n}\left\langle f,\Omega_{n}\right\rangle  &  =\left\langle f,a^{\dag
}\Omega_{n-1}\right\rangle =\frac{1}{\sqrt{2}}\left\langle f,\left(
M_{x}-\partial_{x}\right)  \Omega_{n-1}\right\rangle \\
&  =\frac{1}{\sqrt{2}}\left\langle \left(  M_{x}+\partial_{x}\right)
f,\Omega_{n-1}\right\rangle
\end{align*}
from which it follows that
\[
\sum_{n=1}^{\infty}\left\vert \sqrt{n}\left\langle f,\Omega_{n}\right\rangle
\right\vert ^{2}=\frac{1}{2}\left\Vert \left(  M_{x}+\partial_{x}\right)
f\right\Vert ^{2}<\infty
\]
and therefore $f\in D\left(  \bar{a}\right)  =D\left(  N^{1/2}\right)  .$
Conversely if $f\in D\left(  \mathcal{N}^{1/2}\right)  $ and we let
$f_{m}:=\sum_{k=0}^{m}\left\langle f,\Omega_{k}\right\rangle \Omega_{k}$ for
all $m\in\mathbb{N},$ then $f_{m}\rightarrow f,$ $\bar{a}f_{m}\rightarrow
\bar{a}f$ and $a^{\ast}f_{m}\rightarrow a^{\ast}f$ in $L^{2}.$ Thus it follows
that in the limit as $m\rightarrow\infty,$
\begin{align*}
M_{x}f_{m}  &  =\frac{1}{\sqrt{2}}\left(  \bar{a}+a^{\ast}\right)
f_{m}\rightarrow\frac{1}{\sqrt{2}}\left(  \bar{a}+a^{\ast}\right)  f\text{
and}\\
\partial_{x}f_{m}  &  =\frac{1}{\sqrt{2}}\left(  \bar{a}-a^{\ast}\right)
f_{m}\rightarrow\frac{1}{\sqrt{2}}\left(  \bar{a}-a^{\ast}\right)  f.
\end{align*}
As $M_{x}$ and $\partial_{x}$ are closed operators, it follows that $f\in
D\left(  M_{x}\right)  \cap D\left(  \partial_{x}\right)  .$
\end{proof}

\subsection{Operator Inequalities\label{sub.3.4}}

\begin{notation}
[$\beta_{1},\beta_{2}$ -- Operator Norms]\label{not.3.37} Let $\beta_{1}%
,\beta_{2}\geq0.$ If $T:D\left(  \mathcal{N}^{\beta_{1}}\right)  \rightarrow
D\left(  \mathcal{N}^{\beta_{2}}\right)  $ is a linear map, let
\begin{equation}
\left\Vert T\right\Vert _{\beta_{1}\rightarrow\beta_{2}}:=\sup_{0\neq\psi\in
D\left(  \mathcal{N}^{\beta_{1}}\right)  }\frac{\left\Vert T\psi\right\Vert
_{\beta_{2}}}{\left\Vert \psi\right\Vert _{\beta_{1}}}. \label{equ.3.39}%
\end{equation}
denote the corresponding operator norm. We say that $T$ is $\beta
_{1}\rightarrow\beta_{2}$ bounded if $\left\Vert T\right\Vert _{\beta
_{1}\rightarrow\beta_{2}}<\infty.$ In the special case when $\beta_{1}%
=\beta_{2}=\beta,$ let $\left(  B\left(  D\left(  \mathcal{N}^{\beta}\right)
\right)  ,\left\Vert \cdot\right\Vert _{\beta\rightarrow\beta}\right)  $
denote the Banach space of all $\beta\rightarrow\beta$ bounded linear
operators, $T:D\left(  \mathcal{N}^{\beta}\right)  \rightarrow D\left(
\mathcal{N}^{\beta}\right)  .$
\end{notation}

\begin{remark}
\label{rem.3.38} Let $\beta_{1},\beta_{2},\beta_{3}\geq0.$ As usual, if
$T:D\left(  \mathcal{N}^{\beta_{1}}\right)  \rightarrow D\left(
\mathcal{N}^{\beta_{2}}\right)  $ and $S:D\left(  \mathcal{N}^{\beta_{2}%
}\right)  \rightarrow D\left(  \mathcal{N}^{\beta_{3}}\right)  $ are any
linear operators, then
\begin{equation}
\left\Vert ST\right\Vert _{\beta_{1}\rightarrow\beta_{3}}\leq\left\Vert
S\right\Vert _{\beta_{2}\rightarrow\beta_{3}}\left\Vert T\right\Vert
_{\beta_{1}\rightarrow\beta_{2}}. \label{equ.3.40}%
\end{equation}

\end{remark}

\begin{proposition}
\label{pro.3.39} Let $k=\deg_{\theta}u_{\mathbf{b}}\left(  \theta,\theta
^{\ast}\right)  $ and $\mathcal{A}=u_{\mathbf{b}}\left(  a,a^{\dag}\right)  $
be as in Notations \ref{not.3.17} and \ref{not.3.1}. If $\beta\geq0,$ then
$\mathcal{\bar{A}}D\left(  \mathcal{N}^{\beta+k/2}\right)  \subset D\left(
\mathcal{N}^{\beta}\right)  $ and
\begin{equation}
\left\Vert \mathcal{\bar{A}}\right\Vert _{\beta+\frac{k}{2}\rightarrow\beta
}^{2}\leq k^{k}\left(  k+1\right)  ^{2\beta}\leq\left(  k+1\right)
^{2\beta+k}.\text{ } \label{equ.3.41}%
\end{equation}
Moreover,
\begin{equation}
\left\Vert \mathcal{\bar{A}}f\right\Vert _{\beta}\leq\left\Vert \left(
\mathcal{N}+k\right)  ^{k/2}\left(  \mathcal{N}+k+1\right)  ^{\beta
}f\right\Vert ~\forall~f\in D\left(  \mathcal{N}^{\beta+k/2}\right)  .
\label{equ.3.42}%
\end{equation}

\end{proposition}

\begin{proof}
Let $f\in D\left(  \mathcal{N}^{\beta+k/2}\right)  \subset D\left(
\mathcal{N}^{k/2}\right)  $ and recall from Lemma \ref{lem.3.18} that
$c_{\mathcal{A}}^{\dag}\left(  n\right)  =c_{\mathcal{A}}\left(  n-l\right)  $
and $\left\vert c_{\mathcal{A}}\left(  n\right)  \right\vert ^{2}\leq\left(
n+k\right)  ^{k}.$ Using these facts and the fact that $\mathcal{\bar{A}%
=A}^{\dag\ast}$ (see Theorem \ref{the.3.22}), we find,
\begin{align}
\left\Vert \mathcal{\bar{A}}f\right\Vert _{\beta}^{2}=  &  \sum_{n}\left\vert
\left\langle \mathcal{\bar{A}}f,\Omega_{n}\right\rangle \right\vert
^{2}\left(  1+n\right)  ^{2\beta}=\sum_{n}\left\vert \left\langle
f,\mathcal{A}^{\dag}\Omega_{n}\right\rangle \right\vert ^{2}\left(
1+n\right)  ^{2\beta}1_{n\geq0}\nonumber\\
=  &  \sum_{n}\left\vert \left\langle f,\Omega_{n-l}\right\rangle \right\vert
^{2}\left\vert c_{\mathcal{A}}\left(  n-l\right)  \right\vert \left(
1+n\right)  ^{2\beta}1_{n\geq0}\nonumber\\
=  &  \sum_{n}\left\vert \left\langle f,\Omega_{n}\right\rangle \right\vert
^{2}\left(  1+n+l\right)  ^{2\beta}1_{n+l\geq0}\left\vert c_{\mathcal{A}%
}\left(  n\right)  \right\vert ^{2}\nonumber\\
\leq &  \sum_{n}\left\vert \left\langle f,\Omega_{n}\right\rangle \right\vert
^{2}\left(  n+k+1\right)  ^{2\beta}\left(  n+k\right)  ^{k}\label{equ.3.43}\\
=  &  \left\Vert \left(  \mathcal{N}+k\right)  ^{k/2}\left(  \mathcal{N}%
+k+1\right)  ^{\beta}f\right\Vert _{0}^{2}\nonumber
\end{align}
which proves Eq. (\ref{equ.3.42}). Using
\begin{equation}
n+a\leq a\left(  n+1\right)  \text{ for }a\geq1\text{ and }n\in\mathbb{N}_{0}
\label{equ.3.44}%
\end{equation}
in Eq. (\ref{equ.3.43}) with $a=k$ and $a=k+1$ shows,
\[
\left\Vert \mathcal{\bar{A}}f\right\Vert _{\beta}^{2}\leq k^{k}\left(
k+1\right)  ^{2\beta}\sum_{n}\left\vert \left\langle f,\Omega_{n}\right\rangle
\right\vert ^{2}\left(  1+n\right)  ^{2\beta+k}=k^{k}\left(  k+1\right)
^{2\beta}\left\Vert f\right\Vert _{\beta+k/2}^{2}.
\]
The previous inequality proves Eq. (\ref{equ.3.41}) and also $\mathcal{\bar
{A}}D\left(  \mathcal{N}^{\beta+k/2}\right)  \subset D\left(  \mathcal{N}%
^{\beta}\right)  .$
\end{proof}

\begin{corollary}
\label{cor.3.40}If $P\left(  \theta,\theta^{\ast}\right)  \in\mathbb{C}%
\left\langle \theta,\theta^{\ast}\right\rangle $ and $d=\deg_{\theta}P,$ then
$D\left(  \mathcal{N}^{d/2}\right)  =D\left(  P\left(  \bar{a},a^{\ast
}\right)  \right)  ,$ $P\left(  \bar{a},a^{\ast}\right)  \subseteq
\overline{P\left(  a,a^{\dag}\right)  },$ and
\begin{equation}
\left\Vert P\left(  \bar{a},a^{\ast}\right)  \right\Vert _{\beta
+d/2\rightarrow\beta}\leq\sum_{k=0}^{d}k^{k/2}\left(  k+1\right)  ^{\beta
}\left\vert P_{k}\right\vert \text{ for all }\beta\geq0. \label{equ.3.45}%
\end{equation}

\end{corollary}

\begin{proof}
The operator $P\left(  \bar{a},a^{\ast}\right)  $ is a linear combination of
operators of the form $u_{\mathbf{b}}\left(  \bar{a},a^{\ast}\right)  $ where
$k=\deg_{\theta}u_{\mathbf{b}}\left(  \theta,\theta^{\ast}\right)  \leq d.$ By
Theorem \ref{the.3.22}, it follows that $D\left(  u_{\mathbf{b}}\left(
\bar{a},a^{\ast}\right)  \right)  =D\left(  \mathcal{N}^{k/2}\right)
\supseteq D\left(  \mathcal{N}^{d/2}\right)  $ and hence $D\left(
\mathcal{N}^{d/2}\right)  =D\left(  P\left(  \bar{a},a^{\ast}\right)  \right)
.$ Further, Proposition \ref{pro.3.39} shows
\[
\left\Vert u_{\mathbf{b}}\left(  \bar{a},a^{\ast}\right)  \right\Vert
_{_{\beta+d/2\rightarrow\beta}}\leq\left\Vert u_{\mathbf{b}}\left(  \bar
{a},a^{\ast}\right)  \right\Vert _{_{\beta+k/2\rightarrow\beta}}\leq
k^{k/2}\left(  k+1\right)  ^{\beta}.
\]
This estimate, the triangle inequality, and the definition of $\left\vert
P_{k}\right\vert $ in Eq. (\ref{equ.2.23}) leads directly to the inequality in
Eq. (\ref{equ.3.45}).

If $f\in D\left(  \mathcal{N}^{d/2}\right)  ,$ it follows from Eq.
(\ref{equ.3.45}) and Proposition \ref{pro.3.32} that
\[
P\left(  \bar{a},a^{\ast}\right)  f=\lim_{N\rightarrow\infty}P\left(  \bar
{a},a^{\ast}\right)  \mathcal{P}_{N}f=\lim_{N\rightarrow\infty}P\left(
a,a^{\dag}\right)  \mathcal{P}_{N}f
\]
which shows $f\in D\left(  \overline{P\left(  a,a^{\dag}\right)  }\right)  $
and $\overline{P\left(  a,a^{\dag}\right)  }f=P\left(  \bar{a},a^{\ast
}\right)  f.$
\end{proof}

\begin{notation}
\label{not.3.41}For $x\in\mathbb{R}$ let $\left(  x\right)  _{+}:=\max\left(
x,0\right)  .$
\end{notation}

\begin{lemma}
\label{lem.3.42}If $\mathcal{A}=u_{\mathbf{b}}\left(  a,a^{\dag}\right)  ,$
$k=\deg_{\theta}u_{\mathbf{b}}\left(  \theta,\theta^{\ast}\right)  ,$
$l=\ell\left(  \mathbf{b}\right)  \in\mathbb{Z}$ are as in Notations
\ref{not.3.17} and \ref{not.3.1}, then for all $\beta\geq0$ we have,
\begin{equation}
\left(  \mathcal{N}+1\right)  ^{\beta}\mathcal{\bar{A}}f=\mathcal{\bar{A}%
}\left(  \left(  \mathcal{N}+l\right)  _{+}+1\right)  ^{\beta}f\text{ for all
}f\in D\left(  \mathcal{N}^{\beta+\frac{k}{2}}\right)  . \label{equ.3.46}%
\end{equation}

\end{lemma}

\begin{proof}
Using Proposition \ref{pro.3.39} and Remark \ref{rem.3.38} it is readily
verified that the operators on both sides of Eq. (\ref{equ.3.46}) are bounded
linear operators from $D\left(  \mathcal{N}^{\beta+\frac{k}{2}}\right)  $ to
$L^{2}\left(  m\right)  .$ Since $\mathcal{S}_{0}$ is dense in $D\left(
\mathcal{N}^{\beta+\frac{k}{2}}\right)  $ (see Proposition \ref{pro.3.32}) it
suffices to verify Eq. (\ref{equ.3.46}) for $f=\Omega_{n}$ for all
$n\in\mathbb{N}_{0}$ which is trivial. Indeed, $\mathcal{\bar{A}}\Omega
_{n}=c_{\mathcal{A}}\left(  n\right)  \Omega_{n+l}$ which is zero if $n+l<0$
and hence
\begin{align*}
\left(  \mathcal{N}+1\right)  ^{\beta}\mathcal{\bar{A}}\Omega_{n}  &  =\left(
\left(  n+l\right)  _{+}+1\right)  ^{\beta}\mathcal{\bar{A}}\Omega
_{n}=\mathcal{\bar{A}}\left(  \left(  n+l\right)  _{+}+1\right)  ^{\beta
}\Omega_{n}\\
&  =\mathcal{\bar{A}}\left(  \left(  \mathcal{N}+l\right)  _{+}+1\right)
^{\beta}\Omega_{n}.
\end{align*}

\end{proof}

\begin{proposition}
\label{pro.3.43}Let $k\in\mathbb{N},$ $\mathbf{b\in}\left\{  \theta
,\theta^{\ast}\right\}  ^{k},$ $\mathcal{A},$ and $\ell\left(  \mathbf{b}%
\right)  $ be as in Notation \ref{not.3.1}. For any $\beta\geq0,$ it gets
\begin{align}
&  \left\Vert \left[  \left(  \mathcal{N}+1\right)  ^{\beta},\mathcal{\bar{A}%
}\right]  \left(  \mathcal{N}+1\right)  ^{-\beta}\varphi\right\Vert
\nonumber\\
&  \quad\quad\leq\beta k^{k/2}\left\vert \ell\left(  \mathbf{b}\right)
\right\vert \left(  1+\left\vert \ell\left(  \mathbf{b}\right)  \right\vert
\right)  ^{\left\vert \beta-1\right\vert }\left\Vert \left(  \mathcal{N}%
+1\right)  ^{k/2-1}\mathbf{1}_{\mathcal{N}\geq-l}\varphi\right\Vert
\label{equ.3.47}\\
&  \quad\quad\leq\beta k^{k/2}\left\vert \ell\left(  \mathbf{b}\right)
\right\vert \left(  1+\left\vert \ell\left(  \mathbf{b}\right)  \right\vert
\right)  ^{\left\vert \beta-1\right\vert }\left\Vert \left(  \mathcal{N}%
+1\right)  ^{k/2-1}\varphi\right\Vert \label{equ.3.48}%
\end{align}
for all $\varphi\in D\left(  \mathcal{N}^{k/2}\right)  .$
\end{proposition}

\begin{proof}
Let $l:=\ell\left(  \mathbf{b}\right)  .$ By Lemma \ref{lem.3.42} and the
identity, $\mathcal{\bar{A}}=\mathcal{\bar{A}}\mathbf{1}_{\mathcal{N}+l\geq
0},$ for all $\psi\in D\left(  \mathcal{N}^{k/2+\beta}\right)  $ we have,
\begin{align*}
\left[  \left(  \mathcal{N}+1\right)  ^{\beta},\mathcal{\bar{A}}\right]  \psi
&  =\left[  \left(  \mathcal{N}+1\right)  ^{\beta}\mathcal{\bar{A}%
}-\mathcal{\bar{A}}\left(  \mathcal{N}+1\right)  ^{\beta}\right]  \psi\\
&  =\mathcal{\bar{A}}\left[  \left(  \left(  \mathcal{N}+l\right)
_{+}+1\right)  ^{\beta}-\left(  \mathcal{N}+1\right)  ^{\beta}\right]  \psi\\
&  =\mathcal{\bar{A}}\mathbf{1}_{\mathcal{N}+l\geq0}\left[  \left(  \left(
\mathcal{N}+l\right)  _{+}+1\right)  ^{\beta}-\left(  \mathcal{N}+1\right)
^{\beta}\right]  \psi\\
&  =\mathcal{\bar{A}}\left[  \left(  \mathcal{N}+l+1\right)  ^{\beta}-\left(
\mathcal{N}+1\right)  ^{\beta}\right]  \mathbf{1}_{\mathcal{N}+l\geq0}\psi\\
&  =\mathcal{\bar{A}}\left[  \beta\int_{0}^{l}\left(  \mathcal{N}+1+r\right)
^{\beta-1}dr\right]  \mathbf{1}_{\mathcal{N}+l\geq0}\psi.
\end{align*}
Combining this equation with Eq. (\ref{equ.3.42}) of Proposition
\ref{pro.3.39} shows,
\begin{align*}
\left\Vert \left[  \left(  \mathcal{N}+1\right)  ^{\beta},\mathcal{\bar{A}%
}\right]  \psi\right\Vert  &  \leq\left\Vert \left(  \mathcal{N}+k\right)
^{k/2}\left[  \beta\int_{0}^{l}\left(  \mathcal{N}+1+r\right)  ^{\beta
-1}dr\right]  \mathbf{1}_{\mathcal{N}\geq-l}\psi\right\Vert \\
&  \leq\beta\left\vert \int_{0}^{l}\left\Vert \left(  \mathcal{N}+k\right)
^{k/2}\left(  \mathcal{N}+1+r\right)  ^{\beta-1}\mathbf{1}_{\mathcal{N}\geq
-l}\psi\right\Vert dr\right\vert .\\
&  \leq\beta k^{k/2}\left\vert \int_{0}^{l}\left\Vert \left(  \mathcal{N}%
+1\right)  ^{k/2}\left(  \mathcal{N}+1+r\right)  ^{\beta-1}\mathbf{1}%
_{\mathcal{N}\geq-l}\psi\right\Vert dr\right\vert .
\end{align*}
For $x\geq\max\left(  0,-l\right)  $ and $r$ between $0$ and $l,$ one shows
\[
\left(  x+1+r\right)  ^{\beta-1}\leq\left(  1+\left\vert l\right\vert \right)
^{\left\vert \beta-1\right\vert }\left(  x+1\right)  ^{\beta-1}%
\]
which combined with the previously displayed equation implies,
\begin{equation}
\left\Vert \left[  \left(  \mathcal{N}+1\right)  ^{\beta}%
,\mathcal{\mathcal{\bar{A}}}\right]  \psi\right\Vert \leq\beta k^{k/2}\left(
1+\left\vert l\right\vert \right)  ^{\left\vert \beta-1\right\vert }\left\vert
l\right\vert \left\Vert \left(  \mathcal{N}+1\right)  ^{\frac{k}{2}+\beta
-1}\mathbf{1}_{\mathcal{N}\geq-l}\psi\right\Vert . \label{equ.3.49}%
\end{equation}
Finally, Eq. (\ref{equ.3.47}) easily follows by replacing $\psi$ by $\left(
\mathcal{N}+1\right)  ^{-\beta}\varphi\in D\left(  \mathcal{N}^{k/2}\right)  $
in Eq. (\ref{equ.3.49}).
\end{proof}

\subsection{Truncated Estimates\label{sub.3.5}}

\begin{notation}
[Operator Truncation]\label{not.3.45}If $Q=P\left(  a,a^{\dag}\right)  $ is a
polynomial operator on $L^{2}\left(  m\right)  $ and $M>0,$ let
\begin{equation}
Q_{M}:=\mathbf{1}_{\mathcal{N}\leq M}Q\mathbf{1}_{\mathcal{N}\leq
M}=\mathcal{P}_{M}Q\mathcal{P}_{M}. \label{equ.3.50}%
\end{equation}
and refer to $Q_{M}$ as the \textbf{level-}$M$\textbf{\ truncation of }$Q.$
[Recall that $\mathcal{P}_{M}=$ $\mathbf{1}_{\mathcal{N}\leq M}$ are as in
Notations \ref{not.3.16} and \ref{not.3.29}.]
\end{notation}

\begin{proposition}
\label{pro.3.46}If $k\in\mathbb{N},$ $\beta\geq0,$ $0<M<\infty,$
$\mathbf{b}\in\left\{  \theta,\theta^{\ast}\right\}  ^{k},$ $\mathcal{A}%
=u_{\mathbf{b}}\left(  a,a^{\dag}\right)  ,$ and $\ell\left(  \mathbf{b}%
\right)  $ are as in Notation \ref{not.3.1}, then
\begin{equation}
\left\Vert \mathcal{A}_{M}\right\Vert _{\beta\rightarrow\beta}\leq\left(
M+k\right)  ^{k/2}\left(  1+\left\vert \ell\left(  \mathbf{b}\right)
\right\vert \right)  ^{\beta}\leq k^{k/2}\left(  1+\left\vert \ell\left(
\mathbf{b}\right)  \right\vert \right)  ^{\beta}\left(  M+1\right)  ^{k/2}.
\label{equ.3.51}%
\end{equation}
Consequently if $P\in\mathbb{C}\left\langle \theta,\theta^{\ast}\right\rangle
$ with $d=\deg_{\theta}P,$ then
\begin{equation}
\left\Vert \left[  P\left(  a,a^{\dag}\right)  \right]  _{M}\right\Vert
_{\beta\rightarrow\beta}\leq\sum_{k=0}^{d}\left(  M+k\right)  ^{k/2}\left(
1+k\right)  ^{\beta}\left\vert P_{k}\right\vert \label{equ.3.52}%
\end{equation}
which in particular implies that the map,
\[
P\in\mathbb{C}\left\langle \theta,\theta^{\ast}\right\rangle \rightarrow
\left[  P\left(  a,a^{\dag}\right)  \right]  _{M}\in\left(  B\left(  D\left(
\mathcal{N}^{\beta}\right)  \right)  ,\left\Vert \cdot\right\Vert
_{\beta\rightarrow\beta}\right)  ,
\]
depends continuously on the coefficients of $P.$
\end{proposition}

\begin{proof}
With $l=\ell\left(  \mathbf{b}\right)  ,$ we have for all $n\in\mathbb{N}%
_{0},$
\begin{align}
\mathcal{A}_{M}^{\ast}\Omega_{n}  &  =\left(  \mathcal{P}_{M}\mathcal{A}%
\mathcal{P}_{M}\right)  ^{\ast}\Omega_{n}=\mathcal{P}_{M}\mathcal{A}^{\ast
}\mathcal{P}_{M}\Omega_{n}\nonumber\\
&  =1_{n\leq M}\mathcal{P}_{M}\mathcal{A}^{\dag}\Omega_{n}=1_{n\leq
M}c_{\mathcal{A}}\left(  n-l\right)  \mathcal{P}_{M}\Omega_{n-l}\nonumber\\
&  =1_{n\leq M}1_{n-l\leq M}c_{\mathcal{A}}\left(  n-l\right)  \Omega_{n-l}.
\label{equ.3.53}%
\end{align}
From this identity and simple estimates using Eq. (\ref{equ.3.44}) repeatedly
we find, for $f\in D\left(  \mathcal{N}^{\beta}\right)  ,$
\begin{align*}
\left\Vert \mathcal{A}_{M}f\right\Vert _{\beta}^{2}  &  =\sum_{n}\left\vert
\left\langle \mathcal{A}_{M}f,\Omega_{n}\right\rangle \right\vert ^{2}\left(
1+n\right)  ^{2\beta}\\
&  =\sum_{n}1_{0\leq n\leq M}1_{0\leq n-l\leq M}\left\vert \left\langle
f,\Omega_{n-l}\right\rangle \right\vert ^{2}\left\vert c_{\mathcal{A}}\left(
n-l\right)  \right\vert ^{2}\left(  1+n\right)  ^{2\beta}\\
&  =\sum_{n}1_{0\leq n+l\leq M}1_{0\leq n\leq M}\left\vert \left\langle
f,\Omega_{n}\right\rangle \right\vert ^{2}\left\vert c_{\mathcal{A}}\left(
n\right)  \right\vert ^{2}\left(  1+n+l\right)  ^{2\beta}\\
&  \leq\sum_{n}1_{0\leq n+l\leq M}1_{0\leq n\leq M}\left\vert \left\langle
f,\Omega_{n}\right\rangle \right\vert ^{2}\left(  k+n\right)  ^{k}\left(
1+n+\left\vert l\right\vert \right)  ^{2\beta}\\
&  \leq\left(  M+k\right)  ^{k}\left(  1+\left\vert l\right\vert \right)
^{2\beta}\sum_{n}1_{0\leq n+l\leq M}1_{0\leq n\leq M}\left\vert \left\langle
f,\Omega_{n}\right\rangle \right\vert ^{2}\left(  1+n\right)  ^{2\beta}\\
&  \leq\left(  M+k\right)  ^{k}\left(  1+\left\vert \ell\left(  \mathbf{b}%
\right)  \right\vert \right)  ^{2\beta}\left\Vert f\right\Vert _{\beta}%
^{2}\leq k^{k}\left(  M+1\right)  ^{k}\left(  1+\left\vert \ell\left(
\mathbf{b}\right)  \right\vert \right)  ^{2\beta}\left\Vert f\right\Vert
_{\beta}^{2}.
\end{align*}

\end{proof}

\begin{theorem}
\label{the.3.48}Let $k\in\mathbb{N},$ $\beta\geq0,$ $\mathbf{b}\in\left\{
\theta,\theta^{\ast}\right\}  ^{k},$ and $\mathcal{A}=u_{\mathbf{b}}\left(
a,a^{\dag}\right)  $ be as in Notation \ref{not.3.1}. If $\alpha\geq
\beta+k/2,$ then
\begin{equation}
\left\Vert \mathcal{\bar{A}}-\mathcal{A}_{M}\right\Vert _{\alpha
\rightarrow\beta}\leq\left(  M-k+2\right)  ^{\left(  \beta+k/2-\alpha\right)
}\text{ for all }M\geq k. \label{equ.3.55}%
\end{equation}
Consequently, if $\alpha>\beta+k/2,$ then
\begin{equation}
\lim_{M\rightarrow\infty}\left\Vert \left(  \mathcal{\bar{A}}-\mathcal{A}%
_{M}\right)  \varphi\right\Vert _{\beta}^{2}=0~\forall~\varphi\in D\left(
\mathcal{N}^{\alpha}\right)  . \label{equ.3.56}%
\end{equation}

\end{theorem}

\begin{proof}
Let $M\geq k.$ From Proposition \ref{pro.3.39}, $\mathcal{\bar{A}}%
-\mathcal{A}_{M}$ is a bounded operator from $\left(  D\left(  \mathcal{N}%
^{\alpha}\right)  ,\left\Vert \cdot\right\Vert _{\alpha}\right)  $ to $\left(
D\left(  \mathcal{N}^{\beta}\right)  ,\left\Vert \cdot\right\Vert _{\beta
}\right)  .$ Making use of Eq. (\ref{equ.3.53}) we find
\begin{align*}
\left(  \mathcal{A}^{\dag}-\mathcal{P}_{M}\mathcal{A}^{\dag}\mathcal{P}%
_{M}\right)  \Omega_{n}  &  =c_{\mathcal{A}}\left(  n-l\right)  \left[
1-1_{n\leq M}\cdot1_{n-l\leq M}\right]  \Omega_{n-l}\\
&  =c_{\mathcal{A}}\left(  n-l\right)  1_{n>M\wedge\left(  M+l\right)  }%
\Omega_{n-l}\text{ for all }n\in\mathbb{Z}.
\end{align*}
Hence, if $\varphi\in D\left(  \mathcal{N}^{\alpha}\right)  \subset D\left(
\mathcal{N}^{k/2}\right)  =D\left(  \mathcal{\bar{A}}\right)  ,$ then
\begin{align*}
\left\Vert \left(  \mathcal{\bar{A}}-\mathcal{A}_{M}\right)  \varphi
\right\Vert _{\beta}^{2}  &  =\sum_{n}\left\vert \left\langle \left(
\mathcal{\bar{A}}-\mathcal{A}_{M}\right)  \varphi,\Omega_{n}\right\rangle
\right\vert ^{2}\left(  n+1\right)  ^{2\beta}\\
&  =\sum_{n}\left\vert \left\langle \varphi,\left(  \mathcal{A}^{\dag
}-\mathcal{P}_{M}\mathcal{A}^{\dag}\mathcal{P}_{M}\right)  \Omega
_{n}\right\rangle \right\vert ^{2}\left(  n+1\right)  ^{2\beta}\\
&  =\sum_{n}1_{n>M\wedge\left(  M+l\right)  }\left(  n+1\right)  ^{2\beta
}\left\vert \left\langle \varphi,\Omega_{n-l}\right\rangle \right\vert
^{2}\left\vert c_{\mathcal{A}}\left(  n-l\right)  \right\vert ^{2}\\
&  =\sum_{n}1_{n+l>M\wedge\left(  M+l\right)  }\left(  n+l+1\right)  ^{2\beta
}\left\vert \left\langle \varphi,\Omega_{n}\right\rangle \right\vert
^{2}\left\vert c_{\mathcal{A}}\left(  n\right)  \right\vert ^{2}\\
&  =\sum_{n}\rho\left(  n\right)  \left(  n+1\right)  ^{2\alpha}\left\vert
\left\langle \varphi,\Omega_{n}\right\rangle \right\vert ^{2}\leq\max_{n}%
\rho\left(  n\right)  \left\Vert \varphi\right\Vert _{\alpha}^{2}%
\end{align*}
where
\[
\rho\left(  n\right)  :=1_{n+l>M\wedge\left(  M+l\right)  }\frac{\left(
n+l+1\right)  ^{2\beta}}{\left(  n+1\right)  ^{2\alpha}}\left\vert
c_{\mathcal{A}}\left(  n\right)  \right\vert ^{2}.
\]
This completes the proof since simple estimates using Lemma \ref{lem.3.18} and
the fact that $n\geq M-k+1$ shows,
\[
\rho\left(  n\right)  \leq k^{k}\left(  k+1\right)  ^{2\beta}\left(
M-k+2\right)  ^{2\left(  \beta+k/2-\alpha\right)  }.
\]

\end{proof}

\begin{corollary}
\label{cor.3.50}If $P\left(  \theta,\theta^{\ast}\right)  \in\mathbb{C}%
\left\langle \theta,\theta^{\ast}\right\rangle ,$ $d=\deg_{\theta}P,$
$\beta\geq0,$ and $\alpha\geq\beta+d/2,$ then for any $M\geq d,$
\begin{align}
\left\Vert \left[  P\left(  a,a^{\dag}\right)  \right]  _{M}-P\left(  \bar
{a},a^{\ast}\right)  \right\Vert _{\alpha\rightarrow\beta}  &  \leq\sum
_{k=0}^{d}\left\vert P_{k}\right\vert \left(  M-k+2\right)  ^{\left(
\beta+k/2-\alpha\right)  }\nonumber\\
&  \leq\left(  M-d+2\right)  ^{\left(  \beta+d/2-\alpha\right)  }\left\vert
P\right\vert . \label{equ.3.57}%
\end{align}

\end{corollary}

\begin{proof}
This result a simple consequence of Theorem \ref{the.3.48}, the triangle
inequality, and the elementary estimate,
\[
\left(  M-k+2\right)  ^{\left(  \beta+k/2-\alpha\right)  }\leq\left(
M-d+2\right)  ^{\left(  \beta+d/2-\alpha\right)  }\text{ for }0\leq k\leq d.
\]

\end{proof}

\begin{proposition}
\label{pro.3.52}If $P\left(  \theta,\theta^{\ast}\right)  \in\mathbb{C}%
\left\langle \theta,\theta^{\ast}\right\rangle $ is as in Eq. (\ref{equ.2.20})
and $\left\vert P_{k}\right\vert $ is as in Eq. (\ref{equ.2.23}), then for all
$\beta\geq0,$
\begin{align}
&  \left\Vert \left[  \left(  \mathcal{N}+1\right)  ^{\beta},P\left(
a,a^{\dag}\right)  _{M}\right]  \left(  \mathcal{N}+1\right)  ^{-\beta
}\right\Vert _{0\rightarrow0}\nonumber\\
&  \qquad\leq\sum_{k=1}^{d}\beta k^{k/2}k\left(  1+k\right)  ^{\left\vert
\beta-1\right\vert }\left(  M+1\right)  ^{\left(  k/2-1\right)  _{+}%
}\left\vert P_{k}\right\vert \label{equ.3.58}\\
&  \qquad\leq K\left(  \beta,d\right)  \cdot\sum_{k=1}^{d}\left(  M+1\right)
^{\left(  k/2-1\right)  _{+}}\left\vert P_{k}\right\vert \label{equ.3.59}%
\end{align}
where
\begin{equation}
K\left(  \beta,d\right)  :=\beta d^{1+\frac{d}{2}}\left(  1+d\right)
^{\left\vert \beta-1\right\vert }. \label{equ.3.60}%
\end{equation}

\end{proposition}

\begin{proof}
If $f\in L^{2}\left(  m\right)  ,$ $\mathbf{b\in}\left\{  \theta,\theta^{\ast
}\right\}  ^{k}$ and $\mathcal{A}_{\mathbf{b}}:=u_{\mathbf{b}}\left(
a,a^{\dag}\right)  ,$ then by Proposition \ref{pro.3.43}),
\begin{align*}
\left\Vert \left[  \left(  \mathcal{N}+1\right)  ^{\beta},\left[
\mathcal{A}_{\mathbf{b}}\right]  _{M}\right]  \left(  \mathcal{N}+1\right)
^{-\beta}f\right\Vert  &  =\left\Vert \left[  \left(  \mathcal{N}+1\right)
^{\beta},\mathcal{P}_{M}\mathcal{A}_{\mathbf{b}}\mathcal{P}_{M}\right]
\left(  \mathcal{N}+1\right)  ^{-\beta}f\right\Vert \\
&  =\left\Vert \mathcal{P}_{M}\left[  \left(  \mathcal{N}+1\right)  ^{\beta
},\mathcal{A}_{\mathbf{b}}\right]  \left(  \mathcal{N}+1\right)  ^{-\beta
}\mathcal{P}_{M}f\right\Vert \\
&  \leq\beta k^{k/2}k\left(  1+k\right)  ^{\left\vert \beta-1\right\vert
}\left\Vert \left(  \mathcal{N}+1\right)  ^{k/2-1}\mathcal{P}_{M}f\right\Vert
\\
&  \leq\beta k^{k/2}k\left(  1+k\right)  ^{\left\vert \beta-1\right\vert
}\left(  M+1\right)  ^{\left(  k/2-1\right)  _{+}}\left\Vert f\right\Vert .
\end{align*}
Hence $P\in\mathbb{C}\left\langle \theta,\theta^{\ast}\right\rangle $ with
$d=\deg_{\theta}P$ is given as in Eq. (\ref{equ.2.20}) (so that $P\left(
a,a^{\dag}\right)  $ is as in Eq. (\ref{equ.2.28}) with $\hbar=1$), then by
the triangle inequality we find,
\begin{align*}
&  \left\Vert \left[  \left(  \mathcal{N}+1\right)  ^{\beta},P\left(
a,a^{\dag}\right)  _{M}\right]  \left(  \mathcal{N}+1\right)  ^{-\beta
}\right\Vert _{0\rightarrow0}\\
&  \qquad\leq\sum_{k=1}^{d}\left\Vert \left[  \left(  \mathcal{N}+1\right)
^{\beta},P_{k}\left(  a,a^{\dag}\right)  _{M}\right]  \left(  \mathcal{N}%
+1\right)  ^{-\beta}\right\Vert _{0\rightarrow0}\\
&  \qquad\leq\sum_{k=1}^{d}\beta k^{k/2}k\left(  1+k\right)  ^{\left\vert
\beta-1\right\vert }\left(  M+1\right)  ^{\left(  k/2-1\right)  _{+}%
}\left\vert P_{k}\right\vert
\end{align*}
where the absence of the $k=0$ term is a consequence $P_{0}\left(  a,a^{\dag
}\right)  _{M}$ is proportional to $\mathcal{P}_{M}$ and hence commutes with
$\left(  \mathcal{N}+1\right)  ^{\beta}.$
\end{proof}

\section{Basic Linear ODE Results\label{sec.4}}

\begin{notation}
\label{not.4.1} If $\left(  X,\left\Vert \cdot\right\Vert \right)  $ is a
Banach space, then $B\left(  X\right)  $ is notated as a collection of bounded
linear operators from $X$ to itself and $\left\Vert \cdot\right\Vert
_{B\left(  X\right)  }$ is denoted as an operator norm. (e.g. $\left(
B\left(  D\left(  \mathcal{N}^{\beta}\right)  \right)  ,\left\Vert
\cdot\right\Vert _{\beta\rightarrow\beta}\right)  $ in Notation \ref{not.3.37}%
. )
\end{notation}

\begin{lemma}
[Basic Linear ODE Theorem]\label{lem.4.2} Suppose that $\left(  X,\left\Vert
\cdot\right\Vert \right)  $ is a Banach space and $t\rightarrow C\left(
t\right)  \in B\left(  X\right)  $ is an operator norm continuous map. Then to
each $s\in\mathbb{R}$ there exists a unique solution, $U\left(  t,s\right)
\in B\left(  X\right)  ,$ to the ordinary differential equation,
\begin{equation}
\frac{d}{dt}U\left(  t,s\right)  =C\left(  t\right)  U\left(  t,s\right)
\text{ with }U\left(  s,s\right)  =I. \label{equ.4.1}%
\end{equation}
Moreover, the function $\left(  t,s\right)  \rightarrow U\left(  t,s\right)
\in B\left(  X\right)  $ is operator norm continuously differentiable in each
of its variables and $\left(  t,s\right)  \rightarrow\partial_{t}U\left(
t,s\right)  $ and $\left(  t,s\right)  \rightarrow\partial_{s}U\left(
t,s\right)  $ are operator norm continuous functions into $B\left(  X\right)
,$
\begin{align*}
\partial_{s}U\left(  t,s\right)   &  =-U\left(  t,s\right)  C\left(  s\right)
\text{ with }U\left(  t,t\right)  =I,\text{ and}\\
U\left(  t,s\right)  U\left(  s,\sigma\right)   &  =U\left(  t,\sigma\right)
\text{ for all }s,\sigma,t\in\mathbb{R}.
\end{align*}

\end{lemma}

\begin{proof}
Let $V\left(  t\right)  $ and $W\left(  t\right)  $ in $B\left(  X\right)  $
solve the ordinary differential equations,
\begin{align*}
\frac{d}{dt}V\left(  t\right)   &  =C\left(  t\right)  V\left(  t\right)
\text{ with }V\left(  0\right)  =I\text{ and }\\
\frac{d}{dt}W\left(  t\right)   &  =-W\left(  t\right)  C\left(  t\right)
\text{ with }W\left(  0\right)  =I.\text{ }%
\end{align*}
We then have
\[
\frac{d}{dt}\left[  W\left(  t\right)  V\left(  t\right)  \right]  =-W\left(
t\right)  C\left(  t\right)  V\left(  t\right)  +W\left(  t\right)  C\left(
t\right)  V\left(  t\right)  =0
\]
so that $W\left(  t\right)  V\left(  t\right)  =I$ for all $t.$ Moreover,
$Z\left(  t\right)  :=V\left(  t\right)  W\left(  t\right)  $ solves the
differential equation,
\begin{align*}
\frac{d}{dt}Z\left(  t\right)   &  =-V\left(  t\right)  W\left(  t\right)
C\left(  t\right)  +C\left(  t\right)  V\left(  t\right)  W\left(  t\right) \\
&  =\left[  C\left(  t\right)  ,Z\left(  t\right)  \right]  \text{ with
}Z\left(  0\right)  =V\left(  0\right)  W\left(  0\right)  =I.
\end{align*}
The unique solution to this differential equation is $Z\left(  t\right)  =I$
from which we conclude $V\left(  t\right)  W\left(  t\right)  =I$ for all
$t\in\mathbb{R}.$ In summary, we have shown $W\left(  t\right)  $ and
$V\left(  t\right)  $ are inverses of one another. It is now easy to check
that
\[
U\left(  t,s\right)  =V\left(  t\right)  V\left(  s\right)  ^{-1}=V\left(
t\right)  W\left(  s\right)
\]
from which all of the rest of the stated results easily follow.
\end{proof}

\begin{proposition}
[Operator Norm Bounds]\label{pro.4.3} Suppose that $\left(  K,\left\langle
\cdot,\cdot\right\rangle \right)  $ is a Hilbert space, $A,$ is a self-adjoint
operators on $K$ with $A\geq I,$ and make $D\left(  A\right)  $ into a Hilbert
space using the inner product, $\left\langle \cdot,\cdot\right\rangle _{A},$
defined by
\[
\left\langle \psi,\varphi\right\rangle _{A}:=\left\langle A\psi,A\varphi
\right\rangle \text{ for all }\varphi,\psi\in D\left(  A\right)  .
\]
Further suppose that $t\rightarrow C\left(  t\right)  \in B\left(  K\right)
$[see Notation \ref{not.4.1}] is a $\left\Vert \cdot\right\Vert _{K}$-operator
norm continuous map such that $C\left(  t\right)  D\left(  A\right)  \subset
D\left(  A\right)  $ for all $t$ and the map $t\rightarrow C\left(  t\right)
|_{D\left(  A\right)  }\in B\left(  D\left(  A\right)  \right)  $ is
$\left\Vert \cdot\right\Vert _{A}$-operator norm continuous. Let $U\left(
t,s\right)  \in B\left(  K\right)  $ be as in Lemma \ref{lem.4.2}. Then,

\begin{enumerate}
\item $U\left(  t,s\right)  D\left(  A\right)  \subset D\left(  A\right)  $
for all $s,t\in\mathbb{R},$ and
\[
U\left(  t,s\right)  U\left(  s,\sigma\right)  =U\left(  t,\sigma\right)  .
\]

\item $U\left(  t,s\right)  |_{D\left(  A\right)  }$ solves
\[
\frac{d}{dt}U\left(  t,s\right)  |_{D\left(  A\right)  }=C\left(  t\right)
|_{D\left(  A\right)  }U\left(  t,s\right)  |_{D\left(  A\right)  }\text{ with
}U\left(  s,s\right)  |_{D\left(  A\right)  }=I_{D\left(  A\right)  }%
\]
where the derivative on the left side of this equation is taken relative to
the operator norm on the Hilbert space, $\left(  D\left(  A\right)
,\left\langle \cdot,\cdot\right\rangle _{A}\right)  .$

\item For all $s,t\in\mathbb{R},$
\begin{equation}
\left\Vert U\left(  t,s\right)  \right\Vert _{B\left(  K\right)  }\leq
\exp\left(  \frac{1}{2}\left\vert \int_{s}^{t}\left\Vert C\left(  \tau\right)
+C^{\ast}\left(  \tau\right)  \right\Vert _{B\left(  K\right)  }%
d\tau\right\vert \right)  \label{equ.4.2}%
\end{equation}
where $\left\Vert \cdot\right\Vert _{B\left(  K\right)  }$ is as in Notation
\ref{not.4.1}. Moreover, $U\left(  t,s\right)  $ is unitary on $K$ if
$C\left(  t\right)  $ is skew adjoint for all $t\in\mathbb{R}.$

\item For all $s,t\in\mathbb{R},$
\begin{align}
&  \left\Vert U\left(  t,s\right)  \right\Vert _{B\left(  D\left(  A\right)
\right)  }\nonumber\\
&  \quad\leq\exp\left(  \left\vert \int_{s}^{t}\left[  \frac{1}{2}\left\Vert
C\left(  \tau\right)  +C^{\ast}\left(  \tau\right)  \right\Vert _{B\left(
K\right)  }+\left\Vert \left[  A,C\left(  \tau\right)  \right]  A^{-1}%
\right\Vert _{B\left(  K\right)  }\right]  d\tau\right\vert \right)  .
\label{equ.4.3}%
\end{align}
and
\begin{align}
&  \left\Vert U\left(  t,s\right)  \right\Vert _{B\left(  D\left(  A\right)
\right)  }\nonumber\\
&  \quad\geq\exp\left(  -\left\vert \int_{s}^{t}\left[  \frac{1}{2}\left\Vert
C\left(  \tau\right)  +C^{\ast}\left(  \tau\right)  \right\Vert _{B\left(
K\right)  }+\left\Vert \left[  A,C\left(  \tau\right)  \right]  A^{-1}%
\right\Vert _{B\left(  K\right)  }\right]  d\tau\right\vert \right)
\label{equ.4.4}%
\end{align}
where $\left\Vert \left[  A,C\left(  \tau\right)  \right]  A^{-1}\right\Vert
_{B\left(  K\right)  }$ is defined to be $\infty$ if $\left[  A,C\left(
\tau\right)  \right]  A^{-1}$ is an unbounded operator on $K.$
\end{enumerate}
\end{proposition}

\begin{proof}
Let $U\left(  t,s\right)  $ be as in Lemma \ref{lem.4.2} when $X=K$ and
$U_{A}\left(  t,s\right)  $ be as in Lemma \ref{lem.4.2} when $X=D\left(
A\right)  .$ Further suppose that $\psi_{0}\in D\left(  A\right)  $ and let
$\psi\left(  t\right)  :=U\left(  t,s\right)  \psi_{0}$ and $\psi_{A}\left(
t\right)  :=U_{A}\left(  t,s\right)  \psi_{0}.$ We now prove each item in turn.

\begin{enumerate}
\item Since $\psi\left(  t\right)  $ and $\psi_{A}\left(  t\right)  $ both
solve the differential equation (in the $K$ -- norm) [Note: $\left\Vert
\cdot\right\Vert _{A}\geq\left\Vert \cdot\right\Vert _{K}$]
\begin{equation}
\dot{\varphi}\left(  t\right)  =C\left(  t\right)  \varphi\left(  t\right)
\text{ with }\varphi\left(  s\right)  =\psi_{0}, \label{equ.4.5}%
\end{equation}
it follows by the uniqueness of solutions to ODE that
\[
U\left(  t,s\right)  \psi_{0}=\psi\left(  t\right)  =\psi_{A}\left(  t\right)
=U_{A}\left(  t,s\right)  \psi_{0}\in D\left(  A\right)  .
\]
The results of items 1. and 2. now easily follow.

\item It is well known and easily verified that $U\left(  t,s\right)  $ is
unitary on $K$ if $C\left(  t\right)  $ is skew adjoint. The estimate in Eq.
(\ref{equ.4.2}) is a special case of the estimate in Eq. (\ref{equ.4.3}) when
$A=I$ so it suffices to prove the latter estimate.

\item With $\psi\left(  t\right)  =U\left(  t,s\right)  \psi_{0}=U_{A}\left(
t,s\right)  \psi_{0}\in D\left(  A\right)  $ as above we have,
\begin{align*}
\frac{d}{dt}\left\Vert \psi\right\Vert _{A}^{2}  &  =2\operatorname{Re}%
\left\langle C\psi,\psi\right\rangle _{A}=2\operatorname{Re}\left\langle
AC\psi,A\psi\right\rangle \\
&  =2\operatorname{Re}\left[  \left\langle CA\psi,A\psi\right\rangle
+\left\langle \left[  A,C\right]  \psi,A\psi\right\rangle \right] \\
&  =\left\langle \left(  C+C^{\ast}\right)  A\psi,A\psi\right\rangle
+2\operatorname{Re}\left\langle \left[  A,C\right]  A^{-1}A\psi,A\psi
\right\rangle
\end{align*}
and therefore,
\[
\left\vert \frac{d}{dt}\left\Vert \psi\right\Vert _{A}^{2}\right\vert
\leq\left(  \left\Vert C+C^{\ast}\right\Vert _{B\left(  K\right)
}+2\left\Vert \left[  A,C\right]  A^{-1}\right\Vert _{B\left(  K\right)
}\right)  \left\Vert \psi\right\Vert _{A}^{2}.
\]
This last inequality may be integrated to find,
\[
\left(  \frac{\left\Vert \psi\left(  t\right)  \right\Vert _{A}^{2}%
}{\left\Vert \psi_{0}\right\Vert _{A}^{2}}\right)  ^{\pm1}\leq\exp\left(
\left\vert \int_{s}^{t}\left[  \left\Vert C\left(  \tau\right)  +C^{\ast
}\left(  \tau\right)  \right\Vert _{B\left(  K\right)  }+2\left\Vert \left[
A,C\left(  \tau\right)  \right]  A^{-1}\right\Vert _{B\left(  K\right)
}\right]  d\tau\right\vert \right)
\]
from which Eqs. (\ref{equ.4.3}) and (\ref{equ.4.4}) easily follow.
\end{enumerate}
\end{proof}

\subsection{Truncated Evolutions\label{sec.4.1}}

Now suppose that $P\left(  t:\theta,\theta^{\ast}\right)  \in\mathbb{C}%
\left\langle \theta,\theta^{\ast}\right\rangle $ with $\deg_{\theta}P\left(
t:\theta,\theta^{\ast}\right)  =d\in\mathbb{N}$ is a one parameter family of
\textbf{symmetric} non-commutative polynomials whose coefficients depend
continuously on $t.$ In more detail we may write $P\left(  t:\theta
,\theta^{\ast}\right)  $ as;
\begin{align}
P\left(  t:\theta,\theta^{\ast}\right)   &  =\sum_{k=0}^{d}P_{k}\left(
t:\theta,\theta^{\ast}\right)  \text{ where}\label{equ.4.6}\\
P_{k}\left(  t:\theta,\theta^{\ast}\right)   &  =\sum_{\mathbf{b}\in\left\{
\theta,\theta^{\ast}\right\}  ^{k}}c_{k}\left(  t,\mathbf{b}\right)
u_{\mathbf{b}}\left(  \theta,\theta^{\ast}\right)  \label{equ.4.7}%
\end{align}
and all coefficients, $t\rightarrow c_{k}\left(  t,\mathbf{b}\right)  $ are
continuous in $t.$ Let $Q\left(  t\right)  :=P\left(  t:a,a^{\dag}\right)  $
and for any $M>0$ let $Q_{M}\left(  t\right)  =\mathcal{P}_{M}Q\left(
t\right)  \mathcal{P}_{M}$ be the truncation of $Q\left(  t\right)  $ as in
Notation \ref{not.3.45}. Applying Lemma \ref{lem.4.2} with $C\left(  t\right)
=-iQ_{M}\left(  t\right)  $ shows, for each $M\in\mathbb{N}$ there exists
$U^{M}\left(  t,s\right)  \in B\left(  L^{2}\left(  m\right)  \right)  $ such
that for all $s\in\mathbb{R},$
\begin{equation}
i\frac{d}{dt}U^{M}\left(  t,s\right)  =Q_{M}\left(  t\right)  U^{M}\left(
t,s\right)  \text{ with }U^{M}\left(  s,s\right)  =I. \label{equ.4.8}%
\end{equation}

\begin{theorem}
\label{the.4.5}Let $M>0$ and $U^{M}\left(  t,s\right)  $ be defined as in Eq.
(\ref{equ.4.8}). Then;

\begin{enumerate}
\item $\left(  t,s\right)  \rightarrow U^{M}\left(  t,s\right)  \in B\left(
L^{2}\left(  m\right)  \right)  $ are jointly operator norm continuous in
$\left(  t,s\right)  $ and $U^{M}\left(  t,s\right)  $ is unitary on
$L^{2}\left(  m\right)  $ for each $t,s\in\mathbb{R}.$

\item If $\sigma,s,t\in\mathbb{R},$ then
\begin{equation}
U^{M}\left(  t,s\right)  U^{M}\left(  s,\sigma\right)  =U^{M}\left(
t,\sigma\right)  . \label{equ.4.9}%
\end{equation}

\item If $\beta\geq0$ and $s,t\in\mathbb{R},$ then $U^{M}\left(  t,s\right)
D\left(  \mathcal{N}^{\beta}\right)  =D\left(  \mathcal{N}^{\beta}\right)  ,$
$U^{M}\left(  t,s\right)  |_{D\left(  \mathcal{N}^{\beta}\right)  }$ is
continuous in $\left(  t,s\right)  $ in the $\left\Vert \cdot\right\Vert
_{\beta}$-operator norm topology, $\partial_{t}U^{M}\left(  t,s\right)
|_{D\left(  \mathcal{N}^{\beta}\right)  },$ and $\partial_{s}U^{M}\left(
t,s\right)  |_{D\left(  \mathcal{N}^{\beta}\right)  }$ exists in the
$\left\Vert \cdot\right\Vert _{\beta}$-operator norm topology (see Notation
\ref{not.3.30}) and again are continuous functions of $\left(  t,s\right)  $
in this topology and satisfy
\begin{align}
i\frac{d}{dt}U^{M}\left(  t,s\right)  \varphi &  =Q_{M}\left(  t\right)
U^{M}\left(  t,s\right)  \varphi\label{equ.4.10}\\
i\frac{d}{ds}U^{M}\left(  t,s\right)  \varphi &  =-U^{M}\left(  t,s\right)
Q_{M}\left(  s\right)  \varphi.\text{ } \label{equ.4.11}%
\end{align}

\item If $\beta\geq0$ and $t,s\in\mathbb{R},$ then with $K\left(
\beta,d\right)  <\infty$ as in Eq. (\ref{equ.3.60}) we have%
\begin{equation}
\left\Vert U^{M}\left(  t,s\right)  \right\Vert _{\beta\rightarrow\beta}%
\leq\exp\left(  K\left(  \beta,d\right)  \sum_{k=1}^{d}\left(  M+1\right)
^{\left(  k/2-1\right)  _{+}}\int_{J_{st}}\left\vert P_{k}\left(  \tau
,\theta,\theta^{\ast}\right)  \right\vert d\tau\right)  . \label{equ.4.12}%
\end{equation}
where $J_{st}=\left[  \min\left(  s,t\right)  ,\max\left(  s,t\right)
\right]  ,$ and $\left\Vert \cdot\right\Vert _{\beta\rightarrow\beta}$ is as
in Notation \ref{not.3.37}, $P_{k}$ as in Eq. (\ref{equ.4.7}) and $K\left(
\beta,d\right)  $ is as in Eq. (\ref{equ.3.60}).
\end{enumerate}
\end{theorem}

\begin{remark}
\label{rem.4.6}Taking $t=\sigma$ in Eq. (\ref{equ.4.9}) and using the fact
that $U^{M}\left(  t,s\right)  $ is unitary on $L^{2}\left(  m\right)  ,$ it
follows that
\begin{equation}
U^{M}\left(  t,s\right)  ^{-1}=U^{M}\left(  s,t\right)  =U^{M}\left(
t,s\right)  ^{\ast}. \label{equ.4.13}%
\end{equation}

\end{remark}

\begin{remark}
\label{rem.4.7}From the item 3 of the Theorem and Eq. (\ref{equ.3.34}), we can
conclude that $U^{M}\left(  t,s\right)  \mathcal{S}=\mathcal{S}.$
\end{remark}

\begin{proof}
The continuity of $U^{M}$ in the item 1. and the identity in Eq.
(\ref{equ.4.9}) both follow from Lemma \ref{lem.4.2}. Since $Q_{M}\left(
t\right)  ^{\ast}=Q_{M}\left(  t\right)  $ it follows that $C\left(  t\right)
:=-iQ_{M}\left(  t\right)  $ is skew-adjoint and so the unitary property in
the first item is a consequence of item 3. of Proposition \ref{pro.4.3}. The
remaining item 3. and 4. follow from Proposition \ref{pro.4.3} with
$A:=\left(  \mathcal{N}+I\right)  ^{\beta}$ and $C\left(  t\right)
:=-iQ_{M}\left(  t\right)  .$ The hypothesis that $C\left(  t\right)  D\left(
A\right)  \subset D\left(  A\right)  $ and $t\rightarrow C\left(  t\right)
\in B\left(  D\left(  A\right)  \right)  $ is $\left\Vert \cdot\right\Vert
_{\beta}$-operator norm continuous in $t$ has been verified in Proposition
\ref{pro.3.46}. Moreover, from Eq. (\ref{equ.3.59}) of Proposition
\ref{pro.3.52} we know
\[
\left\Vert \left[  A,C\left(  \tau\right)  \right]  A^{-1}\right\Vert
_{B\left(  L^{2}\left(  m\right)  \right)  }\leq K\left(  \beta,d\right)
\sum_{k=1}^{d}\left(  M+1\right)  ^{\left(  k/2-1\right)  _{+}}\left\vert
P_{k}\left(  \tau,\theta,\theta^{\ast}\right)  \right\vert .
\]
Equation (\ref{equ.4.12}) now follows directly from Eq. (\ref{equ.4.3}) and
the fact that $C\left(  t\right)  $ is skew adjoint. Finally, the inclusion,
$U^{M}\left(  t,s\right)  D\left(  \mathcal{N}^{\beta}\right)  \subseteq
D\left(  \mathcal{N}^{\beta}\right)  ,$ follows by Proposition \ref{pro.4.3}.
The opposite inclusion is then deduced using $U^{M}\left(  t,s\right)
^{-1}=U^{M}\left(  s,t\right)  $ which follows from Eq. (\ref{equ.4.9}).
\end{proof}

\begin{corollary}
\label{cor.4.8} Recall $P\left(  t:\theta,\theta^{\ast}\right)  $ as in Eq.
(\ref{equ.4.6}). Let $\hbar>0,$ $M>0,$ $U_{\hbar}^{M}\left(  t,s\right)  $
denotes the solution to the ordinary differential equation,
\[
i\hbar\frac{d}{dt}U_{\hbar}^{M}\left(  t,s\right)  =\left[  P\left(
t:a_{\hbar},a_{\hbar}^{\dag}\right)  \right]  _{M}U_{\hbar}^{M}\left(
t,s\right)  \text{ with }U_{\hbar}^{M}\left(  s,s\right)  =I,
\]
If $\beta\geq0$ and $s,t\in\mathbb{R},$ then
\begin{equation}
\left\Vert U_{\hbar}^{M}\left(  t,s\right)  \right\Vert _{\beta\rightarrow
\beta}\leq e^{K\left(  \beta,d\right)  \sum_{k=1}^{d}\hbar^{k/2-1}\left(
M+1\right)  ^{\left(  k/2-1\right)  _{+}}\int_{J_{s,t}}\left\vert P_{k}\left(
\tau:\theta,\theta^{\ast}\right)  \right\vert d\tau}, \label{equ.4.14}%
\end{equation}
where $K\left(  \beta,d\right)  <\infty$ is as in Eq. (\ref{equ.3.60}). In
particular if $P_{1}\left(  t:\theta,\theta^{\ast}\right)  \equiv0,$ $\eta
\in(0,1],$and $0<\hbar\leq\eta\leq1,$ then
\begin{equation}
\left\Vert U_{\hbar}^{M}\left(  t,s\right)  \right\Vert _{\beta\rightarrow
\beta}\leq e^{K\left(  \beta,d\right)  \left(  \hbar M+1\right)  ^{\frac{d}%
{2}-1}\sum_{k=2}^{d}\int_{J_{s,t}}\left\vert P_{k}\left(  \tau:\theta
,\theta^{\ast}\right)  \right\vert d\tau}. \label{equ.4.15}%
\end{equation}

\end{corollary}

\begin{proof}
Since
\[
\frac{1}{\hbar}P_{k}\left(  t:a_{\hbar},a_{\hbar}^{\dag}\right)  =\frac
{1}{\hbar}\hbar^{k/2}P_{k}\left(  t:a,a^{\dag}\right)  =\hbar^{k/2-1}%
P_{k}\left(  t:a,a^{\dag}\right)  ,
\]
Eq. (\ref{equ.4.14}) follows from Theorem \ref{the.4.5} after making the
replacement,
\[
P\left(  t:,\theta,\theta^{\ast}\right)  \longrightarrow\sum_{k=0}^{d}%
\hbar^{k/2-1}P_{k}\left(  t:\theta,\theta^{\ast}\right)  .
\]
Equation (\ref{equ.4.15}) then follows from Eq. (\ref{equ.4.14}) since for
$2\leq k\leq d$ and $0<\hbar\leq\eta\leq1,$
\[
\hbar^{k/2-1}\left(  M+1\right)  ^{\left(  k/2-1\right)  _{+}}=\left(  \hbar
M+\hbar\right)  ^{\left(  k/2-1\right)  }\leq\left(  \hbar M+1\right)
^{\frac{d}{2}-1}.
\]

\end{proof}

\section{Quadratically Generated Unitary Groups\label{sec.5}}

Let $P\left(  t:\theta,\theta^{\ast}\right)  \in\mathbb{C}\left\langle
\theta,\theta^{\ast}\right\rangle $ be a continuously varying one parameter
family of \textbf{symmetric} polynomials with $d=\deg_{\theta}P\left(
t:\theta,\theta^{\ast}\right)  \leq2.$ Then $Q\left(  t\right)  :=P\left(
t:a,a^{\dag}\right)  $ may be decomposed as;
\begin{equation}
Q\left(  t\right)  =\sum_{j=0}^{6}c_{j}\left(  t\right)  \mathcal{A}^{\left(
j\right)  } \label{equ.5.1}%
\end{equation}
where $\mathcal{A}^{\left(  j\right)  }$ is a monomial in $a$ and $a^{\dag}$
of degree no bigger than $2$ and $c_{j}\left(  \cdot\right)  $ is continuous
for each $0\leq j\leq6$ and $\mathcal{A}^{\left(  0\right)  }=1$ by
convention. The main goal of this chapter is to record the relevant
information we need about solving the following time dependent Schr\"{o}dinger
equation;
\begin{equation}
i\dot{\psi}\left(  t\right)  =\overline{Q\left(  t\right)  }\psi\left(
t\right)  \text{ with }\psi\left(  s\right)  =\varphi, \label{equ.5.2}%
\end{equation}
where $s\in\mathbb{R}$ and $\varphi\in D\left(  \mathcal{N}\right)  $ and the
derivative is taken in $L^{2}\left(  m\right)  .$

\begin{theorem}
[Uniqueness of Solutions]\label{the.5.1}If $\mathbb{R\ni}t\rightarrow
\psi\left(  t\right)  \in D\left(  \mathcal{N}\right)  $ solves Eq.
(\ref{equ.5.2}) then $\left\Vert \psi\left(  t\right)  \right\Vert =\left\Vert
\varphi\right\Vert $ for all $t\in\mathbb{R}.$ Moreover, there is at most one
solution to Eq. (\ref{equ.5.2}).
\end{theorem}

\begin{proof}
If $\psi\left(  t\right)  $ solves Eq. (\ref{equ.5.2}), then because
$\overline{Q\left(  t\right)  }$ is symmetric on $D\left(  \mathcal{N}\right)
,$
\[
\frac{d}{dt}\left\Vert \psi\left(  t\right)  \right\Vert ^{2}%
=2\operatorname{Re}\left\langle \dot{\psi}\left(  t\right)  ,\psi\left(
t\right)  \right\rangle =2\operatorname{Re}\left\langle -i\overline{Q\left(
t\right)  }\psi\left(  t\right)  ,\psi\left(  t\right)  \right\rangle =0.
\]
Therefore it follows that $\left\Vert \psi\left(  t\right)  \right\Vert
^{2}=\left\Vert \psi\left(  s\right)  \right\Vert ^{2}=\left\Vert
\varphi\right\Vert ^{2}$ which proves the isometry property and because the
equation (\ref{equ.5.2}) is linear this also proves uniqueness of solutions.
\end{proof}

Theorem \ref{the.5.6} below (among other things) guarantees the existence of
solutions to Eq. (\ref{equ.5.2}). This result may be in fact be viewed as an
aspect of the well known metaplectic representation. Nevertheless, we will
provide a full proof as we need some detailed bounds on the solutions to Eq.
(\ref{equ.5.2}).

In order to prove existence to Eq. (\ref{equ.5.2}) we are going to construct
the evolution operator $U\left(  t,s\right)  $ associated to Eq.
(\ref{equ.5.2}) as a limit of the truncated evolution operators, $U^{M}\left(
t,s\right)  ,$ defined by Eq. (\ref{equ.4.8}) with $Q_{M}\left(  t\right)
=\mathcal{P}_{M}Q\left(  t\right)  \mathcal{P}_{M} $ where $Q\left(  t\right)
$ is as in Eq. (\ref{equ.5.1}). The next estimate provides uniform bounds on
$U^{M}\left(  t,s\right)  .$

\begin{corollary}
[Uniform Bounds]\label{cor.5.2}Continuing the notation above if $\beta\geq0,$
$-\infty<S<T<\infty,$ and $M\in\mathbb{N},$ then
\begin{equation}
\left\Vert U^{M}\left(  t,s\right)  \right\Vert _{\beta\rightarrow\beta}%
\leq\exp\left(  K\left(  \beta,S,T,P\right)  \left\vert t-s\right\vert
\right)  \text{ for all }S<s,t\leq T \label{equ.5.3}%
\end{equation}
where
\begin{equation}
K\left(  \beta,S,T,P\right)  =\beta4\cdot3^{\left\vert \beta-1\right\vert
}\sum_{j=1}^{6}\max_{\tau\in\left[  S,T\right]  }\left\vert c_{j}\left(
\tau\right)  \right\vert <\infty. \label{equ.5.4}%
\end{equation}

\end{corollary}

\begin{proof}
This result follows directly from Theorem \ref{the.4.5} and the assumed
continuity of the coefficients of $P\left(  t:\theta,\theta^{\ast}\right)  $
along with the assumption that $d=\deg_{\theta}P\left(  t:\theta,\theta^{\ast
}\right)  \leq2.$
\end{proof}

The next proposition will be a key ingredient in the proof of Proposition
\ref{pro.5.5} below which guarantees that $\lim_{M\rightarrow\infty}%
U^{M}\left(  t,s\right)  $ exists.

\begin{proposition}
\label{pro.5.3}If $\beta\in\mathbb{R}$ and $\psi\in D\left(  \mathcal{N}%
^{\beta+1}\right)  ,$ then for all $-\infty<S<T<\infty$
\begin{align}
\lim_{M\rightarrow\infty}\sup_{K<\infty}\sup_{S\leq s,\tau\leq T}\left\Vert
\left[  \overline{Q\left(  \tau\right)  }-Q_{M}\left(  \tau\right)  \right]
U^{K}\left(  \tau,s\right)  \psi\right\Vert _{\beta}  &  =0\text{
and}\label{equ.5.5}\\
\lim_{M\rightarrow\infty}\sup_{K<\infty}\sup_{S\leq s,\tau\leq T}\left\Vert
U^{K}\left(  \tau,s\right)  \left[  \overline{Q\left(  s\right)  }%
-Q_{M}\left(  s\right)  \right]  \psi\right\Vert _{\beta}  &  =0.
\label{equ.5.6}%
\end{align}

\end{proposition}

\begin{proof}
Let us express $Q\left(  t\right)  $ as in Eq. (\ref{equ.5.1}). Since
\begin{equation}
Q_{M}\left(  t\right)  =\sum_{j=0}^{6}c_{j}\left(  t\right)  \mathcal{A}%
_{M}^{(j)} \label{equ.5.7}%
\end{equation}
where $\mathcal{A}_{M}^{(j)}$ is the truncation of $\mathcal{A}^{(j)}$ as in
Notation \ref{not.3.45}, to complete the proof it suffices to show,
\begin{align}
\lim_{M\rightarrow\infty}\sup_{K<\infty}\sup_{S\leq s,\tau\leq T}\left\Vert
\left[  \mathcal{\bar{A}}-\mathcal{A}_{M}\right]  U^{K}\left(  \tau,s\right)
\psi\right\Vert _{\beta}  &  =0\text{ and}\label{equ.5.8}\\
\lim_{M\rightarrow\infty}\sup_{K<\infty}\sup_{S\leq s,\tau\leq T}\left\Vert
U^{K}\left(  \tau,s\right)  \left[  \mathcal{\bar{A}}-\mathcal{A}_{M}\right]
\psi\right\Vert _{\beta}  &  =0 \label{equ.5.9}%
\end{align}
where $\mathcal{A}$ is a monomial in $a$ and $a^{\dag}$ with degree $2$ or less.

According to Theorem \ref{the.3.48} and Corollary \ref{cor.5.2}, if $\psi\in
D\left(  \mathcal{N}^{\alpha}\right)  $ with $\alpha\geq\beta+1,$ then
\begin{align}
\left\Vert \left[  \mathcal{\bar{A}}-\mathcal{A}_{M}\right]  U^{K}\left(
\tau,s\right)  \psi\right\Vert _{\beta}  &  \leq\left\Vert \left[
\mathcal{\bar{A}}-\mathcal{A}_{M}\right]  U^{K}\left(  \tau,s\right)
\right\Vert _{\alpha\rightarrow\beta}\left\Vert \psi\right\Vert _{\alpha
}\nonumber\\
&  \leq\left\Vert \left[  \mathcal{\bar{A}}-\mathcal{A}_{M}\right]
\right\Vert _{\alpha\rightarrow\beta}\left\Vert U^{K}\left(  \tau,s\right)
\right\Vert _{\alpha\rightarrow\alpha}\left\Vert \psi\right\Vert _{\alpha
}\nonumber\\
&  \leq C\left(  \alpha,\beta,S,T,P\right)  \left(  M+1\right)  ^{\beta
+1-\alpha}\left\Vert \psi\right\Vert _{\alpha} \label{equ.5.10}%
\end{align}
and
\begin{align}
\left\Vert U^{K}\left(  \tau,s\right)  \left[  \mathcal{\bar{A}}%
-\mathcal{A}_{M}\right]  \psi\right\Vert _{\beta}  &  \leq\left\Vert
U^{K}\left(  \tau,s\right)  \right\Vert _{\beta\rightarrow\beta}\left\Vert
\left[  \mathcal{\bar{A}}-\mathcal{A}_{M}\right]  \psi\right\Vert _{\beta
}\nonumber\\
&  \leq\left\Vert U^{K}\left(  \tau,s\right)  \right\Vert _{\beta
\rightarrow\beta}\left\Vert \left[  \mathcal{\bar{A}}-\mathcal{A}_{M}\right]
\right\Vert _{\alpha\rightarrow\beta}\left\Vert \psi\right\Vert _{\alpha
}\nonumber\\
&  \leq\tilde{C}\left(  \alpha,\beta,S,T,P\right)  \left(  M+1\right)
^{\beta+1-\alpha}\left\Vert \psi\right\Vert _{\alpha} \label{equ.5.11}%
\end{align}
from which Eqs. (\ref{equ.5.8}) and (\ref{equ.5.9}) follow if $\psi\in
D\left(  \mathcal{N}^{\alpha}\right)  $ with $\alpha>\beta+1.$

The general case, $\alpha=\beta+1,$ follows by a standard \textquotedblleft%
$3\varepsilon$\textquotedblright argument, the uniform (in $M>0)$ estimates in
Eq. (\ref{equ.5.10}) and (\ref{equ.5.11}) and the density of $\mathcal{S}%
_{0}\subset\mathcal{S}\subset D\left(  \mathcal{N}^{\beta+1}\right)  $ from
Proposition \ref{pro.3.32}.
\end{proof}

\begin{proposition}
\label{pro.5.5}If $\beta\geq0,$ $-\infty<S<T<\infty$ and $\psi\in D\left(
\mathcal{N}^{\beta}\right)  $, then it follows that
\begin{equation}
\lim_{M,K\rightarrow\infty}\sup_{S\leq s,t\leq T}\left\Vert \left[
U^{K}\left(  t,s\right)  -U^{M}\left(  t,s\right)  \right]  \psi\right\Vert
_{\beta}=0. \label{equ.5.12}%
\end{equation}

\end{proposition}

\begin{proof}
By item 3 in Theorem \ref{the.4.5}, we have
\begin{equation}
i\frac{d}{dt}\left[  U^{M}\left(  s,t\right)  U^{K}\left(  t,s\right)
\right]  =U^{M}\left(  s,t\right)  \left[  Q_{K}\left(  t\right)
-Q_{M}\left(  t\right)  \right]  U^{K}\left(  t,s\right)  \label{equ.5.13}%
\end{equation}
in the sense of $\left\Vert \cdot\right\Vert _{\beta}$-operator norm.
Integrating the identity Eq. (\ref{equ.5.13}) gives
\begin{equation}
U^{M}\left(  s,t\right)  U^{K}\left(  t,s\right)  =I-i\int_{s}^{t}U^{M}\left(
s,\tau\right)  \left[  Q_{K}\left(  \tau\right)  -Q_{M}\left(  \tau\right)
\right]  U^{K}\left(  \tau,s\right)  d\tau.\label{equ.5.14}%
\end{equation}
Using Eq. (\ref{equ.4.9}) in Theorem \ref{the.4.5} and multiplying this
identity by $U^{M}\left(  t,s\right)  $ then shows,
\[
U^{K}\left(  t,s\right)  -U^{M}\left(  t,s\right)  =-i\int_{s}^{t}U^{M}\left(
t,\tau\right)  \left[  Q_{K}\left(  \tau\right)  -Q_{M}\left(  \tau\right)
\right]  U^{K}\left(  \tau,s\right)  d\tau.
\]
Applying this equation to $\psi\in D\left(  \mathcal{N}^{\beta+1}\right)  $
and then making use of Corollary \ref{cor.5.2} and the triangle inequality for
integrals shows,
\begin{align*}
&  \left\Vert \left[  U^{K}\left(  t,s\right)  -U^{M}\left(  t,s\right)
\right]  \psi\right\Vert _{\beta}\\
&  \qquad\leq\left\vert \int_{s}^{t}\left\Vert U^{M}\left(  t,\tau\right)
\left[  Q_{K}\left(  \tau\right)  -Q_{M}\left(  \tau\right)  \right]
U^{K}\left(  \tau,s\right)  \psi\right\Vert _{\beta}d\tau\right\vert \\
&  \qquad\leq\int_{s}^{t}\left\Vert U^{M}\left(  t,\tau\right)  \right\Vert
_{\beta\rightarrow\beta}\left\Vert \left[  Q_{K}\left(  \tau\right)
-Q_{M}\left(  \tau\right)  \right]  U^{K}\left(  \tau,s\right)  \psi
\right\Vert _{\beta}d\tau\\
&  \qquad\leq K\left(  \beta,S,T\right)  \left\vert \int_{s}^{t}\left\Vert
\left[  Q_{K}\left(  \tau\right)  -Q_{M}\left(  \tau\right)  \right]
U^{K}\left(  \tau,s\right)  \psi\right\Vert _{\beta}d\tau\right\vert \\
&  \qquad\leq K\left(  \beta,S,T\right)  \left\vert \int_{s}^{t}\left\Vert
\left[  Q_{K}\left(  \tau\right)  -\bar{Q}\left(  \tau\right)  \right]
U^{K}\left(  \tau,s\right)  \psi\right\Vert _{\beta}d\tau\right\vert \\
&  \qquad\qquad+K\left(  \beta,S,T\right)  \left\vert \int_{s}^{t}\left\Vert
\left[  \bar{Q}\left(  \tau\right)  -Q_{M}\left(  \tau\right)  \right]
U^{K}\left(  \tau,s\right)  \psi\right\Vert _{\beta}d\tau\right\vert
\end{align*}
and the latter expression tends to zero locally uniformly in $\left(
t,s\right)  $ as $K,M\rightarrow\infty$ by Proposition \ref{pro.5.3}. This
proves Eq. (\ref{equ.5.12}) for $\psi\in D\left(  \mathcal{N}^{\beta
+1}\right)  .$ Note that $\mathcal{S}$ is dense in $\left(  D\left(
\mathcal{N}^{\beta}\right)  ,\left\Vert \cdot\right\Vert _{\beta}\right)  $
from Proposition \ref{pro.3.32}. The uniform estimate in Eq. (\ref{equ.5.3})
of Corollary \ref{cor.5.2} along with a standard density argument shows Eq.
(\ref{equ.5.12}) holds for $\psi\in D\left(  \mathcal{N}^{\beta}\right)  .$
\end{proof}

\begin{theorem}
\label{the.5.6}Let $Q\left(  t\right)  :=P\left(  t:a,a^{\dag}\right)  $ be as
above, i.e. $P$ is a symmetric non-commutative polynomial of $\left\{
\theta,\theta^{\ast}\right\}  $ of $\deg_{\theta}P\leq2$ and having
coefficients depending continuously on $t\in\mathbb{R}.$ Then there exists a
unique strongly continuous family of unitary operators $\left\{  U\left(
t,s\right)  \right\}  _{t,s\in\mathbb{R}}$ on $L^{2}\left(  m\right)  $ such
that for all $\varphi\in D\left(  \mathcal{N}\right)  ,$ $\psi\left(
t\right)  :=U\left(  t,s\right)  \varphi$ solves Eq. (\ref{equ.5.2}).
Furthermore $\left\{  U\left(  t,s\right)  \right\}  _{t,s\in\mathbb{R}}$
satisfies the following properties;

\begin{enumerate}
\item For all $s,t,\tau\in\mathbb{R}$ we have
\begin{equation}
U\left(  t,s\right)  =U\left(  t,\tau\right)  U\left(  \tau,s\right)  .
\label{equ.5.15}%
\end{equation}

\item For all $\beta\geq0$ and $s,t\in\mathbb{R},$ $U\left(  t,s\right)
D\left(  \mathcal{N}^{\beta}\right)  =D\left(  \mathcal{N}^{\beta}\right)  $
and $\left(  t,s\right)  \rightarrow U\left(  t,s\right)  \varphi$ are jointly
$\left\Vert \cdot\right\Vert _{\beta}$-norm continuous for all $\varphi\in
D\left(  \mathcal{N}^{\beta}\right)  .$

\item If $-\infty<S<T<\infty,$ then
\begin{equation}
C\left(  \beta,S,T\right)  :=\sup_{S\leq s,t\leq T}\left\Vert U\left(
t,s\right)  \right\Vert _{\beta\rightarrow\beta}<\infty. \label{equ.5.16}%
\end{equation}

\item For $\beta\geq0$ and $\varphi\in D\left(  \mathcal{N}^{\beta+1}\right)
, $ $t\rightarrow U\left(  t,s\right)  \varphi$ and $s\rightarrow U\left(
t,s\right)  \varphi$ are strongly $\left\Vert \cdot\right\Vert _{\beta}$
--differentiable (see Definition \ref{def.2.9}) and satisfy
\begin{equation}
i\frac{d}{dt}U\left(  t,s\right)  \varphi=\bar{Q}\left(  t\right)  U\left(
t,s\right)  \varphi\text{ with }U\left(  s,s\right)  \varphi=\varphi
\label{equ.5.17}%
\end{equation}
and
\begin{equation}
i\frac{d}{ds}U\left(  t,s\right)  \varphi=-U\left(  t,s\right)  \bar{Q}\left(
s\right)  \varphi\text{ with }U\left(  s,s\right)  \varphi=\varphi
\label{equ.5.18}%
\end{equation}
where the derivatives are taken relative to the $\beta$ -- norm, $\left\Vert
\cdot\right\Vert _{\beta}.$
\end{enumerate}
\end{theorem}

\begin{proof}
\textbf{Item 1. }Let $\varphi\in D\left(  \mathcal{N}^{\beta}\right)  .$ From
Proposition \ref{pro.5.5} we know that $L_{\varphi}\left(  t,s\right)
:=\lim_{M\rightarrow\infty}U^{M}\left(  t,s\right)  \varphi$ exists locally
uniformly in $\left(  t,s\right)  $ in the $\beta$ -- norm and therefore
$\left(  t,s\right)  \rightarrow L_{\varphi}\left(  t,s\right)  \in D\left(
\mathcal{N}^{\beta}\right)  $ is $\beta$ -- norm continuous jointly in
$\left(  t,s\right)  .$ In particular, this observation with $\beta=0$ allows
us to define
\[
U\left(  t,s\right)  =s-\lim_{M\rightarrow\infty}U^{M}\left(  t,s\right)
\]
where the limit is taken in the strong $L^{2}\left(  m\right)  $ - operator
topology. Since the operator product is continuous under strong convergence,
by taking the strong limit of Eq. (\ref{equ.4.9}) shows the first equality in
Eq. (\ref{equ.5.15}) holds. By taking $s=t$ in Eq. (\ref{equ.5.15}) we
conclude that $U\left(  t,s\right)  $ is invertible and hence is unitary on
$L^{2}\left(  m\right)  $ as it is already known to be an isometry because it
is the strong limit of unitary operators. This proves the item 1. of the theorem.

\textbf{Items 2. }As we have just seen, for any $\varphi\in D\left(
\mathcal{N}^{\beta}\right)  $ we know that $\left(  t,s\right)  \rightarrow
U\left(  t,s\right)  \varphi=L_{\varphi}\left(  t,s\right)  \in D\left(
\mathcal{N}^{\beta}\right)  $ is $\left\Vert \cdot\right\Vert _{\beta}$ --
continuous which proves item 2. Along the way we have shown $U\left(
t,s\right)  D\left(  \mathcal{N}^{\beta}\right)  \subset D\left(
\mathcal{N}^{\beta}\right)  $ and equality then follows using Eq.
(\ref{equ.5.15}).

\textbf{Item 3 }follows by the Eq. (\ref{equ.5.3}) in Corollary \ref{cor.5.2}
where the bounds are independent of $M.$

So it only remains to prove item 4. of the theorem. We begin with proving the
following claim.

\textbf{Claim. }If $\varphi\in D\left(  \mathcal{N}^{\beta+1}\right)  ,$ then
\begin{align}
Q_{M}\left(  \tau\right)  U^{M}\left(  \tau,s\right)  \varphi &
\rightarrow\bar{Q}\left(  \tau\right)  U\left(  \tau,s\right)  \varphi\text{
as }M\rightarrow\infty\text{ and}\label{equ.5.19}\\
U^{M}\left(  \tau,s\right)  Q_{M}\left(  s\right)  \varphi &  \rightarrow
U\left(  \tau,s\right)  \bar{Q}\left(  s\right)  \varphi\text{ as
}M\rightarrow\infty\label{equ.5.20}%
\end{align}
locally uniformly in $\left(  \tau,s\right)  $ in the $\left\Vert
\cdot\right\Vert _{\beta}$ -- topology.

\textbf{Proof of the claim. }Using $\sup_{\tau\in\left[  S,T\right]
}\left\Vert \bar{Q}\left(  \tau\right)  \right\Vert _{\beta+1\rightarrow\beta
}<\infty$ (see Corollary \ref{cor.3.40}) and the simple estimate,
\begin{align*}
&  \left\Vert Q_{M}\left(  \tau\right)  U^{M}\left(  \tau,s\right)
\varphi-\bar{Q}\left(  \tau\right)  U\left(  \tau,s\right)  \varphi\right\Vert
_{\beta}\\
&  \qquad\leq\left\Vert \left[  Q_{M}\left(  \tau\right)  -\bar{Q}\left(
\tau\right)  \right]  U^{M}\left(  \tau,s\right)  \varphi\right\Vert _{\beta
}+\left\Vert \bar{Q}\left(  \tau\right)  \left[  U^{M}\left(  \tau,s\right)
-U\left(  \tau,s\right)  \right]  \varphi\right\Vert _{\beta}\\
&  \qquad\leq\left\Vert \left[  Q_{M}\left(  \tau\right)  -\bar{Q}\left(
\tau\right)  \right]  U^{M}\left(  \tau,s\right)  \varphi\right\Vert _{\beta
}+\left\Vert \bar{Q}\left(  \tau\right)  \right\Vert _{\beta+1\rightarrow
\beta}\left\Vert \left[  U^{M}\left(  \tau,s\right)  -U\left(  \tau,s\right)
\right]  \varphi\right\Vert _{\beta+1},
\end{align*}
the local uniform convergence in Eq. (\ref{equ.5.19}) is now a consequence of
Propositions \ref{pro.5.3} and \ref{pro.5.5}. The local uniform convergence in
Eq. (\ref{equ.5.20}) holds by the same methods now based on the simple
estimate,
\begin{align}
&  \left\Vert U^{M}\left(  \tau,s\right)  Q_{M}\left(  \tau\right)
\varphi-U\left(  \tau,s\right)  \bar{Q}\left(  \tau\right)  \varphi\right\Vert
_{\beta}\nonumber\\
&  \qquad\leq\left\Vert U^{M}\left(  \tau,s\right)  \left[  {Q}_{M}\left(
\tau\right)  -\bar{Q}\left(  \tau\right)  \right]  \varphi\right\Vert _{\beta
}+\left\Vert \left[  U^{M}\left(  \tau,s\right)  -U\left(  \tau,s\right)
\right]  \bar{Q}\left(  \tau\right)  \varphi\right\Vert _{\beta}
\label{equ.5.21}%
\end{align}
along with Propositions \ref{pro.5.3} and \ref{pro.5.5}. Indeed, since (see Eq.
(\ref{equ.5.1})) $\bar{Q}(t)\varphi=\sum_{j=0}^{6}c_{j}\left(  t\right)
\mathcal{\bar{A}}^{\left(  j\right)  }\in D\left(  \mathcal{N}^{\beta}\right)
$ where each $c_{j}\left(  t\right)  $ is continuous in $t,$ the latter term
in Eq. (\ref{equ.5.21}) is estimated by a sum of $7$ terms resulting from the
estimates in Proposition \ref{pro.5.5} with $\psi=\mathcal{\bar{A}}^{\left(
j\right)  }\varphi$ for $0\leq j\leq6.$ This completes the proof of the claim.

\textbf{Item 4. }By integrating Eqs. (\ref{equ.4.10}) and (\ref{equ.4.11}) on
$t$ we find,
\begin{align}
U^{M}\left(  t,s\right)  \varphi &  =\varphi-i\int_{s}^{t}Q_{M}\left(
\tau\right)  U^{M}\left(  \tau,s\right)  \varphi d\tau\text{ and
}\label{equ.5.22}\\
U^{M}\left(  t,s\right)  \varphi &  =\varphi+i\int_{t}^{s}U^{M}\left(
t,\sigma\right)  Q_{M}\left(  \sigma\right)  \varphi d\sigma\label{equ.5.23}%
\end{align}
where the integrands are $\left\Vert \cdot\right\Vert _{\beta}$ -- continuous
and the integrals are taken relative to the $\left\Vert \cdot\right\Vert
_{\beta}$ -- topology. As a consequence of the above claim, we may let
$M\rightarrow\infty$ in Eqs. (\ref{equ.5.22}) and (\ref{equ.5.23}) to find
\begin{align*}
U\left(  t,s\right)  \varphi &  =\varphi-i\int_{s}^{t}\bar{Q}\left(
\tau\right)  U\left(  \tau,s\right)  \varphi d\tau\text{ and }\\
U\left(  t,s\right)  \varphi &  =\varphi+i\int_{t}^{s}U\left(  t,\sigma
\right)  \bar{Q}\left(  \sigma\right)  \varphi d\sigma
\end{align*}
where again the integrands are $\left\Vert \cdot\right\Vert _{\beta}$ --
continuous and the integrals are taken relative to the $\left\Vert
\cdot\right\Vert _{\beta}$ -- topology. Equations (\ref{equ.5.17}) and
(\ref{equ.5.18}) follow directly from the previously displayed equations along
with the fundamental theorem of calculus.
\end{proof}

\begin{remark}
\label{rem.5.9}By taking $t=s$ in Eq. (\ref{equ.5.15}) and using the fact that
$U\left(  t,s\right)  $ is unitary on $L^{2}\left(  m\right)  $, it follows
that
\begin{equation}
U\left(  t,\tau\right)  ^{-1}=U\left(  \tau,t\right)  =U^{\ast}\left(
t,\tau\right)  , \label{equ.5.26}%
\end{equation}
where $U^{\ast}\left(  t,\tau\right)  $ is the $L^{2}\left(  m\right)  $ -
adjoint of $U\left(  \tau,t\right)  .$ Also observe from Item 2. of Theorem
\ref{the.5.6} and Eq. (\ref{equ.3.34}) that
\begin{equation}
U\left(  t,s\right)  \mathcal{S}=\mathcal{S}\text{ for all }s,t\in\mathbb{R}.
\label{equ.5.27}%
\end{equation}

\end{remark}

\begin{remark}
\label{rem.5.10}Recall that if $X$ is a Banach space, $\psi\left(  h\right)
\in X,$ $T\left(  h\right)  \in B\left(  X\right)  $ for $0<\left\vert
h\right\vert <1,$ and $\psi\left(  h\right)  \rightarrow\psi\in X$ and
$T\left(  h\right)  \overset{s}{\rightarrow}T\in B\left(  X\right)  $ as
$h\rightarrow0,$ then $T\left(  h\right)  \psi\left(  h\right)  \rightarrow
T\psi$ as $h\rightarrow0.$
\end{remark}

\begin{theorem}
\label{the.5.12}Let $Q\left(  t\right)  $ and $U\left(  t,s\right)  $ be as in
Theorem \ref{the.5.6} and set $W\left(  t\right)  :=U\left(  t,0\right)  .$ If
$\varphi\in\mathcal{S},$ $R\in\mathbb{C}\left\langle \theta,\theta^{\ast
}\right\rangle ,$ and $\mathcal{R}:=R\left(  a,a^{\dag}\right)  ,$ then
\[
\frac{d}{dt}W\left(  t\right)  ^{\ast}\mathcal{R}W\left(  t\right)
\varphi=iW\left(  t\right)  ^{\ast}\left[  Q\left(  t\right)  ,\mathcal{R}%
\right]  W\left(  t\right)  \varphi
\]
where the derivative may be taken relative to the $\left\Vert \cdot\right\Vert
_{\beta}$ -- topology for any $\beta\geq0.$
\end{theorem}

\begin{proof}
Let $d=\deg_{\theta}R,$ $\psi\left(  t\right)  =\mathcal{R}W\left(  t\right)
\varphi$ and
\[
f\left(  t\right)  :=W\left(  t\right)  ^{\ast}\mathcal{R}W\left(  t\right)
\varphi=W\left(  t\right)  ^{\ast}\psi\left(  t\right)  =U\left(  0,t\right)
\psi\left(  t\right)  .
\]
In the proof we will write $\left\Vert \cdot\right\Vert _{\beta}$-$\frac
{d}{dt}\psi\left(  t\right)  $ to indicate that we are taking the derivative
relative to the $\beta$ -- norm topology.

Using the result of Theorem \ref{the.5.6} and the fact that $\left\Vert
\mathcal{R}\right\Vert _{\beta+d/2\rightarrow\beta}<\infty$ (Corollary
\ref{cor.3.40}) it easily follows that
\begin{equation}
\left\Vert \cdot\right\Vert _{\beta}\text{-}\frac{d}{dt}\psi\left(  t\right)
=-i\mathcal{R}Q\left(  t\right)  W\left(  t\right)  \varphi. \label{equ.5.28}%
\end{equation}
Combining this assertion with Remark \ref{rem.5.10} and the $\beta$ -- norm
strong continuity of $W\left(  t\right)  ^{\ast}$ (again Theorem
\ref{the.5.6}) we may conclude that
\[
\left\Vert \cdot\right\Vert _{\beta}\text{-}\lim_{h\rightarrow0}W\left(
t+h\right)  ^{\ast}\frac{\psi\left(  t+h\right)  -\psi\left(  t\right)  }%
{h}=W\left(  t\right)  ^{\ast}\dot{\psi}\left(  t\right)  =-iW\left(
t\right)  ^{\ast}\mathcal{R}Q\left(  t\right)  W\left(  t\right)  \varphi.
\]
Hence, as
\[
\frac{f\left(  t+h\right)  -f\left(  t\right)  }{h}=W\left(  t+h\right)
^{\ast}\frac{\psi\left(  t+h\right)  -\psi\left(  t\right)  }{h}%
+\frac{W\left(  t+h\right)  ^{\ast}-W\left(  t\right)  ^{\ast}}{h}\psi\left(
t\right)  ,
\]
we may conclude
\begin{align*}
\left\Vert \cdot\right\Vert _{\beta}\text{-}\frac{d}{dt}f\left(  t\right)   &
=\left\Vert \cdot\right\Vert _{\beta}\text{-}\lim_{h\rightarrow0}%
\frac{f\left(  t+h\right)  -f\left(  t\right)  }{h}\\
&  =-iW\left(  t\right)  ^{\ast}\mathcal{R}Q\left(  t\right)  W\left(
t\right)  \varphi+\dot{W}^{\ast}\left(  t\right)  \psi\left(  t\right) \\
&  =-iW\left(  t\right)  ^{\ast}\mathcal{R}Q\left(  t\right)  W\left(
t\right)  \varphi+iW\left(  t\right)  ^{\ast}Q\left(  t\right)  \mathcal{R}%
W\left(  t\right)  \varphi
\end{align*}
which completes the proof.
\end{proof}

\subsection{Consequences of Theorem \ref{the.5.6}\label{sec.5.1}}

\begin{notation}
\label{not.5.13} Let $H\in\mathbb{C}\left\langle \theta,\theta^{\ast
}\right\rangle $ be a symmetric non-commutative polynomial in $\theta$ and
$\theta^{\ast}.$ Let $\alpha\in\mathbb{C}$ and $H_{2}\left(  \alpha
:\theta,\theta^{\ast}\right)  $ as in Eq. (\ref{equ.2.21}) be the degree 2
homogeneous component of $H\left(  \theta+\alpha,\theta^{\ast}+\bar{\alpha
}\right)  .$ From Remark \ref{rem.2.19} and Theorem \ref{the.2.22},
$H^{\mathrm{cl}}\left(  \alpha\right)  $ is real-valued and $H_{2}\left(
\alpha:\theta,\theta^{\ast}\right)  $ is still symmetric.
\end{notation}

\begin{corollary}
\label{cor.5.14}Let $H\in\mathbb{C}\left\langle \theta,\theta^{\ast
}\right\rangle $ be a symmetric non-commutative polynomial in $\theta$ and
$\theta^{\ast},$ $H_{2}\left(  \alpha:\theta,\theta^{\ast}\right)  $ be as in
Notation \ref{not.5.13}, and suppose that $\mathbb{R}\ni t\rightarrow
\alpha\left(  t\right)  \in\mathbb{C}$ is a given continuous function. Then
there exists a unique one parameter strongly continuous family of unitary
operators $\left\{  W_{0}\left(  t\right)  \right\}  _{t\in\mathbb{R}}$ on
$L^{2}\left(  m\right)  $ such that (with $W_{0}^{\ast}\left(  t\right)  $
being the $L^{2}$ - adjoint of $W_{0}\left(  t\right)  )$;

\begin{enumerate}
\item $W_{0}\left(  t\right)  \mathcal{S}=\mathcal{S}$ and $W_{0}^{\ast
}\left(  t\right)  \mathcal{S}=\mathcal{S}.$

\item $W_{0}\left(  t\right)  D\left(  \mathcal{N}^{\beta}\right)  =D\left(
\mathcal{N}^{\beta}\right)  ,$ $W_{0}\left(  t\right)  ^{\ast}D\left(
\mathcal{N}^{\beta}\right)  =D\left(  \mathcal{N}^{\beta}\right)  ,$ and for
all $0\leq T<\infty,$ there exists $C_{T,\beta}=C_{T,\beta}\left(
\alpha\right)  <\infty$ such that
\begin{equation}
\sup_{\left\vert t\right\vert \leq T}\left\Vert W_{0}\left(  t\right)
\right\Vert _{\beta\rightarrow\beta}\vee\left\Vert W_{0}\left(  t\right)
^{\ast}\right\Vert _{\beta\rightarrow\beta}\leq C_{T,\beta}. \label{equ.5.29}%
\end{equation}

\item The maps $t\rightarrow W_{0}\left(  t\right)  \psi$ and $t\rightarrow
W_{0}^{\ast}\left(  t\right)  \psi$ are $\left\Vert \cdot\right\Vert _{\beta}%
$-norm continuous for all $\psi\in D\left(  \mathcal{N}^{\beta}\right)  .$

\item For each $\beta\geq0$ and $\psi\in D\left(  \mathcal{N}^{\beta
+1}\right)  ;$
\begin{equation}
i\left(  \left\Vert \cdot\right\Vert _{\beta}\text{-}\frac{\partial}{\partial
t}\right)  W_{0}\left(  t\right)  \psi=\overline{H_{2}\left(  \alpha\left(
t\right)  :a,a^{\dag}\right)  }W_{0}\left(  t\right)  \psi\text{ with }%
W_{0}\left(  0\right)  \psi=\psi\label{equ.5.30}%
\end{equation}
and
\begin{equation}
-i\left(  \left\Vert \cdot\right\Vert _{\beta}\text{-}\frac{\partial}{\partial
t}\right)  W_{0}\left(  t\right)  ^{\ast}\psi=W_{0}\left(  t\right)  ^{\ast
}\overline{H_{2}\left(  \alpha\left(  t\right)  :a,a^{\dag}\right)  }%
\psi\text{ with }W_{0}\left(  0\right)  ^{\ast}\psi=\psi. \label{equ.5.31}%
\end{equation}

\end{enumerate}

[In Eqs. (\ref{equ.5.30}) and (\ref{equ.5.31}), one may replace $\overline
{H_{2}\left(  \alpha\left(  t\right)  :a,a^{\dag}\right)  }$ by $H_{2}\left(
\alpha\left(  t\right)  :\overline{a},a^{\ast}\right)  $ as both operators are
equal on $D\left(  \mathcal{N}\right)  $ by Corollary \ref{cor.3.40}.]
\end{corollary}

\begin{proof}
The stated results follow from Theorem \ref{the.5.6} and Remark \ref{rem.5.9}
with $Q\left(  t\right)  :=H_{2}\left(  \alpha\left(  t\right)  :a,a^{\dag
}\right)  $ after setting $W_{0}\left(  t\right)  =U\left(  t,0\right)  $ in
which case that $W_{0}\left(  t\right)  ^{\ast}=U\left(  t,0\right)  ^{\ast
}=U\left(  0,t\right)  .$
\end{proof}

\begin{corollary}
\label{cor.5.15}If $\alpha\in\mathbb{C},$ $U\left(  \alpha\right)  $ is as in
Definition \ref{def.1.6}, and $U\left(  \alpha\right)  ^{\ast}$ is the
$L^{2}\left(  m\right)  $-adjoint of $U\left(  \alpha\right)  ,$ then for any
$\beta\geq0;$

\begin{enumerate}
\item $U\left(  \alpha\right)  \mathcal{S}=\mathcal{S}$ and $U\left(
\alpha\right)  ^{\ast}\mathcal{S}=\mathcal{S}$ (also seen in Proposition
\ref{pro.2.6}),

\item $U\left(  \alpha\right)  D\left(  \mathcal{N}^{\beta}\right)  =D\left(
\mathcal{N}^{\beta}\right)  $ and $U\left(  \alpha\right)  ^{\ast}D\left(
\mathcal{N}^{\beta}\right)  =D\left(  \mathcal{N}^{\beta}\right)  ,$ and

\item the following operator norm bounds hold,%
\begin{equation}
\left\Vert U\left(  \alpha\right)  \right\Vert _{\beta\rightarrow\beta}%
\vee\left\Vert U\left(  \alpha\right)  ^{\ast}\right\Vert _{\beta
\rightarrow\beta}\leq\exp\left(  8\beta\cdot3^{\left\vert \beta-1\right\vert
}\left\vert \alpha\right\vert \right)  . \label{equ.5.32}%
\end{equation}

\end{enumerate}
\end{corollary}

\begin{proof}
Let $\alpha\left(  t\right)  =t\alpha,$%
\[
H\left(  t:\theta,\theta^{\ast}\right)  =\dot{\alpha}\left(  t\right)
\theta^{\ast}-\overline{\dot{\alpha}\left(  t\right)  }\theta
+i\operatorname{Im}\left(  \alpha\left(  t\right)  \overline{\dot{\alpha
}\left(  t\right)  }\right)  =\alpha\theta^{\ast}-\overline{\alpha}\theta
\]
so that
\[
Q\left(  t\right)  =\alpha a^{\dag}-\overline{\alpha}a+i\operatorname{Im}%
\left(  t\alpha\overline{\alpha}\right)  =\alpha a^{\dag}-\overline{\alpha}a.
\]
By Proposition \ref{pro.2.10}, if $\varphi\in D\left(  \mathcal{N}\right)  ,$
$\psi\left(  t\right)  :=U\left(  t\alpha\right)  U\left(  s\alpha\right)
^{\ast}\varphi,$ then $\psi$ satisfies Eq. (\ref{equ.5.2}) and therefore items
1. and 2. follow Theorem \ref{the.5.6} and Remark \ref{rem.5.9}. To get the
explicit upper bound in Eq. (\ref{equ.5.32}), we apply Corollary \ref{cor.5.2}
with $S=0,$ $T=1,$ $P\left(  t,\theta,\theta^{\ast}\right)  =\alpha
\theta^{\ast}-\bar{\alpha}\theta$ in order to conclude, for any $M\in\left(
0,\infty\right)  $, that%
\[
\left\Vert U^{M}\left(  \alpha\right)  \right\Vert _{\beta\rightarrow\beta
}\leq\exp\left(  \beta4\cdot3^{\left\vert \beta-1\right\vert }\left[
\left\vert \alpha\right\vert +\left\vert \bar{\alpha}\right\vert \right]
\right)  =\exp\left(  8\beta\cdot3^{\left\vert \beta-1\right\vert }\left\vert
\alpha\right\vert \right)
\]
Letting $M\rightarrow\infty$ (as in the proof of Theorem \ref{the.5.6}) then
implies
\[
\left\Vert U\left(  \alpha\right)  \right\Vert _{\beta\rightarrow\beta}%
\leq\exp\left(  8\beta\cdot3^{\left\vert \beta-1\right\vert }\left\vert
\alpha\right\vert \right)  .
\]
Using $U\left(  \alpha\right)  ^{\ast}=U\left(  -\alpha\right)  ,$ the
previous equation is sufficient to prove the estimated in Eq. (\ref{equ.5.32}).
\end{proof}

\begin{corollary}
\label{cor.5.17}Let $U\left(  \alpha\right)  $ be as in Definition
\ref{def.1.6}, $U\left(  \alpha\right)  ^{\ast}$ be the $L^{2}\left(
m\right)  $-adjoint of $U\left(  \alpha\right)  ,$ $\mathbb{R}\ni
t\rightarrow\alpha\left(  t\right)  \in\mathbb{C}$ be a $C^{1}$ function, and
\[
Q\left(  t\right)  :=\dot{\alpha}\left(  t\right)  a^{\dag}-\overline
{\dot{\alpha}\left(  t\right)  }a+i\operatorname{Im}\left(  \alpha\left(
t\right)  \overline{\dot{\alpha}\left(  t\right)  }\right)  .
\]
Then for any $\beta\geq0;$

\begin{enumerate}
\item the maps $t\rightarrow U\left(  \alpha\left(  t\right)  \right)  \psi$
and $t\rightarrow U\left(  \alpha\left(  t\right)  \right)  ^{\ast}\psi$ are
$\left\Vert \cdot\right\Vert _{\beta}$-continuous for all $\psi\in D\left(
\mathcal{N}^{\beta}\right)  ,$ and

\item for each $\beta\geq0$ and $\psi\in D\left(  \mathcal{N}^{\beta
+1}\right)  ;$
\begin{equation}
i\left(  \left\Vert \cdot\right\Vert _{\beta}\text{-}\frac{\partial}{\partial
t}\right)  U\left(  \alpha\left(  t\right)  \right)  \psi=\overline{Q\left(
t\right)  }U\left(  \alpha\left(  t\right)  \right)  \psi\label{equ.5.33}%
\end{equation}
and
\begin{equation}
-i\left(  \left\Vert \cdot\right\Vert _{\beta}\text{-}\frac{\partial}{\partial
t}\right)  U\left(  \alpha\left(  t\right)  \right)  ^{\ast}\psi=U\left(
\alpha\left(  t\right)  \right)  ^{\ast}\overline{Q\left(  t\right)  }\psi.
\label{equ.5.34}%
\end{equation}

\end{enumerate}
\end{corollary}

\begin{proof}
Let
\[
H\left(  t:\theta,\theta^{\ast}\right)  :=\dot{\alpha}\left(  t\right)
\theta^{\ast}-\overline{\dot{\alpha}\left(  t\right)  }\theta
+i\operatorname{Im}\left(  \alpha\left(  t\right)  \overline{\dot{\alpha
}\left(  t\right)  }\right)
\]
so that $Q\left(  t\right)  =H\left(  t:a,a^{\dag}\right)  .$ By Proposition
\ref{pro.2.10} if $\varphi\in D\left(  \mathcal{N}\right)  ,$ $\psi\left(
t\right)  :=U\left(  \alpha\left(  t\right)  \right)  U\left(  \alpha\left(
s\right)  \right)  ^{\ast}\varphi,$ then $\psi$ satisfies Eq. (\ref{equ.5.2})
and therefore the corollary again follows from Theorem \ref{the.5.6} and
Remark \ref{rem.5.9}.
\end{proof}

\begin{theorem}
[Properties of $a\left(  t\right)  $]\label{the.5.18}Let $H\in\mathbb{C}%
\left\langle \theta,\theta^{\ast}\right\rangle $ be symmetric and
$H^{\text{cl}}\in\mathbb{C}\left[  z,\bar{z}\right]  $ be the symbol of $H,$
($H^{\text{cl }}$ is necessarily real valued by Remark \ref{rem.2.19}.)
Further suppose that $\alpha\left(  t\right)  \in\mathbb{C}$ satisfying
Hamilton's equations of motion (see Eq. (\ref{equ.2.3}) has global solutions,
$a\left(  t\right)  $ and $a^{\dag}\left(  t\right)  $ are the operators on
$\mathcal{S}$ as described in Eqs. (\ref{equ.1.8}), and (\ref{equ.1.9}), and
$W_{0}\left(  t\right)  $ is the unitary operator in Corollary \ref{cor.5.14}.
Then for all $t\in\mathbb{R}$ the following identities hold;
\begin{align}
W_{0}\left(  t\right)  ^{\ast}aW_{0}\left(  t\right)   &  =a\left(  t\right)
,\text{ \quad}W_{0}\left(  t\right)  ^{\ast}a^{\dag}W_{0}\left(  t\right)
=a^{\dag}\left(  t\right)  ,\label{equ.5.35}\\
W_{0}\left(  t\right)  ^{\ast}\overline{a}W_{0}\left(  t\right)   &
=\overline{a\left(  t\right)  },\text{ \quad}W_{0}\left(  t\right)  ^{\ast
}a^{\ast}W_{0}\left(  t\right)  =a^{\ast}\left(  t\right)  ,~\quad
\label{equ.5.36}\\
W_{0}\left(  t\right)  ^{\ast}\overline{a^{\dag}}W_{0}\left(  t\right)   &
=\overline{a^{\dag}\left(  t\right)  }\label{equ.5.37}\\
D\left(  \overline{a\left(  t\right)  }\right)   &  =D\left(  \sqrt
{\mathcal{N}}\right)  =D\left(  a^{\ast}\left(  t\right)  \right)
\label{equ.5.38}\\
a^{\ast}\left(  t\right)   &  =\overline{a^{\dag}\left(  t\right)
},\label{equ.5.39}\\
\overline{a\left(  t\right)  }  &  =\gamma\left(  t\right)  \overline
{a}+\delta\left(  t\right)  a^{\ast},\text{\quad and}~\label{equ.5.40}\\
a^{\ast}\left(  t\right)   &  =\overline{\delta\left(  t\right)  }\bar
{a}+\overline{\gamma\left(  t\right)  }a^{\ast}, \label{equ.5.41}%
\end{align}
where the closures and adjoints are taken relative to the $L^{2}\left(
m\right)  $-inner product.
\end{theorem}

\begin{proof}
Recall from Proposition \ref{pro.2.3} that
\[
v\left(  t\right)  :=\frac{\partial^{2}H^{\text{cl}}}{\partial\alpha
\partial\overline{\alpha}}\left(  \alpha\left(  t\right)  \right)
\in\mathbb{R}\text{ and }u\left(  t\right)  :=\frac{\partial^{2}H^{\text{cl}}%
}{\partial\overline{\alpha}^{2}}\left(  \alpha\left(  t\right)  \right)
\in\mathbb{C}.
\]
With this notation, the commutator formulas in Corollary \ref{cor.3.8} with
$\alpha=\alpha\left(  t\right)  $ may be written as,
\begin{align*}
\left[  H_{2}\left(  \alpha\left(  t\right)  :a,a^{\dag}\right)  ,a\right]
&  =-v\left(  t\right)  a-u\left(  t\right)  \left(  \alpha\left(  t\right)
\right)  a^{\dag}\\
\left[  H_{2}\left(  \alpha\left(  t\right)  :a,a^{\dag}\right)  ,a^{\dag
}\right]   &  =\bar{u}\left(  t\right)  a+v\left(  t\right)  a^{\dag}.
\end{align*}

For $\varphi\in\mathcal{S},$ let
\[
\psi\left(  t\right)  :=W_{0}\left(  t\right)  ^{\ast}aW_{0}\left(  t\right)
\varphi\text{ and }\psi^{\dag}\left(  t\right)  :=W_{0}\left(  t\right)
^{\ast}a^{\dag}W_{0}\left(  t\right)  \varphi.
\]
From Theorem \ref{the.5.12} with $W\left(  t\right)  =W_{0}\left(  t\right)
,$ $Q\left(  t\right)  =H_{2}\left(  \alpha\left(  t\right)  :a,a^{\dag
}\right)  ,$ and $\mathcal{R}=a$ and $\mathcal{R}=a^{\dag},$ we find
\begin{align*}
i\frac{d}{dt}\psi\left(  t\right)   &  =W_{0}\left(  t\right)  ^{\ast}\left[
v\left(  t\right)  a+u\left(  t\right)  \left(  \alpha\left(  t\right)
\right)  a^{\dag}\right]  W_{0}\left(  t\right)  \varphi\\
&  =v\left(  t\right)  \psi\left(  t\right)  +u\left(  t\right)  \psi^{\dag
}\left(  t\right) \\
i\frac{d}{dt}\psi^{\dag}\left(  t\right)   &  =-W_{0}\left(  t\right)  ^{\ast
}\left[  \bar{u}\left(  t\right)  a+v\left(  t\right)  a^{\dag}\right]
W_{0}\left(  t\right)  \varphi\\
&  =-\bar{u}\left(  t\right)  \psi\left(  t\right)  +v\left(  t\right)
\psi^{\dag}\left(  t\right)  .
\end{align*}
In other words,
\[
i\frac{d}{dt}\left[
\begin{array}
[c]{c}%
\psi\left(  t\right) \\
\psi^{\dag}\left(  t\right)
\end{array}
\right]  =\left[
\begin{array}
[c]{cc}%
v\left(  t\right)  & u\left(  t\right) \\
-\bar{u}\left(  t\right)  & -\bar{v}\left(  t\right)
\end{array}
\right]  \left[
\begin{array}
[c]{c}%
\psi\left(  t\right) \\
\psi^{\dag}\left(  t\right)
\end{array}
\right]  \in L^{2}\left(  m\right)  \times L^{2}\left(  m\right)  .
\]
This linear differential equation has a unique solution which, using
Proposition \ref{pro.2.3}, is given by
\[
\left[
\begin{array}
[c]{c}%
\psi\left(  t\right) \\
\psi^{\dag}\left(  t\right)
\end{array}
\right]  =\Lambda\left(  t\right)  \left[
\begin{array}
[c]{c}%
\psi\left(  0\right) \\
\psi^{\dag}\left(  0\right)
\end{array}
\right]  =\Lambda\left(  t\right)  \left[
\begin{array}
[c]{c}%
a\varphi\\
a^{\dag}\varphi
\end{array}
\right]
\]
where $\Lambda\left(  t\right)  $ is the $2\times2$ matrix given in Eq.
(\ref{equ.2.6}). This completes the proof of Eq. (\ref{equ.5.35}) since
\[
\text{ }\left[
\begin{array}
[c]{c}%
W_{0}\left(  t\right)  ^{\ast}aW_{0}\left(  t\right)  \varphi\\
W_{0}\left(  t\right)  ^{\ast}a^{\dag}W_{0}\left(  t\right)  \varphi
\end{array}
\right]  =\left[
\begin{array}
[c]{c}%
\psi\left(  t\right) \\
\psi^{\dag}\left(  t\right)
\end{array}
\right]  \text{ and }\Lambda\left(  t\right)  \left[
\begin{array}
[c]{c}%
a\varphi\\
a^{\dag}\varphi
\end{array}
\right]  =\left[
\begin{array}
[c]{c}%
a\left(  t\right)  \varphi\\
a^{\dag}\left(  t\right)  \varphi
\end{array}
\right]  .
\]

The statements in Eqs. (\ref{equ.5.36}), (\ref{equ.5.37}) and (\ref{equ.5.38})
are easy consequences of the fact that $W_{0}\left(  t\right)  $ is a unitary
operator on $L^{2}\left(  m\right)  $ which preserves $D\left(  \mathcal{N}%
\right)  $ (see Corollary \ref{cor.5.14}). Using Eqs. (\ref{equ.5.36}) and
(\ref{equ.5.37}) along with Theorem \ref{the.3.22} shows,
\[
\overline{a^{\dag}\left(  t\right)  }=W_{0}\left(  t\right)  ^{\ast}%
\overline{a^{\dag}}W_{0}\left(  t\right)  =W_{0}\left(  t\right)  ^{\ast
}a^{\ast}W_{0}\left(  t\right)  =a\left(  t\right)  ^{\ast}%
\]
which gives Eq. (\ref{equ.5.39}).

If $\varphi\in D\left(  \mathcal{N}\right)  ,$ using item 3. of Theorem
\ref{the.3.22} and the formula for $a\left(  t\right)  $ and $a^{\dag}\left(
t\right)  $ in Eqs. (\ref{equ.1.8}) and (\ref{equ.1.9}) we find
\begin{align*}
\lim_{M\rightarrow\infty}a\left(  t\right)  \mathcal{P}_{M}\varphi &
=\lim_{M\rightarrow\infty}\left[  \gamma\left(  t\right)  a\mathcal{P}%
_{M}\varphi+\delta\left(  t\right)  a^{\dag}\mathcal{P}_{M}\varphi\right] \\
&  =\gamma\left(  t\right)  \overline{a}\varphi+\delta\left(  t\right)
a^{\ast}\varphi
\end{align*}
\begin{align*}
\lim_{M\rightarrow\infty}a^{\dag}\left(  t\right)  \mathcal{P}_{M}\varphi &
=\lim_{M\rightarrow\infty}\left[  \delta\left(  t\right)  a\mathcal{P}%
_{M}\varphi+\overline{\gamma\left(  t\right)  }a^{\dag}\mathcal{P}_{M}%
\varphi\right] \\
&  =\overline{\delta\left(  t\right)  }\bar{a}\varphi+\overline{\gamma\left(
t\right)  }a^{\ast}\varphi.
\end{align*}
The above two equations along with Corollary \ref{cor.3.40} show Eqs.
(\ref{equ.5.40}) and (\ref{equ.5.41}).
\end{proof}

\section{Bounds on the Quantum Evolution\label{sec.6}}

Throughout this section and the rest of the paper, let $H\in\mathbb{R}%
\left\langle \theta,\theta^{\ast}\right\rangle $ be a non-commutative
polynomial satisfying Assumption \ref{ass.1}. Before getting to the proof of
the main theorems we need to address some domain issues. Recall as in
Assumption \ref{ass.1} we let $H_{\hbar}:=\overline{H\left(  a_{\hbar
},a_{\hbar}^{\dag}\right)  }.$

The following abstract proposition (Stone's theorem) is a routine application
of the spectral theorem, see \citep[p.265]{Reed1980} for details.

\begin{proposition}
\label{pro.6.1} Supposed $H$ is a self-adjoint operator on a separable Hilbert
space, $\mathcal{K},$ and there is a $C\in\mathbb{R}$ and $\varepsilon>0$ such
that $H+CI\geq\varepsilon I.$ For any $\beta\geq0$ let $\left\Vert
\cdot\right\Vert _{\left(  H+CI\right)  ^{\beta}}~\left(  \geq\varepsilon
\left\Vert \cdot\right\Vert _{\mathcal{K}}\right)  $ be the Hilbertian norm on
$D\left(  \left(  H+CI\right)  ^{\beta}\right)  $ defined by,
\[
\left\Vert f\right\Vert _{\left(  H+CI\right)  ^{\beta}}=\left\Vert \left(
H+CI\right)  ^{\beta}f\right\Vert _{\mathcal{K}}\text{ }\forall~f\in D\left(
\left(  H+CI\right)  ^{\beta}\right)  .
\]
Then for all $t\in\mathbb{R}$ and $\beta\geq0,$
\begin{align*}
e^{-itH}D(\left(  H+CI\right)  ^{\beta})  &  =D(\left(  H+CI\right)  ^{\beta
})\text{ and}\\
\left\Vert e^{-itH}\psi\right\Vert _{\left(  H+CI\right)  ^{\beta}}  &
=\left\Vert \psi\right\Vert _{\left(  H+CI\right)  ^{\beta}}\text{ }%
\forall~\psi\in D(\left(  H+CI\right)  ^{\beta}).
\end{align*}
Moreover, if $\beta\geq0$ and $\varphi\in D(\left(  H+CI\right)  ^{\beta+1}),$
then
\[
\left\Vert \cdot\right\Vert _{\left(  H+CI\right)  ^{\beta}}-\frac{d}%
{dt}e^{-iHt}\varphi=-iHe^{-iHt}\varphi=-ie^{-iHt}H\varphi.
\]

\end{proposition}

In this section we are going to show, as a consequence of Proposition
\ref{pro.6.5} below, that
\begin{equation}
e^{iH_{\hbar}t/\hbar}\bar{a}e^{-iH_{\hbar}t/\hbar}\mathcal{S}\text{ and
}e^{iH_{\hbar}t/\hbar}a^{\ast}e^{-iH_{\hbar}t/\hbar}\mathcal{S}\subseteq
\mathcal{S}. \label{equ.6.1}%
\end{equation}

\begin{lemma}
\label{lem.6.4} For any unbounded operator $T$ and constant $C\in\mathbb{R},$
then for any $n\in\mathbb{N}_{0},$
\[
D\left(  \left(  T+C\right)  ^{n}\right)  =D\left(  T^{n}\right)  .
\]

\end{lemma}

\begin{proof}
We first show by induction that $D\left(  \left(  T+C\right)  ^{n}\right)
\subset D\left(  T^{n}\right)  $ for all $n\in\mathbb{N}.$ The case $n=1$ is
trivial. Then the induction step is
\begin{align*}
f  &  \in D\left(  \left(  T+C\right)  ^{n+1}\right)  \implies f\in D\left(
\left(  T+C\right)  ^{n}\right)  \text{ and }\left(  T+C\right)  ^{n}f\in
D\left(  T+C\right) \\
&  \implies f\in D\left(  \left(  T+C\right)  ^{n}\right)  \text{ and }\left(
T+C\right)  ^{n}f\in D\left(  T\right) \\
&  \implies f\in D\left(  T^{n}\right)  \text{ and }\left(  T+C\right)
^{n}f\in D\left(  T\right)
\end{align*}
But
\[
\left(  T+C\right)  ^{n}f=T^{n}f+\sum_{k=0}^{n-1}\binom{n}{k}C^{n-k}%
T^{k}f=T^{n}f+g
\]
where $g\in D\left(  T\right)  $ and hence
\[
T^{n}f=\left(  T+C\right)  ^{n}f-g\in D\left(  T\right)  \implies f\in
D\left(  T^{n+1}\right)  .
\]
finishing the inductive step.

To finish the proof, we replace $T$ by $T-C$ above to learn%
\[
D\left(  T^{n}\right)  =D\left(  \left(  T-C+C\right)  ^{n}\right)  \subset
D\left(  \left(  T-C\right)  ^{n}\right)
\]
and then replace $C$ by $-C$ to find $D\left(  T^{n}\right)  \subset D\left(
\left(  T+C\right)  ^{n}\right)  .$
\end{proof}

\begin{proposition}
\label{pro.6.5}Let $H\left(  \theta,\theta^{\ast}\right)  $ and $\eta>0$ be as
in Assumption \ref{ass.1}, then $\exp\left(  -iH_{\hbar}t\right)  $ leaves
$\mathcal{S}$ invariant and more explicitly, it is $\exp\left(  -iH_{\hbar
}t\right)  \mathcal{S}=\mathcal{S}$ for all $t\in\mathbb{R}.$
\end{proposition}

\begin{proof}
The fact that $\mathcal{S}\subseteq H_{\hbar}^{n}$ for all $n\in\mathbb{N}$
along with Eq. (\ref{equ.1.14}) in the Assumption \ref{ass.1} and Eq.
(\ref{equ.3.51}), we learn that
\[
\mathcal{S}(\mathbb{R})\subset\bigcap_{n=1}^{\infty}D\left(  H_{\hbar}%
^{n}\right)  \subseteq\bigcap_{n=1}^{\infty}D(\mathcal{N}_{\hbar}%
^{n})=\mathcal{S}(\mathbb{R})
\]
This shows $\mathcal{S}(\mathbb{R})=\bigcap_{n=1}^{\infty}D\left(  H_{\hbar
}^{n}\right)  $ and this finishes the proof since, see Proposition
\ref{pro.6.1}, $\exp\left(  -iH_{\hbar}t\right)  $ leaves $\bigcap
_{n=1}^{\infty}D\left(  H_{\hbar}^{n}\right)  \ $invariant, i.e.,$\exp\left(
-iH_{\hbar}t\right)  \mathcal{S}\subseteq\mathcal{S}$ for all $t\in
\mathbb{R}.$ By multiplying $\exp\left(  iH_{\hbar}t\right)  $ on both sides,
we yield $\mathcal{S}\subseteq\exp\left(  iH_{\hbar}t\right)  \mathcal{S}.$
Therefore, $\exp\left(  -iH_{\hbar}t\right)  \mathcal{S}=\mathcal{S}$ is
resulted if we replacing $t$ to $-t.$
\end{proof}

\begin{lemma}
\label{lem.6.6}If $P\in\mathbb{C}\left\langle \theta,\theta^{\ast
}\right\rangle ,$ $\delta:=\deg_{\theta}P\in\mathbb{N}_{0},$ and $C\left(
P\right)  :=\sum_{k=0}^{\delta}\left\vert P_{k}\right\vert k^{k/2},$ then
\begin{equation}
\left\Vert P\left(  \bar{a}_{\hbar},a_{\hbar}^{\ast}\right)  \psi\right\Vert
\leq C\left(  P\right)  \left\Vert \left(  I+\mathcal{N}_{\hbar}\right)
^{\delta/2}\psi\right\Vert ~\forall~0<\hbar\leq1\text{ and }\psi\in D\left(
\mathcal{N}^{\delta/2}\right)  . \label{equ.6.2}%
\end{equation}

\end{lemma}

\begin{proof}
Let $P_{k}$ be the degree $k$ homogeneous component of $P$ as in Eq.
(\ref{equ.2.22}). Then according to Corollary \ref{cor.3.40} with $\beta=0$
and $d=k$ we have,
\begin{align*}
\left\Vert P_{k}\left(  \bar{a}_{\hbar},a_{\hbar}^{\ast}\right)
\psi\right\Vert  &  =\hbar^{k/2}\left\Vert P_{k}\left(  \bar{a},a^{\ast
}\right)  \psi\right\Vert \\
&  \leq\left\vert P_{k}\right\vert k^{k/2}\hbar^{k/2}\left\Vert \psi
\right\Vert _{k/2}\\
&  =\left\vert P_{k}\right\vert k^{k/2}\hbar^{k/2}\left\Vert \left(
I+\mathcal{N}\right)  ^{k/2}\psi\right\Vert \\
&  =\left\vert P_{k}\right\vert k^{k/2}\left\Vert \left(  \hbar I+\mathcal{N}%
_{\hbar}\right)  ^{k/2}\psi\right\Vert \\
&  \leq\left\vert P_{k}\right\vert k^{k/2}\left\Vert \left(  I+\mathcal{N}%
_{\hbar}\right)  ^{k/2}\psi\right\Vert \leq\left\vert P_{k}\right\vert
k^{k/2}\left\Vert \left(  I+\mathcal{N}_{\hbar}\right)  ^{\delta/2}%
\psi\right\Vert .
\end{align*}
Summing this inequality on $k$ using $P=\sum_{k=0}^{\delta}P_{k}$ and the
triangle inequality leads directly to Eq. (\ref{equ.6.2}).
\end{proof}

The next important result may be found in Heinz \cite{Heinz1951}, also see
Kato \cite[Theorem 2]{Kato1952} and \cite[Proposition 10.14, p.232]%
{Schmudgen2012}.

\begin{theorem}
[L\"{o}wner-Heinz inequality]\label{the.6.7}Let $A$ and $B$ be non-negative
self-adjoint operators on a Hilbert space. If $A\leq B$ (see Notation
\ref{not.1.11}), then $A^{r}\leq B^{r}$ for $0\leq r\leq1.$
\end{theorem}

\begin{corollary}
\label{cor.6.8}Let $H\left(  \theta,\theta^{\ast}\right)  \in\mathbb{R}%
\left\langle \theta,\theta^{\ast}\right\rangle ,$ $1>\eta>0,$ and $C$ be as in
Assumption \ref{ass.1} and set $\tilde{C}:=C+1.$ Then for each $\beta\geq0,$
there exists constants $\widetilde{C}_{\beta}<\infty$ and $\widetilde
{D}_{\beta}<\infty$ such that, for all $0\leq\hbar<\eta,$
\begin{align}
\left(  \mathcal{N}_{\hbar}+I\right)  ^{\beta}  &  \leq\widetilde{C}_{\beta
}\left(  H_{\hbar}+\widetilde{C}\right)  ^{\beta}\text{ and }\label{equ.6.3}\\
\left(  H_{\hbar}+\widetilde{C}\right)  ^{\beta}  &  \leq\widetilde{D}_{\beta
}\left(  \mathcal{N}_{\hbar}+I\right)  ^{\beta d/2}. \label{equ.6.4}%
\end{align}

\end{corollary}

\begin{proof}
Using the simple estimate,
\begin{equation}
\left(  x+1\right)  ^{\beta}\leq2^{\left(  \beta-1\right)  _{+}}\left(
x^{\beta}+1\right)  \text{ }\forall~x,\beta\geq0, \label{e.6.5}%
\end{equation}
along with Eq. (\ref{equ.1.14}) implies,
\begin{align}
\left(  \mathcal{N}_{\hbar}+I\right)  ^{\beta}  &  \preceq2^{\left(
\beta-1\right)  _{+}}\left(  \mathcal{N}_{\hbar}^{\beta}+I\right)
\preceq2^{\left(  \beta-1\right)  _{+}}\left(  C_{\beta}\left(  H_{\hbar
}+C\right)  ^{\beta}+I\right) \nonumber\\
&  \preceq2^{\left(  \beta-1\right)  _{+}}C_{\beta}\left(  H_{\hbar
}+C+I\right)  ^{\beta}, \label{equ.6.6}%
\end{align}
wherein we have assumed $C_{\beta}\geq1$ without loss of generality. Lemma
10.10 of \citep[p.230]{Schmudgen2012} asserts, if $A$ and $B$ are non-negative
self-adjoint operators and $A\preceq B,$ then $A\leq B.$ Therefore we can
deduce from Eq. (\ref{equ.6.6}) that
\[
\left(  \mathcal{N}_{\hbar}+I\right)  ^{\beta}\leq2^{\left(  \beta-1\right)
_{+}}C_{\beta}\left(  H_{\hbar}+C+I\right)  ^{\beta}%
\]
which gives Eq. (\ref{equ.6.3}).

We now turn to the proof of Eq. (\ref{equ.6.4}). For $n\in\mathbb{N},$ let
$P^{\left(  n\right)  }\in\mathbb{C}\left\langle \theta,\theta^{\ast
}\right\rangle $ be defined by
\[
P^{\left(  n\right)  }\left(  \theta,\theta^{\ast}\right)  :=\left(  H\left(
\theta,\theta^{\ast}\right)  +\tilde{C}\right)  ^{n}%
\]
so that $\deg_{\theta}P^{\left(  n\right)  }=dn$ and for $\psi\in D\left(
\mathcal{N}^{dn/2}\right)  ,$ we have
\[
\left(  H_{\hbar}+\widetilde{C}\right)  ^{n}\psi=P^{\left(  n\right)  }\left(
\bar{a}_{\hbar},a_{\hbar}^{\ast}\right)  \psi.
\]
With these observations, we may apply Lemma \ref{lem.6.6} to find for any
$0<\hbar<\eta\leq1$ that
\[
\left\Vert \left(  H_{\hbar}+\widetilde{C}\right)  ^{n}\psi\right\Vert \leq
C\left(  P^{\left(  n\right)  }\right)  \left\Vert \left(  I+\mathcal{N}%
_{\hbar}\right)  ^{\frac{dn}{2}}\psi\right\Vert ~\forall~\psi\in D\left(
\mathcal{N}^{dn/2}\right)  .
\]
The last displayed equation is equivalent (see Notation \ref{not.1.11}) to the
operator inequality,
\[
\left(  H_{\hbar}+\widetilde{C}\right)  ^{2n}\leq C\left(  P^{\left(
2n\right)  }\right)  \left(  I+\mathcal{N}_{\hbar}\right)  ^{dn}.
\]
Hence if $0\leq\beta\leq2n,$ we may apply the L\"{o}wner-Heinz inequality
(Theorem \ref{the.6.7}) with $r=\beta/2n$ to conclude
\[
\left(  H_{\hbar}+\widetilde{C}\right)  ^{\beta}\leq\left[  C\left(
P^{\left(  n\right)  }\right)  \right]  ^{\beta/2n}\left(  I+\mathcal{N}%
_{\hbar}\right)  ^{\beta d/2}.
\]
As $n\in\mathbb{N}$ was arbitrary, the proof is complete.
\end{proof}

\begin{theorem}
\label{the.6.11}Let $H\left(  \theta,\theta^{\ast}\right)  \in\mathbb{R}%
\left\langle \theta,\theta^{\ast}\right\rangle ,$ $d=\deg_{\theta}H,$ and
$1>\eta>0$ be as in Assumption \ref{ass.1} and suppose $0<\hbar<\eta\leq1.$

\begin{enumerate}
\item If $\beta\geq0$ then
\begin{equation}
e^{-iH_{\hbar}t/\hbar}D\left(  \mathcal{N}^{\beta d/2}\right)  \subseteq
D\left(  \mathcal{N}^{\beta}\right)  . \label{equ.6.7}%
\end{equation}
and there exists $C_{\beta}<\infty$ such that
\begin{equation}
\left\Vert e^{-iH_{\hbar}t/\hbar}\right\Vert _{\beta d/2\rightarrow\beta}\leq
C_{\beta}\hbar^{-\beta}\text{ for all }t\in\mathbb{R}. \label{equ.6.8}%
\end{equation}

\item If $\beta\geq0$ and $\psi\in D\left(  \mathcal{N}^{\left(
\beta+1\right)  d/2}\right)  \subset D\left(  H_{\hbar}^{\beta+1}\right)  ,$
then
\[
e^{-iH_{\hbar}t/\hbar}\psi,~H_{\hbar}e^{-iH_{\hbar}t/\hbar}\psi,~\text{and
}e^{-iH_{\hbar}t/\hbar}H_{\hbar}\psi
\]
are all in $D\left(  \mathcal{N}^{\beta}\right)  $ for all $t\in\mathbb{R}$
and moreover,
\begin{equation}
i\hbar\left(  \left\Vert \cdot\right\Vert _{\beta}\text{-}\frac{d}{dt}\right)
e^{-iH_{\hbar}t/\hbar}\psi=H_{\hbar}e^{-iH_{\hbar}t/\hbar}\psi=e^{-iH_{\hbar
}t/\hbar}H_{\hbar}\psi, \label{equ.6.9}%
\end{equation}
where, as before, $\left\Vert \cdot\right\Vert _{\beta}$-$\frac{d}{dt}$
indicates the derivative is taken in $\beta$ -- norm topology.
\end{enumerate}
\end{theorem}

\begin{proof}
If $\beta\geq0,$ it follows from Corollary \ref{cor.6.8} (with $\beta$
replaced by $2\beta)$ that
\begin{equation}
D\left(  \mathcal{N}^{\beta d/2}\right)  =D\left(  \mathcal{N}_{\hbar}^{\beta
d/2}\right)  \subset D\left(  \left(  H_{\hbar}+\widetilde{C}\right)  ^{\beta
}\right)  \subset D\left(  \mathcal{N}_{\hbar}^{\beta}\right)  =D\left(
\mathcal{N}^{\beta}\right)  \label{equ.6.10}%
\end{equation}
and
\[
\left\Vert \psi\right\Vert _{\left(  \mathcal{N}_{\hbar}+I\right)  ^{\beta}%
}\leq\sqrt{\widetilde{C}_{2\beta}}\left\Vert \psi\right\Vert _{\left(
H_{\hbar}+\widetilde{C}\right)  ^{\beta}}~\forall~\psi\in D\left(  \left(
H_{\hbar}+\widetilde{C}\right)  ^{\beta}\right)  .
\]
Moreover if $0<\hbar<\eta\leq1,$ a simple calculus inequality shows
\[
\hbar^{\beta}\left\Vert \psi\right\Vert _{\beta}=\hbar^{\beta}\left\Vert
\psi\right\Vert _{\left(  \mathcal{N}+I\right)  ^{\beta}}\leq\left\Vert
\psi\right\Vert _{\left(  \mathcal{N}_{\hbar}+I\right)  ^{\beta}}%
\]
and hence
\begin{equation}
\left\Vert \psi\right\Vert _{\beta}\leq\hbar^{-\beta}\sqrt{\widetilde
{C}_{2\beta}}\left\Vert \psi\right\Vert _{\left(  H_{\hbar}+\widetilde
{C}\right)  ^{\beta}}~\forall~\psi\in D\left(  \left(  H_{\hbar}+\widetilde
{C}\right)  ^{\beta}\right)  . \label{equ.6.11}%
\end{equation}

From Proposition \ref{pro.6.1} we know for all $t\in\mathbb{R}$ that
\begin{align*}
e^{-iH_{\hbar}t/\hbar}D\left(  \left(  H_{\hbar}+\widetilde{C}\right)
^{\beta}\right)   &  =D\left(  \left(  H_{\hbar}+\widetilde{C}\right)
^{\beta}\right)  \text{ and }\\
\left\Vert e^{-iH_{\hbar}t/\hbar}\psi\right\Vert _{\left(  H_{\hbar
}+\widetilde{C}\right)  ^{\beta}}  &  =\left\Vert \psi\right\Vert _{\left(
H_{\hbar}+\widetilde{C}\right)  ^{\beta}}.
\end{align*}
Combining these statements with Eqs. (\ref{equ.6.10}) and (\ref{equ.6.11})
respectively shows,%
\[
e^{-iH_{\hbar}t/\hbar}D\left(  \mathcal{N}^{\beta d/2}\right)  \subset
e^{-iH_{\hbar}t/\hbar}D\left(  \left(  H_{\hbar}+\widetilde{C}\right)
^{\beta}\right)  =D\left(  \left(  H_{\hbar}+\widetilde{C}\right)  ^{\beta
}\right)  \subset D\left(  \mathcal{N}^{\beta}\right)  .
\]
Moreover, if $\varphi\in D\left(  \mathcal{N}^{\beta d/2}\right)  \subset
D\left(  \left(  H_{\hbar}+\widetilde{C}\right)  ^{\beta}\right)  ,$ then%
\[
\left\Vert e^{-iH_{\hbar}t/\hbar}\varphi\right\Vert _{\beta}\leq\hbar^{-\beta
}\sqrt{\widetilde{C}_{2\beta}}\left\Vert e^{-iH_{\hbar}t/\hbar}\varphi
\right\Vert _{\left(  H_{\hbar}+\widetilde{C}\right)  ^{\beta}}=\hbar^{-\beta
}\sqrt{\widetilde{C}_{2\beta}}\left\Vert \varphi\right\Vert _{\left(
H_{\hbar}+\widetilde{C}\right)  ^{\beta}}.
\]
However, from Eq. (\ref{equ.6.4}) (again with $\beta\rightarrow2\beta)$ we
also know
\[
\left\Vert \varphi\right\Vert _{\left(  H_{\hbar}+\widetilde{C}\right)
^{\beta}}\leq\sqrt{\widetilde{D}_{2\beta}}\cdot\left\Vert \varphi\right\Vert
_{\left(  \mathcal{N}_{\hbar}+I\right)  ^{\beta d/2}}\leq\sqrt{\widetilde
{D}_{2\beta}}\cdot\left\Vert \varphi\right\Vert _{\left(  \mathcal{N}%
+I\right)  ^{\beta d/2}}.
\]
Combining the last two displayed equations proves the estimate in Eq.
(\ref{equ.6.8}) with $C_{\beta}:=\sqrt{\widetilde{C}_{2\beta}\cdot
\widetilde{D}_{2\beta}}.$

If we now further assume that $\psi\in D\left(  \mathcal{N}^{\left(
\beta+1\right)  d/2}\right)  $, then $\psi\in D\left(  H_{\hbar}^{\beta
+1}\right)  $ by Eq. (\ref{equ.6.10}) then, by Proposition \ref{pro.6.1}, it
follows that
\[
H_{\hbar}e^{-iH_{\hbar}t/\hbar}\psi\,=e^{-iH_{\hbar}t/\hbar}H_{\hbar}\psi\in
D\left(  \left(  H_{\hbar}+\widetilde{C}\right)  ^{\beta}\right)  \subset
D\left(  \mathcal{N}^{\beta}\right)
\]
and
\begin{equation}
i\hbar\left(  \left\Vert \cdot\right\Vert _{H_{\hbar}^{\beta}}\text{-}\frac
{d}{dt}\right)  \psi\left(  t\right)  =H_{\hbar}\psi\left(  t\right)
=e^{-iH_{\hbar}t/\hbar}H_{\hbar}\psi_{0}. \label{equ.6.12}%
\end{equation}
Owing to Eq. (\ref{equ.6.11}) the $\beta$ -- norm is weaker than $\left\Vert
\cdot\right\Vert _{H_{\hbar}^{\beta}}$ -- norm and hence Eq. (\ref{equ.6.12})
directly implies the weaker Eq. (\ref{equ.6.9}).
\end{proof}

\section{A Key One Parameter Family of Unitary Operators\label{sec.7}}

In this section (except for Lemma \ref{lem.7.3}) we will always suppose that
$H\left(  \theta,\theta^{\ast}\right)  $ and $1\geq\eta>0$ are as in
Assumption \ref{ass.1}, $\alpha_{0}\in\mathbb{C},$ and $\alpha\left(
t\right)  $ denotes the solution to Hamilton's classical equations
(\ref{equ.1.1}) of motion with $\alpha\left(  0\right)  =\alpha_{0}.$ From
Corollary \ref{cor.3.10}, $U_{\hbar}\left(  \alpha_{0}\right)  \psi$ is a
state on $L^{2}\left(  m\right)  $ which has position and momentum
concentrated at $\xi_{0}+i\pi_{0}=\sqrt{2}\alpha_{0}$ in the limit as
$\hbar\downarrow0.$ Thus if quantum mechanics is to limit to classical
mechanics as $\hbar\downarrow0,$ one should expect that the quantum evolution,
$\psi_{\hbar}\left(  t\right)  :=e^{-iH_{\hbar}t/\hbar}U_{\hbar}\left(
\alpha_{0}\right)  \psi,$ of the state, $U_{\hbar}\left(  \alpha_{0}\right)
\psi,$ should be concentrated near $\alpha\left(  t\right)  $ in phase space
as $\hbar\downarrow0.$ One possible candidate for these approximate states
would be $U_{\hbar}\left(  \alpha\left(  t\right)  \right)  \psi$ or more
generally any state of the form, $U_{\hbar}\left(  \alpha\left(  t\right)
\right)  W_{0}\left(  t\right)  \psi,$ where $\left\{  W_{0}\left(  t\right)
:t\in\mathbb{R}\right\}  $ are unitary operators on $L^{2}\left(  m\right)  $
which preserve $\mathcal{S}.$ All states of this form concentrate their
position and momentum expectations near $\sqrt{2}\alpha\left(  t\right)  ,$
see Remark \ref{rem.3.11}. These remarks then motivate us to consider the one
parameter family of unitary operators $V_{\hbar}\left(  t\right)  $ defined
by,
\begin{equation}
V_{\hbar}\left(  t\right)  :=U_{\hbar}\left(  -\alpha\left(  t\right)
\right)  e^{-iH_{\hbar}t/\hbar}U_{\hbar}\left(  \alpha_{0}\right)  =U_{\hbar
}\left(  \alpha\left(  t\right)  \right)  ^{\ast}e^{-iH_{\hbar}t/\hbar
}U_{\hbar}\left(  \alpha_{0}\right)  . \label{equ.7.1}%
\end{equation}

Because of Propositions \ref{pro.2.6} and \ref{pro.6.5}, we know $V_{\hbar
}\left(  t\right)  \mathcal{S=S}$ for all $0<\hbar<\eta$ and in particular,
$V_{\hbar}\left(  t\right)  \mathcal{S}=\mathcal{S}\subset D\left(  P\left(
a,a^{\dag}\right)  \right)  $ for any $P\left(  \theta,\theta^{\ast}\right)
\in\mathbb{C}\left\langle \theta,\theta^{\ast}\right\rangle .$ The main point
of this section is to study the basic properties of this family of unitary
operators with an eye towards showing that $\lim_{\hbar\downarrow0}V_{\hbar
}\left(  t\right)  $ exists (modulo a phase factor). Our first task is to
differentiate $V_{\hbar}\left(  t\right)  $ for which we will need the
following differentiation lemma.

\begin{lemma}
[Product Rule]\label{lem.7.1} Let $P\left(  \theta,\theta^{\ast}\right)
\in\mathbb{C}\left\langle \theta,\theta^{\ast}\right\rangle ,$ $k:=\deg
_{\theta}P\left(  \theta,\theta^{\ast}\right)  \in\mathbb{N}_{0},$ and
$P:=P\left(  a,a^{\dag}\right)  .$ Suppose that $U\left(  t\right)  $ and
$T\left(  t\right)  $ are unitary operators on $L^{2}\left(  m\right)  $ which
preserve $\mathcal{S}. $ We further assume;

\begin{enumerate}
\item for each $\varphi\in\mathcal{S},$ $t\rightarrow U\left(  t\right)
\varphi$ and $t\rightarrow T\left(  t\right)  \varphi$ are $\left\Vert
\cdot\right\Vert _{\beta}$ -- differentiable for all $\beta\geq0.$ We denote
the derivative by $\dot{U}\left(  t\right)  \varphi$ and $\dot{T}\left(
t\right)  \varphi$ respectively. [Notice that $\dot{U}\left(  t\right)
\varphi$ and $\dot{T}\left(  t\right)  \varphi$ are all in $\cap_{\beta\geq
0}D\left(  \mathcal{N}^{\beta}\right)  =\mathcal{S},$ see Eq. (\ref{equ.3.34})
for the last equality, i.e. $\dot{U}\left(  t\right)  $ and $\dot{T}\left(
t\right)  $ preserves $\mathcal{S}.$]

\item For each $\beta\geq0$ there exists $\alpha\geq0$ and $\varepsilon>0$
such that
\[
K:=\sup_{\left\vert \Delta\right\vert \leq\varepsilon}\left\Vert U\left(
t+\Delta\right)  \right\Vert _{\alpha\rightarrow\beta}<\infty.
\]
Then for any $\beta\geq0,$
\begin{equation}
\left\Vert \cdot\right\Vert _{\beta}\text{-}\frac{d}{dt}\left[  U\left(
t\right)  PT\left(  t\right)  \varphi\right]  =\dot{U}\left(  t\right)
PT\left(  t\right)  \varphi+U\left(  t\right)  P\dot{T}\left(  t\right)
\varphi. \label{equ.7.2}%
\end{equation}

\end{enumerate}
\end{lemma}

\begin{proof}
Let $\varphi\in\mathcal{S}$ and then define $\varphi\left(  t\right)
=U\left(  t\right)  PT\left(  t\right)  \varphi.$ To shorten notation let
$\Delta f$ denote $f\left(  t+\Delta\right)  -f\left(  t\right)  .$ We then
have,
\[
\frac{\Delta\varphi}{\Delta}=\left[  U\left(  t+\Delta\right)  P\frac{\Delta
T}{\Delta}+\frac{\Delta U}{\Delta}PT\left(  t\right)  \right]  \varphi
\]
and so
\begin{align}
\frac{\Delta\varphi}{\Delta}  &  -U\left(  t\right)  P\dot{T}\left(  t\right)
\varphi-\dot{U}\left(  t\right)  PT\left(  t\right)  \varphi\nonumber\\
&  =U\left(  t+\Delta\right)  P\left[  \frac{\Delta T}{\Delta}-\dot{T}\left(
t\right)  \right]  \varphi+\left[  \Delta U\right]  P\dot{T}\left(  t\right)
\varphi+\left[  \frac{\Delta U}{\Delta}-\dot{U}\left(  t\right)  \right]
PT\left(  t\right)  \varphi. \label{equ.7.3}%
\end{align}
Using the assumptions of the theorem it follows that for each $\beta<\infty,$
since $P\dot{T}\left(  t\right)  \varphi\in\mathcal{S},$ we may conclude that
\[
\left\Vert \left[  \Delta U\right]  P\dot{T}\left(  t\right)  \varphi
\right\Vert _{\beta}\rightarrow0\text{ as }\Delta\rightarrow0,\text{ and}%
\]
\[
\left\Vert \left[  \frac{\Delta U}{\Delta}-\dot{U}\left(  t\right)  \right]
PT\left(  t\right)  \varphi\right\Vert _{\beta}\rightarrow0\text{ as }%
\Delta\rightarrow0.
\]
Furthermore, using the assumptions along with Eq. (\ref{equ.3.41}) in the
Proposition \ref{pro.3.39}, it follows that when $\triangle\rightarrow0,$
\begin{multline*}
\left\Vert U\left(  t+\Delta\right)  P\left[  \frac{\Delta T}{\Delta}-\dot
{T}\left(  t\right)  \right]  \varphi\right\Vert _{\beta}\\
\leq\left\Vert U\left(  t+\Delta\right)  \right\Vert _{\alpha\rightarrow\beta
}\left\Vert P\right\Vert _{\alpha+\frac{k}{2}\rightarrow\alpha}\left\Vert
\left[  \frac{\Delta T}{\Delta}-\dot{T}\left(  t\right)  \right]
\varphi\right\Vert _{\alpha+\frac{k}{2}}\rightarrow0.
\end{multline*}
which combined with Eq. (\ref{equ.7.3}) shows $\varphi\left(  t\right)
=U\left(  t\right)  PT\left(  t\right)  \varphi$ is $\left\Vert \cdot
\right\Vert _{\beta}$ -- differentiable and the derivative is given as in Eq.
(\ref{equ.7.2}).
\end{proof}

\begin{lemma}
\label{lem.7.3}If $\alpha:\mathbb{R\rightarrow C}$ is \textbf{any} $C^{1}$ --
function and $V_{\hbar}\left(  t\right)  $ is defined as in Eq. (\ref{equ.7.1}%
), then for all $\psi\in\mathcal{S},$ $t\rightarrow V_{\hbar}\left(  t\right)
\psi$ and $t\rightarrow V_{\hbar}^{\ast}\left(  t\right)  \psi$ are
$\left\Vert \cdot\right\Vert _{\beta}$-norm differentiable for all
$\beta<\infty$ and moreover,
\begin{align}
\frac{d}{dt}V_{\hbar}\left(  t\right)  \psi &  =\Gamma_{\hbar}\left(
t\right)  V_{\hbar}\left(  t\right)  \psi\text{ and }\label{equ.7.4}\\
\frac{d}{dt}V_{\hbar}^{\ast}\left(  t\right)  \psi &  =-V_{\hbar}^{\ast
}\left(  t\right)  \Gamma_{\hbar}\left(  t\right)  \psi\label{equ.7.5}%
\end{align}
where
\begin{equation}
\Gamma_{\hbar}\left(  t\right)  :=\frac{1}{\hbar}\left(  \overline{\dot
{\alpha}\left(  t\right)  }a_{\hbar}-\dot{\alpha}\left(  t\right)  a_{\hbar
}^{\dag}+i\operatorname{Im}\left(  \alpha\left(  t\right)  \overline
{\dot{\alpha}\left(  t\right)  }\right)  -iH\left(  a_{\hbar}+\alpha\left(
t\right)  ,a_{\hbar}^{\dag}+\bar{\alpha}\left(  t\right)  \right)  \right)  .
\label{equ.7.6}%
\end{equation}

\end{lemma}

\begin{proof}
Let $U\left(  t\right)  :=U_{\hbar}\left(  -\alpha\left(  t\right)  \right)
=U\left(  -\alpha\left(  t\right)  /\sqrt{\hbar}\right)  ,$ $T\left(
t\right)  :=e^{-iH_{\hbar}t/\hbar}$ and $\varphi:=U_{\hbar}\left(  \alpha
_{0}\right)  \psi.$ From Propositions \ref{pro.2.6} and \ref{pro.2.10} we know
$U\left(  t\right)  \mathcal{S=\mathcal{S}}$ and
\begin{equation}
i\frac{d}{dt}U\left(  t\right)  f=Q\left(  t\right)  U\left(  t\right)
f\text{ for }f\in\mathcal{S}.\label{equ.7.7}%
\end{equation}
where
\begin{equation}
Q\left(  t\right)  =i\left(  -\frac{\dot{\alpha}\left(  t\right)  }%
{\sqrt{\hbar}}a^{\dag}+\frac{\overline{\dot{\alpha}\left(  t\right)  }}%
{\sqrt{\hbar}}a\right)  -\frac{1}{\hbar}\operatorname{Im}\left(  \alpha\left(
t\right)  \overline{\dot{\alpha}\left(  t\right)  }\right)  .\label{equ.7.8}%
\end{equation}
As $Q\left(  t\right)  $ is linear in $a$ and $a^{\dag},$ we may apply
Corollaries \ref{cor.5.15} and \ref{cor.5.17} in order to conclude that
$U\left(  t\right)  $ satisfies the hypothesis in Lemma \ref{lem.7.1}.
Moreover, by Proposition \ref{pro.6.5} and the item 2 in Theorem
\ref{the.6.11}, we also know that $T\left(  t\right)  \mathcal{S}=\mathcal{S}$
and it satisfies the hypothesis of Lemma \ref{lem.7.1}. Therefore by taking
$P\left(  \theta,\theta^{\ast}\right)  =1$ (so $P=I)$ in Lemma \ref{lem.7.1},
we learn
\begin{align*}
\frac{d}{dt}V_{\hbar}\left(  t\right)  \psi= &  \dot{U}\left(  t\right)
T\left(  t\right)  \varphi+U\left(  t\right)  \dot{T}\left(  t\right)
\varphi\\
= &  \left[  \left(  -\frac{\dot{\alpha}\left(  t\right)  }{\sqrt{\hbar}%
}a^{\dag}+\frac{\overline{\dot{\alpha}\left(  t\right)  }}{\sqrt{\hbar}%
}a\right)  +\frac{i}{\hbar}\operatorname{Im}\left(  \alpha\left(  t\right)
\overline{\dot{\alpha}\left(  t\right)  }\right)  \right]  U\left(  t\right)
T\left(  t\right)  \varphi\\
&  \qquad+U\left(  t\right)  \frac{H_{\hbar}}{i\hbar}T\left(  t\right)
\varphi\\
= &  \frac{1}{\hbar}\left[  \left(  -\dot{\alpha}\left(  t\right)  a_{\hbar
}^{\dag}+\overline{\dot{\alpha}\left(  t\right)  }a_{\hbar}^{\dag}\right)
+i\operatorname{Im}\left(  \alpha\left(  t\right)  \overline{\dot{\alpha
}\left(  t\right)  }\right)  \right]  V_{\hbar}\left(  t\right)  \psi\\
&  \qquad+U_{\hbar}\left(  -\alpha\left(  t\right)  \right)  \frac{H_{\hbar}%
}{i\hbar}U_{\hbar}\left(  \alpha\left(  t\right)  \right)  U_{\hbar}\left(
-\alpha\left(  t\right)  \right)  T\left(  t\right)  \varphi\\
= &  \Gamma_{\hbar}\left(  t\right)  V_{\hbar}\left(  t\right)  \psi,
\end{align*}
wherein the last equality we have used Proposition \ref{pro.2.6} to conclude,
\[
U_{\hbar}\left(  -\alpha\left(  t\right)  \right)  H\left(  a_{\hbar}%
,a_{\hbar}^{\dag}\right)  U_{\hbar}\left(  \alpha\left(  t\right)  \right)
=H\left(  a_{\hbar}+\alpha\left(  t\right)  ,a_{\hbar}^{\dag}+\bar{\alpha
}\left(  t\right)  \right)  .
\]
This completes the proof of Eq. (\ref{equ.7.4}). We now turn to the proof of
Eq. (\ref{equ.7.5}).

Now let $U\left(  t\right)  =U_{\hbar}^{\ast}\left(  \alpha_{0}\right)
e^{iH_{\hbar}t/\hbar}$ and $T\left(  t\right)  :=U_{\hbar}\left(
\alpha\left(  t\right)  \right)  $ and observe by taking adjoint of Eq.
(\ref{equ.7.1}) that
\[
V_{\hbar}^{\ast}\left(  t\right)  :=U_{\hbar}^{\ast}\left(  \alpha_{0}\right)
e^{iH_{\hbar}t/\hbar}U_{\hbar}\left(  \alpha\left(  t\right)  \right)
=U\left(  t\right)  T\left(  t\right)  .
\]
Working as above, we again easily show that both $U\left(  t\right)  $ and
$T\left(  t\right)  $ satisfy the hypothesis of Lemma \ref{lem.7.1} and
moreover by replacing $\alpha$ by $-\alpha$ in Eq. (\ref{equ.7.8}) we know
\[
i\frac{d}{dt}T\left(  t\right)  \psi=T\left(  t\right)  \left[  i\left(
\frac{\dot{\alpha}\left(  t\right)  }{\sqrt{\hbar}}a^{\dag}-\frac
{\overline{\dot{\alpha}\left(  t\right)  }}{\sqrt{\hbar}}a\right)  +\frac
{1}{\hbar}\operatorname{Im}\left(  \alpha\left(  t\right)  \overline
{\dot{\alpha}\left(  t\right)  }\right)  \right]  \psi.\text{ }%
\]
We now apply Lemma \ref{lem.7.1} with $P\left(  \theta,\theta^{\ast}\right)
=1$ and $\varphi=\psi$ along with some basic algebraic manipulations to show
Eq. (\ref{equ.7.5}) is also valid.
\end{proof}

Specializing our choice of $\alpha\left(  t\right)  $ in Lemma \ref{lem.7.3}
leads to the following important result.

\begin{theorem}
\label{the.7.5}Let $\Gamma_{\hbar}\left(  t\right)  $ be as in Eq.
(\ref{equ.7.6}). If $\alpha\left(  t\right)  $ satisfies Hamilton's equations
of motion (Eq. (\ref{equ.1.1}), $V_{\hbar}\left(  t\right)  $ is defined as in
Eq. (\ref{equ.7.1}), then
\begin{align}
\Gamma_{\hbar}\left(  t\right)  =  &  \frac{i}{\hbar}\operatorname{Im}\left(
\alpha\left(  t\right)  \overline{\dot{\alpha}\left(  t\right)  }\right)
-\frac{i}{\hbar}H^{\text{cl}}\left(  \alpha\left(  t\right)  \right)
\nonumber\\
&  \quad-iH_{2}\left(  \alpha\left(  t\right)  :a,a^{\dag}\right)  -\frac
{i}{\hbar}H_{\geq3}\left(  \alpha\left(  t\right)  :a_{\hbar},a_{\hbar}^{\dag
}\right)  , \label{equ.7.9}%
\end{align}
on $\mathcal{S}$ where $H^{\text{cl}},$ $H_{2}$ and $H_{\geq3}$ are as in Eq.
(\ref{equ.2.31}) by replacing $P$ by $H.$
\end{theorem}

\begin{proof}
From the expansion of $H\left(  \theta+\alpha,\theta^{\ast}+\bar{\alpha
}\right)  $ described in Eq. (\ref{equ.2.29}) and Theorem \ref{the.2.22} we
have
\begin{align}
H  &  \left(  a_{\hbar}+\alpha\left(  t\right)  ,a_{\hbar}^{\dag}+\bar{\alpha
}\left(  t\right)  \right) \nonumber\\
&  =H^{\text{cl}}\left(  \alpha\left(  t\right)  \right)  +\left(
\frac{\partial H^{\text{cl}}}{\partial\alpha}\right)  \left(  \alpha\left(
t\right)  \right)  a_{\hbar}+\left(  \frac{\partial H^{\text{cl}}}%
{\partial\overline{\alpha}}\right)  \left(  \alpha\left(  t\right)  \right)
a_{\hbar}^{\dag}\nonumber\\
&  +H_{2}\left(  \alpha\left(  t\right)  :a_{\hbar},a_{\hbar}^{\dag}\right)
+H_{\geq3}\left(  \alpha\left(  t\right)  :a_{\hbar},a_{\hbar}^{\dag}\right)
. \label{equ.7.10}%
\end{align}
So if $\alpha\left(  t\right)  $ satisfies Hamilton's equations of motion,
\begin{equation}
i\dot{\alpha}\left(  t\right)  =\left(  \frac{\partial}{\partial\bar{\alpha}%
}H^{\text{cl}}\right)  \left(  \alpha\left(  t\right)  \right)  \text{ with
}\alpha\left(  0\right)  =\alpha_{0}, \label{equ.7.11}%
\end{equation}
it follows using Eq. (\ref{equ.7.10}) in Eq. (\ref{equ.7.6}) that we may
cancel all the terms linear in $a_{\hbar}$ or $a_{\hbar}^{\dag}$ in which case
$\Gamma_{\hbar}\left(  t\right)  $ in Eq. (\ref{equ.7.6}) may be written as in
Eq. (\ref{equ.7.9}).
\end{proof}

In order to remove a (non-essential) highly oscillatory phase
factor\footnote{As usual in quantum mechanics, the overall phase factor will
not affect the expected values of observables and so we may safely ignore it
in this introductory description.} from $V_{\hbar}\left(  t\right)  $ let
\begin{equation}
f\left(  t\right)  :=\int_{0}^{t}\left(  H^{\text{cl}}\left(  \alpha\left(
\tau\right)  \right)  -\operatorname{Im}\left(  \alpha\left(  \tau\right)
\overline{\dot{\alpha}\left(  \tau\right)  }\right)  \right)  d\tau
\label{equ.7.12}%
\end{equation}
and then define
\begin{equation}
W_{\hbar}\left(  t\right)  =e^{\frac{i}{\hbar}f\left(  t\right)  }V_{\hbar
}\left(  t\right)  =e^{\frac{i}{\hbar}f\left(  t\right)  }U_{\hbar}\left(
-\alpha\left(  t\right)  \right)  e^{-iH_{\hbar}t/\hbar}U_{\hbar}\left(
\alpha_{0}\right)  . \label{equ.7.13}%
\end{equation}
More generally for $s,t\in\mathbb{R},$ let
\begin{equation}
W_{\hbar}\left(  t,s\right)  =W_{\hbar}\left(  t\right)  W_{\hbar}^{\ast
}\left(  s\right)  =e^{\frac{i}{\hbar}\left[  f\left(  t\right)  -f\left(
s\right)  \right]  }U_{\hbar}\left(  -\alpha\left(  t\right)  \right)
e^{-iH_{\hbar}\left(  t-s\right)  /\hbar}U_{\hbar}\left(  \alpha\left(
s\right)  \right)  . \label{equ.7.14}%
\end{equation}

\begin{proposition}
\label{pro.7.6} Let $H\left(  \theta,\theta^{\ast}\right)  \in\mathbb{R}%
\left\langle \theta,\theta^{\ast}\right\rangle $ and $\eta>0$ satisfy
Assumption \ref{ass.1}, $d=\deg_{\theta}H,$ and $W_{\hbar}\left(  t,s\right)
$ be as in Eq. (\ref{equ.7.14}). Then
\begin{equation}
W_{\hbar}\left(  t,s\right)  D\left(  \mathcal{N}^{\beta\frac{d}{2}}\right)
\subseteq D\left(  \mathcal{N}^{\beta}\right)  \text{~}\forall~s,t\in
\mathbb{R}\text{ and }\beta\geq0. \label{equ.7.15}%
\end{equation}
Moreover, we have $W_{\hbar}\left(  t,s\right)  \mathcal{S}=\mathcal{S}$ for
all $s,t\in\mathbb{R}.$
\end{proposition}

\begin{proof}
Eq. (\ref{equ.7.15}) is a direct consequence from $U_{\hbar}\left(
\alpha\left(  \cdot\right)  \right)  \mathcal{N}^{\beta}=\mathcal{N}^{\beta}$
in Corollary \ref{cor.5.15} and $e^{-iH_{\hbar}t/\hbar}D\left(  \mathcal{N}%
^{\beta\frac{d}{2}}\right)  \subseteq D\left(  \mathcal{N}^{\beta}\right)  $
from the item 1 in Theorem \ref{the.6.11}. Then, by Eq. (\ref{equ.3.34}), it
follows that $W_{\hbar}\left(  t,s\right)  \mathcal{S}\subseteq\mathcal{S}.$
By multiplying $W_{\hbar}\left(  t,s\right)  ^{-1}=W_{\hbar}\left(
s,t\right)  $ on both sides of the last inclusion, we can conclude that
$W_{\hbar}\left(  t,s\right)  \mathcal{S}=\mathcal{S}.$
\end{proof}

\begin{definition}
\label{def.7.7}For $\hbar>0$ and $t\in\mathbb{R},$ $L_{\hbar}\left(  t\right)
$ be the operator on $\mathcal{S}$ defined as,
\begin{align}
L_{\hbar}\left(  t\right)   &  =\frac{1}{\hbar}\left(  H\left(  a_{\hbar
}+\alpha\left(  t\right)  ,a_{\hbar}^{\dag}+\bar{\alpha}\left(  t\right)
\right)  -H^{\text{cl}}\left(  \alpha\left(  t\right)  \right)  -H_{1}\left(
\alpha\left(  t\right)  :a_{\hbar},a_{\hbar}^{\dag}\right)  \right)
\nonumber\\
&  =H_{2}\left(  \alpha\left(  t\right)  :a,a^{\dag}\right)  +\frac{1}{\hbar
}H_{\geq3}\left(  \alpha\left(  t\right)  :a_{\hbar},a_{\hbar}^{\dag}\right)
. \label{equ.7.16}%
\end{align}

\end{definition}

\begin{theorem}
\label{the.7.8}Both $t\rightarrow W_{\hbar}\left(  t,s\right)  $ and
$s\rightarrow W_{\hbar}\left(  t,s\right)  $ are strongly continuous on
$L^{2}\left(  m\right)  .$ Moreover, if $\psi\in\mathcal{S}$ and $\beta\geq0,$
then
\begin{align}
i\left(  \left\Vert \cdot\right\Vert _{\beta}\text{-}\partial_{t}\right)
W_{\hbar}\left(  t,s\right)  \psi &  =L_{\hbar}\left(  t\right)  W_{\hbar
}\left(  t,s\right)  \psi,\text{ and}\label{equ.7.17}\\
i\left(  \left\Vert \cdot\right\Vert _{\beta}\text{-}\partial_{s}\right)
W_{\hbar}\left(  t,s\right)  \psi &  =-W_{\hbar}\left(  t,s\right)  L_{\hbar
}\left(  s\right)  \psi. \label{equ.7.18}%
\end{align}

\end{theorem}

\begin{proof}
The strong continuity of $W_{\hbar}\left(  t,s\right)  $ in $s$ and in $t$
follows from the strong continuity of both $U\left(  \alpha\left(  t\right)
\right)  $ and $e^{-iH_{\hbar}t/\hbar},$ see Corollary \ref{cor.5.14} and
Proposition \ref{pro.6.1}. The derivative formulas in Eqs. (\ref{equ.7.17})
and (\ref{equ.7.18}) follow directly from Lemma \ref{lem.7.3} and Theorem
\ref{the.7.5} along with the an additional term coming from the product rule
involving the added scalar factor, $e^{\frac{i}{\hbar}\left[  f\left(
t\right)  -f\left(  s\right)  \right]  }.$
\end{proof}

For the rest of the paper the following notation will be in force.

\begin{notation}
\label{not.7.10}Let $\alpha_{0}\in\mathbb{C},$ $H\left(  \theta,\theta\right)
\in\mathbb{\mathbb{R}}\left\langle \theta,\theta^{\ast}\right\rangle $ satisfy
the Assumption \ref{ass.1}, $t\rightarrow\alpha\left(  t\right)  $ solve the
Hamiltonian' s equation Eq. (\ref{equ.1.1}) with $\alpha\left(  0\right)
=\alpha_{0},$ and $H_{2}\left(  \alpha\left(  \tau\right)  :\theta
,\theta^{\ast}\right)  $ be the degree $2$ homogeneous component of $H\left(
\theta+\alpha\left(  \tau\right)  ,\theta^{\ast}+\overline{\alpha}\left(
\tau\right)  \right)  $ as in Proposition \ref{pro.3.5}. Further let
\begin{equation}
W_{0}\left(  t,s\right)  :=W_{0}\left(  t\right)  W_{0}^{\ast}\left(
s\right)  \label{equ.7.19}%
\end{equation}
where $W_{0}\left(  t\right)  $ is the unique one parameter strongly
continuous family of unitary operators satisfying,
\begin{equation}
i\frac{\partial}{\partial t}W_{0}\left(  t\right)  =\overline{H_{2}\left(
\alpha\left(  t\right)  :a,a^{\dag}\right)  }W_{0}\left(  t\right)  \text{
with }W_{0}\left(  0\right)  =I\label{equ.7.20}%
\end{equation}
as described in Corollary \ref{cor.5.14}.
\end{notation}

\begin{remark}
\label{rem.7.11}Since
\[
\frac{i}{\hbar}H_{\geq3}\left(  \alpha\left(  t\right)  :a_{\hbar},a_{\hbar
}^{\dag}\right)  =i\sqrt{\hbar}\sum_{l\geq3}\hbar^{\left(  l-3\right)
/2}H_{l}\left(  \alpha\left(  t\right)  ,a,a^{\dag}\right)  ,
\]
it follows that $L_{\hbar}\left(  t\right)  $ in Eq. (\ref{equ.7.16})
satisfies,
\[
\lim_{\hbar\downarrow0}L_{\hbar}\left(  t\right)  \psi=H_{2}\left(
\alpha\left(  t\right)  :a,a^{\dag}\right)  \psi\text{ for all }\psi
\in\mathcal{S}.
\]
From this observation it is reasonable to expect\,$W_{\hbar}\left(  t\right)
\rightarrow W_{0}\left(  t\right)  $ where $W_{0}\left(  t\right)  $ is as in
Notation \ref{not.7.10}. This is in fact the key content of this paper, see
Theorem \ref{the.9.4} below. To complete the proof we will still need a fair
number of preliminary results.
\end{remark}

\subsection{Crude Bounds on $W_{\hbar}$\label{sec.7.1}}

\begin{theorem}
\label{the.7.12}Suppose that $H\left(  \theta,\theta^{\ast}\right)
\in\mathbb{R}\left\langle \theta,\theta^{\ast}\right\rangle $ and
$0<\hbar<\eta\leq1$ satisfy Assumption \ref{ass.1}, $d=\deg_{\theta}H,$ and
$W_{\hbar}\left(  t,s\right)  $ is as in Eq. (\ref{equ.7.14}). Then for all
$\beta\geq0,$ there exists $C_{\beta,H}<\infty$ depending only on $\beta\geq0$
and $H$ such that, for all $s,t\in\mathbb{R},$
\begin{align}
W_{\hbar}\left(  t,s\right)  D\left(  \mathcal{N}^{\beta d/2}\right)   &
\subset D\left(  \mathcal{N}^{\beta}\right)  \text{ and}\nonumber\\
\left\Vert \mathcal{N}^{\beta}W_{\hbar}\left(  t,s\right)  \psi\right\Vert  &
\leq\hbar^{-\beta}C_{\beta,H}\left\Vert \psi\right\Vert _{\frac{\beta d}{2}}.
\label{equ.7.21}%
\end{align}
[This bound is crude in the sense that $\hbar^{-\beta}C_{\beta,H}%
\uparrow\infty$ as $\hbar\downarrow0.$ We will do much better later in Theorem
\ref{the.9.1}.]
\end{theorem}

\begin{proof}
Let $\beta\geq0.$ From Proposition \ref{pro.7.6} it follows that $W_{\hbar
}\left(  t,s\right)  D\left(  \mathcal{N}^{\beta d/2}\right)  \subseteq
D\left(  \mathcal{N}^{\beta}\right)  .$ Moreover,
\begin{align*}
&  \left\Vert \mathcal{N}^{\beta}W_{\hbar}\left(  t,s\right)  \psi\right\Vert
\leq\left\Vert W_{\hbar}\left(  t,s\right)  \psi\right\Vert _{\beta}\\
&  \qquad=\left\Vert U_{\hbar}\left(  -\alpha\left(  t\right)  \right)
e^{-iH_{\hbar}\left(  t-s\right)  /\hbar}U_{\hbar}\left(  \alpha\left(
s\right)  \right)  \psi\right\Vert _{\beta}\\
&  \qquad\leq\left\Vert U_{\hbar}^{\ast}\left(  \alpha\left(  t\right)
\right)  \right\Vert _{\beta\rightarrow\beta}\left\Vert e^{-iH_{\hbar}\left(
t-s\right)  /\hbar}\right\Vert _{\beta d/2\rightarrow\beta}\left\Vert
U_{\hbar}\left(  \alpha\left(  s\right)  \right)  \right\Vert _{\beta
d/2\rightarrow\beta d/2}\left\Vert \psi\right\Vert _{\beta d/2}.
\end{align*}
Note that $\kappa:=\sup_{t\in\mathbb{R}}\left\vert \alpha\left(  t\right)
\right\vert <\infty$ from Proposition \ref{pro.3.13} , then by the Corollary
\ref{cor.5.15}, there exists a constant $C=C\left(  \beta,d,\kappa\right)  $
such that
\[
\sup_{t\in\mathbb{R}}\left\Vert U_{\hbar}^{\ast}\left(  \alpha\left(
t\right)  \right)  \right\Vert _{\beta\rightarrow\beta}\vee\sup_{s\in
\mathbb{R}}\left\Vert U_{\hbar}\left(  \alpha\left(  s\right)  \right)
\right\Vert _{\beta d/2\rightarrow\beta d/2}\leq C\left(  \beta,d,\kappa
\right)  .
\]
Then, combing all above inequalities along with Eq. (\ref{equ.6.8}) in Theorem
\ref{the.6.11}, we have
\[
\left\Vert \mathcal{N}^{\beta}W_{\hbar}\left(  t,s\right)  \psi\right\Vert
\leq C_{\beta,H}\hbar^{-\beta}\left\Vert \psi\right\Vert _{\beta d/2}%
\]
and therefore, Eq. (\ref{equ.7.21}) follows immediately.
\end{proof}

\section{Asymptotics of the Truncated Evolutions\label{sec.8}}

As in Section \ref{sec.7}, we assume that $H\left(  \theta,\theta^{\ast
}\right)  \in\mathbb{R}\left\langle \theta,\theta^{\ast}\right\rangle $ and
$\eta>0$ are as in Assumption \ref{ass.1}, $\alpha_{0}\in\mathbb{C},$ and
$\alpha\left(  t\right)  $ denotes the solution to Eq. (\ref{equ.1.1}) with
$\alpha\left(  0\right)  =\alpha_{0}.$ Further let $L_{\hbar}\left(  t\right)
$ be as in Eq. (\ref{equ.7.16}), i.e.
\begin{equation}
L_{\hbar}\left(  t\right)  =\sum_{k=2}^{d}\hbar^{\frac{k}{2}-1}H_{k}\left(
\alpha\left(  t\right)  :a,a^{\dag}\right)  . \label{equ.8.1}%
\end{equation}

\begin{definition}
[Truncated Evolutions]\label{def.8.1}For $0\leq M<\infty$ and $0<\hbar
<\infty,$ let $L_{\hbar}^{M}\left(  t\right)  =\mathcal{P}_{M}L_{\hbar}\left(
t\right)  \mathcal{P}_{M}$ be the level $M$ truncation of $L_{\hbar}\left(
t\right)  $ (see Notation \ref{not.3.45}) and let $W_{\hbar}^{M}\left(
t,s\right)  $ be the associated \textbf{truncated evolution }defined to be the
solution to the ordinary differential equation,
\begin{equation}
i\frac{d}{dt}W_{\hbar}^{M}\left(  t,s\right)  =L_{\hbar}^{M}\left(  t\right)
W_{\hbar}^{M}\left(  t,s\right)  \text{ with }W^{M}\left(  s,s\right)  =I
\label{equ.8.2}%
\end{equation}
as in Section \ref{sec.4.1}. We further let $W_{\hbar}^{M}\left(  t\right)
=W_{\hbar}^{M}\left(  t,0\right)  .$
\end{definition}

From the results of Theorem \ref{the.4.5} with $Q_{M}\left(  t\right)
=L_{\hbar}^{M}\left(  t\right)  $ and $U^{M}\left(  t,s\right)  =W_{\hbar}%
^{M}\left(  t,s\right)  ,$ we know that $W_{\hbar}^{M}\left(  t,s\right)  $ is
unitary on $L^{2}\left(  m\right)  $ and
\[
W_{\hbar}^{M}\left(  t,s\right)  =W_{\hbar}^{M}\left(  t,0\right)  W_{\hbar
}^{M}\left(  0,s\right)  =W_{\hbar}^{M}\left(  t\right)  W_{\hbar}^{M}\left(
s\right)  ^{\ast}%
\]
and in particular, $W_{\hbar}^{M}\left(  t\right)  ^{\ast}=W_{\hbar}%
^{M}\left(  0,t\right)  .$

\begin{proposition}
\label{pro.8.2}Suppose that $H\left(  \theta,\theta^{\ast}\right)
\in\mathbb{R}\left\langle \theta,\theta^{\ast}\right\rangle $ and $\eta>0$
satisfy Assumption \ref{ass.1}, $d=\deg_{\theta}H>0\in2\mathbb{N},$ and
further let $W_{\hbar}\left(  t,s\right)  ,$ $W_{0}\left(  t,s\right)  $ and
$W_{\hbar}^{M}\left(  t,s\right)  $ be as in Eq. (\ref{equ.7.14}), Notation
\ref{not.7.10}, and Definition \ref{def.8.1} respectively. If $\psi\in
D\left(  \mathcal{N}^{\frac{d}{2}}\right)  $ and $0<\hbar<\eta,$ then
\begin{equation}
W_{\hbar}\left(  t,s\right)  \psi-W_{\hbar}^{M}\left(  t,s\right)  \psi
=i\int_{s}^{t}W_{\hbar}\left(  t,\tau\right)  \left[  L_{\hbar}^{M}\left(
\tau\right)  -\overline{L_{\hbar}\left(  \tau\right)  }\right]  W_{\hbar}%
^{M}\left(  \tau,s\right)  \psi d\tau\label{equ.8.3}%
\end{equation}
and
\begin{equation}
W_{\hbar}\left(  t,s\right)  \psi-W_{0}\left(  t,s\right)  \psi=i\int_{s}%
^{t}W_{\hbar}\left(  t,\tau\right)  \left[  H_{2}\left(  \alpha\left(
\tau\right)  :\bar{a},a^{\ast}\right)  -\overline{L_{\hbar}\left(
\tau\right)  }\right]  W_{0}\left(  \tau,s\right)  \psi d\tau\label{equ.8.4}%
\end{equation}
where $L_{\hbar}\left(  t\right)  $ and $H_{2}\left(  \alpha\left(
\tau\right)  :\bar{a},a^{\ast}\right)  $ are as in Eqs. (\ref{equ.7.16}) and
(\ref{equ.7.20}) and $L_{\hbar}^{M}\left(  \tau\right)  =\mathcal{P}%
_{M}L_{\hbar}\left(  t\right)  \mathcal{P}_{M}$ as in Definition
\ref{def.8.1}. [The integrands in Eqs. (\ref{equ.8.3}) and (\ref{equ.8.4}) are
$L^{2}\left(  m\right)  $-norm continuous functions of $\tau$ and therefore
the integrals above are well defined.]
\end{proposition}

\begin{proof}
Let $B\left(  D\left(  \mathcal{N}^{\frac{d}{2}}\right)  ,L^{2}\left(
m\right)  \right)  $ denote the space of bounded linear operators from
$D\left(  \mathcal{N}^{\frac{d}{2}}\right)  $ to $L^{2}\left(  m\right)  .$
The integrals in Eq. (\ref{equ.8.3}) and (\ref{equ.8.4}) may be interpreted as
$L^{2}\left(  m\right)  $ -- valued Riemann integrals because their integrands
are $L^{2}\left(  m\right)  $ -- continuous functions of $\tau.$ This is
consequence of the observations that both
\begin{align*}
F\left(  \tau\right)   &  :=W_{\hbar}\left(  t,\tau\right)  \left[  L_{\hbar
}^{M}\left(  \tau\right)  -\overline{L_{\hbar}\left(  \tau\right)  }\right]
W_{\hbar}^{M}\left(  \tau,s\right)  \text{ and}\\
G\left(  \tau\right)   &  :=W_{\hbar}\left(  t,\tau\right)  \left[
H_{2}\left(  \alpha\left(  \tau\right)  :\bar{a},a^{\ast}\right)
-\overline{L_{\hbar}\left(  \tau\right)  }\right]  W_{0}\left(  \tau,s\right)
\end{align*}
are strongly continuous $B\left(  D\left(  \mathcal{N}^{\frac{d}{2}}\right)
,L^{2}\left(  m\right)  \right)  $ -- valued functions of $\tau.$ To verify
this assertion recall that;

\begin{enumerate}
\item $\tau\rightarrow W_{\hbar}^{M}\left(  \tau,s\right)  $ is $\left\Vert
\cdot\right\Vert _{d/2\rightarrow d/2}$ continuous by Item 3. of Theorem
\ref{the.4.5} and $\tau\rightarrow W_{0}\left(  \tau,s\right)  \psi$ is
$\left\Vert \cdot\right\Vert _{\frac{d}{2}}$ -- continuous by Corollary
\ref{cor.5.14}.

\item Both $L_{\hbar}^{M}\left(  \tau\right)  -\overline{L_{\hbar}\left(
\tau\right)  }$ and $H_{2}\left(  \alpha\left(  \tau\right)  :\bar{a},a^{\ast
}\right)  -\overline{L_{\hbar}\left(  \tau\right)  }$ are easily seen to be
strongly continuous as functions of $\tau$ with values in $B\left(  D\left(
\mathcal{N}^{\frac{d}{2}}\right)  ,L^{2}\left(  m\right)  \right)  $ by using
Corollary \ref{cor.3.40} and noting that the coefficients of the four
operators depend continuously on $\tau.$

\item The map, $\tau\rightarrow W_{\hbar}\left(  t,\tau\right)  $ is strongly
continuous on $L^{2}\left(  m\right)  $ by Theorem \ref{the.7.8}.
\end{enumerate}

As strong continuity is preserved under operator products, it follows that
both $F\left(  \tau\right)  $ and $G\left(  \tau\right)  $ are strongly continuous.

By Remark \ref{rem.4.7} and Proposition \ref{pro.7.6} we know that $W_{\hbar
}^{M}\left(  t,s\right)  \mathcal{S}=\mathcal{S}$ and $W_{\hbar}\left(
t,s\right)  \mathcal{S}=\mathcal{S}.$ Moreover, from item 3. of Theorem
\ref{the.4.5} and Theorem \ref{the.7.8}, if $\varphi\in\mathcal{S},$ then both
$t\rightarrow W_{\hbar}^{M}\left(  t,s\right)  \varphi$ and $t\rightarrow
W_{\hbar}\left(  t,s\right)  \varphi$ and are $\left\Vert \cdot\right\Vert
_{\beta}$-differentiable for $\beta\geq0.$ Since $W_{\hbar}\left(  t,s\right)
$ is unitary (see Eq. (\ref{equ.7.14})), it follows that $\sup_{t,s\in
\mathbb{R}}\left\Vert W_{\hbar}\left(  t,s\right)  \right\Vert _{0\rightarrow
0}=1.$ Therefore, by applying Lemma \ref{lem.7.1} with $U\left(  \tau\right)
=W_{\hbar}\left(  t,\tau\right)  ,$ $P\left(  \theta,\theta^{\ast}\right)
=1,$ and $T\left(  \tau\right)  =W_{\hbar}^{M}\left(  \tau,s\right)  $ while
making use of Eqs. (\ref{equ.7.18}) and (\ref{equ.8.2}) to find,
\[
i\frac{d}{d\tau}W_{\hbar}\left(  t,\tau\right)  W_{\hbar}^{M}\left(
\tau,s\right)  \varphi=F\left(  \tau\right)  \varphi.
\]
A similar arguments using Corollary \ref{cor.5.14} in place of Theorem
\ref{the.4.5} shows,
\[
i\frac{d}{d\tau}W_{\hbar}\left(  t,\tau\right)  W_{0}\left(  \tau,s\right)
\varphi=G\left(  \tau\right)  \varphi.
\]
Equations (\ref{equ.8.3}) and (\ref{equ.8.4}) now follow for $\psi=\varphi
\in\mathcal{S}$ by integrating the last two displayed equations and making use
of the fundamental theorem of calculus.

By the uniform boundedness principle (or by direct estimates already
provided), it follows that
\[
\sup_{\tau\in J_{s,t}}\left\Vert F\left(  \tau\right)  \right\Vert _{\frac
{d}{2}\rightarrow0}<\infty\text{ and }\sup_{\tau\in J_{s,t}}\left\Vert
G\left(  \tau\right)  \right\Vert _{\frac{d}{2}\rightarrow0}<\infty,
\]
where $J_{s,t}:=\left[  \min\left(  s,t\right)  ,\max\left(  s,t\right)
\right]  . $ Because of these observation and the fact that $\mathcal{S}$ is
dense in $D\left(  \mathcal{N}^{\frac{d}{2}}\right)  ,$ it follows that by a
standard \textquotedblleft$\varepsilon/3$ -- argument\textquotedblright that
Eqs. (\ref{equ.8.3}) and (\ref{equ.8.4}) are valid for all $\psi\in D\left(
\mathcal{N}^{\frac{d}{2}}\right)  .$
\end{proof}

\begin{theorem}
\label{the.8.3}Let $0<\eta\leq1$, $H\left(  \theta,\theta^{\ast}\right)
\in\mathbb{R}\left\langle \theta,\theta^{\ast}\right\rangle $ be a polynomial
of degree $d$ satisfying Assumption \ref{ass.1} and $d\geq2$ be an even
number. Then for all $\beta\geq d/2$ and $-\infty<S<T<\infty,$ there exists a
constant, $K\left(  \beta,\alpha_{0},H,S,T\right)  <\infty$ such that
\begin{equation}
\sup_{S<s,t<T}\left\Vert W_{\hbar}\left(  t,s\right)  -W_{\hbar}^{\hbar^{-1}%
}\left(  t,s\right)  \right\Vert _{\beta\rightarrow0}\leq K\left(
\beta,\alpha_{0},H,S,T\right)  \hbar^{\beta-1} \text{~}\forall~0<\hbar<\eta.
\label{equ.8.5}%
\end{equation}

\end{theorem}

\begin{proof}
Since $W_{\hbar}\left(  t,s \right)  $ and $W^{\hbar^{-1}}_{\hbar}\left(
t,s\right)  $ are unitary from Theorem \ref{the.4.5} and Eq. (\ref{equ.7.14})
and $\left\Vert \cdot\right\Vert _{\beta}\geq\left\Vert \cdot\right\Vert _{0}
$ in Remark \ref{rem.3.31}, it follows
\begin{equation}
\sup_{S<s,t<T}\left\Vert W_{\hbar}\left(  t,s\right)  -W_{\hbar}^{\hbar^{-1}%
}\left(  t,s\right)  \right\Vert _{\beta\rightarrow0}\leq1,
\end{equation}
and hence Eq. (\ref{equ.8.5}) holds if $\eta\wedge d^{-1}\leq\hbar< \eta$. The
remaining thing to show is Eq.(\ref{equ.8.5}) still holds for $0<\hbar
<\eta\wedge d^{-1}$.

Let $\psi\in D\left(  \mathcal{N}^{\beta}\right)  \subset D\left(
\mathcal{N}^{d/2}\right)  .$ Taking the $L^{2}\left(  m\right)  $ -- norm of
Eq. (\ref{equ.8.3}) implies,
\begin{equation}
\left\Vert \left[  W_{\hbar}\left(  t,s\right)  -W_{\hbar}^{M}\left(
t,s\right)  \right]  \psi\right\Vert \leq\int_{J_{s,t}}\left\Vert W_{\hbar
}\left(  t,\tau\right)  \left[  L_{\hbar}^{M}\left(  \tau\right)
-\overline{L_{\hbar}\left(  \tau\right)  }\right]  W_{\hbar}^{M}\left(
\tau,s\right)  \psi\right\Vert d\tau, \label{equ.8.7}%
\end{equation}
where
\begin{align}
&  \left\Vert W_{\hbar}\left(  t,\tau\right)  \left[  L_{\hbar}^{M}\left(
\tau\right)  -\overline{L_{\hbar}\left(  \tau\right)  }\right]  W_{\hbar}%
^{M}\left(  \tau,s\right)  \psi\right\Vert \nonumber\\
&  \qquad=\left\Vert \left[  L_{\hbar}^{M}\left(  \tau\right)  -\overline
{L_{\hbar}\left(  \tau\right)  }\right]  W_{\hbar}^{M}\left(  \tau,s\right)
\psi\right\Vert \nonumber\\
&  \qquad\leq\left\Vert L_{\hbar}^{M}\left(  \tau\right)  -\overline{L_{\hbar
}\left(  \tau\right)  }\right\Vert _{\beta\rightarrow0}\left\Vert W_{\hbar
}^{M}\left(  \tau,s\right)  \right\Vert _{\beta\rightarrow\beta}\left\Vert
\psi\right\Vert _{\beta}. \label{equ.8.8}%
\end{align}
In order to simplify this estimate further, let
\[
P\left(  \hbar,t:\theta,\theta^{\ast}\right)  =\sum_{k=2}^{d}\hbar^{\frac
{k}{2}-1}H_{k}\left(  \alpha\left(  t\right)  :\theta,\theta^{\ast}\right)  ,
\]
in which case, $L_{\hbar}\left(  t\right)  =P\left(  \hbar,t:a,a^{\dag
}\right)  .$ It follows from Corollary \ref{cor.3.50} with $\beta=0$ and
$\alpha\rightarrow\beta$ that (for $M\geq d)$
\begin{align*}
\left\Vert L_{\hbar}^{M}\left(  \tau\right)  -\overline{L_{\hbar}\left(
\tau\right)  }\right\Vert _{\beta\rightarrow0}  &  \leq\sum_{k=2}^{d}%
\hbar^{\frac{k}{2}-1}\left\vert H_{k}\left(  \alpha\left(  t\right)
:\theta,\theta^{\ast}\right)  \right\vert \left(  M-k+2\right)  ^{k/2-\beta}\\
&  \leq K\left(  \alpha_{0},H\right)  \hbar^{-1}\sum_{k=2}^{d}\left(  \hbar
M-k\hbar+2\hbar\right)  ^{k/2}\left(  M-k+2\right)  ^{-\beta}\text{ }%
\end{align*}
and from Eq. (\ref{equ.4.15}) that
\begin{align*}
\left\Vert W_{\hbar}^{M}\left(  \tau,s\right)  \right\Vert _{\beta
\rightarrow\beta}  &  \leq e^{K\left(  \beta,d\right)  \left(  \hbar
M+1\right)  ^{\frac{d}{2}-1}\sum_{k=2}^{d}\int_{J_{s,\tau}}\left\vert
\hbar^{\frac{k}{2}-1}H_{k}\left(  \alpha\left(  \sigma\right)  :\theta
,\theta^{\ast}\right)  \right\vert d\sigma}\\
&  \leq e^{\tilde{K}\left(  \beta,d,H\right)  \left(  \hbar M+1\right)
^{\frac{d}{2}-1}\left\vert t-s\right\vert }.
\end{align*}
Thus reducing to the case where $M=\hbar^{-1}$ (i.e. $M\hbar=1)$ we see there
exists $\tilde{K}\left(  \beta,\alpha_{0},H,S,T\right)  <\infty$ such that
\[
\left\Vert L_{\hbar}^{\hbar^{-1}}\left(  \tau\right)  -\overline{L_{\hbar
}\left(  \tau\right)  }\right\Vert _{\beta\rightarrow0}\left\Vert W_{\hbar
}^{\hbar^{-1}}\left(  \tau,s\right)  \right\Vert _{\beta\rightarrow\beta}%
\leq\tilde{K}\left(  \beta,\alpha_{0},H,S,T\right)  \hbar^{\beta-1}%
\]
which combined with Eqs. (\ref{equ.8.7}) and (\ref{equ.8.8}) implies Eq.
(\ref{equ.8.5}) with $K\left(  \beta,\alpha_{0},H,S,T\right)  =\tilde
{K}\left(  \beta,\alpha_{0},H,S,T\right)  \left[  T-S\right]  .$
\end{proof}

\section{Proof of the main Theorems\label{sec.9}}

The next theorem combines the crude bound in Theorem \ref{the.7.12} with the
asymptotics of the truncated evolutions in Theorem \ref{the.8.3} in order to
give a much improved version of Theorem \ref{the.7.12}.

\begin{theorem}
[$N$ -- Sobolev Boundedness of $W_{\hbar}\left(  t\right)  $]\label{the.9.1}%
Suppose that $H\left(  \theta,\theta^{\ast}\right)  \in\mathbb{R}\left\langle
\theta,\theta^{\ast}\right\rangle $ and $\eta>0$ satisfy Assumption
\ref{ass.1}, $d=\deg_{\theta}H>0\in2\mathbb{N},$ and $W_{\hbar}\left(
t,s\right)  $ and $W_{\hbar}\left(  t\right)  $ be as in Eqs. (\ref{equ.7.14})
and (\ref{equ.7.13}) respectively. Then for each $\beta\geq0, $ $-\infty
<S<T<\infty,$ there exists $K_{\beta}\left(  S,T\right)  <\infty$ such that
for all $\psi\in D\left(  \mathcal{N}^{\left(  2\beta+1\right)  d}\right)  ,$
all $0<\hbar<\eta\leq1,$ and all $S\leq s,t\leq T$ we have
\begin{equation}
\left\Vert \mathcal{N}^{\beta}W_{\hbar}\left(  t,s\right)  \psi\right\Vert
\leq K_{\beta}\left(  S,T\right)  \left\Vert \psi\right\Vert _{\left(
2\beta+1\right)  d}, \label{equ.9.1}%
\end{equation}
and
\begin{equation}
\sup_{S\leq s,t\leq T}\left\Vert W_{\hbar}\left(  t,s\right)  \right\Vert
_{\left(  2\beta+1\right)  d\rightarrow\beta}\leq\tilde{K}_{\beta}\left(
S,T\right)  , \label{equ.9.2}%
\end{equation}
where
\begin{equation}
\tilde{K}_{\beta}\left(  S,T\right)  :=\left(  1+K_{\beta}\left(  S,T\right)
\right)  2^{\left(  \beta-1\right)  _{+}}. \label{equ.9.3}%
\end{equation}
In particular this estimate implies, for $0<\hbar<\eta\leq1,$
\begin{equation}
\sup_{S\leq t\leq T}\left[  \left\Vert W_{\hbar}\left(  t\right)  \right\Vert
_{\left(  2\beta+1\right)  d\rightarrow\beta}\vee\left\Vert W_{\hbar}^{\ast
}\left(  t\right)  \right\Vert _{\left(  2\beta+1\right)  d\rightarrow\beta
}\right]  \leq\tilde{K}_{\beta}\left(  S,T\right)  . \label{equ.9.4}%
\end{equation}
[The bound in Eq. (\ref{equ.9.2}) improves on the crude bound in Eq.
(\ref{equ.8.5}) in that the bound now does not blow up as $\hbar\downarrow0.]$
\end{theorem}

\begin{remark}
\label{rem.9.3}The bound in Eq.(\ref{equ.9.1}) is not tight in that the index,
$\left(  2\beta+1\right)  d,$ of the norm on the right side of this equation
is not claimed to be optimal.
\end{remark}

\begin{proof}
The case $\beta=0$ is a trivial and so we now assume $\beta>0.$ If $\psi\in
D\left(  \mathcal{N}^{\left(  2\beta+1\right)  d}\right)  ,$ then by
Proposition \ref{pro.7.6} $W_{\hbar}\left(  t,s\right)  \psi\in D\left(
\mathcal{N}^{2\left(  2\beta+1\right)  }\right)  .$ Some simple algebra then
shows $\left\langle W_{\hbar}\left(  t,s\right)  \psi,\mathcal{N}^{2\beta
}W_{\hbar}\left(  t,s\right)  \psi\right\rangle =A+B,$ where
\begin{align*}
A  &  :=\left\langle W_{\hbar}^{\hbar^{-1}}\left(  t,s\right)  \psi
,\mathcal{N}^{2\beta}W_{\hbar}^{\hbar^{-1}}\left(  t,s\right)  \psi
\right\rangle \text{ and }\\
B  &  :=\left\langle \left[  W_{\hbar}\left(  t,s\right)  -W_{\hbar}%
^{\hbar^{-1}}\left(  t,s\right)  \right]  \psi,\mathcal{N}^{2\beta}W_{\hbar
}\left(  t,s\right)  \psi\right\rangle \\
&  +\left\langle \mathcal{N}^{2\beta}W_{\hbar}^{\hbar^{-1}}\left(  t,s\right)
\psi,\left[  W_{\hbar}\left(  t,s\right)  -W_{\hbar}^{\hbar^{-1}}\left(
t,s\right)  \right]  \psi\right\rangle .
\end{align*}
The $\left\vert B\right\vert $ term is bounded by the following two terms.
\begin{align*}
\left\vert B\right\vert  &  \leq\left\Vert \left[  W_{\hbar}\left(
t,s\right)  -W_{\hbar}^{\hbar^{-1}}\left(  t,s\right)  \right]  \psi
\right\Vert \cdot\left\Vert \mathcal{N}^{2\beta}W_{\hbar}\left(  t,s\right)
\psi\right\Vert \\
&  +\left\Vert \left[  W_{\hbar}\left(  t,s\right)  -W_{\hbar}^{\hbar^{-1}%
}\left(  t,s\right)  \right]  \psi\right\Vert \cdot\left\Vert \mathcal{N}%
^{2\beta}W_{\hbar}^{\hbar^{-1}}\left(  t,s\right)  \psi\right\Vert .
\end{align*}
Therefore, using Eq. (\ref{equ.4.15}) in Corollary \ref{cor.4.8}, Theorem
\ref{the.8.3} with $\beta$ replaced by $\frac{d}{2}+2\beta,$ and Theorem
\ref{the.7.12}, it follows that
\begin{align}
\left\vert B\right\vert  &  \leq\left\Vert \left[  W_{\hbar}\left(
t,s\right)  -W_{\hbar}^{\hbar^{-1}}\left(  t,s\right)  \right]  \psi
\right\Vert \cdot\left(  \left\Vert \mathcal{N}^{2\beta}W_{\hbar}\left(
t,s\right)  \psi\right\Vert +\left\Vert \mathcal{N}^{2\beta}W_{\hbar}%
^{\hbar^{-1}}\left(  t,s\right)  \psi\right\Vert \right) \nonumber\\
&  \leq C\hbar^{2\beta+\frac{d}{2}-1}\left\Vert \psi\right\Vert _{\frac{d}%
{2}+2\beta}\cdot\left(  \hbar^{-2\beta}\left\Vert \left(  \mathcal{N}%
+I\right)  ^{\beta d}\psi\right\Vert +\left\Vert \left(  \mathcal{N}+I\right)
^{2\beta}\psi\right\Vert \right) \nonumber\\
&  \leq C\hbar^{\frac{d}{2}-1}\left\Vert \psi\right\Vert _{\frac{d}{2}+2\beta
}\left(  \left\Vert \psi\right\Vert _{\beta d}+\hbar^{2\beta}\left\Vert
\psi\right\Vert _{2\beta}\right) \nonumber\\
&  \leq C\hbar^{\frac{d}{2}-1}\left\Vert \psi\right\Vert _{\left(
2\beta+1\right)  d}^{2}<\infty\text{ for all }S\leq s,t\leq T\text{ and }
0<\hbar<\eta. \label{equ.9.5}%
\end{align}
In the last inequality we have used, $\frac{d}{2}+2\beta\leq\left(
2\beta+1\right)  d$ when $\beta>0$ and $d\geq2.$ Corollary \ref{cor.4.8}
directly implies there exists $C>0$ such that
\[
\left\vert A\right\vert =\left\Vert \mathcal{N}^{\beta}W_{\hbar}^{\hbar^{-1}%
}\left(  t,s\right)  \psi\right\Vert _{\beta}^{2}\leq C\left\Vert
\psi\right\Vert _{\beta}^{2}\leq C\left\Vert \psi\right\Vert _{\left(
2\beta+1\right)  d}^{2}%
\]
for all $S\leq s,t\leq T$ and therefore, we get
\begin{equation}
\left\Vert \mathcal{N}^{\beta}W_{\hbar}\left(  t,s\right)  \psi\right\Vert
^{2}=\left\langle W_{\hbar}\left(  t,s\right)  \psi,\mathcal{N}^{2\beta
}W_{\hbar}\left(  t,s\right)  \psi\right\rangle \leq\left(  K_{\beta}\left(
S,T\right)  \right)  ^{2}\left\Vert \psi\right\Vert _{\left(  2\beta+1\right)
d}^{2} \label{equ.9.6}%
\end{equation}
for an appropriate constant $K_{\beta}\left(  S,T\right)  .$ Equation
(\ref{equ.9.1}) is proved and Eq. (\ref{equ.9.2}) is a consequence of Eq.
(\ref{equ.9.1}) and the inequality in Eq. (\ref{e.6.5}). Equation
(\ref{equ.9.2}) also implies Eq. (\ref{equ.9.4}) because $W_{\hbar
}(t)=W_{\hbar}\left(  t,0\right)  $and $W_{\hbar}^{\ast}\left(  t\right)
=W_{\hbar}\left(  0,t\right)  .$
\end{proof}

\begin{theorem}
\label{the.9.4}Suppose that $H\left(  \theta,\theta^{\ast}\right)
\in\mathbb{R}\left\langle \theta,\theta^{\ast}\right\rangle $ and $0<\eta
\leq1$ satisfy Assumptions \ref{ass.1}. Let $d=\deg_{\theta}H\in2\mathbb{N},$
$W_{\hbar}\left(  t,s\right)  ,$ and $W_{0}\left(  t,s\right)  $ be as in Eq.
(\ref{equ.7.14}) and Notation \ref{not.7.10} respectively. Then $W_{\hbar
}\left(  t,s\right)  \overset{s}{\rightarrow}W_{0}\left(  t,s\right)  $ as
$\hbar\downarrow0.$ Moreover for all $\beta\geq0$ and $-\infty<S<T<\infty$
there exists $K=K_{\beta}\left(  S,T\right)  <\infty$ such that, for
$0<\hbar<\eta\leq1$,
\begin{equation}
\sup_{S\leq s,t\leq T}\left\Vert \mathcal{N}^{\beta}\left(  W_{0}\left(
t,s\right)  -W_{\hbar}\left(  t,s\right)  \right)  \psi\right\Vert \leq
K\sqrt{\hbar}\left\Vert \psi\right\Vert _{\frac{d}{2}\left(  4\beta+3\right)
}~\forall~\psi\in D\left(  \mathcal{N}^{\frac{d}{2}\left(  4\beta+3\right)
}\right)  \label{equ.9.7}%
\end{equation}
and, with $\tilde{K}:=\left(  1+K\right)  2^{\left(  \beta-1\right)  _{+}},$%
\begin{equation}
\sup_{s,t\in\left[  S,T\right]  }\left\Vert W_{0}\left(  t,s\right)
-W_{\hbar}\left(  t,s\right)  \right\Vert _{\frac{d}{2}\left(  4\beta
+3\right)  \rightarrow\beta}\leq\tilde{K}\sqrt{\hbar}. \label{equ.9.8}%
\end{equation}
In particular, for $0<\hbar<\eta\leq1,$
\begin{equation}
\sup_{S\leq t\leq T}\left\Vert W_{0}\left(  t\right)  -W_{\hbar}\left(
t\right)  \right\Vert _{\frac{d}{2}\left(  4\beta+3\right)  \rightarrow\beta
}\vee\left\Vert W_{0}^{\ast}\left(  t\right)  -W_{\hbar}^{\ast}\left(
t\right)  \right\Vert _{\frac{d}{2}\left(  4\beta+3\right)  \rightarrow\beta
}\leq\tilde{K}\sqrt{\hbar}. \label{equ.9.9}%
\end{equation}

\end{theorem}

\begin{proof}
The claimed strong convergence now follows from Eq. (\ref{equ.9.7}) with
$\beta=0$ along with a standard density argument. To simplify notation, let
\[
p=d\left(  2\beta+1\right)  \text{ and }q=\frac{d}{2}\left(  4\beta+3\right)
=p+\frac{d}{2}.
\]
If $\psi\in D\left(  \mathcal{N}^{q}\right)  \subseteq D\left(  \mathcal{N}%
^{\frac{d}{2}}\right)  ,$ then by Eq. (\ref{equ.8.4}) in Proposition
\ref{pro.8.2}, Eq. (\ref{equ.7.16}), and Corollary \ref{cor.3.40},
\begin{align*}
W_{\hbar}\left(  t,s\right)  \psi-W_{0}\left(  t,s\right)  \psi &  =i\int
_{s}^{t}W_{\hbar}\left(  t,\tau\right)  \left[  H_{2}\left(  \alpha\left(
\tau\right)  :\bar{a},a^{\ast}\right)  -\bar{L}_{\hbar}\left(  \tau\right)
\right]  W_{0}\left(  \tau,s\right)  \psi d\tau\\
&  =i\int_{s}^{t}W_{\hbar}\left(  t,\tau\right)  \left[  \frac{1}{\hbar
}H_{\geq3}\left(  \alpha\left(  \tau\right)  :\bar{a}_{\hbar},a_{\hbar}^{\ast
}\right)  \right]  W_{0}\left(  \tau,s\right)  \psi d\tau.
\end{align*}
Then, by using theorem \ref{the.9.1}, we find for all $0<\hbar<\eta\leq1 $ and
$S\leq s,t\leq T$ (with $d=\deg_{\theta} H)$ that
\begin{align}
&  \left\Vert \left(  W_{\hbar}\left(  t,s\right)  -W_{0}\left(  t,s\right)
\right)  \psi\right\Vert _{\beta}\nonumber\\
&  \leq\int_{J_{s,t}}\left\Vert W_{\hbar}\left(  t,\tau\right)  \left[
\frac{1}{\hbar}H_{\geq3}\left(  \alpha\left(  \tau\right)  :\bar{a}_{\hbar
},a_{\hbar}^{\ast}\right)  \right]  W_{0}\left(  \tau,s\right)  \psi
\right\Vert _{\beta}d\tau\nonumber\\
&  \leq\int_{S}^{T}\left\Vert W_{\hbar}\left(  t,\tau\right)  \right\Vert
_{p\rightarrow\beta}\left\Vert \left[  \frac{1}{\hbar}H_{\geq3}\left(
\alpha\left(  \tau\right)  :\bar{a}_{\hbar},a_{\hbar}^{\ast}\right)  \right]
W_{0}\left(  \tau,s\right)  \psi\right\Vert _{p}d\tau\nonumber\\
&  \leq K\int_{S}^{T}\left\Vert \frac{1}{\hbar}H_{\geq3}\left(  \alpha\left(
\tau\right)  :\bar{a}_{\hbar},a_{\hbar}^{\ast}\right)  \right\Vert
_{q\rightarrow p}\left\Vert W_{0}\left(  t,\tau\right)  \right\Vert
_{q\rightarrow q}\left\Vert \psi\right\Vert _{q}d\tau\nonumber\\
&  \leq K\sqrt{\hbar}\int_{S}^{T}\left\Vert H_{\geq3}\left(  \alpha\left(
\tau\right)  , \sqrt{\hbar} :\bar{a},a^{\ast}\right)  \right\Vert
_{q\rightarrow p}\left\Vert W_{0}\left(  t,\tau\right)  \right\Vert
_{q\rightarrow q}d\tau\left\Vert \psi\right\Vert _{q}. \label{e.9.10}%
\end{align}
where $H_{\geq3}\left(  \alpha\left(  \tau\right)  ,\sqrt{\hbar}:\theta
,\theta^{\ast}\right)  \in\mathbb{R}\left[  \alpha\left(  \tau\right)
,\sqrt{\hbar}\right]  \left\langle \theta,\theta^{\ast}\right\rangle $ is a
polynomial in $\left(  \alpha\left(  \tau\right)  ,\sqrt{\hbar}:\theta
,\theta^{\ast}\right)  $ which is a sum of terms homogeneous of degree three
or more in the $\left\{  \theta,\theta^{\ast}\right\}  $ -- grading. By Eq.
(\ref{equ.3.45}) in Corollary \ref{cor.3.40} and Eq. (\ref{equ.5.29}) in
Corollary \ref{cor.5.14},
\[
\sup_{S\leq t\leq T}\int_{S}^{T}\left\Vert H_{\geq3}\left(  \alpha\left(
\tau\right)  :\bar{a}_{\hbar},a_{\hbar}^{\ast}\right)  \right\Vert
_{q\rightarrow p}\left\Vert W_{0}\left(  t,\tau\right)  \right\Vert
_{q\rightarrow q}d\tau<\infty
\]
which along with Eq. (\ref{e.9.10}) completes the proof of Eq. (\ref{equ.9.7}%
). Equation (\ref{equ.9.8}) follows directly from Eq. (\ref{equ.9.7}) after
making use of Eq. (\ref{e.6.5}). Equation (\ref{equ.9.9}) is a special case of
Eq. (\ref{equ.9.8}) because of the identities; $W_{\hbar}\left(  t\right)
=W_{\hbar}\left(  t,0\right)  ,$ $W_{\hbar}^{\ast}\left(  t\right)  =W_{\hbar
}\left(  0,t\right)  ,$ $W_{0}\left(  t\right)  =W_{0}\left(  t,0\right)  $
and $W_{0}\left(  t\right)  ^{\ast}=W_{0}\left(  0,t\right)  .$
\end{proof}

\subsection{Proof of Theorem \ref{the.1.20}\label{sub.9.1}}

We now finish this paper by showing that Eqs. (\ref{equ.9.4}) and
(\ref{equ.9.9}) can be used to prove the main theorems of this paper, namely
Theorem \ref{the.1.20} and Corollaries \ref{cor.1.22} and \ref{cor.1.24}. For
the rest of Section \ref{sec.9}, we always assume that $H\in\mathbb{R}%
\left\langle \theta,\theta^{\ast}\right\rangle $ and $1\geq\eta>0$ satisfy
Assumption \ref{ass.1}, \thinspace$d=\deg_{\theta}H>0\in2\mathbb{N},$
$W_{\hbar}\left(  t\right)  $ is defined as in Eq. (\ref{equ.7.13}), and
$W_{0}\left(  t\right)  $ is as in Notation \ref{not.7.10}.

\begin{notation}
\label{not.9.5}For $\hbar\geq0,$ let%
\begin{equation}
a\left(  \hbar:t\right)  :=W_{\hbar}^{\ast}\left(  t\right)  aW_{\hbar}\left(
t\right)  \text{ and }a^{\dag}\left(  \hbar:t\right)  :=W_{\hbar}^{\ast
}\left(  t\right)  a^{\dag}W_{\hbar}\left(  t\right)  \label{equ.9.11}%
\end{equation}
as operator on $\mathcal{S}.$ It should be noted that under Assumption
\ref{ass.1} we have $a^{\dag}\left(  \hbar:t\right)  =a\left(  \hbar:t\right)
^{\dag}$ for $0\leq\hbar<\eta.$
\end{notation}

According to Theorem \ref{the.5.18}, if $a\left(  t\right)  $ and $a^{\dag
}\left(  t\right)  $ are as in Eqs. (\ref{equ.1.8}) and (\ref{equ.1.9})
respectively then satisfies,%
\begin{align}
a\left(  t\right)   &  =W_{0}^{\ast}\left(  t\right)  aW_{0}\left(  t\right)
=a\left(  0:t\right)  \text{ and}\label{equ.9.12}\\
a^{\dag}\left(  t\right)   &  =W_{0}^{\ast}\left(  t\right)  a^{\dag}%
W_{0}\left(  t\right)  =a^{\dag}\left(  0:t\right)  \label{equ.9.13}%
\end{align}
as operators on $\mathcal{S}.$ For this reason we will typically write
$a\left(  t\right)  $ and $a^{\dag}\left(  t\right)  $ for $a\left(
0:t\right)  $ and $a\left(  0:t\right)  $ respectively.

By Proposition \ref{pro.2.6} and Eq.(\ref{equ.7.13}), the operator $A_{\hbar
}\left(  t\right)  $ defined in Eq. (\ref{equ.1.22}) satisfies,
\begin{align}
U_{\hbar}^{\ast}\left(  \alpha_{0}\right)  A_{\hbar}\left(  t\right)
U_{\hbar}\left(  \alpha_{0}\right)   &  =U_{\hbar}^{\ast}\left(  \alpha
_{0}\right)  e^{itH_{\hbar}/\hbar}a_{\hbar}e^{-itH_{\hbar}/\hbar}U_{\hbar
}\left(  \alpha_{0}\right) \nonumber\\
&  =W_{\hbar}^{\ast}\left(  t\right)  \left(  a_{\hbar}+\alpha\left(
t\right)  \right)  W_{\hbar}\left(  t\right) \nonumber\\
&  =\alpha\left(  t\right)  +\sqrt{\hbar}W_{\hbar}^{\ast}\left(  t\right)
aW_{\hbar}\left(  t\right) \nonumber\\
&  =\alpha\left(  t\right)  +\sqrt{\hbar}~a\left(  \hbar:t\right)  \text{ on
}\mathcal{S}. \label{equ.9.14}%
\end{align}

\begin{notation}
\label{not.9.8}For $t\in\mathbb{R}$ and $0\leq\hbar<\eta,$ let
\begin{align*}
B_{\theta}\left(  \hbar:t\right)   &  :=\overline{a\left(  \hbar:t\right)
}=W_{\hbar}^{\ast}\left(  t\right)  \bar{a}W_{\hbar}\left(  t\right)  \text{
and}\\
B_{\theta^{\ast}}\left(  \hbar:t\right)   &  :=a\left(  \hbar:t\right)
^{\ast}=W_{\hbar}^{\ast}\left(  t\right)  a^{\ast}W_{\hbar}\left(  t\right)  .
\end{align*}
When $\hbar=0$ we will denote $B_{b}\left(  0:t\right)  $ more simply as
$B_{b}\left(  t\right)  $ for $b\in\left\{  \theta,\theta^{\ast}\right\}  .$
\end{notation}

\begin{lemma}
\label{lem.9.9}Let $\eta>0$ and $d>0\in2\mathbb{N}$ be as in Theorem
\ref{the.9.1}, $b\in\left\{  \theta,\theta^{\ast}\right\}  ,$ $t\in\left[
S,T\right]  ,$ and $B_{b}\left(  \hbar:t\right)  $ be as in Notation
\ref{not.9.8}. Then, for any $\beta\geq0,$ there exists a constant $C\left(
\beta,S,T\right)  >0$ such that
\begin{equation}
\sup_{t\in\left[  S,T\right]  }\max_{b\in\left\{  \theta,\theta^{\ast
}\right\}  }\left\Vert B_{b}\left(  \hbar:t\right)  \right\Vert _{g\left(
\beta\right)  \rightarrow\beta}\leq C\left(  \beta,S,T\right)  \text{ for
}0<\hbar<\eta\label{equ.9.15}%
\end{equation}
where $g\left(  \beta\right)  =4d^{2}\beta+2d\left(  d+1\right)  .$
\end{lemma}

\begin{proof}
For definiteness, suppose that $b=\theta^{\ast}$ as the case $b=\theta$ is
proved analogously. If $q=\left(  2\beta+1\right)  d$ and
\[
p=\left[  2\left(  q+\frac{1}{2}\right)  +1\right]  d=4d^{2}\beta+2d\left(
d+1\right)  ,
\]
then%
\[
\left\Vert B_{b}\left(  \hbar:t\right)  \right\Vert _{p\rightarrow\beta}%
\leq\left\Vert W_{\hbar}^{\ast}\left(  t\right)  \right\Vert _{q\rightarrow
\beta}\left\Vert a^{\ast}\right\Vert _{q+\frac{1}{2}\rightarrow q}\left\Vert
W_{\hbar}\left(  t\right)  \right\Vert _{p\rightarrow q+\frac{1}{2}}%
\]
which combined with the estimates in Eqs. (\ref{equ.3.41}) and (\ref{equ.9.4})
gives the estimate in Eq. (\ref{equ.9.15}).
\end{proof}

\begin{lemma}
\label{lem.9.11}Let $\beta\geq0,$ $b\in\left\{  \theta,\theta^{\ast}\right\}
,$ $-\infty<S<T<\infty,$ $\eta>0,$ and $d>0\in2\mathbb{N}$ be the same as
Lemma \ref{lem.9.9}. Then there exists a constant $C\left(  \beta,S,T\right)
>0$ such that
\begin{equation}
\sup_{t\in\left[  S,T\right]  }\left\Vert B_{b}\left(  \hbar:t\right)
-B_{b}\left(  t\right)  \right\Vert _{r\left(  \beta\right)  \rightarrow\beta
}\leq C\left(  \beta,S,T\right)  \sqrt{\hbar}\text{ for }0\leq\hbar
<\eta\label{equ.9.16}%
\end{equation}
where $r\left(  \beta\right)  =\left(  4d^{2}\right)  \beta+\left(
3d+2\right)  d. $
\end{lemma}

\begin{proof}
Let us suppose that $b=\theta$ as the proof for $b=\theta^{\ast}$ is very
similar. Given $p\geq\beta$ (to be chosen later) we have,%
\begin{align}
&  \left\Vert B_{b}\left(  \hbar:t\right)  -B_{b}\left(  t\right)  \right\Vert
_{p\rightarrow\beta}\nonumber\\
&  \qquad=\left\Vert W_{\hbar}^{\ast}\left(  t\right)  \bar{a}W_{\hbar}\left(
t\right)  -W_{0}^{\ast}\left(  t\right)  \bar{a}W_{0}\left(  t\right)
\right\Vert _{p\rightarrow\beta}\nonumber\\
&  \qquad\leq\left\Vert \left[  W_{\hbar}^{\ast}\left(  t\right)  -W_{0}%
^{\ast}\left(  t\right)  \right]  \bar{a}W_{\hbar}\left(  t\right)
\right\Vert _{p\rightarrow\beta}+\left\Vert W_{0}^{\ast}\left(  t\right)
\bar{a}\left[  W_{\hbar}\left(  t\right)  -W_{0}\left(  t\right)  \right]
\right\Vert _{p\rightarrow\beta}. \label{equ.9.17}%
\end{align}
Using Eqs. (\ref{equ.3.41}), (\ref{equ.9.4}), and (\ref{equ.9.9}), there
exists a constant $C_{1}:=C_{1}\left(  \beta,S,T\right)  $ such that the first
term will become
\begin{align*}
&  \left\Vert \left[  W_{\hbar}^{\ast}\left(  t\right)  -W_{0}^{\ast}\left(
t\right)  \right]  \bar{a}W_{\hbar}\left(  t\right)  \right\Vert
_{p_{1}\rightarrow\beta}\\
&  \qquad\leq\left\Vert \left[  W_{\hbar}^{\ast}\left(  t\right)  -W_{0}%
^{\ast}\left(  t\right)  \right]  \right\Vert _{q_{1}\rightarrow\beta
}\left\Vert \bar{a}\right\Vert _{q_{1}+\frac{1}{2}\rightarrow q_{1}}\left\Vert
W_{\hbar}\left(  t\right)  \right\Vert _{p_{1}\rightarrow q_{1}+\frac{1}{2}%
}\leq C_{1}\sqrt{\hbar}%
\end{align*}
where
\[
q_{1}=\frac{d}{2}\left(  4\beta+3\right)  \text{ and }p_{1}=\left(  2\left(
q_{1}+\frac{1}{2}\right)  +1\right)  d=\left(  4d^{2}\right)  \beta+\left(
3d+2\right)  d.
\]
Likewise, using Eqs. (\ref{equ.3.41}), (\ref{equ.5.29}) and (\ref{equ.9.9}),
there exists a constant $C_{2}:=C_{2}\left(  \beta,S,T\right)  $ such that the
second term will become
\begin{align*}
&  \left\Vert W_{0}^{\ast}\left(  t\right)  \bar{a}\left[  W_{\hbar}\left(
t\right)  -W_{0}\left(  t\right)  \right]  \right\Vert _{p_{2}\rightarrow
\beta}\\
&  \qquad\leq\left\Vert W_{0}^{\ast}\left(  t\right)  \right\Vert
_{q_{2}\rightarrow\beta}\left\Vert \bar{a}\right\Vert _{q_{2}+\frac{1}%
{2}\rightarrow q_{2}}\left\Vert W_{\hbar}\left(  t\right)  -W_{0}\left(
t\right)  \right\Vert _{p_{2}\rightarrow q_{2}+\frac{1}{2}}\leq C_{2}%
\sqrt{\hbar}%
\end{align*}
where
\[
q_{2}=\beta\text{ and }p_{2}=\frac{d}{2}\left(  4\left(  q_{2}+\frac{1}%
{2}\right)  +3\right)  =\left(  2d\right)  \beta+\frac{5d}{2}.
\]
Since $d\geq2$ and $\beta\geq0,$ it follows that $p_{2}\leq p_{1}$ and so
taking $p=p_{1}$ in Eq. (\ref{equ.9.17}) and making use of the previous
estimates proves Eq. (\ref{equ.9.16}).
\end{proof}

\begin{notation}
\label{not.9.12}For $n\in\mathbb{N},$ let $d=\deg_{\theta} H>0$ and
\begin{equation}
\sigma_{n}:=\left(  4d^{2}\right)  2d\left(  d+1\right)  \frac{\left(
4d^{2}\right)  ^{n}-1}{4d^{2}-1}+\left(  3d+2\right)  d. \label{equ.9.18}%
\end{equation}

\end{notation}

\begin{lemma}
\label{lem.9.13}Let $S$, $T$, $d$ and $\eta$ be the same as Lemma
\ref{lem.9.9} and $\sigma_{n}$ be as in Notation \ref{not.9.8} for
$n\in\mathbb{N}.$ Then there exists $C_{n}\left(  S,T\right)  <\infty$ such
that for any $\mathbf{b}=\left(  b_{1},\ldots,b_{n}\right)  \in\left\{
\theta,\theta^{\ast}\right\}  ^{n},$ $0\leq\hbar<\eta,$ and $\left(
t_{1},\ldots,t_{n}\right)  \in\left[  S,T\right]  $ we have
\begin{equation}
\left\Vert B_{1}\left(  \hbar\right)  \dots B_{n}\left(  \hbar\right)
-B_{1}\dots B_{n}\right\Vert _{\sigma_{n}\rightarrow0}\leq C_{n}\left(
S,T\right)  \sqrt{\hbar}, \label{equ.9.19}%
\end{equation}
where $B_{i}\left(  \hbar\right)  :=B_{b_{i}}\left(  \hbar:t_{i}\right)  $ and
$B_{i}:=B_{i}\left(  0\right)  =B_{b_{i}}\left(  t_{i}\right)  $ for $1\leq
i\leq n,$ see Notation \ref{not.9.8}.
\end{lemma}

\begin{proof}
By a telescoping series arguments,
\begin{align*}
B_{1}\left(  \hbar\right)   &  \dots B_{n}\left(  \hbar\right)  -B_{1}\dots
B_{n}\\
&  =\sum_{i=1}^{n}\left[  B_{1}\left(  \hbar\right)  \dots B_{i}\left(
\hbar\right)  B_{i+1}\dots B_{n}-B_{1}\left(  \hbar\right)  \dots
B_{i-1}\left(  \hbar\right)  B_{i}\dots B_{n}\right] \\
&  =\sum_{i=1}^{n}B_{1}\left(  \hbar\right)  \dots B_{i-1}\left(
\hbar\right)  \left[  B_{i}\left(  \hbar\right)  -B_{i}\right]  B_{i+1}\dots
B_{n}%
\end{align*}
and therefore
\begin{align}
&  \left\Vert B_{1}\left(  \hbar\right)  \dots B_{n}\left(  \hbar\right)
-B_{1}\dots B_{n}\right\Vert _{\sigma_{n}\rightarrow0}\nonumber\\
&  \qquad\leq\sum_{i=1}^{n}\left\Vert B_{1}\left(  \hbar\right)  \dots
B_{i-1}\left(  \hbar\right)  \left[  B_{i}\left(  \hbar\right)  -B_{i}\right]
B_{i+1}\dots B_{n}\right\Vert _{\sigma_{n}\rightarrow0}. \label{equ.9.21}%
\end{align}
To finish the proof it suffices to show for $1\leq i\leq n$ that
\[
\left\Vert B_{1}\left(  \hbar\right)  \dots B_{i-1}\left(  \hbar\right)
\left[  B_{i}\left(  \hbar\right)  -B_{i}\right]  B_{i+1}\dots B_{n}%
\right\Vert _{\sigma_{n}\rightarrow0}\leq C\sqrt{\hbar}.
\]

Now
\begin{multline*}
\left\Vert B_{1}\left(  \hbar\right)  \dots B_{i-1}\left(  \hbar\right)
\left[  B_{i}\left(  \hbar\right)  -B_{i}\right]  B_{i+1}\dots B_{n}%
\right\Vert _{\sigma_{n}\rightarrow0}\\
\leq\left\Vert B_{1}\left(  \hbar\right)  \dots B_{i-1}\left(  \hbar\right)
\right\Vert _{v\rightarrow0}\left\Vert B_{i}\left(  \hbar\right)
-B_{i}\right\Vert _{u\rightarrow v}\left\Vert B_{i+1}\dots B_{n}\right\Vert
_{\sigma_{n}\rightarrow u}%
\end{multline*}
where we will choose all $\sigma_{n},$ $u,$ and $v\geq0$ appropriately. First
off if $\beta\geq0$ and $\mathcal{A}=\overline{a}$ or $a^{\ast},$ then (see
Proposition \ref{pro.3.39}) $\mathcal{A}:D\left(  \mathcal{N}^{\beta+\frac
{1}{2}}\right)  \rightarrow D\left(  \mathcal{N}^{\beta}\right)  $ and (see
Corollary \ref{cor.5.14}) $W_{0}\left(  t\right)  :\mathcal{N}^{\beta
}\rightarrow\mathcal{N}^{\beta}$ are bounded operators and therefore,%
\begin{equation}
\left\Vert B_{i+1}\dots B_{n}\right\Vert _{\sigma_{n}\rightarrow u}%
<\infty\text{ if }\sigma_{n}=u+\frac{1}{2}\left(  n-i\right)  .
\label{equ.9.22}%
\end{equation}
Also, with $r\left(  v\right)  $ as in Lemma \ref{lem.9.11}, there exists $C$
such that, for $0<\hbar<\eta, $
\begin{equation}
\left\Vert B_{i}\left(  \hbar\right)  -B_{i}\right\Vert _{u\rightarrow v}\leq
C\sqrt{\hbar}\text{ if }u=r\left(  v\right)  . \label{equ.9.23}%
\end{equation}
Using Lemma \ref{lem.9.9}, there exists $C>0$ such that, for $0< \hbar<\eta,$
\[
\left\Vert B_{1}\left(  \hbar\right)  \dots B_{i-1}\left(  \hbar\right)
\right\Vert _{v\rightarrow0}\leq C
\]
provided that
\begin{equation}
v=g^{i-1}\left(  0\right)  =2d\left(  d+1\right)  \frac{\left(  4d^{2}\right)
^{i}-1}{4d^{2}-1}. \label{equ.9.24}%
\end{equation}
If we let $1\leq i\leq n$ and
\begin{align*}
\sigma_{n}\left(  i\right)   &  =r\left(  g^{i-1}\left(  0\right)  \right)
+\frac{1}{2}\left(  n-i\right) \\
&  =\left(  4d^{2}\right)  2d\left(  d+1\right)  \frac{\left(  4d^{2}\right)
^{i}-1}{4d^{2}-1}+\left(  3d+2\right)  d+\frac{1}{2}\left(  n-i\right)  ,
\end{align*}
then the by the above bounds it follows that%
\begin{equation}
\left\Vert B_{1}\left(  \hbar\right)  \dots B_{i-1}\left(  \hbar\right)
\left[  B_{i}\left(  \hbar\right)  -B_{i}\right]  B_{i+1}\dots B_{n}%
\right\Vert _{\sigma_{n}\left(  i\right)  \rightarrow0}<\infty.
\label{equ.9.25}%
\end{equation}
One shows $\sigma_{n}\left(  i\right)  $ is increasing in $i$ and therefore
$\max_{1\leq i\leq n}\sigma_{n}\left(  i\right)  =\sigma_{n}\left(  n\right)
=\sigma_{n}$ where $\sigma_{n}$ is as in Notation \ref{not.9.12}. Equation
(\ref{equ.9.19}) now follows from Eqs. (\ref{equ.9.21}) and (\ref{equ.9.25})
with $\sigma_{n}\left(  i\right)  $ increased to $\sigma_{n}.$
\end{proof}

We finish the proof of Theorem \ref{the.1.20} with Lemma \ref{lem.9.13}.

\begin{proof}
[Proof of Theorem \ref{the.1.20}]Note that we have already shown that
$A_{\hbar}\left(  t_{i}\right)  $ and $A_{\hbar}^{\dag}\left(  t_{i}\right)  $
preserve $\mathcal{S}$ from Eq. (\ref{equ.6.1}) and $U_{\hbar}\left(
\alpha_{0}\right)  \mathcal{S}=\mathcal{S}$ and $U_{\hbar}\left(  \alpha
_{0}\right)  ^{\ast}\mathcal{S}=\mathcal{S}$ from Proposition \ref{pro.2.6}.
To show Eq.(\ref{equ.1.23}), for $\psi\in\mathcal{S},$ we have
\begin{align}
&  \left\langle P\left(  \left\{  A_{\hbar}\left(  t_{i}\right)
-\alpha\left(  t_{i}\right)  ,A_{\hbar}^{\dag}\left(  t_{i}\right)
-\overline{\alpha}\left(  t_{i}\right)  \right\}  _{i=1}^{n}\right)
\right\rangle _{U_{\hbar}\left(  \alpha_{0}\right)  \psi}\nonumber\\
=  &  \left\langle P\left(  \left\{  U_{\hbar}^{\ast}\left(  \alpha
_{0}\right)  A_{\hbar}\left(  t_{i}\right)  U_{\hbar}\left(  \alpha
_{0}\right)  -\alpha\left(  t_{i}\right)  ,U_{\hbar}^{\ast}\left(  \alpha
_{0}\right)  A_{\hbar}^{\dag}\left(  t_{i}\right)  U_{\hbar}\left(  \alpha
_{0}\right)  -\overline{\alpha}\left(  t_{i}\right)  \right\}  _{i=1}%
^{n}\right)  \right\rangle _{\psi}\nonumber\\
=  &  \left\langle P\left(  \left\{  \sqrt{\hbar}a\left(  \hbar:t_{i}\right)
,\sqrt{\hbar}a^{\dag}\left(  \hbar:t_{i}\right)  \right\}  _{i=1}^{n}\right)
\right\rangle _{\psi} \label{equ.9.26}%
\end{align}
where $\left\langle \cdot\right\rangle _{\psi}$ is defined in Definition
\ref{def.1.8} and the last step is asserted by Eq. (\ref{equ.9.14}). Supposed
$p=\operatorname{deg}\left(  P\left(  \left\{  \theta,\theta^{\ast}\right\}
_{i=1}^{n}\right)  \right)  $ and $p_{\min}$ is then minimum degree of each
non-constant term in $P\left(  \left\{  \theta_{i},\theta_{i}^{\ast}\right\}
_{i=1}^{n}\right)  .$ As $p=0$ is a trivial case, we assume $p>0.$ Then, it
follows
\begin{equation}
P\left(  \left\{  \theta_{i},\theta_{i}^{\ast}\right\}  _{i=1}^{n}\right)
=P_{0}+\sum_{k=p_{min}}^{p}P_{k}\left(  \left\{  \theta_{i},\theta_{i}^{\ast
}\right\}  _{i=1}^{n}\right)  \label{equ.9.27}%
\end{equation}
where $P_{0}\in\mathbb{C}$ and%
\[
P_{k}\left(  \left\{  \theta_{i},\theta_{i}^{\ast}\right\}  _{i=1}^{n}\right)
=\sum_{b_{1},.\dots,b_{k}\in\left\{  \theta_{i},\theta_{i}^{\ast}\right\}
_{i=1}^{n}}c\left(  b_{1},\dots,b_{k}\right)  b_{1}\dots b_{k}%
\]
is a homogeneous polynomial of $\left\{  \theta_{i},\theta_{i}^{\ast}\right\}
_{i=1}^{n}$ with degree $k.$ Plugging Eq.(\ref{equ.9.27}) into
Eq.(\ref{equ.9.26}) gives,
\begin{align}
&  \left\langle P\left(  \left\{  \sqrt{\hbar}a\left(  \hbar:t_{i}\right)
,\sqrt{\hbar}a^{\dag}\left(  \hbar:t_{i}\right)  \right\}  _{i=1}^{n}\right)
\right\rangle _{\psi}\nonumber\\
=  &  P_{0}+\sum_{k=p_{\min}}^{p}\hbar^{\frac{k}{2}}\left\langle P_{k}\left(
\left\{  a\left(  \hbar:t_{i}\right)  ,a^{\dag}\left(  \hbar:t_{i}\right)
\right\}  _{i=1}^{n}\right)  \right\rangle _{\psi} \label{equ.9.28}%
\end{align}
wherein we have used the fact that $P_{k}$ is a homogeneous polynomial of
degree $k$ in $\left\{  \theta_{i},\theta_{i}^{\ast}\right\}  _{i=1}^{n}.$ By
Lemma \ref{lem.9.13}, for $0<\hbar<\eta,$ we have
\[
\left\Vert P_{k}\left(  \left\{  a\left(  \hbar:t_{i}\right)  ,a^{\dag}\left(
\hbar:t_{i}\right)  \right\}  _{i=1}^{n}\right)  \psi\right\Vert =\left\Vert
P_{k}\left(  \left\{  a\left(  t_{i}\right)  ,a^{\dag}\left(  t_{i}\right)
\right\}  _{i=1}^{n}\right)  \psi\right\Vert +O\left(  \sqrt{\hbar}\right)  .
\]
Therefore, for $k\geq1,$ we have
\begin{multline}
\hbar^{\frac{k}{2}}\left\langle P_{k}\left(  \left\{  a\left(  \hbar
:t_{i}\right)  ,a^{\dag}\left(  \hbar:t_{i}\right)  \right\}  _{i=1}%
^{n}\right)  \right\rangle _{\psi}\\
=\hbar^{\frac{k}{2}}\left\langle P_{k}\left(  \left\{  a\left(  t_{i}\right)
,a^{\dag}\left(  t_{i}\right)  \right\}  _{i=1}^{n}\right)  \right\rangle
_{\psi}+O\left(  \hbar^{\frac{k+1}{2}}\right)  . \label{equ.9.29}%
\end{multline}
Applying Eq.(\ref{equ.9.29}) to Eq.(\ref{equ.9.28}), we have
\begin{align*}
&  \left\langle P\left(  \left\{  \sqrt{\hbar}a\left(  \hbar:t_{i}\right)
,\sqrt{\hbar}a^{\dag}\left(  \hbar:t_{i}\right)  \right\}  _{i=1}^{n}\right)
\right\rangle _{\psi}\\
&  \qquad=P_{0}+\sum_{k=p_{\min}}^{p}\hbar^{\frac{k}{2}}\left\langle
P_{k}\left(  \left\{  a\left(  t_{i}\right)  ,a^{\dag}\left(  t_{i}\right)
\right\}  _{i=1}^{n}\right)  \right\rangle _{\psi}+O\left(  \hbar^{\frac
{k+1}{2}}\right) \\
&  \qquad=\left\langle P\left(  \left\{  \sqrt{\hbar}a\left(  t_{i}\right)
,\sqrt{\hbar}a^{\dag}\left(  t_{i}\right)  \right\}  _{i=1}^{n}\right)
\right\rangle _{\psi}+O\left(  \hbar^{\frac{p_{\min}+1}{2}}\right)  .
\end{align*}
Therefore, Eq.(\ref{equ.1.23}) follows immediately.
\end{proof}

\subsection{Proof of Corollary \ref{cor.1.22}}

Let $P\left(  \left\{  \theta_{i},\theta_{i}^{\ast}\right\}  _{i=1}%
^{n}\right)  \in\mathbb{C}\left\langle \left\{  \theta_{i},\theta_{i}^{\ast
}\right\}  _{i=1}^{n}\right\rangle $ be a non-commutative polynomial, $\psi
\in\mathcal{S}$ and $\left\{  t_{1},\dots,t_{n}\right\}  \subseteq\mathbb{R}.$
With out loss of generality, we may assume $\operatorname{deg}\left(
P\right)  \geq1.$ We define, (see Notation \ref{not.2.20}),
\begin{align*}
\widetilde{P}\left(  \left\{  \alpha\left(  t_{i}\right)  :\theta_{i}%
,\theta_{i}^{\ast}\right\}  _{i=1}^{n}\right)   &  =P\left(  \left\{
\theta_{i}+\alpha\left(  t_{i}\right)  ,\theta_{i}^{\ast}+\overline{\alpha
}\left(  t_{i}\right)  \right\}  _{i=1}^{n}\right)  \\
&  \in\mathbb{C}\left[  \left\{  \alpha\left(  t_{i}\right)  ,\overline
{\alpha\left(  t_{i}\right)  }\right\}  _{i=1}^{n}\right]  \left\langle
\left\{  \theta_{i},\theta_{i}^{\ast}\right\}  _{i=1}^{n}\right\rangle .
\end{align*}
Note that $\operatorname{deg}_{\theta}\left(  \widetilde{P}\right)
=\operatorname{deg}\left(  P\right)  $ (see Notation \ref{not.2.20}) and
$\widetilde{p}_{\min}\geq1$ because $\operatorname{deg}\left(  \widetilde
{P}\right)  \geq1.$ By Theorem \ref{the.1.20} , for $0<\hbar<\eta,$ we have
\begin{align*}
&  \left\langle P\left(  \left\{  A_{\hbar}\left(  t_{i}\right)  ,A_{\hbar
}^{\dag}\left(  t_{i}\right)  \right\}  _{i=1}^{n}\right)  \right\rangle
_{U_{\hbar}\left(  \alpha_{0}\right)  \psi}\\
&  =\left\langle \widetilde{P}\left(  \left\{  \alpha\left(  t_{i}\right)
:A_{\hbar}\left(  t_{i}\right)  -\alpha\left(  t_{i}\right)  ,A_{\hbar}^{\dag
}\left(  t_{i}\right)  -\overline{\alpha}\left(  t_{i}\right)  \right\}
_{i=1}^{n}\right)  \right\rangle _{U_{\hbar}\left(  \alpha_{0}\right)  \psi}\\
&  =\left\langle \widetilde{P}\left(  \left\{  \alpha\left(  t_{i}\right)
:\sqrt{\hbar}a\left(  t_{i}\right)  ,\sqrt{\hbar}a^{\dag}\left(  t_{i}\right)
\right\}  _{i=1}^{n}\right)  \right\rangle _{\psi}+O\left(  \hbar
^{\frac{\widetilde{p}_{\min}+1}{2}}\right)  \\
&  =\left\langle P\left(  \left\{  \alpha\left(  t_{i}\right)  +\sqrt{\hbar
}a\left(  t_{i}\right)  ,\overline{\alpha}\left(  t_{i}\right)  +\sqrt{\hbar
}a^{\dag}\left(  t_{i}\right)  \right\}  _{i=1}^{n}\right)  \right\rangle
_{\psi}+O\left(  \hbar^{\frac{\widetilde{p}_{\min}+1}{2}}\right)  \\
&  =\left\langle P\left(  \left\{  \alpha\left(  t_{i}\right)  +\sqrt{\hbar
}a\left(  t_{i}\right)  ,\overline{\alpha}\left(  t_{i}\right)  +\sqrt{\hbar
}a^{\dag}\left(  t_{i}\right)  \right\}  _{i=1}^{n}\right)  \right\rangle
_{\psi}+O\left(  \hbar\right)  .
\end{align*}
The last equality is because $\widetilde{p}_{\min}$ is at least $1$.
Therefore, Eq. (\ref{equ.1.25}) follows.

\subsection{Proof of Corollary \ref{cor.1.24}}

By Eqs. (\ref{equ.1.8}) and (\ref{equ.1.9}) in Definition \ref{def.1.3}, the
term $\left\langle P_{1}\left(  \left\{  \alpha\left(  t_{i}\right)  :a\left(
t_{i}\right)  ,a^{\dag}\left(  t_{i}\right)  \right\}  _{i=1}^{n}\right)
\right\rangle _{\psi}$ in Eq.(\ref{equ.1.26}) is bounded independent of
$\hbar$ for $\psi\in\mathcal{S}.$ Therefore, by setting $\hbar\rightarrow0$ in
Eq.(\ref{equ.1.26}), Eq.(\ref{equ.1.28}) follows. To show Eq.(\ref{equ.1.29}),
let $p_{\min}$ be the minimum degree of all non constant terms in $P\left(
\left\{  \theta_{i},\theta_{i}^{\ast}\right\}  _{i=1}^{n}\right)  .$ We assume
$p_{\min}\geq1$ as usual. Otherwise, it means $P$ is a constant polynomial
which is a trivial case in Eq. (\ref{equ.1.29}). With the same notations as in
Eq. (\ref{equ.9.27}), we have
\[
P\left(  \left\{  \theta_{i},\theta_{i}^{\ast}\right\}  _{i=1}^{n}\right)
=P_{0}+\sum_{k=p_{\min}}^{p}P_{k}\left(  \left\{  \theta_{i},\theta_{i}^{\ast
}\right\}  _{i=1}^{n}\right)  .
\]
Then, we apply Eq.(\ref{equ.1.23}) on each term $P_{k}$ where $k\geq1,$ and
get
\begin{align}
&  \left\langle P_{k}\left(  \left\{  A_{\hbar}\left(  t_{i}\right)
-\alpha\left(  t_{i}\right)  ,A_{\hbar}^{\dag}\left(  t_{i}\right)
-\overline{\alpha}\left(  t_{i}\right)  \right\}  _{i=1}^{n}\right)
\right\rangle _{U_{\hbar}\left(  \alpha_{0}\right)  \psi}\nonumber\\
=  &  \left\langle P_{k}\left(  \left\{  \sqrt{\hbar}a\left(  t_{i}\right)
,\sqrt{\hbar}a^{\dag}\left(  t_{i}\right)  \right\}  _{i=1}^{n}\right)
\right\rangle _{\psi}+O\left(  \hbar^{\frac{k+1}{2}}\right) \nonumber\\
=  &  \hbar^{\frac{k}{2}}\left(  \left\langle P_{k}\left(  \left\{  a\left(
t_{i}\right)  ,a^{\dag}\left(  t_{i}\right)  \right\}  _{i=1}^{n}\right)
\right\rangle _{\psi}+O\left(  \hbar^{\frac{1}{2}}\right)  \right)  .
\label{equ.9.30}%
\end{align}
By applying Eq.(\ref{equ.9.30}), we have
\begin{align*}
&  \left\langle P\left(  \left\{  \frac{A_{\hbar}\left(  t_{i}\right)
-\alpha\left(  t_{i}\right)  }{\sqrt{\hbar}},\frac{A_{\hbar}^{\dag}\left(
t_{i}\right)  -\overline{\alpha}\left(  t_{i}\right)  }{\sqrt{\hbar}}\right\}
_{i=1}^{n}\right)  \right\rangle _{U_{\hbar}\left(  \alpha_{0}\right)  \psi}\\
&  =P_{0}+\sum_{k=p_{\min}}^{p}\frac{1}{\hbar^{\frac{k}{2}}}\left\langle
P_{k}\left(  \left\{  A_{\hbar}\left(  t_{i}\right)  -\alpha\left(
t_{i}\right)  ,A_{\hbar}^{\dag}\left(  t_{i}\right)  -\overline{\alpha}\left(
t_{i}\right)  \right\}  _{i=1}^{n}\right)  \right\rangle _{U_{\hbar}\left(
\alpha_{0}\right)  \psi}\\
&  =P_{0}+\sum_{k=p_{\min}}^{p}\left\langle P_{k}\left(  \left\{  a\left(
t_{i}\right)  ,a^{\dag}\left(  t_{i}\right)  \right\}  _{i=1}^{n}\right)
\right\rangle _{\psi}+O\left(  \hbar^{\frac{1}{2}}\right) \\
&  =\left\langle P\left(  \left\{  a\left(  t_{i}\right)  ,a^{\dag}\left(
t_{i}\right)  \right\}  _{i=1}^{n}\right)  \right\rangle _{\psi}+O\left(
\hbar^{\frac{1}{2}}\right)  .
\end{align*}
Eq.(\ref{equ.1.29}) follows.

\section{Appendix: Main Theorems in terms of the standard CCRs\label{sec.10}}

Let
\[
\hat{a}_{\hbar}=\frac{1}{\sqrt{2}}\left(  M_{x}+\hbar\frac{d}{dx}\right)
\text{ and }\hat{a}_{\hbar}^{\dag}=\frac{1}{\sqrt{2}}\left(  M_{x}-\hbar
\frac{d}{dx}\right)
\]
(as an operator on $\mathcal{S})$ be the more standard representation for the
annihilation and creation operators form of the CCRs used in the physics
literature. We will reformulate Theorem \ref{the.1.20}, Corollaries
\ref{cor.1.22} and \ref{cor.1.24} in the standard CCRs. The following lemma
(whose proof is left to the reader) implements the equivalence of our
representation of the canonical commutation relations (CCRs) to the standard
representation of the CCRs.

\begin{lemma}
\label{lem.10.1}For $\rho>0,$ let $S_{\rho}:L^{2}\left(  \mathbb{R}\right)
\rightarrow L^{2}\left(  \mathbb{R}\right)  $ be the unitary map defined by
\[
\left(  S_{\rho}f\right)  \left(  x\right)  :=\sqrt{\rho}f\left(  \rho
x\right)  \text{ for }x\in\mathbb{R}.
\]
Then $S_{\rho}\mathcal{S}=\mathcal{S}$ and it follows that
\[
\hat{a}_{\hbar}=S_{\hbar^{-1/2}}a_{\hbar}S_{\hbar^{1/2}}\text{ and }\hat
{a}_{\hbar}^{\dag}=S_{\hbar^{-1/2}}a_{\hbar}^{\dag}S_{\hbar^{1/2}}%
\]

\end{lemma}

\begin{definition}
\label{def.10.2}For $\hbar>0$ and $\alpha:=\left(  \xi+i\pi\right)  /\sqrt{2},
$ let
\[
\hat{U}_{\hbar}\left(  \alpha\right)  =\exp\left(  \frac{1}{\hbar}\left(
\overline{\alpha\hat{a}_{\hbar}^{\dag}-\bar{\alpha}\hat{a}_{\hbar}}\right)
\right)
\]
be the unitary operator on $L^{2}\left(  \mathbb{R}\right)  $ which implements
translation by $\left(  \xi,\pi\right)  $ in phase space.
\end{definition}

Using the more standard representation of the CCRs instead, we have an
immediate corollary from Theorem \ref{the.1.20}.

\begin{theorem}
\label{the.10.4}Suppose $H\left(  \theta,\theta^{\ast}\right)  \in
\mathbb{R}\left<  \theta,\theta^{\ast}\right>  $ is a non-commutative
polynomial in two indeterminates, $d=deg H>0$ and $0<\eta\leq1$ satisfying the
same assumptions in Theorem \ref{the.1.20}. Let $\hat{H}_{\hbar}%
:=\overline{H\left(  \hat{a}_{\hbar},\hat{a}_{\hbar}^{\dag}\right)  }.$ We
define
\[
\hat{A}_{\hbar}\left(  t\right)  :=e^{i\hat{H}_{\hbar}t/\hbar}\hat{a}_{\hbar
}e^{-i\hat{H}_{\hbar}t/\hbar}%
\]
denote $\hat{a}_{\hbar}$ in the Heisenberg picture. Furthermore for all
$\psi\in\mathcal{S},$ $\alpha_{0}\in\mathbb{C},$ $0<\hbar<\eta,$ real numbers
$\left\{  t_{i}\right\}  _{i=1}^{n}\subset\mathbb{R},$ and non-commutative
polynomial, $P\left(  \left\{  \theta_{i},\theta_{i}^{*}\right\}  _{i=1}%
^{n}\right)  \in\mathbb{C}\left\langle \left\{  \theta_{i},\theta_{i}%
^{*}\right\}  _{i=1}^{n}\right\rangle ,$ in $2n$ -- indeterminants where
$p_{\min}$ be the minimum degree of all non constant terms in $P\left(
\left\{  \theta_{i},\theta_{i}^{\ast}\right\}  _{i=1}^{n}\right)  ,$ the
following weak limits (in the sense of non-commutative probability) hold;
\begin{align}
&  \left\langle P\left(  \left\{  \hat{A}_{\hbar}\left(  t_{i}\right)
-\alpha\left(  t_{i}\right)  ,\hat{A}_{\hbar}^{\dag}\left(  t_{i}\right)
-\overline{\alpha}\left(  t_{i}\right)  \right\}  _{i=1}^{n}\right)
\right\rangle _{\hat{U}_{\hbar}\left(  \alpha_{0}\right)  S_{\hbar^{-1/2}}%
\psi}\nonumber\\
&  \quad\quad=\left\langle P\left(  \left\{  \sqrt{\hbar}a\left(
t_{i}\right)  ,\sqrt{\hbar}a^{\dag}\left(  t_{i}\right)  \right\}  _{i=1}%
^{n}\right)  \right\rangle _{\psi}+O\left(  \hbar^{\frac{p_{\min}+1}{2}%
}\right)  . \label{equ.10.1}%
\end{align}
where $a\left(  t\right)  $ and $a^{\dag}\left(  t\right)  $ are as in Eqs.
(\ref{equ.1.8}) and (\ref{equ.1.9}).
\end{theorem}

\begin{proof}
By the Lemma \ref{lem.10.1}, we have
\[
A_{\hbar}\left(  t_{i}\right)  =S_{\hbar^{1/2}}\hat{A}_{\hbar}\left(
t_{i}\right)  S_{\hbar^{-1/2}}\text{ and }U_{\hbar}\left(  \alpha_{0}\right)
=S_{\hbar^{1/2}}\hat{U}_{\hbar}\left(  \alpha_{0}\right)  S_{\hbar^{-1/2}}%
\]
on $\mathcal{S}.$ Therefore,
\begin{align*}
&  \left\langle P\left(  \left\{  \hat{A}_{\hbar}\left(  t_{i}\right)
-\alpha\left(  t_{i}\right)  ,\hat{A}_{\hbar}^{\dag}\left(  t_{i}\right)
-\overline{\alpha}\left(  t_{i}\right)  \right\}  _{i=1}^{n}\right)
\right\rangle _{\hat{U}_{\hbar}\left(  \alpha_{0}\right)  S_{\hbar^{-1/2}}%
\psi}\\
=  &  \left\langle P\left(  \left\{  S_{\hbar^{\frac{1}{2}}}\hat{A}_{\hbar
}\left(  t_{i}\right)  S_{\hbar^{-\frac{1}{2}}}-\alpha\left(  t_{i}\right)
,S_{\hbar^{\frac{1}{2}}}\hat{A}_{\hbar}^{\dag}\left(  t_{i}\right)
S_{\hbar^{-\frac{1}{2}}}-\overline{\alpha}\left(  t_{i}\right)  \right\}
_{i=1}^{n}\right)  \right\rangle _{S_{\hbar^{\frac{1}{2}}}\hat{U}_{\hbar
}\left(  \alpha_{0}\right)  S_{\hbar^{-\frac{1}{2}}}\psi}\\
=  &  \left\langle P\left(  \left\{  A_{\hbar}\left(  t_{i}\right)
-\alpha\left(  t_{i}\right)  ,A_{\hbar}^{\dag}\left(  t_{i}\right)
-\overline{\alpha}\left(  t_{i}\right)  \right\}  _{i=1}^{n}\right)
\right\rangle _{U_{\hbar}\left(  \alpha_{0}\right)  \psi}.
\end{align*}
Then, Eq.(\ref{equ.10.1}) follows by applying Eq.(\ref{equ.1.23})..
\end{proof}

Likewise we can show two corollaries of Theorem \ref{the.10.4} below which
behave like Corollaries \ref{cor.1.22} and \ref{cor.1.24}.

\begin{corollary}
\label{cor.10.5}Under the same notations and assumptions in Theorem
\ref{the.10.4}, then ,for $0<\hbar<\eta,$ we have
\begin{align}
&  \left\langle P\left(  \left\{  \hat{A}_{\hbar}\left(  t_{i}\right)
,\hat{A}_{\hbar}^{\dag}\left(  t_{i}\right)  \right\}  _{i=1}^{n}\right)
\right\rangle _{\hat{U}_{\hbar}\left(  \alpha_{0}\right)  S_{\hbar^{-1/2}}%
\psi}\nonumber\\
&  \qquad=\left\langle P\left(  \left\{  \alpha\left(  t_{i}\right)
+\sqrt{\hbar}a\left(  t_{i}\right)  ,\overline{\alpha}\left(  t_{i}\right)
+\sqrt{\hbar}a^{\dag}\left(  t_{i}\right)  \right\}  \right)  \right\rangle
_{\psi}+O\left(  \hbar\right)  . \label{equ.10.2}%
\end{align}

\end{corollary}

\begin{proof}
It is a similar proof as Theorem \ref{the.10.4}. Using Lemma \ref{lem.10.1},
we can conclude
\[
\left\langle P\left(  \left\{  \hat{A}_{\hbar}\left(  t_{i}\right)  ,\hat
{A}_{\hbar}^{\dag}\left(  t_{i}\right)  \right\}  _{i=1}^{n}\right)
\right\rangle _{\hat{U}_{\hbar}\left(  \alpha_{0}\right)  S_{\hbar^{-1/2}}%
\psi}=\left\langle P\left(  \left\{  A_{\hbar}\left(  t_{i}\right)  ,A_{\hbar
}^{\dag}\left(  t_{i}\right)  \right\}  _{i=1}^{n}\right)  \right\rangle
_{U_{\hbar}\left(  \alpha_{0}\right)  \psi}.
\]
Then, the rest of the proof is simply to apply Eq.(\ref{equ.1.25}) and hence,
Eq.(\ref{equ.10.2}) follows.
\end{proof}

\begin{corollary}
\label{cor.10.6}Under the same notations and assumptions in Theorem
\ref{the.10.4}, let $\hat{\psi}_{\hbar}=\hat{U}_{\hbar}\left(  \alpha
_{0}\right)  S_{\hbar^{-1/2}}\psi.$ As $\hbar\rightarrow0^{+},$ we have
\[
\left\langle P\left(  \left\{  \hat{A}_{\hbar}\left(  t_{i}\right)  ,\hat
{A}_{\hbar}^{\dag}\left(  t_{i}\right)  \right\}  _{i=1}^{n}\right)
\right\rangle _{\hat{\psi}_{\hbar}}\rightarrow P\left(  \left\{  \alpha\left(
t_{i}\right)  ,\overline{\alpha}\left(  t_{i}\right)  \right\}  _{i=1}%
^{n}\right)  .
\]
and
\begin{equation}
\left\langle P\left(  \left\{  \frac{\hat{A}_{\hbar}\left(  t_{i}\right)
-\alpha\left(  t_{i}\right)  }{\sqrt{\hbar}},\frac{\hat{A}_{\hbar}^{\dag
}\left(  t_{i}\right)  -\bar{\alpha}\left(  t_{i}\right)  }{\sqrt{\hbar}%
}\right\}  _{i=1}^{n}\right)  \right\rangle _{\hat{\psi}_{\hbar}}%
\rightarrow\left\langle P\left(  \left\{  a\left(  t_{i}\right)  ,a^{\dag
}\left(  t_{i}\right)  \right\}  _{i=1}^{n}\right)  \right\rangle _{\psi}.
\label{equ.10.3}%
\end{equation}
We abbreviate this convergence by saying
\[
\operatorname*{Law}\nolimits_{\hat{\psi}_{\hbar}}\left(  \left\{  \frac
{\hat{A}_{\hbar}\left(  t_{i}\right)  -\alpha\left(  t_{i}\right)  }%
{\sqrt{\hbar}},\frac{\hat{A}_{\hbar}^{\dag}\left(  t_{i}\right)  -\bar{\alpha
}\left(  t_{i}\right)  }{\sqrt{\hbar}}\right\}  _{i=1}^{n}\right)
\rightarrow\operatorname*{Law}\nolimits_{\psi}\left(  \left\{  a\left(
t_{i}\right)  ,a^{\dag}\left(  t_{i}\right)  \right\}  _{i=1}^{n}\right)  .
\]

\end{corollary}

\begin{proof}
Similar to the proof in Theorem \ref{the.10.4}, by using Lemma \ref{lem.10.1},
we have
\[
\left\langle P\left(  \left\{  \hat{A}_{\hbar}\left(  t_{i}\right)  ,\hat
{A}_{\hbar}^{\dag}\left(  t_{i}\right)  \right\}  _{i=1}^{n}\right)
\right\rangle _{\hat{\psi}_{\hbar}}=\left\langle P\left(  \left\{  A_{\hbar
}\left(  t_{i}\right)  ,A_{\hbar}^{\dag}\left(  t_{i}\right)  \right\}
_{i=1}^{n}\right)  \right\rangle _{U_{\hbar}\left(  \alpha_{0}\right)  \psi,}%
\]
and
\begin{align*}
&  \left\langle P\left(  \left\{  \frac{\hat{A}_{\hbar}\left(  t_{i}\right)
-\alpha\left(  t_{i}\right)  }{\sqrt{\hbar}},\frac{\hat{A}_{\hbar}^{\dag
}\left(  t_{i}\right)  -\bar{\alpha}\left(  t_{i}\right)  }{\sqrt{\hbar}%
}\right\}  _{i=1}^{n}\right)  \right\rangle _{\hat{\psi}_{\hbar}}\\
=  &  \left\langle P\left(  \left\{  \frac{A_{\hbar}\left(  t_{i}\right)
-\alpha\left(  t_{i}\right)  }{\sqrt{\hbar}},\frac{A_{\hbar}^{\dag}\left(
t_{i}\right)  -\bar{\alpha}\left(  t_{i}\right)  }{\sqrt{\hbar}}\right\}
_{i=1}^{n}\right)  \right\rangle _{U_{\hbar}\left(  \alpha_{0}\right)  \psi}.
\end{align*}
Therefore, the corollary is a direct consequence of Corollary \ref{cor.1.24}.
\end{proof}

 \bibliographystyle{plain}
\bibliography{semi-classical}
\end{document}